\documentclass{amsart}
\usepackage{epsfig}
\usepackage{amsmath,amssymb}
\usepackage{slashed}
\usepackage{subcaption}
\usepackage{graphicx}
\numberwithin{equation}{section}

\usepackage{todonotes}

\usepackage{enumerate}

\newtheorem{theorem}{Theorem}[section]
\newtheorem{lemma}[theorem]{Lemma}
\newtheorem{corollary}[theorem]{Corollary}
\newtheorem{proposition}[theorem]{Proposition}

\allowdisplaybreaks \numberwithin{equation}{section}

\def\Tr{\mathop{\rm Tr}\nolimits}
\def\res{\mathop{\rm Res}}
\newcommand{\PP}{{\mathbb P}}

\def\div{\mathop{\rm Div}\nolimits}
\def\Jac{\mathop{\rm Jac}\nolimits}
\def\diag{\mathop{\rm Diag}\nolimits}

\def\adj{\mathop{\rm Adj}\nolimits}
\def\Adj{\mathop{\rm  Adj}\nolimits}
\def\im{\mathop{\imath}}

\def\sn{\mathop{\rm sn}\nolimits}
\def\cn{\mathop{\rm cn}\nolimits}
\def\dn{\mathop{\rm dn}\nolimits}

\usepackage[colorlinks = true,
            linkcolor = blue,
            urlcolor  = red,
            citecolor = red,
            anchorcolor = blue,
            urlbordercolor= red ]{hyperref}

\makeatletter\let\@wraptoccontribs\wraptoccontribs\makeatother

\begin{document}

\title[Charge $2$ Monopole]{The Charge $2$ Monopole via the ADHMN construction}

\author{ H.W. Braden}
\address{School of Mathematics and Maxwell Institute of Mathematical Sciences, University of Edinburgh, Edinburgh.}
\email{hwb@ed.ac.uk}
\author{V.Z. Enolski\vspace{0.15in}}
\address{National University of Kyiv-Mohyla Academy, 
Hryhoriya Skovorody St, 2,
Kyiv, Ukraine. 
}
\email{ enolsky@ukma.edu.ua, venolski@googlemail.com}

\contrib[(with an Appendix by]{David E. Braden, Peter Braden and H.W. Braden)}

\begin{abstract}
Recently we have shown how one may use use integrable systems techniques to implement the ADHMN
construction and obtain general analytic formulae for the charge $n$ $su(2)$ Euclidean monopole. 
Here we do this for the case of charge $2$: so answering an open problem of some 30 years standing. A comparison with known results and other approaches
is made and new results presented.

\end{abstract}
\thanks{EMPG-19-05}
\maketitle
\tableofcontents

\section{Introduction}
This paper describes the exact solution of the gauge and Higgs fields for charge two $su(2)$ Euclidean monopoles. Despite BPS monopoles having been studied for over 35 years, and having uncovered extraordinarily beautiful structures, such analytic reconstruction has (with the exception of some partial results that will later be recalled) proved too hard. We often know more about the moduli space of these solutions than we do the actual fields. This is particularly true in the charge two setting: the Atiyah-Hitchin manifold, the moduli space of the centred charge two monopoles, is a well-studied and rich object
and yet the analytic solution of the fields has proved elusive. Recently a general program for reconstructing the gauge theory data for $su(2)$ Euclidean monopoles of general charge has been given, circumventing a number of
previously intractable steps. This lowest charge case is a useful testing ground and will produce a
number of new results. (The spherically symmetric  case for charge one and coincident charge $n$ monopoles is amenable to other approaches.) We will compare our results with some of the numerical studies that have been
undertaken. 
Although constructing exact solutions -- be they of gravity or gauge theory -- is often viewed as a rather recondite area of research analytic solutions give at the very least some control over 
numerical results.

The algebro-geometric construction of $su(2)$ Euclidean monopoles described here is built upon the substantive work of a number of authors. Particularly relevant (with more detail following) are:
\begin{enumerate}[(i)]
\item  Nahm's modification of the ADHM construction of instantons \cite{Nahm:1979yw, nahm_80}. 
This introduces $n\times n$ matrices $T_i(z)$ ($j=1,\ldots,4$) that satisfy a system of ordinary differential equations (Nahm's equations) and an operator $\Delta$. 

\item Nahm's equations may be written as a Lax pair $\dot L=[L,M]$. Here there is a spectral parameter
$L=L(\zeta)$, $\zeta\in \mathbb{P}\sp1$, and the characteristic equation $P(\eta,\zeta):=\det(\eta-L(\zeta))=0$ defines a spectral curve $\mathcal{C}\subset T\mathbb{P}\sp1$ where the mini-twistor space
$T\mathbb{P}\sp1$ is the geometric setting for Hitchin's description of monopoles \cite{nahm82c, hitchin_82}.

\item The gauge and Higgs fields are then constructed from integrals (over $z$) of bilinears involving the two normalizable solutions to the Weyl equation $\Delta\sp\dagger \boldsymbol{v}=0$.  Hitchin proved 
\cite{hitchin_83} that the regularity of these fields placed certain constraints on the curve $\mathcal{C}$. We shall describe a curve
satisfying Hitchin's constraints as a \emph{monopole spectral curve}.

\item These integrals may in fact be performed using formulae of Panagopoulos \cite{panagopo83, Braden2018b}.

\item Ercolani and Sinha showed how one could use integrable systems techniques to solve for a gauge transform of the Nahm data in terms of a Baker-Akhiezer function $\Phi_{BA}$ associated to $\mathcal{C}$
\cite{ercolani_sinha_89}. Here one of Hitchin's constraints on the curve is reexpressed in terms of the direction $\mathbf{U}$ of flow on the Jacobian  $\Jac(\mathcal{C})$. The Ercolani-Sinha vector $\mathbf{U}$ is a half-period \cite{houghton_manton_romao_00, Braden2010d}.

\item Using a lesser known ansatz of Nahm the authors showed how one might solve for $\boldsymbol{v}$
in terms of the Baker-Akhiezer function $\Phi_{BA}$ and the same (unkown) gauge transformation 
employed by Ercolani and Sinha \cite{nahm82c, Braden2018b}.

\item Finally it has been shown how to eliminate the unknown gauge transformation to reconstruct the
gauge and Higgs fields  \cite{Braden2018b}.

\end{enumerate}
At this stage one has a way of analytically constructing the gauge and Higgs fields given a monopole
spectral curve. Several remarks are however in order.
The number of known monopole spectral curves is few: although Hitchin's constraints on a curve are algebro-geometric in nature a constructive solution is still lacking \cite{Braden2018c}. 
The construction outlined does not yet
provide a solution to the Nahm equations in standard ($T_4=0$) gauge. Notwithstanding such questions
we may now in principle analytically construct solutions.

To provide the context to the contents and new results of this paper we must first recall the various analytic
approaches to studying BPS monopoles.

\subsection{Three Analytic Approaches}

There have been three approaches to constructing analytic solutions of the $su(2)$ monopole equations on
$\mathbb{R}^3$: via the $\mathcal{A}_k$ ansatz of Atiyah-Ward; via an ansatz of {Forg{\'a}cs,  Horv{\'a}th  and Palla } that emerged from their study of axially symmetric monopoles and the Ernst equation; and via
Nahm's modification of the ADHM construction of instantons. 
We shall briefly describe these. 
In all three approaches the spectral curve of the monopole appears and the importance of this curve was gradually elucidated. Further, most authors focussed on calculating the Higgs field and the gauge invariant quantity $\frac12 \Tr\Phi^2$. With appropriate
choices points on the spatial axes are related to points on the $n=2$ spectral curve by biquadratic
equations rather than the more general quartic equation and this meant the Higgs field on the coordinate axes was more amenable to study. One early result \cite{bpp82}[(7.2)] was that the Higgs field at the origin 
gave (in units described in the sequel)
\begin{equation}\label{higgsorigin}
-\frac12 \Tr\Phi^2 \big|_{(0,0,0)}= \frac{  (K(1+{k'}^2)-2E)^2}{K^2k^4}.
\end{equation}
One of the simplifying features of monopoles is that the energy density $\mathcal{E}(\boldsymbol{x})$ is
related to $\frac12 \Tr\Phi^2$ via Ward's formula \cite{Ward1981b}
\begin{equation*}
\mathcal{E}(\boldsymbol{x})=-\frac12 \nabla^2 \Tr \Phi^2.
\end{equation*}
Once one could calculate $\frac12 \Tr\Phi^2$ in any of these approaches it was possible to numerically calculate the Laplacian and subsequently the energy density: the culmination of these (amalgams of analytic and numerical)  studies were plots and a video using an early supercomputer (see below).

\subsubsection{$\mathcal{A}_k$ ansatz}
Based on Ward's identification \cite{Ward1977} of self-dual solutions to the Yang-Mills equations and appropriate vector bundles over twistor space,  Atiyah and Ward \cite{atiyah_ward_77} developed
a series of ans\"atze, the $\mathcal{A}_k$ ansatz, that reduced the construction of $su(2)$ instantons  to constructing  patching functions $g$ for gauge bundles. In terms of this data Corrigan, Fairlie, Yates and Goddard \cite{corrigan_fairlie_yates_goddard} showed how to reconstruct the gauge fields making connection with
Yang's  study of $su(2)$ instantons \cite{yang77} and Yang's equation.

Now Manton in \cite{manton78} had noted that the field equations for
BPS monopoles corresponded to the equations of static self-duality and Ward in \cite{Ward1981b} described how to
modify the patching function data to reproduce such solutions. Ward's initial ansatz produced axially symmetric\footnote{One of the surprises discovered about BPS monopoles was that an axial symmetric monopole corresponded to coincident
charges \cite{hor80}. } charge 2 monopoles and for a particular choice of constant he saw regular solutions.
Prasad and Rossi \cite{prasad_rossi_81} then produced the appropriate Atiyah-Ward patching function for the axially symmetric charge $n$ monopole.
Ward  \cite{Ward1981a} subsequently generalized his ansatz to account for separated charge $2$ monopoles and
Corrigan and Goddard \cite{corrigan_goddard_81} extended this to the general charge $n$ monopole with $4n-1$ degrees of freedom. One shortcoming with this approach was that the regularity of the gauge fields was left unproven: although the spectral curve of the monopole makes its appearance the full conditions for regularity were not obtained until Hitchin's work \cite{hitchin_83}.

Ward concludes in \cite{Ward1981a}:
\lq\lq It seems likely that the expressions for general $n$-monopole solutions, as
functions of $x$, $y$  and $z$ are so complicated that there would be little point in
trying to write them out. Of course, since we have explicit formulae, the fields
could be computed numerically to any desired degree of accuracy. One attraction
of the technique presented here is that the matrices $g$ are relatively simple,
even when the corresponding space-time fields $A_\mu$ are extremely complicated.
So one can deduce much about instantons and monopoles (such as their
existence!) without having to write down space-time expressions for them.\rq\rq\

There have been a few works that have sought to apply the Atiyah-Ward construction. Brown, Prasad and
Rossi \cite{MR632527} explored the uniqueness and assumptions of \cite{Ward1981a, corrigan_goddard_81}; their results differed in cases of non-regular monopoles. In \cite{MR658809}
O'Raifeartaigh, Rouhani and Singh looked at solving the Corrigan-Goddard constraints for $n$ monopoles
while in \cite{ors82} they studied the $n=2$ monopole in detail. This latter work presents the Higgs field
in terms of various infinite sums and their derivatives: their  \lq very complicated\rq\ expression was evaluated numerically for the
axis joining the monopoles where the zero was found to be \lq very close\rq\ and \rq barely distinguishable\rq\ from $\pm kK(k)/2$ (in our later notation) \cite[\S6, \S9]{ors82}. They write that they \lq\lq cannot guess a \lq natural\rq\ analytic expression\rq\rq\ describing this position. Brown \cite{brown83}  later evaluates these infinite sums in terms of elliptic functions. Brown in fact evaluates the Higgs field on each of the axes using
the Corrigan-Fairlie-Goddard-Yates formalism reproducing for one axis the earlier result of Brown, Prasad and Panagopoulos \cite{bpp82} (see below) obtained via the Nahm equations with the corresponding
value of $\frac12 \Tr\Phi^2$ at the origin (\ref{higgsorigin}). Without denying the importance of the Brown's work we believe that there are errors in his formulae describing behaviour on the other axes, in particular his values of the Higgs field at the origin of his $y$ and $z$ axes differ from (\ref{higgsorigin}).

\subsubsection{The Forg{\'a}cs,  Horv{\'a}th and Palla Ansatz}
Again in \cite{manton78} Manton introduced an ansatz for axially symmetric BPS monopoles that
he was unable to solve. In a series of papers Forg{\'a}cs,  Horv{\'a}th and Palla 
\cite{forgacs_horvath_palla_80, forgacs_horvath_palla_81ym, MR610412}
used the Ernst equation to study such monopoles separate to the developments of the Atiyah-Ward construction. In \cite{forgacs_horvath_palla_81} they obtained a suitable B\"acklund transformation 
reproducing\footnote{Compare (8), (9) of \cite{Ward1981b} with (22), (23) of 
\cite{forgacs_horvath_palla_81}.}  Ward's results while in \cite{forgacs_horvath_palla_81b} they look at
$n=2$, $3$, $5$ giving determinantal expressions for $\Tr\,\Phi^2$ and from this  plots for
the energy density evaluated numerically.

Forg{\'a}cs,  Horv{\'a}th and Palla subsequently
generalized their ansatz \cite{MR610412, forgacs_horvath_palla_82b, forgacs_horvath_palla_83b} to account for separated monopoles; this also made connection to Yang's equation. In  
\cite{forgacs_horvath_palla_82b, forgacs_horvath_palla_83b} their ansatz gives the Higgs field, and from this they numerically calculate the energy density plotting this for the (in our conventions) $x_2=0$ plane.
Based on the numerical evaluation of their ansatz Forg{\'a}cs,  Horv{\'a}th and Palla \cite[(21)] {forgacs_horvath_palla_82b} gave the zeros of the Higgs field
to be (in our units) $\pm kK(k)/2$ while in their later work\footnote{
This followed two analytic works: the already noted  \cite[\S6, \S9]{ors82} where 
the zero was numerically  found to be very close to $\pm kK(k)/2$;  and in \cite{bpp82} expansions for the zeros of the Higgs field were given for $k$
near $0$ and $1$, the latter being situated near $\pm kK(k)/2$.}
 \cite[\S6] {forgacs_horvath_palla_83b}
they expressed that their earlier result was to be viewed as a very good approximation of the zeros.

Using the Forg{\'a}cs,  Horv{\'a}th and Palla ansatz for the $n=2$ Higgs field Hey, Merlin, Ricketts, Vaughn and Williams \cite{hey_maerlin_ricketts_vaughn_williams, merlin_ricketts_87} made use of a very early supercomputer to determine the Higgs field and consequently (numerically) the
energy density over a region of $\mathbb{R}\sp3$ for various monopole separations. Together with
Atiyah and Hitchin this was used to produce a video describing monopole collisions
\cite{atiyah_hitchin_merlin_pottinger_ricketts}.

\subsubsection{Nahm's modification of the ADHM construction}
Nahm's modification of the ADHM construction was developed in \cite{Nahm:1979yw, nahm_80, Nahm:1981xg} and in \cite{nahm82c} he described the algebraic geometry underlying this together with his
\lq\lq lesser known\rq\rq\ ansatz. Brown, Prasad and Panagopoulos \cite{bpp82} used Nahm's formalism to
explicitly solve for the Higgs field on a portion of the axis joining two separated monopoles. This was possible because Nahm's $4\times4$ matrix equation $\Delta\sp\dagger v=0$ (see below) actually factorizes into two
$2\times2$ matrix equations. We shall show that this holds true for each axis and indeed the same Lam\'e equation results with appropriate shifts for each axis.
A significant early step in tying Nahm's work with integrable systems was then made by Ercolani and Sinha \cite{ercolani_sinha_89} who first made connection with the Baker-Akhiezer function; Houghton, Manton and
Rom\~ao  \cite{houghton_manton_romao_00} revisted this making connection with the Corrigan-Goddard
constraints \cite{corrigan_goddard_81} and the $\mathcal{A}_k$ ansatz. In a number of works culminating in
\cite{Braden2018b} the authors have shown how given a spectral curve one may solve for the monopole gauge data; this paper will, amongst other things, do this for the $n=2$ case.

\subsubsection{Spectral Curves}
As noted above a spectral curve underlies each of the analytic approaches just described. Hitchin \cite{hitchin_82}[Theorem 7.6] shows this curve determines the bundle described by the $\mathcal{A}_k$ ansatz and in \cite{hitchin_83} that it is the spectral curve of Nahm's integrable system. Also in
 \cite{hitchin_83}  Hitchin gives the necessary and sufficient conditions for a spectral curve to yield a  nonsingular monopole. Hurtubise \cite{hurtubise_83} then evaluated these constraints to produce the 
 $n=2$ spectral curve. 
\subsection{Overview and Principal Results}
While it has been known for a long time then that the spectral curve fully determines a monopole it has remained less clear how to implement this. Our approach here is to follow Nahm's construction: in
\cite{Braden2018b} we have described this for general $n$ and here we will do this concretely for $n=2$.
We will review this approach in Section \ref{background}. 
Whatever approach is adopted one needs an understanding of the spectral curve and the integrals of certain meromorphic forms on it. Sections 3-6 will determine many of the basic properties of the $n=2$ curve, its parameterizations and needed integrals. A given point $\boldsymbol{x}\in \mathbb{R}\sp3$
corresponds to (generically) $4$ points on the curve (by what we describe below as the Atiyah-Ward
constraint). We uncover a number of new special addition theorems for $\theta$-functions whose
arguments are the Abelian images of these  points as well new relations for sums of non-complete first and second kind integrals. The explicit answers and derivations for the charge two monopole depend significantly on these.
The results  of these sections will enable us to make contact with earlier results.
(Appendix A will relate the different forms of this curve used by workers over the years.)
As remarked upon above, the coordinate axes (under appropriate choices) have a number of
simplifying properties and these are described in \ref{sectionloci}.
Section \ref{sectionbitangency} describes spatial points whose twistor lines are bitangent to the spectral curve: these points will also be distinguished in various ways described in the sequel.

Only in sections 7-9 do we come to the data in Nahm's modification of the ADHM construction: section
\ref{sectionW} describes a fundamental matrix  $W$ of solutions to a matrix first order differential equation
 $\Delta W=0$ in terms of the function theory on the curve (in particular the Baker-Akhiezer function); section \ref{sectionV} describes the adjoint of
this equation  $\Delta\sp\dagger V=0$; and section \ref{sectionprojector} describes the projector from $V$
to two normalizable solutions $\Delta\sp\dagger \boldsymbol{v}=0$. (Both $\Delta$ and its adjoint
$\Delta\sp\dagger$ are describe more fully below.) From $\boldsymbol{v}$ one may construct the gauge and Higgs field. We illustrate this by constructing the Higgs field in the simpler setting for the $x_2$ axis in section \ref{sectionx2axis} recovering (\ref{higgsorigin}). In Appendix \ref{applame}
we  show this yields the result of Brown, Prasad and Panagopoulos \cite{bpp82} obtained via Lam\'e's
equation.
In section \ref{sectionHiggs} we turn to general formulae for the Higgs field and energy density where we obtain the new result  for the energy density at the origin in Proposition \ref{energy_density_origin},
\begin{align*} 
\mathcal{E}_{\boldsymbol{x}=0}(k)= \frac{32}{k^8{k'}^2K^4} \left[ k^2(K^2{k'}^2 +E^2-4EK+2K^2+k^2)
-2(E-K)^2    \right]^2,
\end{align*}
with the known limiting values $\mathcal{E}_{\boldsymbol{x}=0}(0) =\frac{8}{\pi^4}(\pi^2-8)^2$ at $k=0$ (coincident  monopoles) and  $ \mathcal{E}_{\boldsymbol{x}=0}(1)=0$ at $k=1$.
Section \ref{sectionHiggsonaxes} evaluates the general formulae for the Higgs field on the coordinate axes;
here we are able to give the equation (\ref{higgszeroposition}) describing the zero of the Higgs field.
(Again in Appendix \ref{applame} we obtain these solutions via Lam\'e's equation.) Figure \ref{HFCombined}
illustrates these results for different scales.
Finally in section \ref{sectionk0}  we take the  $k=0$ limit of our results reproducing Ward's expressions
\cite{Ward1981b}  amongst others. Throughout the text we will defer a significant number of proofs and computations to the five Appendices.
\begin{figure}
\centering
\begin{subfigure}[b]{0.49\linewidth}
\includegraphics[width=0.7\textwidth]{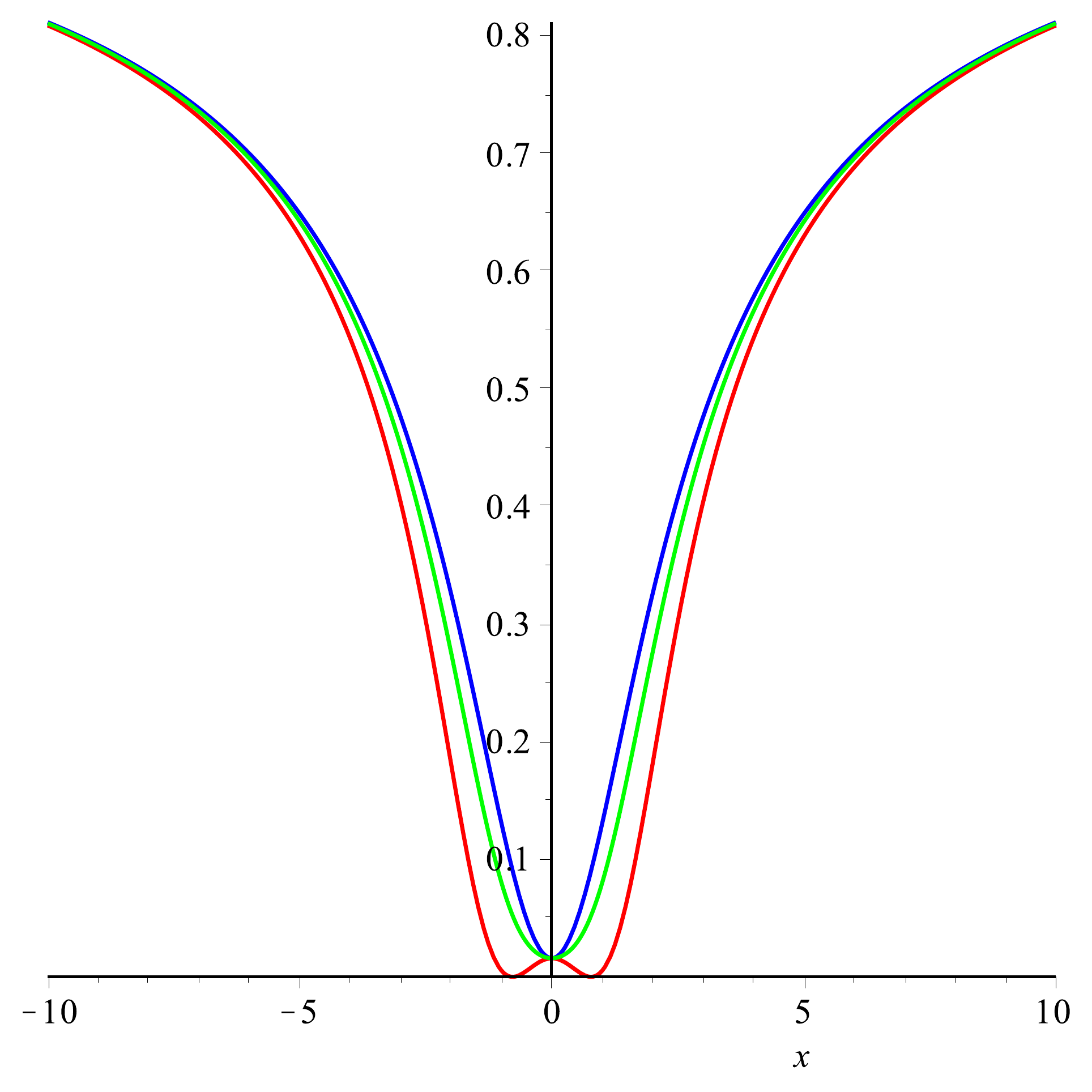}
\end{subfigure}
\begin{subfigure}[b]{0.49\linewidth}
\includegraphics[width=0.7\textwidth]{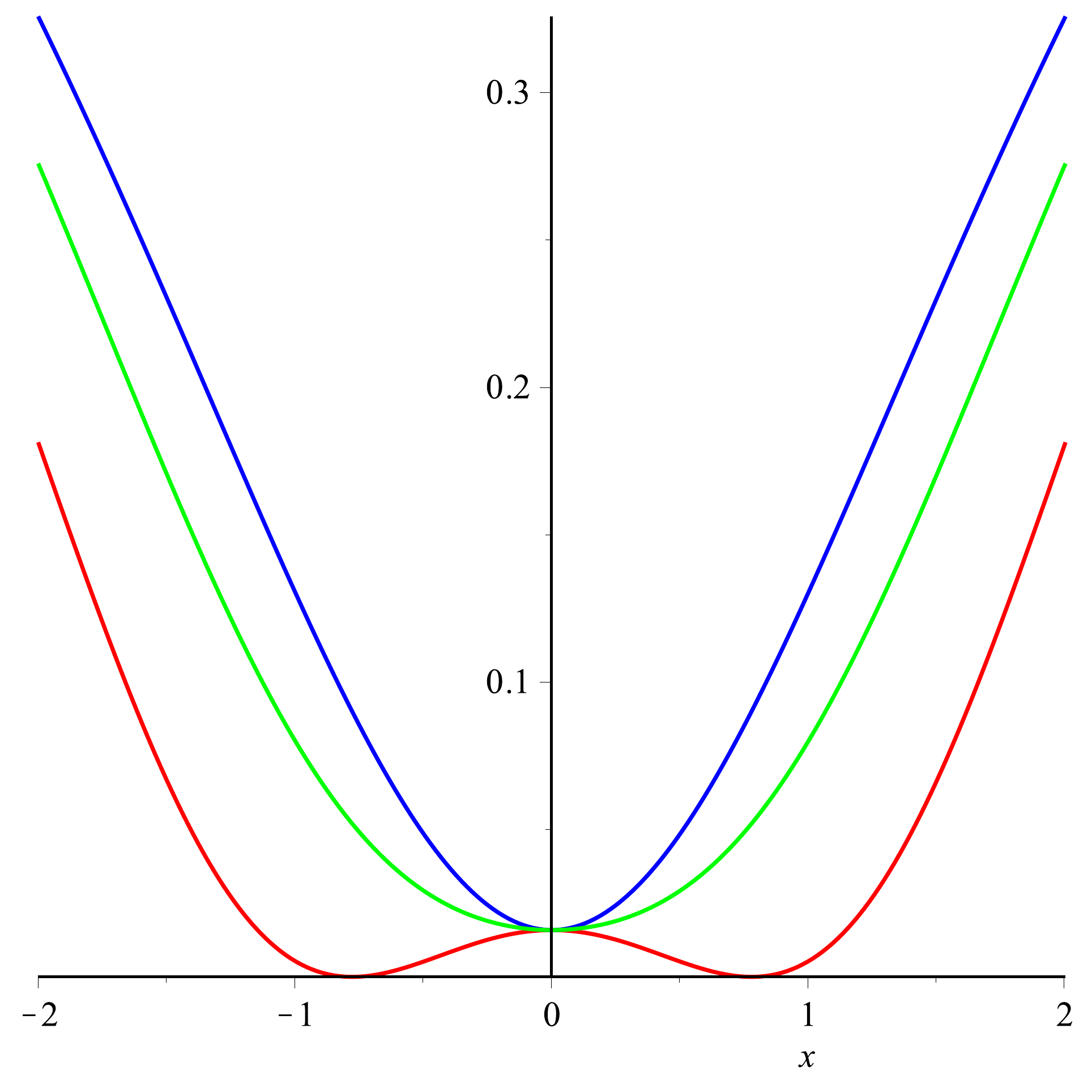}
\end{subfigure}
\caption{$-\frac12 \Tr\Phi^2$:   $x_1$-axis red; $x_2$-axis blue; $x_3$-axis green.  $k=0.8$. }
 \label{HFCombined} 
\end{figure}

We conclude this introduction by comparing our analytic results with numeric computations. Figures
\ref{AnalvsNumerHiggs} and \ref{EnergyDensityX1axis} compare our results with the numerical results
underlying\footnote{We thank Paul Sutcliffe for making these available for comparison.} the charge $2$ results of \cite{manton_sutcliffe_book} for $-\frac12 \Tr\Phi^2$ and the energy density respectively.
The clear lesson is how well these results agree. From Figure \ref{AnalvsNumerHiggs} the values of
 $- \frac12 \Tr\Phi^2$ are essentially indistinguishable for $x_\star<5$; for larger $x_\star$ one sees a
 divergence (attributable to the large intermediate quantities involved in the calculation) but in the range
of the second plot the is only $1\sim2\%$.
Figure \ref{EnergyDensityX1axis} compares the energy density: again these are essentially indistinguishable. Closer inspection of the analytic result (red) shows $4$ anomalous evaluations: these are
close to points of bitangency mentioned above (and described more fully below); these could be removed by
using l'Hopital's rule, but we have included them here to illustrate their presence.

\begin{figure}
\centering
\begin{subfigure}[t]{0.49\linewidth}
\includegraphics[width=0.9\textwidth]{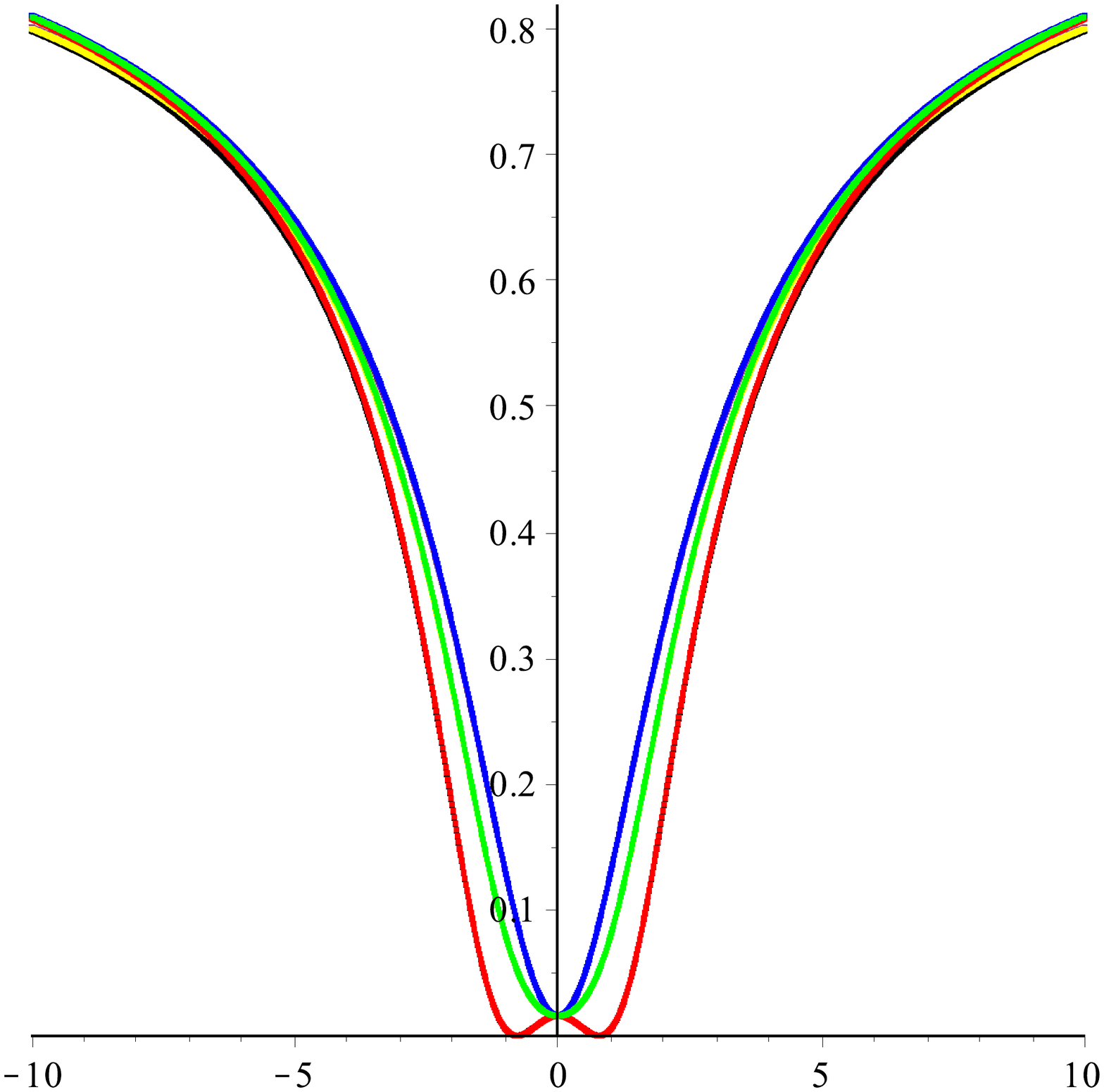}
\end{subfigure}
\begin{subfigure}[t]{0.49\linewidth}
\includegraphics[width=0.9\textwidth]{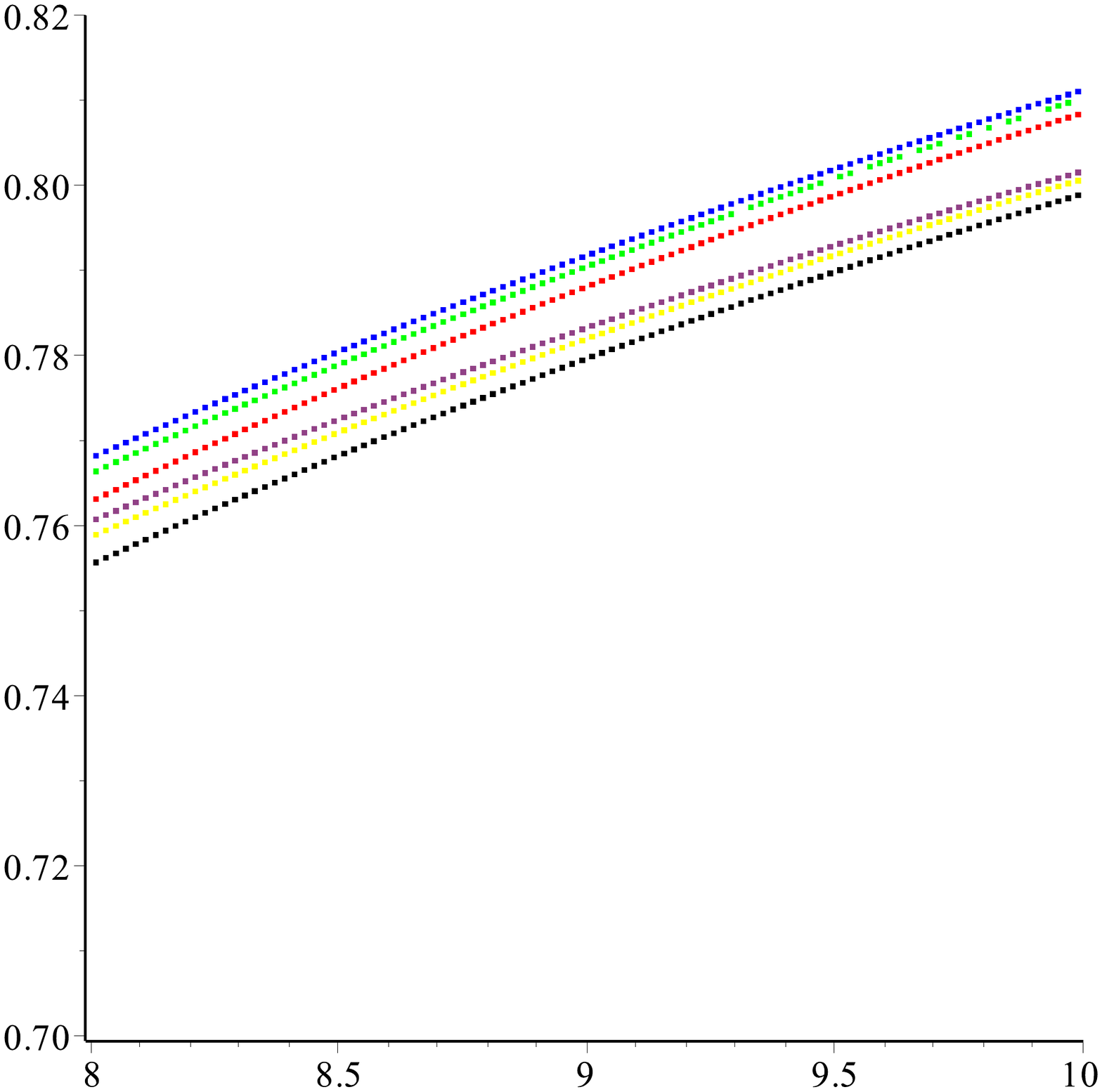}
\end{subfigure}
\caption{$-\frac12 \Tr\Phi^2$ Analytic vs Numerical:   $x_1$-axis red vs black; $x_2$-axis blue vs violet; $x_3$-axis green vs yellow. The second plot focusses on a smaller interval. $k=0.8$.  }
 \label{AnalvsNumerHiggs} 
\end{figure}

\begin{figure}
\centering
\begin{subfigure}[t]{0.3\linewidth}
\includegraphics[width=0.9\textwidth]{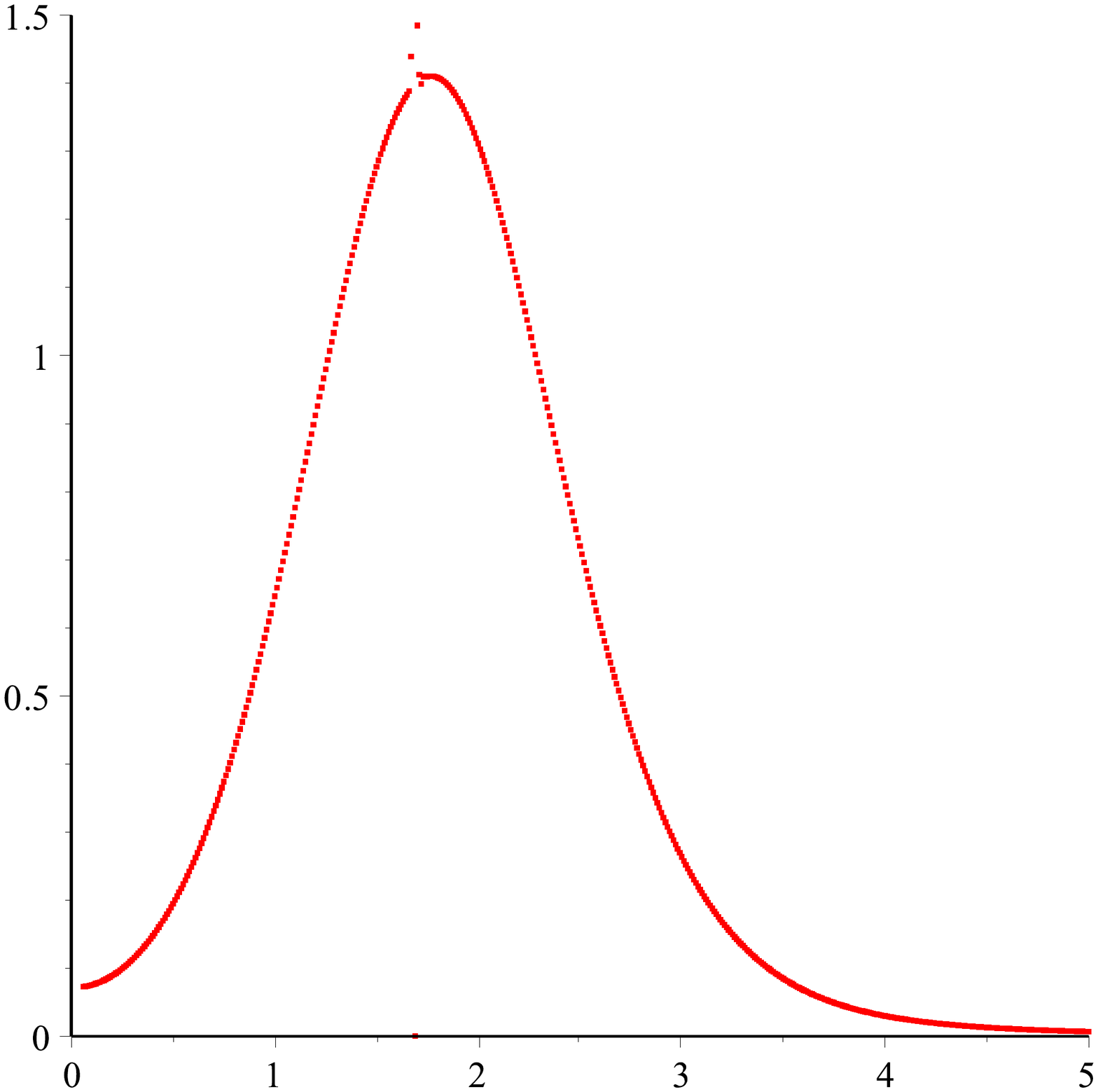}
\end{subfigure}
\begin{subfigure}[t]{0.3\linewidth}
\includegraphics[width=0.9\textwidth]{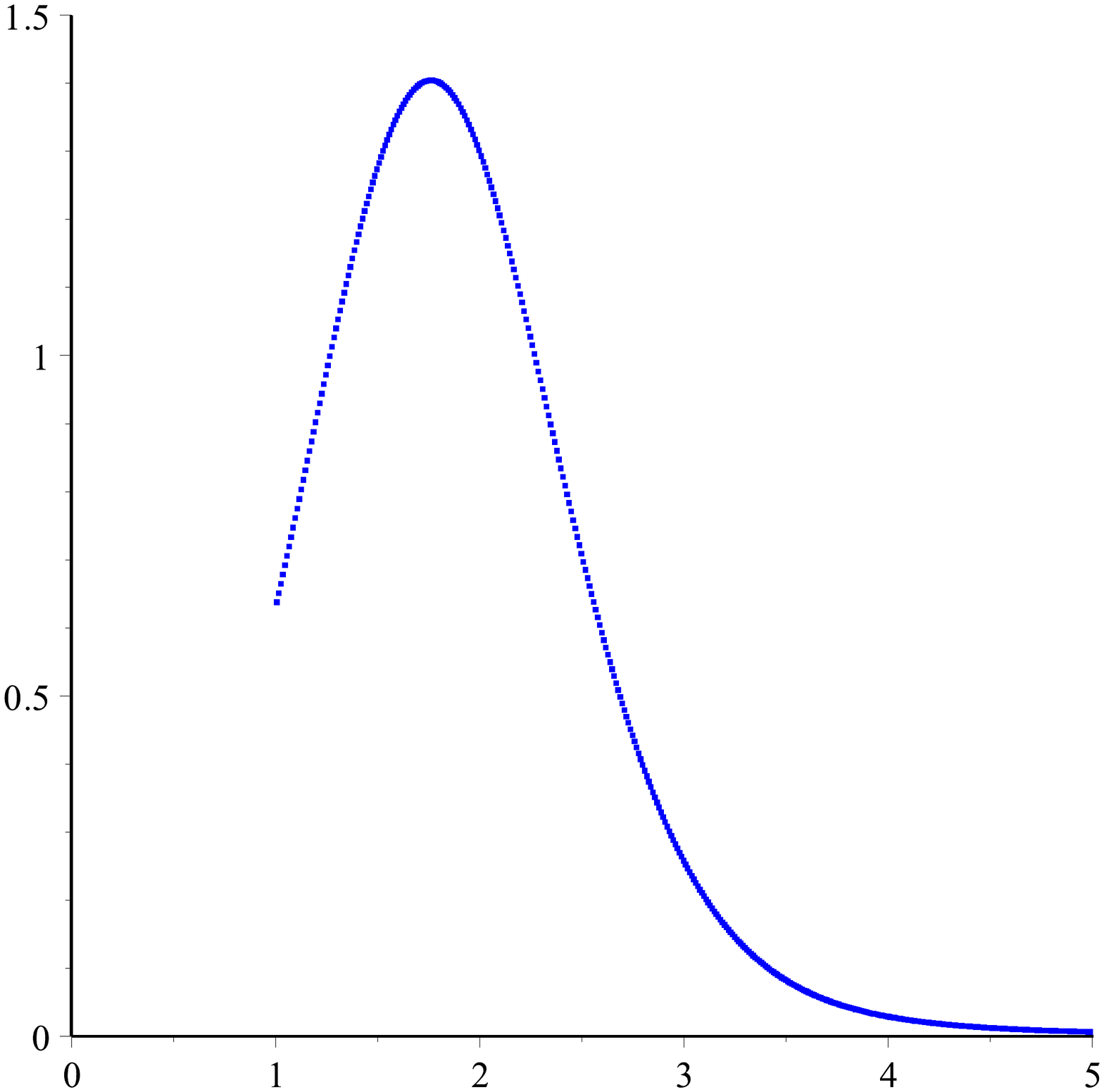}
\end{subfigure}
\begin{subfigure}[t]{0.3\linewidth}
\includegraphics[width=0.9\textwidth]{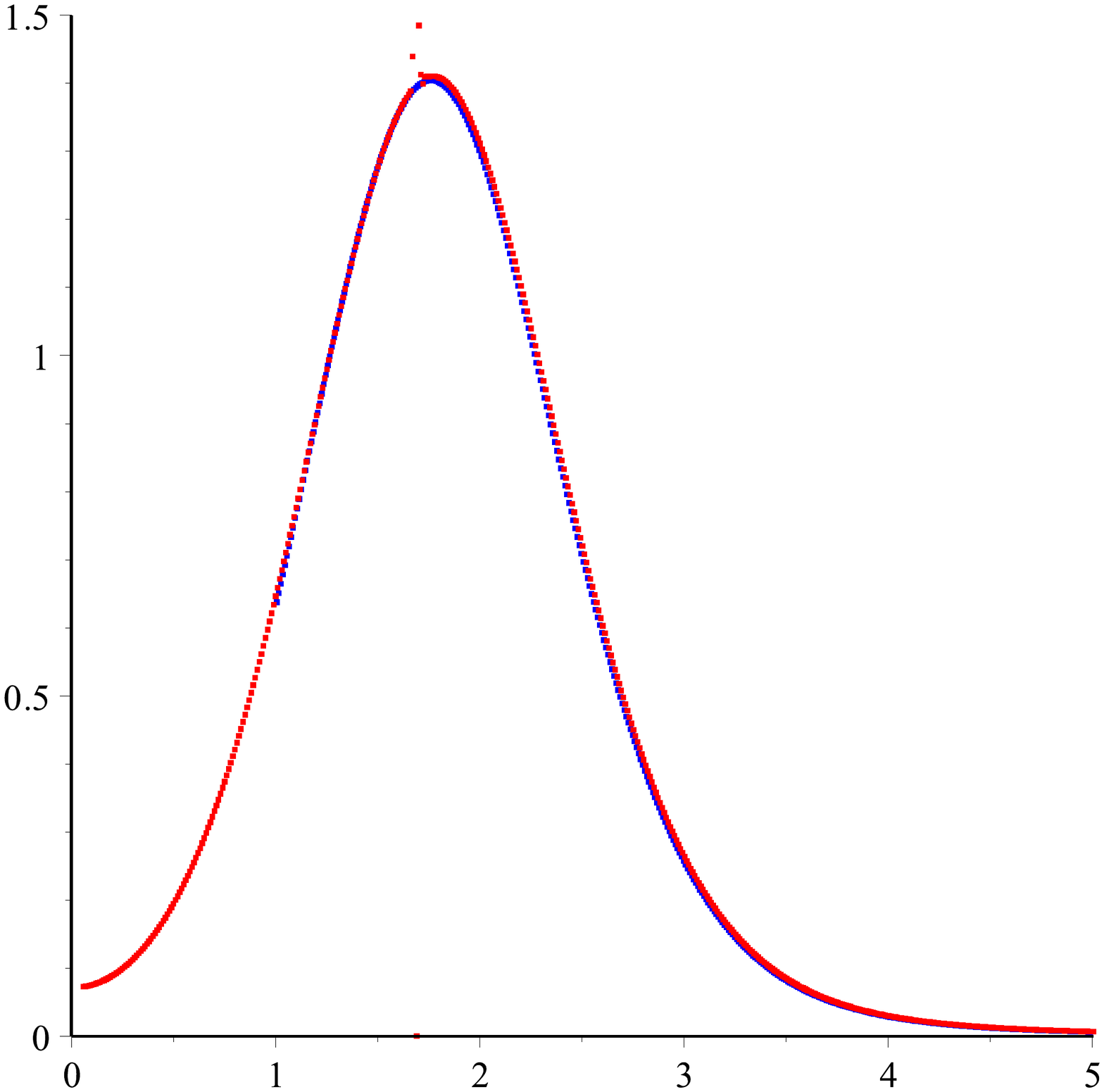}
\end{subfigure}
\caption{Energy Density on the  $x_1$-axis: analytic (red), numerical (blue) and comparison for  $k=0.99$. }
 \label{EnergyDensityX1axis} 
\end{figure}

Figures \ref{energydensityk=80}-\ref{energydensityk=99} give a number of views of the energy density as a function of $k$ ($k=0$ being coincident, and $k=1$ infinitely separated). In Appendix \ref{appBBB}
(by David E. Braden, Peter Braden and H.W. Braden) we describe and give links to both the scripts that generate the monopole numerics and tools to enable their visualisation. Three tools are given: two are interactive, and the third graphical. The first visualiser encodes energy density as opacity while the
second defines a energy density threshold above which to consider as solid (the mesh can be visualised with many mesh viewers, or even 3D printed). Screenshots of these are given in Figure \ref{bothvisualisers}. The third method of visualizing the data is a \lq Tomogram\rq\ that takes slices through the volume. We can plot the contours on
these images, or use colour to represent the density at that slice (see Figure \ref{bothtoms}).
The second last column of these figures correspond to the $k$ value of Figure  \ref{bothvisualisers}.

\begin{figure}
\centering
\begin{subfigure}[t]{0.49\linewidth}
\includegraphics[width=1.5\textwidth]{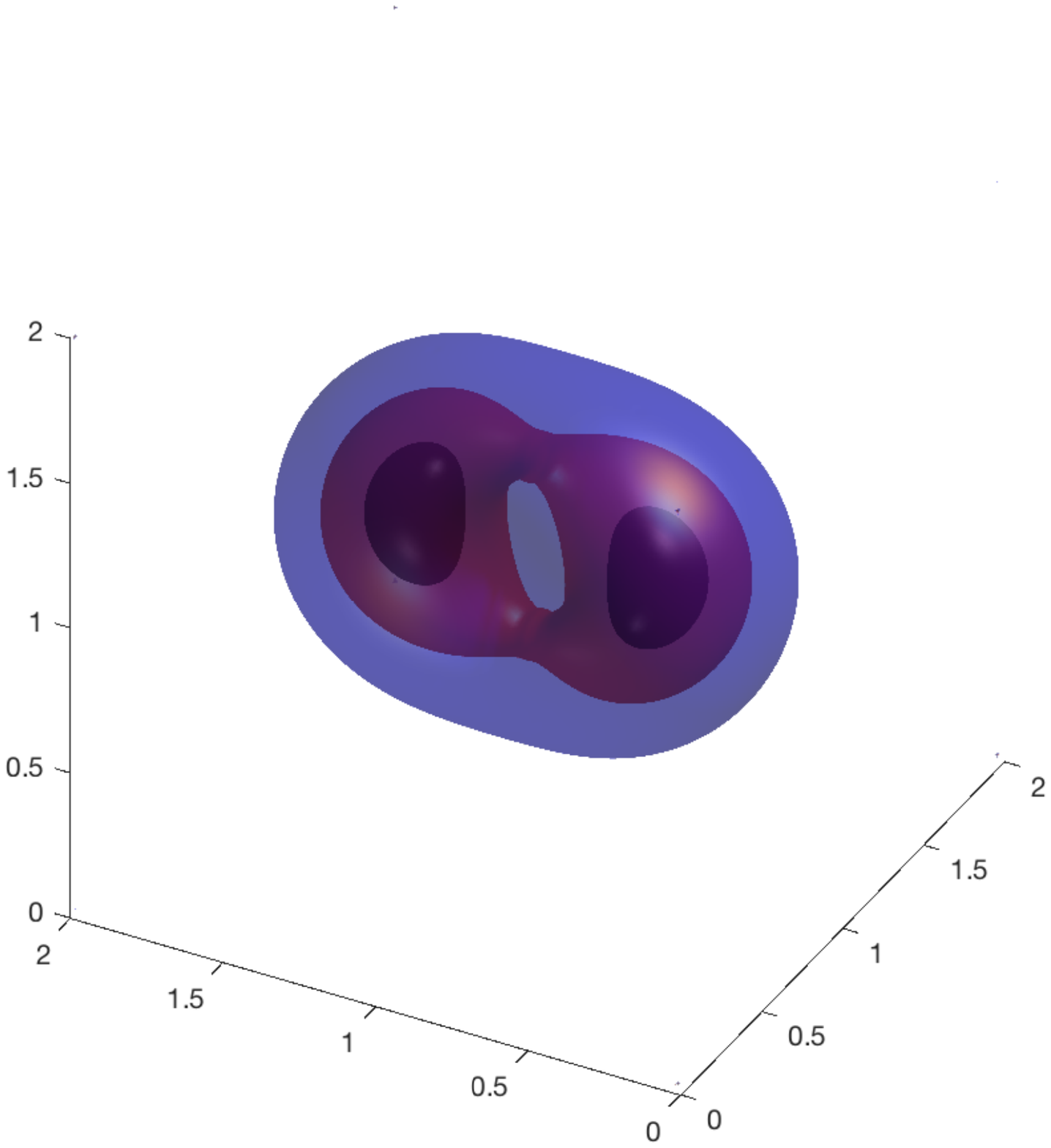}
\end{subfigure}
\begin{subfigure}[t]{0.49\linewidth}
\includegraphics[width=1.5\textwidth]{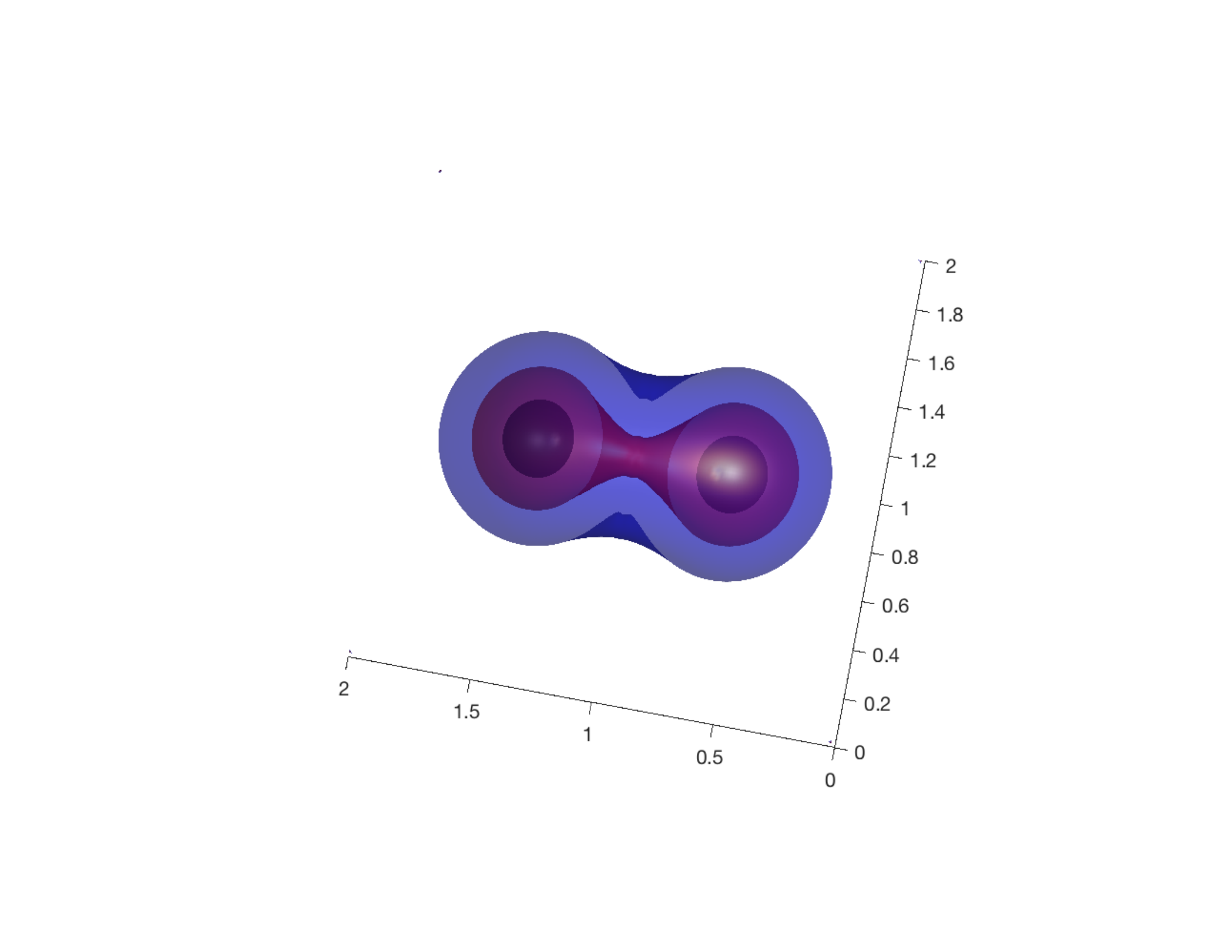}
\end{subfigure}
\caption{Two views of the Energy density $\mathcal{E}(x)$  for $k=0.8$.  Blue corresponds to the isocontour
$\mathcal{E}(x)=0.2$, red to $\mathcal{E}(x)=0.42$,  and dark red to $\mathcal{E}(x)=0.7$. } 
 \label{energydensityk=80}
\end{figure}

\begin{figure}
\centering
\begin{subfigure}[t]{0.49\linewidth}
\includegraphics[width=1.5\textwidth]{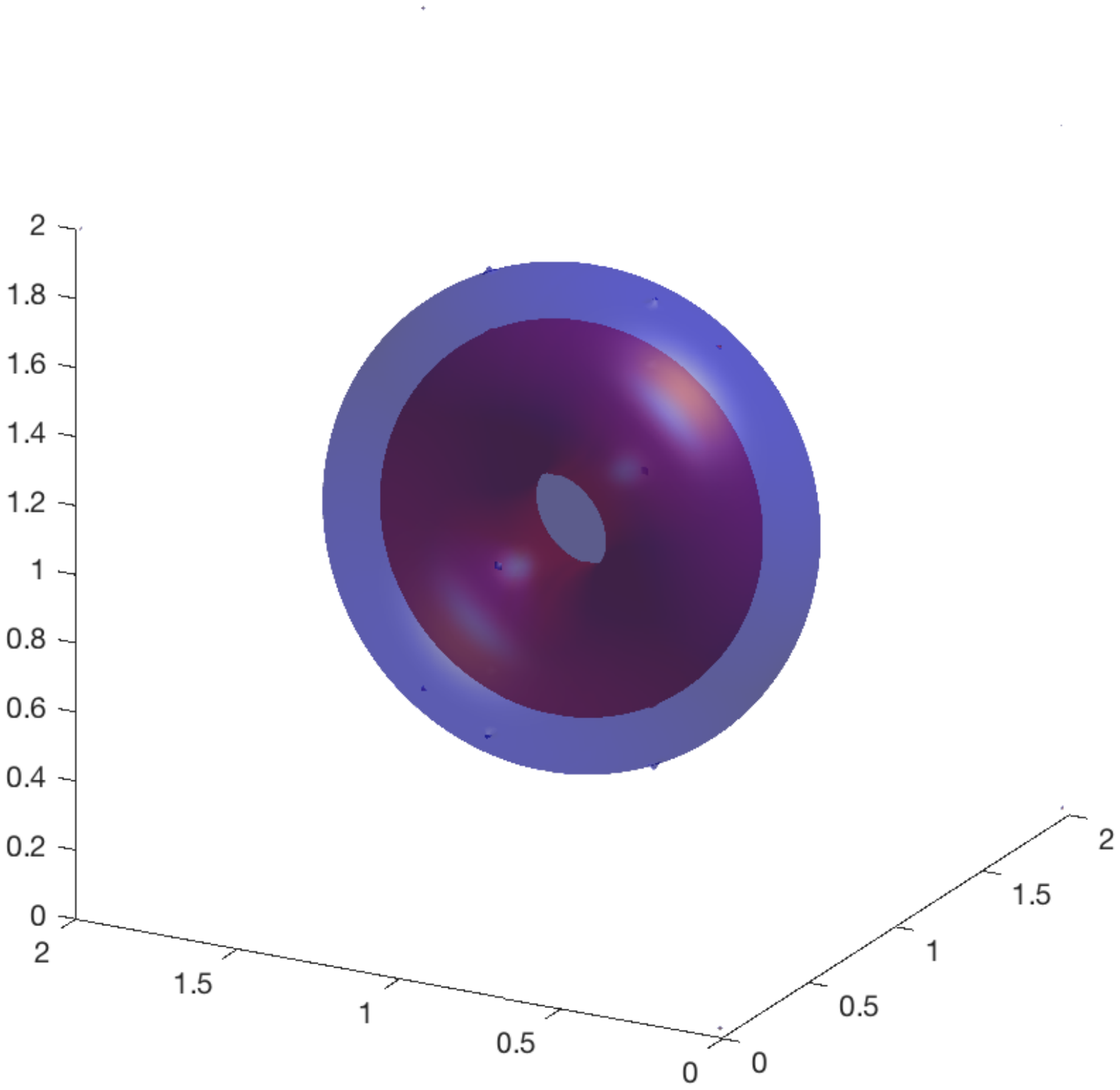}
\end{subfigure}
\begin{subfigure}[t]{0.49\linewidth}
\includegraphics[width=1.5\textwidth]{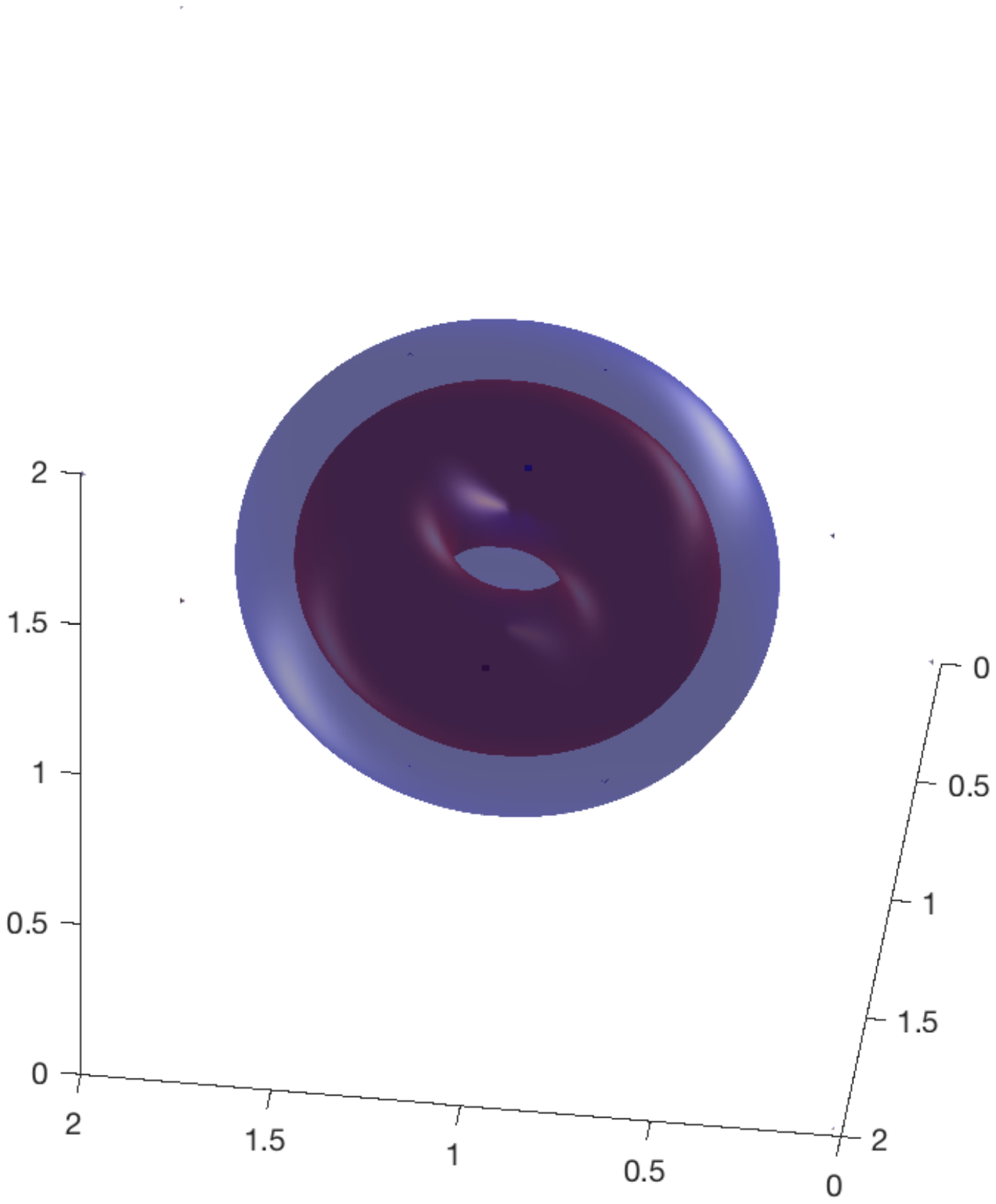}
\end{subfigure}
\caption{Two views of the Energy density $\mathcal{E}(x)$  for $k=0.05$.  Blue corresponds to the isocontour
$\mathcal{E}(x)=0.2$, red to $\mathcal{E}(x)=0.42$. The energy density $\mathcal{E}(x)=0.7$ is not
achieved. } 
 \label{energydensityk=05}
\end{figure}

\begin{figure}
\centering
\begin{subfigure}[t]{0.49\linewidth}
\includegraphics[width=1.5\textwidth]{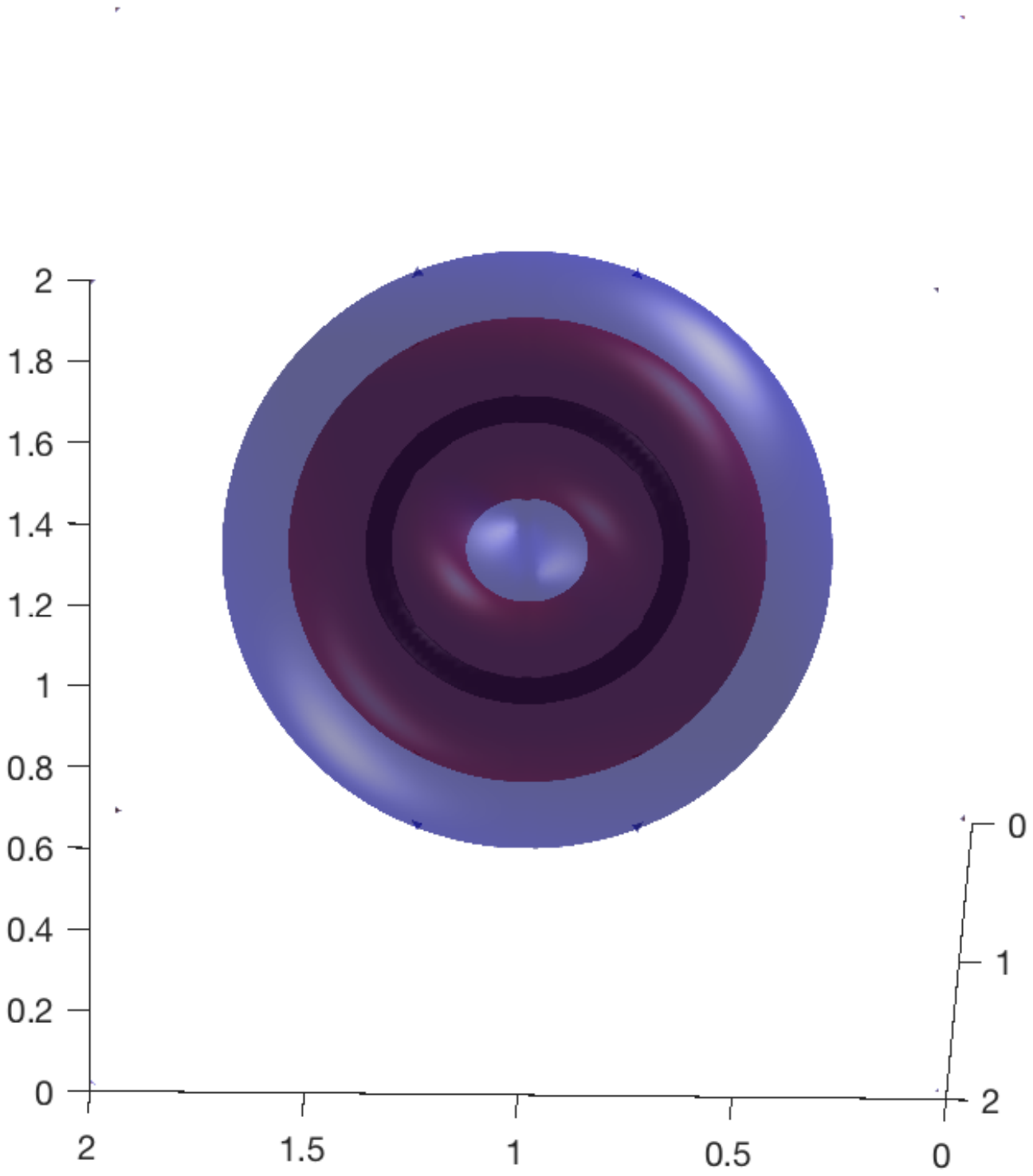}
\end{subfigure}
\begin{subfigure}[t]{0.49\linewidth}
\includegraphics[width=1.5\textwidth]{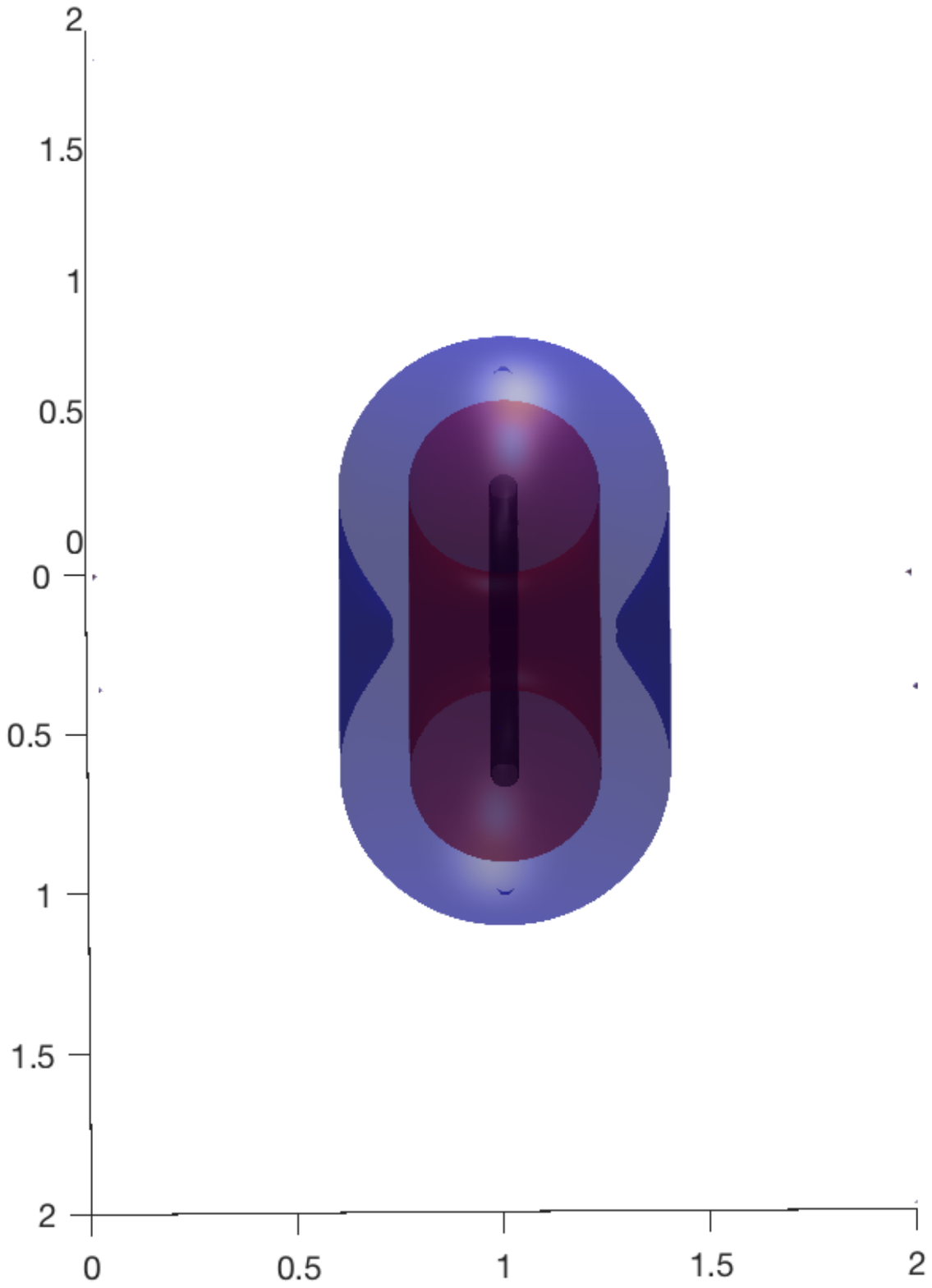}
\end{subfigure}
\caption{Two views of the Energy density $\mathcal{E}(x)$  for $k=0.25$.  Blue corresponds to the isocontour
$\mathcal{E}(x)=0.2$, red to $\mathcal{E}(x)=0.42$, dark red to $\mathcal{E}(x)=0.65$. } 
 \label{energydensityk=02}
\end{figure}

\begin{figure}
\centering
\includegraphics[width=1.0\textwidth]{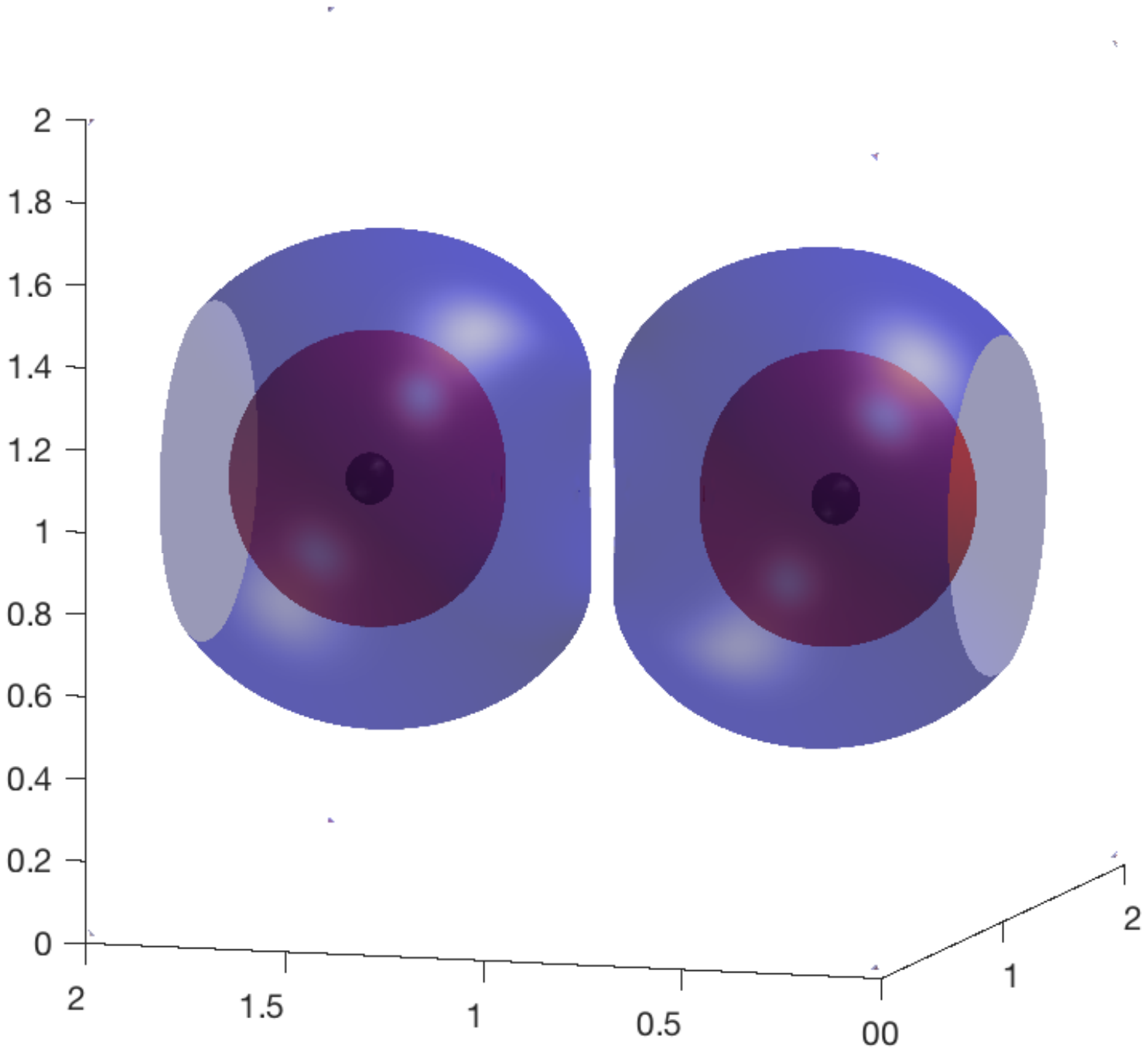}
\caption{The Energy density $\mathcal{E}(x)$  for $k=0.99$.  Blue corresponds to the isocontour
$\mathcal{E}(x)=0.09$, red to $\mathcal{E}(x)=0.42$, dark red to $\mathcal{E}(x)=1.35$. } 
 \label{energydensityk=99}
\end{figure}

\begin{figure}
\centering
\begin{subfigure}[t]{0.49\linewidth}
\includegraphics[width=1.5\textwidth]{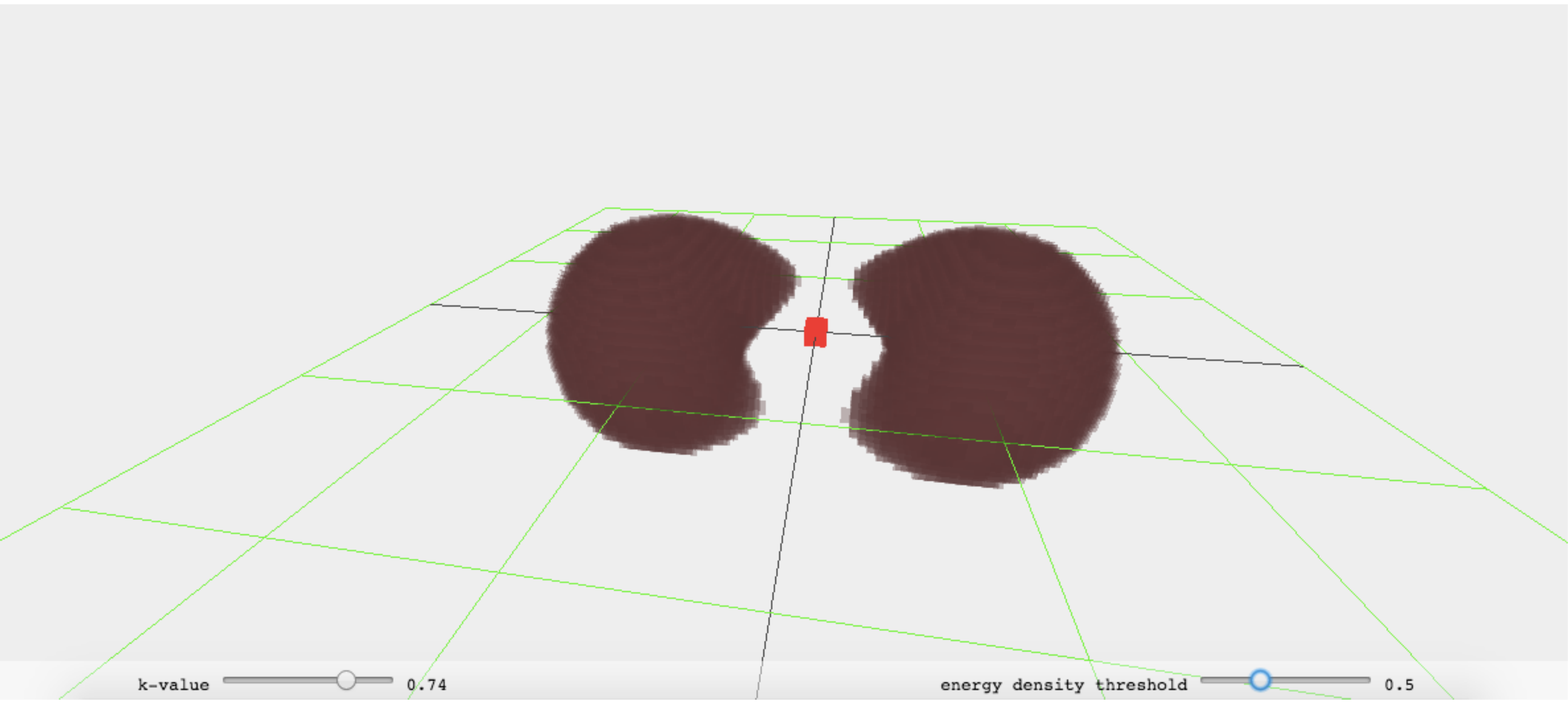}
\end{subfigure}
\begin{subfigure}[t]{0.49\linewidth}
\includegraphics[width=1.5\textwidth]{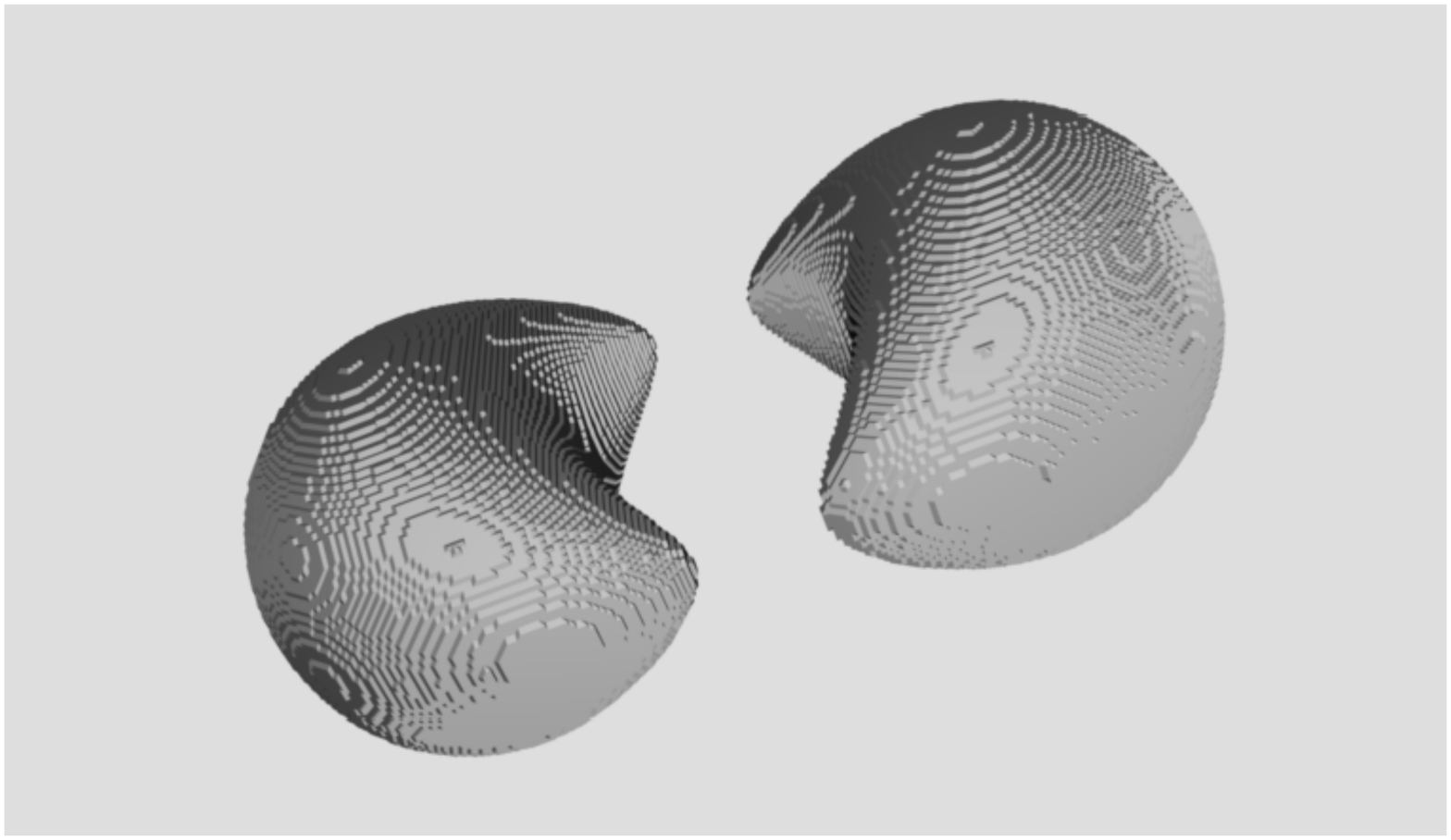}
\end{subfigure}
\caption{Two interactive visualisers: the first represents energy density   by opacity while the second uses
energy density to give a threshold producing a solid above the
threshold. Here $k=0.74$ and the
energy density threshold is $0.5$. } 
 \label{bothvisualisers}
\end{figure}

\begin{figure}
\centering
\begin{subfigure}[t]{0.49\linewidth}
\includegraphics[width=1.3\textwidth]{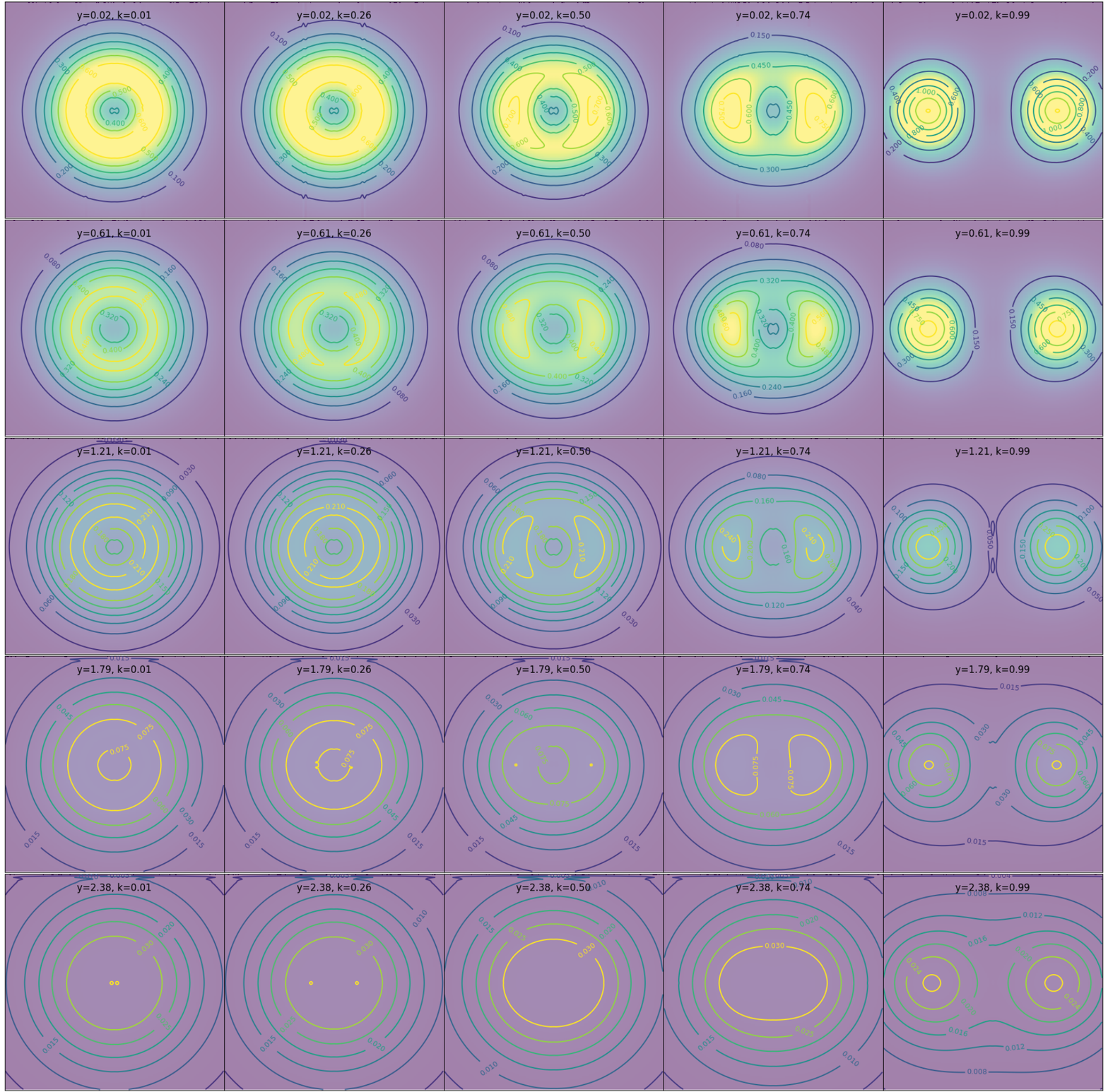}
\end{subfigure}
\begin{subfigure}[t]{0.49\linewidth}
\includegraphics[width=1.3\textwidth]{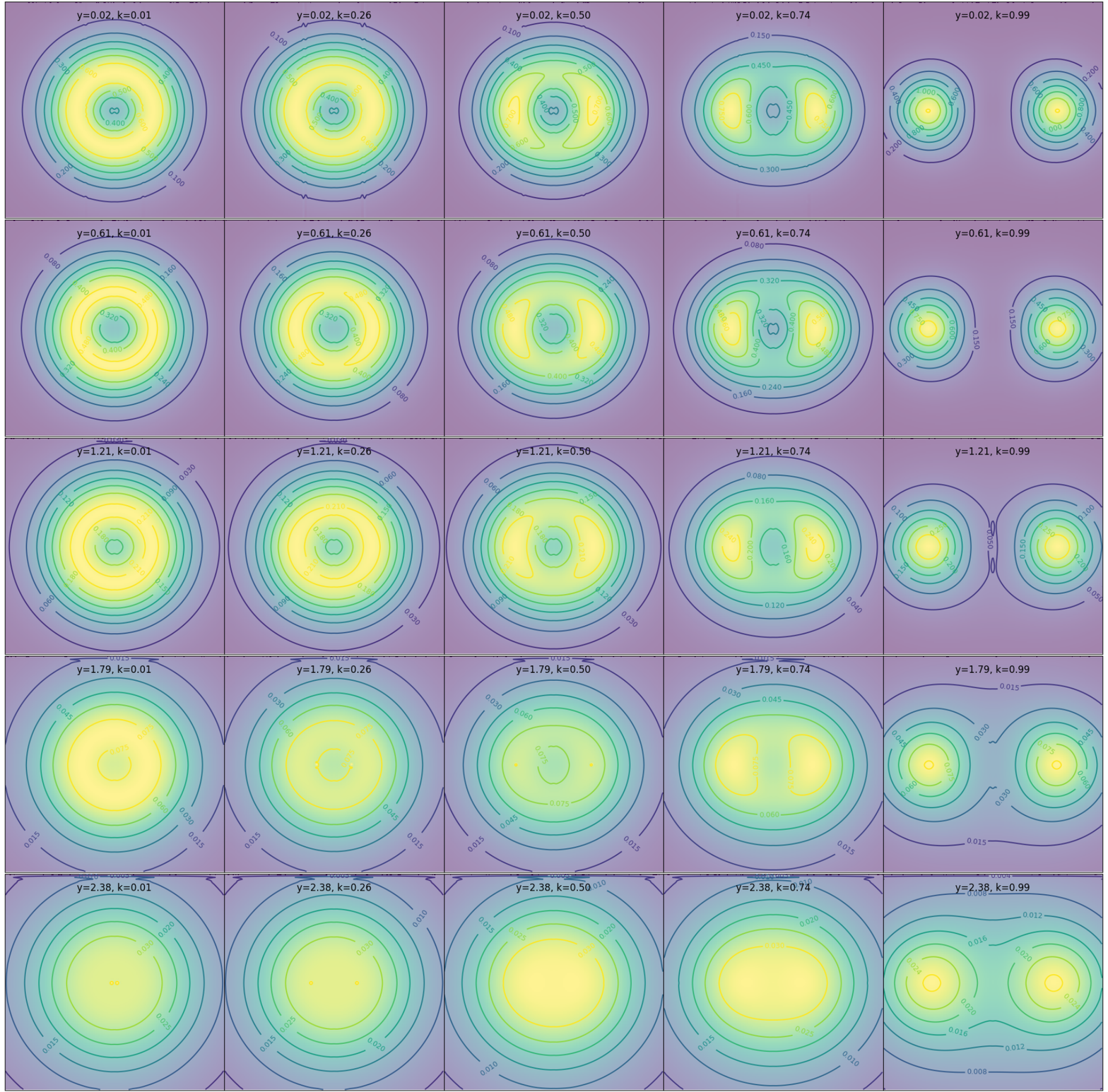}
\end{subfigure}
\caption{A Monopole Tomogram with (a) uniform colouring and (b) nonuniform colouring.} 
 \label{bothtoms}
\end{figure}

\section{Background}\label{background}

To make this paper more self-contained we will elaborate a little on the points noted in the construction:
the ADHM construction and Panagopoulos formulae; the spectral curve and Hitchin's constraints;
the Ercolani-Sinha Baker-Akhiezer function for the curve; and Nahm's lesser known ansatz. Here we will 
simply cite the critical formulae.

The field equations for the three dimensional Yang-Mills-Higgs Lagrangian with gauge group $SU(2)$
\[L= -\frac12 \mathrm{Tr}\, F_{ij} F^{ij}+\mathrm{Tr}\,D_{i}\Phi\, D^{i}\Phi,\]
are
\begin{equation} D_i\Phi= \frac12 \sum_{j,k=1}^3\epsilon_{ijk}
F_{jk},\quad i=1,2,3.\label{bogomolny}\end{equation} 
Here  $\Phi$ is the Higgs field, $
F_{ij}=\partial_i A_j -\partial_j A_i+[A_i,A_j]$ is the
curvature of the (spatial) connection of the gauge field
$A_i(\boldsymbol{x})$ and $D_i$  the covariant derivative
$D_i\Phi=\partial_i\Phi+[A_i,\Phi]$,
$\boldsymbol{x}=(x_1,x_2,x_3)\in \mathbb{R}^3$.
These equations may be viewed as a reduction of the self-dual Yang Mills equations to three dimensions under the assumption that all fields are independent of time. Configurations minimizing the energy of the system are given by the {\it Bogomolny equation} (\ref{bogomolny}).
 A solution with the boundary conditions
\[\left. \sqrt{-\frac12 \mathrm{Tr}\,\Phi(r)^2}\right|_{r\rightarrow\infty}\sim
1-\frac{n}{2r}+O(r^{-2}),\quad r=\sqrt{x_1^2+x_2^2+x_3^2} ,\] 
is called a {\it monopole} of  charge $n$.
The aim is to construct the Higgs and gauge field satisfying the Bogomolny equation and this
boundary condition.

\subsection{The ADHM construction and Panagopoulos formulae}
Nahm, in modifying 
the Atiyah-Drinfeld-Hitchin-Manin (ADHM) construction of instanton solutions to the (Euclidean) self-dual Yang-Mills equations, introduced the operator
\begin{align}\label{defdelta}
\Delta&=\im \dfrac{d}{dz}+x_4-\im T_4+\sum_{j=1}\sp{3}   \sigma_j\otimes( T_j+\im x_j 1_n),
\end{align}
where the $T_j(z)$ are $n\times n$ matrices and $\sigma_j$ the Pauli matrices. Here $n$ is the charge
of the $su(2)$ monopole.
Following the instanton construction the operator $\Delta\sp\dagger\Delta$  must commute with quaternions which happens if and only if ${T_i}\sp\dagger=-T_i$, $T_4\sp\dagger =-T_4$ and
\begin{equation}\label{fullnahm}
\dot{{T}_i} =[T_4,T_i]+\frac12\sum_{j,k=1}^3\epsilon_{ijk}[T_j(z),T_k(z)].
\end{equation}
Equations (\ref{fullnahm}) are known as Nahm's equations; one often encounters them in the more familiar gauge with $T_4=0$. When $\Delta\sp\dagger\Delta$ commutes\footnote{Throughout the superscript $\dagger$ means conjugated and transposed. We will at times emphasise the vectorial nature of an object by 
 printing this in bold,  e.g for vector
$\boldsymbol{a}^{\dagger}=\overline{\boldsymbol{a}}^T$. }
 with quaternions it is a positive operator; in particular this means that $\left(\Delta\sp\dagger\Delta\right)(z)$ is an invertible operator
and consequently $\Delta$ has no zero modes. 
The ADHM construction further requires $\Delta$ to be quaternionic linear, which means that
$T_i(z)=-{\overline{T}}_i(-z)$,  $T_4(z)=-{\overline{T}}_4(-z)$.
To describe monopoles the matrices  $T_j(z)$ are further required to be regular for  $z\in(-1,1)$ and have simple poles at $z=\pm1$, the residues of which define an
irreducible $n$-dimensional representation of the $su(2)$ algebra. Hitchin's analysis
\cite{hitchin_83}[\S2]  of the equation 
$\Delta\sp\dagger \boldsymbol{v}=0$ tells us this
has two normalizable solutions and it is in terms of these that  the
Atiyah-Drinfeld-Hitchin-Manin-Nahm (ADHMN) construction gives the gauge and Higgs field solutions.
\begin{theorem} [\bf ADHMN] \label{ADHMN} The charge $n$
monopole solution of the Bogomolny equation (\ref{bogomolny}) is given by
\begin{align}
\Phi_{ab}(\boldsymbol{x})&=\im\int_{-1}^1\mathrm{d}z\,z \boldsymbol{v}_a\sp
\dagger(z,\boldsymbol{x})\boldsymbol{v}_b(z,\boldsymbol{x})
,\quad a,b=1,2,\label{higgs}\\
A_{i\, ab}(\boldsymbol{x})&=\int_{-1}^1\mathrm{d}z\, \boldsymbol{v}_a\sp\dagger(z,\boldsymbol{x})
\frac{\partial}{\partial x_i} \boldsymbol{v}_b(z,\boldsymbol{x}),
\quad
i=1,2,3,\quad a,b=1,2.\label{gauges}
\end{align}
Here the two ($a=1,2$) $2n$-column vectors ${{\boldsymbol{v}}}_{a}(z,\boldsymbol{x})
=({v}_1^{(a)}(z,\boldsymbol{x}),\ldots,
{v}_{2n}^{(a)}(z,\boldsymbol{x}))^T$ form an orthonormal basis on
the interval $z\in[-1,1]$
\begin{equation}\int_{-1}^1\mathrm{d}z\, \boldsymbol{v}_a\sp\dagger(z,\boldsymbol{x})
\boldsymbol{v}_b(z,\boldsymbol{x})
=\delta_{ab}
,\label{norm}
\end{equation}
for the normalizable solutions to the Weyl equation
\begin{align}\Delta\sp\dagger \boldsymbol{v}=0,\label{deltadagger}
\end{align}
where 
\begin{align}
\Delta^{\dagger}&=\im \dfrac{d}{dz}+x_4-\im T_4-\sum_{j=1}\sp{3}   \sigma_j\otimes( T_j+\im x_j 1_n).
\label{weylequ}\end{align} 
The normalizable solutions form a two-dimensional subspace of the full $2n$-dimensional solution space
to the formal adjoint equation (\ref{deltadagger}).
The $n\times n$-matrices $T_j(z)$, $T_4(z)$, called Nahm data,  satisfy
Nahm's equation (\ref{fullnahm}) and the $T_j(z)$ are
required to be regular for  $z\in(-1,1)$ and have simple poles at $z=\pm1$, the residues of which define an
irreducible $n$-dimensional representation of the $su(2)$ algebra;
further
\begin{equation}
T_i(z)=-T_i^{\dagger}(z),\quad
T_4(z)=-T_4^{\dagger}(z),\quad
T_i(z)=T_i^{T}(-z),\quad
T_4(z)=T_4^{T}(-z)
.\label{constraint}\end{equation}
\end{theorem}

Although the integrations in (\ref{higgs}, \ref{gauges}) look intractable work of Panagopoulos enables their evaluation. Define the Hermitian matrices
\begin{equation}\label{pandefs}
\mathcal{H}=-\sum_{j=1}^3 x_j \sigma_j\otimes 1_n,\qquad 
\mathcal{F}=\imath \sum_{j=1}^3 \sigma_j\otimes T_j, \qquad
\mathcal{Q}=\frac{1}{r^2} \mathcal{H}\mathcal{F}\mathcal{H}-\mathcal{F}.
\end{equation}
Then
\begin{proposition}[Panagopoulos \cite{panagopo83, Braden2018b}]
 \label{panagopoulos} 
\begin{align}\label{pannorm}
 \int \mathrm{d}z\, \boldsymbol{{v}}_a\sp\dagger \boldsymbol{{v}}_b
&= \boldsymbol{{v}}_a\sp\dagger \mathcal{Q}^{-1}
\boldsymbol{{v}}_b.\\
\int \mathrm{d}z\, z \boldsymbol{{v}}_a\sp\dagger \boldsymbol{{v}}_b
&=
\boldsymbol{{v}}_a\sp\dagger\mathcal{Q}^{-1} \left( z+\mathcal{H}\,\frac{x_{i}}{r^2}\frac{\partial}{\partial x_{i}}   \right)\boldsymbol{{v}}_b.
\label{panhiggs}\\
\int \boldsymbol{{v}}_a\sp\dagger\frac{\partial}{\partial
x_i}\boldsymbol{{v}}_b \mathrm{d}z
&=\boldsymbol{{v}}_a\sp\dagger\mathcal{Q}^{-1} \left[
\frac{\partial}{\partial x_i}+\mathcal{H}\frac{z}{r^2}\,
x_i+\mathcal{H}\frac{\imath}{r^2} \left(
\boldsymbol{x}\times\boldsymbol{\nabla}\right)_i \right]\boldsymbol{{v}}_b.\label{pangauge}
\end{align}
\end{proposition}

At this stage we see that to reconstruct the gauge and Higgs fields we need knowledge of the normalizable solutions to $\Delta\sp\dagger \boldsymbol{v}=0$ at the endpoints $z=\pm1$. We will construct a fundamental matrix $V=(\boldsymbol{v}_1,\ldots,\boldsymbol{v}_{2n})$  of solutions to this equation
and then extract the normalizable solutions using a ($2n\times 2$ matrix) projector $\mu$
$$V\mu=(\boldsymbol{v}_1,\boldsymbol{v}_2).$$
The work of \cite{Braden2018b} shows that $\mu$ is $z$-independent and so may be removed from 
the integrals; thus for example the matrix 
$$
\left(  \int \mathrm{d}z\, \boldsymbol{{v}}_a\sp\dagger \boldsymbol{{v}}_b\right)
=\int dz\,\mu\sp{\dagger}V\sp{\dagger}V\mu=
\mu\sp{\dagger}\left(\int dz\,V\sp{\dagger}V\right)\mu
= \mu\sp{\dagger} V\sp{\dagger}\mathcal{Q}^{-1} V \mu
.$$
We also note a further result of \cite{Braden2018b} that will prove useful:
\begin{theorem}\label{constancythm} With the notation above, and for $W=(V^\dagger)^{-1}$
\begin{equation}
\left(V\sp\dagger \mathcal{Q}^{-1}\mathcal{H} V\right)(z)=\textrm{constant},
\qquad
\left(W\sp\dagger \mathcal{Q}\mathcal{H} W\right)(z)=\textrm{constant}.
\label{constancy}
\end{equation}
\end{theorem}

Towards constructing the fundamental matrix $V$  we next turn to the spectral curve.

\subsection{The Spectral Curve and Hitchin's constraints} One may readily associate an integrable system and spectral curve to Nahm's equations. Hitchin's seminal work \cite{hitchin_83} provided a geometric setting for this, the global geometry yielding necessary and sufficient conditions for such to be a monopole spectral curve. Here we will recall the salient features.

Upon setting (with ${T_i}\sp\dagger=-T_i$, $T_4\sp\dagger =-T_4$)
\begin{equation*}
\alpha=T_4+\im T_3,\quad \beta = T_1+iT_2,
\quad
L=L(\zeta):=\beta -(\alpha+\alpha\sp\dagger)\zeta-\beta\sp\dagger \zeta^2, \quad
M=M(\zeta):=-\alpha-\beta\sp\dagger \zeta,
\end{equation*}
one finds
\begin{equation}
\begin{split}\label{integrability}
\dot{{T}_i} =[T_4,T_i]+\frac12\sum_{j,k=1}^3\epsilon_{ijk}[T_j(z),T_k(z)]
&\Longleftrightarrow
\dot L=[L,M]\\
&\Longleftrightarrow\quad
\left\{\begin{aligned}
\left[\dfrac{d }{dz}-\alpha,\beta\right]&=0,\\
\dfrac{d (\alpha+\alpha\sp\dagger)}{dz}&=[\alpha,\alpha\sp\dagger]+[\beta,\beta\sp\dagger].
\end{aligned}\right.
\end{split}
\end{equation}
Focussing on the first equivalence, 
Nahm's equations may be expressed as a Lax pair, to which we may associate the
spectral curve $\mathcal{C}$ given by 
\begin{equation}
P(\zeta,\eta):=\det(\eta-L(\zeta))=\eta^n+a_1(\zeta)\eta^{n-1}+\ldots+a_n(\zeta)=0, \quad
\mathrm{deg}\, a_k(\zeta) \leq 2k.
\label{curve1}
\end{equation}
The genus of $\mathcal{C}$ is $g=(n-1)^2$. Hitchin's construction shows that the spectral curve
naturally lies in mini-twistor space\footnote{
If we set $\boldsymbol{y}=\left(\frac{1+\zeta^2}{2\im},\frac{1-\zeta^2}2,-\zeta \right)\in \mathbb{C}\sp3$,
then $\boldsymbol{y}\cdot \boldsymbol{ y}=0$ and $\boldsymbol{ y}\cdot \boldsymbol{\overline y}={(1+|\zeta|\sp2)\sp2}/2$. Thus with $\boldsymbol{T}=(T_1,T_2,T_3)$,
$\boldsymbol{x}=(x_1,x_2,x_3)$, then 
$$L(\zeta):=2\im \boldsymbol{ y}\cdot\boldsymbol{T}=
(T_1+\im T_2)-2\im T_3\,\zeta+(T_1-\im T_2)\,\zeta\sp2.$$ 
We have $\eta=2\, \boldsymbol{ y}\cdot\boldsymbol{x}$. 
} $T\PP\sp1$, the space of lines in $\mathbb{R}\sp3$. The spectral curve is an algebraic curve $\mathcal{C} \subset T\PP\sp1$.
If $\zeta$ is the inhomogeneous coordinate on the Riemann sphere then $(\zeta,\eta)$ are the standard local coordinates on $T\PP\sp1$ defined by $(\zeta,\eta)\rightarrow\eta\frac{d}{d\zeta}$. 
The mini-twistor correspondence relates $(x_1,x_2,x_3)\in\mathbb{R}\sp3$ with $(\zeta,\eta)$ by
\begin{equation}
\eta=(x_2-\im x_1)- 2 x_3\,\zeta-(x_2+\im x_1)\,\zeta^2.
\label{AtiyahWard}
\end{equation}
The anti-holomorphic involution 
\begin{equation}\label{invol}
\mathfrak{J}:\, (\zeta,\eta)\rightarrow
(-\frac1{\overline\zeta},-\frac{\overline\eta}{{\overline\zeta}\sp2}),
\end{equation}
which takes a point on $\PP\sp1$ to its antipodal point reversing the orientation of a line, endows $T\PP\sp1$ with its standard real structure.
The hermiticity properties of the Nahm matrices mean that $\mathcal{C}$ is invariant under $\mathfrak{J}$.

If the homogeneous coordinates of $ \PP\sp1$ are $[\zeta_0,\zeta_1]$ we
consider the standard covering of this by the open sets
$U_0=\{[\zeta_0,\zeta_1]\,|\,\zeta_0\ne0\}$ and
$U_1=\{[\zeta_0,\zeta_1]\,|\,\zeta_1\ne0\}$, with
$\zeta=\zeta_1/\zeta_0$ the usual coordinate on $U_0$. 
We denote by $\hat U_{0,1}$ the pre-images of these sets under the
projection map $\pi:T\PP\sp1\rightarrow\PP\sp1$. Let
 $L^{\lambda}$ denote the holomorphic line bundle on
$T\PP\sp1$ defined by the transition function
$g_{01}=\rm{exp}(-\lambda\eta/\zeta)$ on $\hat U_{0}\cap \hat
U_{1}$, and let $L^{\lambda}(m)\equiv
L^{\lambda}\otimes\pi\sp*\mathcal{O}(m)$ be similarly defined in
terms of the transition function
$g_{01}=\zeta^m\exp{(-\lambda\eta/\zeta)}$. A holomorphic section
of such line bundles is given in terms of holomorphic functions
$f_\alpha$ on $\hat U_\alpha$ satisfying
$f_\alpha=g_{\alpha\beta}f_\beta$. We denote line bundles on
$\mathcal{C}$ in the same way, where now we have holomorphic
functions $f_\alpha$ defined on $\mathcal{C}\cap\hat U_\alpha$.
Hitchin's conditions for a monopole spectral curve are:
\begin{itemize}
\item[\bf{H1}] Reality conditions: $\mathcal{C}$ is invariant under $\mathfrak{J}$,
$\displaystyle{
a_k(\zeta)=(-1)^k\zeta^{2k}\overline{a_k(-1/{\overline{\zeta}}\,)}}$.

\item[\bf{H2}] 
$L^2$ is trivial on ${\mathcal{C}}$ and $L\sp1(n-1)$ is real.

\item[\bf{H3}] $H^0({\mathcal{C}},L^{\lambda}(n-2))=0$ for  $\lambda\in(0,2)$.
\end{itemize}

\subsection{The mini-twistor correspondence and the Abel-Jacobi map}\label{minitwistorabjac}
Given the mini-twistor correspondence (\ref{AtiyahWard}) and the spectral curve (\ref{curve1}), a point $\boldsymbol{x}\in\mathbb{R}\sp3$ yields an equation of degree $2n$ in $\zeta$ and
 gives us $2n$ points $P_i=(\zeta_i,\eta_i)$ (perhaps with multiplicity) on the curve $\mathcal{C}$.
 We will refer to this degree $2n$ equation in $\zeta$ and $\boldsymbol{x}$ as the Atiyah-Ward equation
 (it having appeared in their work).
 As both the curve and the correspondence satisfy the antiholomorphic involution $\mathfrak{I}$,
so to do the solutions and we may choose an ordering such that
\begin{equation}
\label{orderining}
 P_{i+n}= \mathfrak{I}(P_i).
 \end{equation}
Recall that the fundamental matrices $V$ and $W$ were $2n\times 2n$; the points $P_i$ will be used to label the columns of $W$.

The points $\{P_i\}$ satisfy a number of relations. Let $\phi(P)=\int_{P_0}\sp{P}\boldsymbol{v}$
denote the Abel-Jacobi map for our curve $\mathcal{C}$ and for a choice of suitably normalized
holomorphic differentials $\boldsymbol{v}$, Abel's theorem says that
$\sum_{\mathfrak{p}\in  \div(w)} \phi(\mathfrak{p})$ lies in the period lattice $\Lambda$ for any function $w$.
Consider first  the function 
$w(P)=-\eta+(x_2-\imath x_1)-2\zeta x_3-(x_2+\imath x_1)\zeta^2$
on $\mathcal{C}$ which has divisor
\[ \div(w)= P_1+\ldots+P_{2n}-2(\infty_1+\ldots+\infty_n).\]
Then
\begin{equation}
 \sum_{i=1}\sp{2n}\int_{P_0}\sp{P_i}\boldsymbol{v} - 2\sum_{i=1}\sp{n}\int_{P_0}\sp{\infty_i}
\boldsymbol{v} 
\in \Lambda.
 \end{equation}
Further identities arise by taking
$f(P):=\int_{P_0}^P\gamma$
for some (possibly meromorphic) differential $\gamma$ and appropriate functions $w(P)$; a dissection of $\mathcal{C}$
along a canonical homology basis $\{\mathfrak{a}_i,\mathfrak{b}_i\}_{i=1}\sp{g}$
(suitably avoiding poles) yields
\begin{align}\nonumber
0=\frac1{2 i\pi}\int_{\mathcal{C}}d\,f(P)\wedge d\,\ln w(P)&=
\frac1{2 i\pi}\int_{\partial \mathcal{C}} f(P)\,d\, \ln w(P)=
\sum_{{P}\in  \div(w)} \res\left( f({P})\,d\, \ln w({P})\right)  \\
&=
\sum_{j=1}\sp{g}\frac1{2 i\pi}\left[
\oint_{\mathfrak{a}_i} d\,f(P) \oint_{\mathfrak{b}_i} d\, \ln w(P)
- \oint_{\mathfrak{b}_i} d\,f(P) \oint_{\mathfrak{a}_i} d\, \ln w(P)
\right] \label{residues}
\end{align}

Generically the points $P_i$ are distinct with  non-generic points corresponding to points of bitangency of the spectral curve. There are typically a number of components to these loci and Hurtubise's study of the asymptotic behaviour of the Higgs field \cite{hurtubise85a} discussed one of these .

\subsection{The Ercolani-Sinha Construction}
Ercolani and Sinha \cite{ercolani_sinha_89} sought to use integrable systems techniques to solve Nahm's equations by solving
\begin{align}
(L-\eta)U&=0,
\nonumber \\
\left[  \frac{d}{dz}+M\right]U&=0.\label{mscatter}
\end{align}

To understand the Ercolani-Sinha results its useful to focus on the second of the equivalences of (\ref{integrability}) which expresses Nahm's equations in the form of a complex and a real equation (respectively) \cite{donaldson84}. The complex Nahm
equation is readily solved,
\begin{equation}
\beta g = g \nu, \quad
\left( \frac{d}{dz}-\alpha\right)g=0 
\Longleftrightarrow
\beta = g\nu g\sp{-1} ,\quad  \alpha=\dot g g\sp{-1},
\end{equation}
where $\nu$ is constant and generically diagonal, $\nu=\diag(\nu_1,\ldots,\nu_n)$;  by 
conjugating\footnote{ $\tilde\beta =g(0)\sp{-1}\beta g(0)$,  $\tilde g(z) =g(0)\sp{-1} g(z)$,
$\tilde\alpha =g(0)\sp{-1}\alpha g(0)$.}  by
the constant matrix $g(0)$ we may assume $\beta(0)=\nu$ and $g(0)=1_n$.
The meaning of $\nu$ follows from the equation of the curve (\ref{curve1}).
For large $\zeta$ we see that 
$\det(\eta/\zeta^2-L/\zeta^2)\sim \prod_{i=1}\sp{n}(\eta/\zeta^2+\nu_i\sp\dagger)$ and so 
$\eta/\zeta\sim -\nu_i\sp\dagger\zeta$; we shall denote by $\{\infty_i\}_{i=1}\sp{n}$ the preimages of $\zeta=\infty$ with this behaviour.
The real equation is more difficult; Donaldson  proved the existence of this equation in the monopole context in \cite{donaldson84}. Upon setting
\begin{equation}\label{defh}
h=g\sp\dagger g
\end{equation}
 then
\begin{equation}\label{defdhh}
\dot h h\sp{-1} =g\sp{\dagger }(\alpha+\alpha\sp\dagger)g\sp{\dagger \,-1}, \qquad h(0)=1_n,
\end{equation}
and the real equation yields the (possibly) nonabelian Toda equation
\begin{equation}\label{nonabtoda}
\frac{d}{dz} \left( \dot h h\sp{-1} \right)= \left[ h\nu h\sp{-1}, \nu\sp\dagger\right].
\end{equation}

Now (\ref{mscatter})  is not a standard scattering equation, but upon setting
$U=g\sp{\dagger\,-1}\Phi$ 
we may use the complex equation to transform (\ref{mscatter}) into a standard scattering equation
for $\Phi$,
\begin{equation}\label{standardscattering}
 \left[  \frac{d}{dz}-g\sp{\dagger }(\alpha+\alpha\sp\dagger)g\sp{\dagger \,-1}\right]\Phi=
\zeta \nu\sp\dagger\Phi.
\end{equation}
\lq\lq Standard\rq\rq\ here simply means that the matrix  $\zeta \nu\sp\dagger$
on the right-hand side is $z$-independent.
In terms of $h$ we have (\ref{defdhh}) and
\begin{align}
g\sp{\dagger}Lg\sp{\dagger\,-1}=h\nu h\sp{-1}-\dot h h\sp{-1}\zeta-\nu\sp\dagger\zeta^2.
\label{conjL}
\end{align}
The point to note is that we can solve the standard scattering equation
(\ref{standardscattering}) explicitly in terms of the function theory of $\mathcal{C}$ by what is known
as a Baker-Akhiezer function  \cite{krichever977a}. If $\Phi:=(\Phi_i)$ (with $i$ labelling the rows) the
Baker-Akhiezer function is uniquely specified by requiring the behaviour 
\begin{equation}\label{BAnorm}
\lim_{ P=P(\zeta,\eta)\rightarrow \infty_j  }\Phi_i(z, P)e^{z \frac{\eta}{\zeta}(P) }=
\lim_{ P=P(\zeta,\eta)\rightarrow \infty_j  }\Phi_i(z, P)e^{-z \zeta \nu_j\sp\dagger}=\delta_{ij}
\end{equation}
and that $\Phi(z,P)$
is meromorphic for $P\in\mathcal{C}\setminus\{\infty_1,\ldots,\infty_n\}$ with poles at a suitably generic
degree $g+n-1$ divisor $\delta$. If 
$\widehat\Phi (z,\zeta)$ is the $n\times n$ matrix
whose columns\footnote{
$
\widehat{\Phi}(z,\zeta):= ( \boldsymbol{\Phi}_1(z,P_1),\ldots,  
\boldsymbol{\Phi}_n(z,P_n)     )$, 
where $P_i=(\zeta,\eta_i)$. 
} are the Baker-Akhiezer functions for the preimages of $\zeta$,
then $g\sp{\dagger}Lg\sp{\dagger\,-1}=\widehat\Phi\diag(\eta_1\ldots,\eta_n)\widehat\Phi\sp{-1}$.
Thus the Baker-Akhiezer function enables us to solve for (the gauge transform)
$g\sp{\dagger}Lg\sp{\dagger\,-1}$ and so too $\dot h h\sp{-1}$; Ercolani and Sinha 
\cite{ercolani_sinha_89} gave an expression
for $\dot h h\sp{-1}$.

The Baker-Akhiezer function may be explicitly constructed:  the asymptotics of the essential singularity of (\ref{standardscattering}) is
encoded by seeking an abelian differential $\gamma_{\infty}$ on the curve such that near $\infty_j$ ($j=1,\ldots,n$)
\[z \nu_j\sp{\dagger} \zeta\sim -z \eta/\zeta \sim z\left[\int\limits_{P_0}^P\gamma_{\infty}(P)-\tilde\nu_j\right] 
. \]
This behaviour defines a differential $\gamma_{\infty}$ of the second kind on $\mathcal{C}$ which
is unique if we further require the $\mathfrak{a}$-normalization
$\oint_{\mathfrak{a}_k}\gamma_{\infty}(P)=0$, ($k=1,\ldots,g$).
The constant $\tilde\nu_{j}$ here is defined by $\tilde\nu_{j}=\lim_{P\rightarrow \infty_j}
\left[\int\limits_{P_0}^P\gamma_{\infty}(P)+\frac{\eta}{\zeta}
\right]$.

In the Baker-Akiezer description the flow of line bundles given by Hitchin's  exponential transition functions
corresponds to a flow on the Jacobian of $\mathcal{C}$ in the direction of the winding vector
$\boldsymbol{U}$ of 
$\mathfrak{b}$-periods of the differential $\gamma_{\infty}(P)$,
\begin{equation}
\boldsymbol{U}=\frac{1}{2 \pi\imath} \left(
\oint_{\mathfrak{b}_1}\gamma_{\infty},\ldots,
\oint_{\mathfrak{b}_g}\gamma_{\infty}
\right)^T.
\end{equation}
This connection with Hitchin's monopole constraints comes from
\begin{lemma}[Ercolani-Sinha Constraints] The following are equivalent:
\begin{enumerate}[(i)]
 \item $L\sp2$ is trivial on $\mathcal{C}$.

\item There exists a 1-cycle
$\mathfrak{es}=\boldsymbol{n}\cdot{\mathfrak{a}}+
\boldsymbol{m}\cdot{\mathfrak{b}}$ such that for every holomorphic
differential $\Omega=\left({\beta_0\eta^{n-2}+\beta_1(\zeta)\eta^{n
-3}+\ldots+\beta_{n-2}(\zeta)}\right) d\zeta/({{\partial\mathcal{P}}/{\partial
\eta}})$,
\begin{equation}
\oint\limits_{\mathfrak{es}}\Omega=-2\beta_0,\label{HMREScond}
\end{equation}
\item
$\displaystyle{
2\boldsymbol{U}\in \Lambda\Longleftrightarrow
\boldsymbol{U}=\frac{1}{2\pi\imath}\left(\oint_{\mathfrak{b}_1}\gamma_{\infty},
\ldots,\oint_{\mathfrak{b}_g}\gamma_{\infty}\right)\sp{T}= \frac12
\boldsymbol{n}+\frac12\tau\boldsymbol{m} .}$

\end{enumerate}
\end{lemma}
Thus for a monopole spectral curve  we require that  $\boldsymbol{U}$ is a half-period.
Further, from $\mathbf{H3}$ the vector $\boldsymbol{U}$ should be {\it primitive}, i.e. $\lambda\boldsymbol{U}$ belongs to the period lattice $\Lambda$ if and only if $\lambda=0$ or $\lambda=2$. Although a general Baker-Akhiezer function depends on a generic divisor
$\delta$ the real structure demanded by Hitchin's $\mathbf{H2}$ imposes constraints on this. This reduces \cite{ercolani_sinha_89} to 
\begin{equation}\label{esrealst}
c_{ij}=-c_{ji},\qquad\text{where}\quad
c_{ij}:=\lim_{ P=P(\zeta,\eta)\rightarrow \infty_j  }\zeta\,\Phi_i(0, P).
\end{equation}
Braden and Fedorov \cite{Braden2008} show that these constraints may always be solved for.

It is worth clarifying what the Ercolani-Sinha construction does and does not yield. 
Given the Baker-Akhiezer function the construction yields the gauge transformed $T_i':=g\sp{\dagger}T_i g\sp{\dagger \,-1} $ which satisfy
$$\dot T_i' = [ \frac12 \dot h h\sp{-1}, T_i']+ [T_j',T_k']$$
(and cyclic). Here $T_3'= \frac12 \dot h h\sp{-1}$ and the $i=3$ equation becomes (\ref{nonabtoda}).
Although Ercolani and Sinha only solved for $\dot h h\sp{-1}$ the recent work of \cite{Braden2018b} shows how one
may obtain $h$. Thus the Ercolani-Sinha construction \emph{will} yield solutions of the Nahm equations, but not in the standard gauge with $T_4=0$. To obtain a solution of the Nahm equations with $T_4=0$ requires
\begin{equation}
\label{standardgauge}
\alpha=\alpha\sp\dagger \Longleftrightarrow h\sp{-1}\dot h =2 g\sp{-1}\dot g
\Longleftrightarrow  \dot h =2 g\sp{\dagger}\dot g=2\dot g\sp{\dagger} g,
\end{equation}
viewed as a differential equation for $g$ with $h$ specified; the solution for $g$ is only defined up to
left multiplication by a constant unitary matrix.
Although a solution exists we cannot as yet explicitly write one down; such however is not needed to 
solve for the gauge and Higgs fields.

To make connection with the notation of Ercolani and Sinha \cite{ercolani_sinha_89}  that we will at times employ, we record
\begin{equation}
\label{esnotation}
g\sp{\dagger \,-1}:=C,\qquad \nu\sp\dagger=-\diag(\rho_1,\ldots,\rho_n).
\end{equation}

\subsection{A lesser known ansatz of Nahm and the construction of \texorpdfstring{$V$}{V}}
It remains to give the fundamental matrix $V$ of solutions to $\Delta\sp\dagger \boldsymbol{v}=0$. 
The solution we follow is based on another ansatz of Nahm: we construct a fundamental solution $W$ to the
equation $\Delta \boldsymbol{w}=0$ and then take $V=\left(W\sp\dagger\right)\sp{-1}$.

\begin{theorem}[Nahm \cite{nahm82c}; the modification of \cite{Braden2018b}]
Let $\boldsymbol{\sigma}=(\sigma_1,\sigma_2,\sigma_3)$ and let  $\boldsymbol{\widehat{ u}}(\boldsymbol{x})$  be a unit vector independent of $z$. Let 
$|s>$ be an arbitrarily normalized spinor not in $\ker
(1_2+\boldsymbol{\widehat{ u}}(\boldsymbol{x})\cdot\boldsymbol{\sigma})$.
Then
\begin{equation}\label{Nansatz}
\boldsymbol{w}:=\boldsymbol{w}(\zeta)
=(1_2+\boldsymbol{\widehat{ u}}(\boldsymbol{x})\cdot\boldsymbol{\sigma})\,
e\sp{-\im z\left[(x_1-\im x_2)\zeta-\im x_3 -x_4\right]}|s>\otimes\,
U(z) 
\end{equation}
satisfies $\Delta \boldsymbol{w}=0$ if and only if
\begin{align}
0&=\left(L(\zeta)-\eta\right) U(z) ,\label{eqlax1}\\
0&=\left( \dfrac{\mathrm{d}}{\mathrm{d}z}+M(\zeta)
\right) U(z),\label{eqlax2} 
\end{align}
where
\begin{equation}
\eta=(x_2-\im x_1)-2 x_3\zeta-(x_2+\im x_1)\zeta^2,
\label{etadef}
\end{equation}
and $L(\zeta)$ and $M(\zeta)$, as above, satisfy the Lax equation $\dot L=[L,M]$.
\end{theorem}
Although early workers sought to explicitly perform these integrations we may use the connection with integrable systems to solve 
\begin{equation}\label{UBA}
U(z) =g\sp{\dagger\,-1}\,\Phi
\end{equation}
in terms of the earlier (and unknown) gauge transformation $C:=g\sp{\dagger\,-1}$ and the Baker-Akhiezer function $\Phi$. We may write the $k$-th column to the fundamental matrix $W =\left(\boldsymbol{w}\sp{(k)}(x.z)\right)$ of $\Delta \boldsymbol{w}=0$ as
\begin{align}
\boldsymbol{w}\sp{(k)}(z,x)&=(1_2+\boldsymbol{\widehat{ u}}(\zeta)\cdot\boldsymbol{\sigma})\,
e\sp{-i z\left[(x_1-i x_2)\zeta-i x_3 -x_4\right]}|s>\otimes\,
C(z)\boldsymbol{\Phi}(z,P_k)
\nonumber
\\
&=\left(1_2\otimes C(z)\right) \left(
(1_2+\boldsymbol{\widehat{ u}}(\zeta)\cdot\boldsymbol{\sigma})\,
e\sp{-i z\left[(x_1-i x_2)\zeta-i x_3 -x_4\right]}|s>\otimes\,
\boldsymbol{\Phi}(z,P_k)
\right) 
\label{nahmW}
\\
&:=\left(1_2\otimes C(z)\right)\,\boldsymbol{\widehat{w}}\sp{(k)}(z,x)\nonumber
\end{align}
and $P_k=(\zeta_k,\eta_k)$ are the $2n$ solutions to the mini-twistor correspondence described earlier.
 These $2n$ points come in
 $n$ pairs of points related by the antiholomorphic involution $\mathfrak{J}$. To each point we have the associated values $\boldsymbol{\widehat{u}}(\zeta_j)$ and for each of these we solve for  $U(z)$ yielding a $2n\times 1$ matrix $\boldsymbol{w}(P_j)$. Taking each of the $2n$ solutions we obtain a $2n\times 2n$ matrix of solutions $W$. As noted earlier, there may be non-generic points for which $\zeta_i=\zeta_j$
correspond to points of bitangency of the spectral curve; at such points we
modify this discussion by taking a derivative $\boldsymbol{w}'(P_j)$.

The equation for $\Delta \boldsymbol{w}=0$ may be rewritten as
\begin{equation}\label{transw}
0=\left[ \frac{d}{dz}-\mathcal{H}-\mathcal{F}' -\frac12 \dot h h\sp{-1}\,1_2\right]\boldsymbol{\widehat{w}}
\end{equation}
where we have used (\ref{standardgauge}) and set
\[
 \mathcal{F}'=
 \imath \sum_{j=1}^3 \sigma_j\otimes  g\sp{\dagger} T_j g\sp{\dagger\,-1}
 =\begin{pmatrix} \frac12 \dot h h\sp{-1} & -\imath \nu\sp\dagger\\
 \imath h\nu h\sp{-1} &-\frac12 \dot h h\sp{-1}
 \end{pmatrix}.
\]
Both $h$ and the solution $\boldsymbol{\widehat{w}}$ of (\ref{transw}) are determined entirely in terms of
the Baker-Akhiezer function and the only unknown in the account at this stage is the gauge transform
$g$. But as shown in \cite{Braden2018b}, this unknown gauge transform combines in all of the integrals in
Proposition \ref{panagopoulos} into the known $h(z)$:
\[
\mu\sp\dagger V\sp\dagger \mathcal{Q}^{-1} \mathcal{O}V\mu
=
\mu\sp\dagger  {\widehat V}\sp\dagger\left[ \left(1_2\otimes g\sp\dagger\right) \mathcal{Q}^{-1} \left(1_2\otimes g
\right) \right] \mathcal{O}{\widehat V}\mu
=
\mu\sp\dagger  {\widehat V}\sp\dagger \mathcal{Q}'\sp{-1}
\left(1_2\otimes h\right) \mathcal{O}{\widehat V}\mu
.
\]
Here $\mathcal{O}$ is one of the operators appearing on the right-hand side of Proposition \ref{panagopoulos} and (using the definitions (\ref{pandefs}))
\[
\mathcal{Q}'=
 \left(1_2\otimes g\sp{\dagger}\right) \mathcal{Q} \left(1_2\otimes g\sp{\dagger\,-1}\right) 
 :=
\frac{1}{r^2} \mathcal{H}\mathcal{F}'\mathcal{H}-\mathcal{F}',
\qquad
 \mathcal{F}'=\imath \sum_{j=1}^3 \sigma_j\otimes  g\sp{\dagger} T_j g\sp{\dagger\,-1}
.
\]
The conclusion is that we may reconstruct the gauge and Higgs fields from just a knowledge of the Baker-Akhiezer function.

\subsection{Remarks}
At this stage we have presented the ingredients needed to reconstruct the gauge and Higgs fields apart from
the general construction of the Baker-Akhiezer function.

Although we only need an expansion of the solutions of $V$ at the end points $z=\pm 1$
to reconstruct the gauge theory data and we have in fact described the solution $V(z)$ for all $z$. 
The asymptotic behaviour of the Nahm matrices given by the ADMHN theorem tells us that
 $T_j(z)$ expanded in the vicinity of the end point $z=1-\xi$ behaves as
\begin{equation*}
T_j(1-\xi)=-\im\frac{l_j}{\xi}+O(1),\quad j=1,2,3,
\end{equation*}
where (the Hermitian) $l_j$  define the irreducible $n$-dimensional representation of the $su(2)$ Lie algebra,
$ [l_j,l_k]=\imath\, \epsilon_{jkl}\, l_l $.
Then (\ref{deltadagger})  behaves in the vicinity of the pole as 
\begin{equation}
\left[\frac{\mathrm{d}}{\mathrm{d} \xi}-\frac1{\xi}\left(\sum_{j=1}^3 \sigma_j\otimes l_j \right)+
\left(\sum_{j=1}^3 \sigma_j \otimes x_j 1_2 \right)  +\mathcal{O}(1)\right]\boldsymbol{v}(1-\xi,\boldsymbol{x})=0.\label{approxdelta1}
\end{equation}
One can show (see for example \cite{weinbyi06}) that $\sum_{j=1}^3  \sigma_j\otimes l_j$ has only two distinct eigenvalues, $\lambda_a=(n-1)/2$ with multiplicity $n+1$ and $\lambda_b=-(n+1)/2$
with multiplicity $n-1$.
If $\boldsymbol{a}_i$ are eigenvectors associated with $\lambda_a$ ($i=1,\ldots,n+1$), and $\boldsymbol{b}_j$ eigenvectors associated with $\lambda_b$ ($j=1,\ldots,n-1$), then  (\ref{approxdelta1}) has solutions
$\boldsymbol{v}=\xi^{\lambda_a}\boldsymbol{a}_i+\ldots$ and $\boldsymbol{v}=\xi^{\lambda_b}\boldsymbol{b}_j+\ldots$.
Therefore normalizable solutions must lie in the subspace with positive $\lambda_a=(n-1)/2$ and so
we require that $\boldsymbol{v}(1,\boldsymbol{x})$ is orthogonal to the subspace with eigenvalue $-(n+1)/2$, i.e.
\[ \lim_{z\to 1\sp-}\boldsymbol{v}(z,\boldsymbol{x})^T\cdot\boldsymbol{b}_j
=0,\quad j=1,\ldots,n-1.  \]
These $n-1$ conditions coming from the behaviour at $z=1$ thus yield a $n+1$ dimensional space of
solutions to $\Delta\sp\dagger \boldsymbol{v}=0$. A similar analysis at $z=-1$ again yields a further $n-1$ constraints resulting in two normalisable solutions on the interval. Now although local analysis at each of the end points lets us construct normalizable solutions, the difficulty is in relating normalizable solutions at both
ends: $V(z)$ does this for us while numerically this has been done via shooting methods
(see \cite{houghton_sutcliffe_tetrahedral} for an algebraic implementation of these).

\section{Basic Properties of the Spectral Curve}\label{sectioncurve}%

\subsection{The Curve}

The spectral curve $\mathcal{C}$ for $n=2$ was constructed by Hurtubise \cite{hurtubise_83} and we shall employ the Ercolani-Sinha \cite{ercolani_sinha_89} choice of homology basis (see Fig. \ref{fig:homology}) and form of the curve,
\begin{equation}
0=\eta^2+\frac{K^2}4\left( \zeta^4+2(k^2-k'^2)\zeta^2+1\right),   \label{curve}
\end{equation}
where $k'\sp2=1-k^2$ and ${K}={K}(k)$ is a complete elliptic integral\footnote{
$${K}(k)=\int_0\sp{\pi/2}\frac{du}{\sqrt{1-{k^2}\sin\sp2 u}}.$$
}. Here $\eta$ is related to the spatial coordinates by (\ref{AtiyahWard}). With our conventions  the monopoles are on the $x_1$ axis (for $k>0$) and at $k=0$ the monopoles are axially symmetric about the $x_2$ axis. These properties, together with
a comparison with other curve conventions in the literature, are given for convenience in Appendix \ref{comparisoncurves}.

\subsection{Homology, differentials and the Ercolani-Sinha vector}
The roots of the quartic $\zeta^4+2(k^2-k'^2)\zeta^2+1$ are $\pm k'\pm\im k$; these give us the branch points.
With $k'=\cos\alpha$, $k=\sin\alpha$ they
 be written as $\pm e\sp{\pm i\alpha}$ and these lie
on the unit circle. We may take $0\le\alpha\le\pi/4$. We choose
cuts between $-k'+ik=-e\sp{- i\alpha}$ and $k'+ik=e\sp{ i\alpha}$
as well as $-k'-ik$ and $k'-ik$. Let $\mathfrak{b}$ encircle
$-k'+ik$ and $k'+ik$ with $\mathfrak{a}$ encircling $k'+ik$ and
$-k'+ik$ on the two sheets as shown in Figure \ref{fig:homology}. 
\begin{figure}
  \centering
  \includegraphics[scale=0.5]{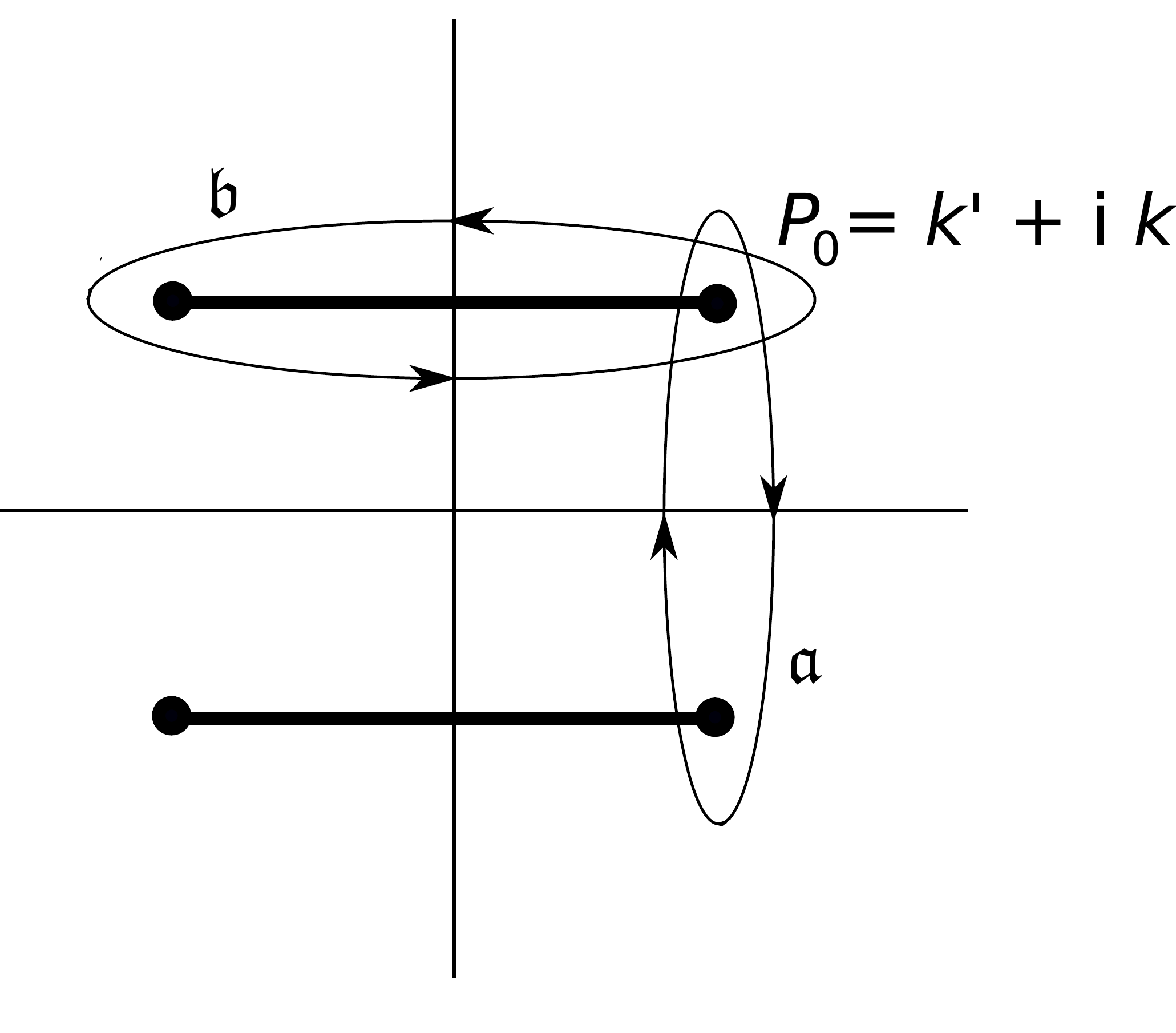}
  \caption{The homology basis for the curve with dark lines representing the cuts.}
  \label{fig:homology}
\end{figure}
We take as our
assignment of sheets ($j=1$, $2$, with analytic continuation from
$\zeta=0$ avoiding the cuts) to be
$$\eta_j
=(-1)\sp{j}\,i\frac{K}{2}\sqrt{\zeta^4+2(k^2-k'^2)\zeta\sp2+1}.$$
With our choice of homology basis the normalized holomorphic differential is then\footnote{
See Appendix \ref{appendixperiods}.
}
\begin{equation}\label{holdiff}
\boldsymbol{v}= \frac{d\zeta} {4\eta}
\end{equation}
and the period matrix  for $\mathcal{C}$ is  $\tau=i\mathbf{K}'/\mathbf{K}$. 
Comparison with (\ref{HMREScond}) shows that the Ercolani-Sinha constraint is satisfied for 
for the normalization of our curve (\ref{curve})
 and $\mathfrak{es}=-\mathfrak{a}$. Thus we have
\begin{equation}
\boldsymbol{U}=-1/2.
\end{equation} 
Also from the (\ref{esnotation}) and with our assignment of sheets for $\mathcal{C}$ we have that
$$\rho_1=-\frac{i}{2}\,\mathbf{K},\qquad
\rho_2=\frac{i}{2}\,\mathbf{K}.$$
The normalized second kind differential  $\gamma_{\infty}$ written in the curve coordinates is
\begin{equation}\label{gammainftdef}
\gamma_{\infty}(P) = \frac{K^2}{4\eta}\left( \zeta^2 - \frac{2E-K}{K}  \right)\mathrm{d} \,\zeta.
\end{equation}

\subsection{Abel Maps and Notation.}

In \cite{Braden2010d} we used the Abel map  $\phi(P)=\int_{P_0}\sp{P}\boldsymbol{v}$
with respect to a fixed base point $P_0=(k'+i k, 0_1)$ of the $\mathfrak{a}$-normalized differential $\boldsymbol{v}$, 
and using symmetry established that
$$\begin{matrix} \boldsymbol{\phi}(\infty_1)=\dfrac{1+\tau}{4}=-\boldsymbol{\phi}(\infty_2),&&
\boldsymbol{\phi}(0_1)=\dfrac{1-\tau}{4}=-\boldsymbol{\phi}(0_2).
\end{matrix}
$$
For a degree zero divisor the choice of $P_0$ doesn't matter. Here we shall often use $\infty_1$ as a limit of our integrals and to distinguish this we introduce
 \[\alpha(P)=\int_{\infty_1}\sp{P}\boldsymbol{v}=\phi(P)-\phi(\infty_1);\quad
 \phi(P)=\alpha(P)-\alpha(P_0).\qquad
 \]
It will be convenient to introduce 
the shorthand notation 
\[
\theta_i[D]:=\theta_i\left(\sum_{P\in D}\alpha(P) \right).
\]
Upon using the identification of $\theta(z)$ with the Jacobi theta
function $\theta_3(z):=\theta_3(z|\tau)=\sum_{n\in\mathbb{Z}}\exp(\im\pi[ n^2\tau +2nz])$ and the periodicities of the Jacobi theta
functions $\theta_*(z)$ we note that
\begin{align*}
\theta(-[z+1]/2 - 1-\tau)&=-e\sp{-i\pi(z+\tau)}\,\theta_4(z/2),\\
\theta(- [z+1]/2 -
(1+\tau)/2)&=e\sp{-i\pi(z/2+\tau/4)}\,\theta_2(z/2).
\end{align*}

We record for later use:
\begin{align}\phi(\infty_1)-\phi(0_1)&= \frac{\tau}{2},
&\phi(\infty_1)-\phi(\infty_2)&= \frac{1+\tau}{2}, \nonumber \\
\label{abelimages}\\
\phi(\infty_1)-\phi(0_2)&= \frac{1}{2},
&\phi(\infty_1)=-\alpha(P_0)&=\frac{1+\tau}{4}.\nonumber
\label{abelimages}
\end{align}

\subsection{Parameterization of the Curve} We establish in Appendix \ref{appparcuve}  that
\begin{lemma}\label{parcurve}
 With $\theta_i:=\theta_i(0)$ 
the curve (\ref{curve}) is parameterized by
\begin{equation}\label{eqparam}
\zeta=-i\,\frac{\theta_2[P]\theta_4[P]}{\theta_1[P]\theta_3[P]},\quad
\eta=  
\frac{i\pi\,\theta_3\theta_2^2\theta_4^2}4\,\frac{\theta_3[2 P]}{\theta_1[P]^2\theta_3[P]^2}.
\end{equation}
\end{lemma}

The following $\theta$-quotients are also expressible in terms of
coordinates and parameters of the curve
\begin{align}\begin{split}
\frac{\theta_1[P]^2}{\theta_4[P]^2}&=\frac{K(\zeta^2-1)-2\imath\eta}{2Kk\zeta^2},\\
\frac{\theta_2[P]^2}{\theta_4[P]^2}&=\frac{K(\zeta^2(k^2-{k'}^2)+1)+2\imath\eta}{2Kkk'\zeta^2},\\
\frac{\theta_3[P]^2}{\theta_4[P]^2}&=\frac{K(\zeta^2+1)+2\imath\eta}{2Kk'\zeta^2}.\end{split}\label{uniform1}
\end{align}

\subsection{The Baker-Akhiezer function}
We shall now gather together a number of functions on $\mathcal{C}$ including the Baker-Akhiezer
function. (The construction of these from first principles is described in \cite{Braden2010d}.) 

The unique meromorphic functions $g_{1,2}(P)$ on the curve such that 
\[ g_j(\infty_l)=\delta_{jl}\]
and with poles for $P$ such that $\alpha(P)=\pm 1/4$ are
\begin{align*}
g_1(P)&=\frac{\theta_2[P]\theta_3[P]}{\theta_2[P]\theta_3[P]-\theta_1[P]\theta_4[P]}
=
\frac{1+\zeta^2+2 i \eta/\boldsymbol{K}}{1+\zeta^2+2 i \eta/\boldsymbol{K}+2 i k' \zeta},\\
g_2(P)
&=\frac{\theta_1[P]\theta_4[P]}{\theta_1[P]\theta_4[P]-\theta_2[P]\theta_3[P]}
=
\frac{2 i k' \zeta}{1+\zeta^2+2 i \eta/\boldsymbol{K}+2 i k' \zeta}.
\end{align*}
The pole behaviour of these may be seen from
 \begin{align*}
\theta_1(\alpha(P)-1/4)\theta_4(\alpha(P)+1/4)\,\theta_2\theta_3&=
\theta_1[P]\theta_4[P]\,\theta_2(1/4)\theta_3(1/4)-
\theta_1(1/4)\theta_4(1/4)\,\theta_2[P]\theta_3[P]\\
&=-\theta_1(1/4)\theta_4(1/4)\left(
\theta_2[P]\theta_3[P]-\theta_1[P]\theta_4[P]\right)
\end{align*}
which holds as a consequence of
\[ \theta_1(x+y)\theta_4(x-y)\,\theta_2\theta_3=
\theta_1(x)\theta_4(x)\,\theta_2(y)\theta_3(y)+
\theta_1(y)\theta_4(y)\,\theta_2(x)\theta_3(x)\]
and
\[ \theta_1(1/4)=\theta_2(1/4),\qquad \theta_3(1/4)=\theta_4(1/4). \]

There is a unique $\mathfrak{a}$-normalized differential $\gamma_\infty$ on $\mathcal{C}$ with second order poles at $\infty_{1,2}$ such that in the vicinity of $P=\infty_{1,2}$ we have 
$\int_{P_0}\sp{P}\gamma_{\infty}\sim - {\eta}/{\zeta}$. We set
\begin{align*}
\tilde{ \nu}_i:=\tilde {\nu}_i(P_0)&=\lim_{P\to\infty_i}\left(\int_{P_0}\sp{P}\gamma_{\infty}(P')+\frac{\eta}
{\zeta}(P)\right).
\end{align*}
The following lemma (proved in Appendix \ref{proofintgammainf}) will be useful.
\begin{lemma}\label{intgammainf}
\begin{align}
\int_{P_0}^P\gamma_{\infty}(P')&=\frac14\left\{ \frac{\theta_1'[P]}{\theta_1[P]}+
\frac{\theta_1'[P-\infty_2]}{\theta_1[P-\infty_2]}\right\} =
\frac14\left\{ \frac{\theta_1'[P]}{\theta_1[P]}+
\frac{\theta_3'[P]}{\theta_3[P]}-\imath\pi\right\} , \label{gammatheta}\\
 \oint_{\mathfrak{b}}\gamma_{\infty}&=2\pi i \boldsymbol{U}=-i\pi ,\\
\tilde\nu_1&=-\frac{i\pi}{4},\quad \tilde\nu_2=\frac{i\pi}{4}, \quad \tilde\nu_2-\tilde\nu_1=\frac{i\pi}{2},\quad 
\tilde\nu_1+\tilde\nu_2=0.\label{nus}
\end{align}
\end{lemma}
It is important to note  that this lemma relates the choice of contours on each side of the identity by the
vanishing of each side at $P=P_0$; adding a $\mathfrak{b}$-cycle then to one side is compensated by adding a $\mathfrak{b}$-cycle to the other and so on.

If we define
$
\beta_i(P)=\int_{P_0}^P \gamma_{\infty}-\tilde\nu_i, 
$
then $\beta_1(P)=\beta_2(P)+\frac{\imath\pi}{2}$ and we are able to 
work with just the one function, which we will choose to be 
\begin{align}\label{betadef2}
\beta_1(P)=
\frac14\left\{ \frac{\theta_1'[P]}{\theta_1[P]}+
\frac{\theta_3'[P]}{\theta_3[P]}\right\} =\int_{P_0}\sp{P}\gamma_{\infty}(P')+\frac{\im\pi}4.
\end{align}
Combining these expressions yields the Baker-Akhiezer function $\Phi(z,P)$ for our problem;
its chief properties are given by:

\begin{lemma}\label{bafunction}
 $\Phi(z,P)$ defined by
\begin{align}\label{bafull}
\Phi(z,P)&=\chi(P)\begin{pmatrix}-\theta_3(\alpha(P))\theta_2(\alpha(P)-z/2) \\
\theta_1(\alpha(P))\theta_4(\alpha(P)-z/2)
\end{pmatrix}\,
\frac{e\sp{\beta_1(P) z}}{\theta_2(z/2)}\\
\intertext{where}
\label{bachi}
\chi(P)&=
\frac{\theta_2(1/4)\theta_3(1/4)}{\theta_3(0)
\theta_1(\alpha(P)-1/4)\theta_4(\alpha(P)+1/4)}
\end{align}
satisfies
\begin{enumerate}[(i)]
\item $ \Phi(z,P)$ is meromorphic for $P\in\mathcal{C}\setminus\{\infty_1,\infty_2\}$ and with poles at
$\alpha(P)=\pm1/4$.
\item $ \Phi(z,P)$ has simple poles at $z=\pm1$ and is regular for $z\in(-1,1)$.
\item $\displaystyle{\qquad
\lim_{P=P(\zeta,\eta)\rightarrow\infty_i} \Phi(z,P)\,e\sp{-z \nu_i\sp\dagger {\zeta}(P)}=
\lim_{P=P(\zeta,\eta)\rightarrow\infty_i} \Phi(z,P)\,e\sp{z \frac{\eta} {\zeta}(P)}=
\begin{pmatrix} \delta_{i1}\\ \delta_{i2} \end{pmatrix}.}$

\item $\displaystyle{  \Phi(0,P)= \begin{pmatrix} g_1(P)\\  g_2(P) \end{pmatrix}}$.
\item
$\displaystyle{  c_{12}= \lim_{P=P(\zeta,\eta)\rightarrow\infty_2} \zeta\,g_1(P)=-\im k'=
-c_{21}=- \lim_{P=P(\zeta,\eta)\rightarrow\infty_1} \zeta\,g_2(P) }$.

\end{enumerate}
Hence $\Phi(z,P)$ is a Baker-Akhiezer function for the charge $2$ spectral curve.
\end{lemma}
We remark that the reality conditions (\ref{esrealst}) determine the pole structure of $\Phi$, here 
$\alpha(P)=\pm1/4$, only up to a discrete number of choices (see \cite{Braden2010d}). We will work throughout with the above.

\subsection{Nahm Data and expansions}
The Nahm data for the $n=2$ spectral curve has long been known. With 
$T_j(z)=\frac{\sigma_j}{2\imath}\, f_j(z)$ Nahm's equation reduce to the equations of the spinning top
$ \dot f_1=f_2\,f_3$ (and cyclic) with solutions
\begin{equation}
\begin{split}\label{nahmsolution}
f_1(z)&={K}\,\frac{\mathrm{dn}\,{K}z
}{\mathrm{cn}\,{K}z}=\frac{\pi\theta_2 \theta_3 }{2}
\,\frac{\theta_3(z/2)}{\theta_2(z/2)},\quad f_2(z)={K}
k'\,\frac{\mathrm{sn}\,{K}z
}{\mathrm{cn}\,{K}z}=\frac{\pi\theta_3 \theta_4 }{2}
\,\frac{\theta_1(z/2)}{\theta_2(z/2)},\\
 f_3(z)&={K}k'\,\frac{1}{\mathrm{cn}\,{K}z}=\frac{\pi\theta_2 \theta_4 }{2}
\,\frac{\theta_4(z/2)}{\theta_2(z/2)}.
\end{split}
\end{equation}
(This choice of solution yields the spectral curve (\ref{curve})%
\footnote{
\begin{align*}
|\eta 1_2-L(\zeta)|&=
\eta^2+ \frac14(f_1^2-f_2^2)\zeta^4+\frac12(f_1^2+f_2^2 -2\,f_3^2)\zeta^2+\frac14(f_1^2-f_2^2)
=\eta^2+\frac14\boldsymbol{K}^2\left(\zeta^4+2(k^2-k'\sp2)\zeta^2+1\right)
\end{align*}
}.) These solutions were derived from first principles for the $n=2$ curve in the work of \cite{ercolani_sinha_89} and (with corrections in) \cite{Braden2010d}. We shall rederive this solution using the recent general approach of  \cite{Braden2018b}; this enables us to introduce a number of functions and their expansions that will be used throughout.

As noted in our general description, $h(z)$ may be constructed directly
for a monopole spectral curve \cite{Braden2018b}. The $n=2$ example of that reference gives
\begin{equation}\label{charge2hfinal}
h(z)=\frac1{K}\begin{pmatrix}f_1&-f_2\\-f_2&f_1 \end{pmatrix},
\qquad
h\sp{-1}(z)=\frac1{K}\begin{pmatrix}f_1&f_2\\f_2&f_1 \end{pmatrix},
\qquad
\dot h h\sp{-1}=-{f_3}\begin{pmatrix}0&1\\1&0 \end{pmatrix}.
\end{equation}
To put the Nahm data into standard gauge one solves the differential equation $\alpha=\alpha\sp\dagger$
(with $h\sp{-1}=C\sp\dagger C$) 
$$C\sp{-1}\dot C= \dot C C\sp{-1}=\frac{f_3}2\begin{pmatrix}0&1\\1&0\end{pmatrix}.
$$
Upon writing 
\begin{align}
C(z)&=\left(\begin{array}{cc}F(z)&G(z)\\G(z)&F(z)
\end{array} \right),
\end{align}
we find  $ \dot{F}={f_3} G/2$, $ \dot{G}={f_3} F/2$ with solution
\[ F=\cosh\left(\int_0\sp{z} f_3(s)ds/2 \right)=\left[ p(z)+1/p(z)\right]/2,\quad
G=\sinh\left(\int_0\sp{z} f_3(s)ds/2 \right)=\left[ p(z)-1/p(z)\right]/2,\]
where\footnote{
Here we have made use of
\begin{equation*}
\int \frac{du}{\mathrm{cn}\,u } =\frac{1}{k'}\,\mathrm{ln}\frac{
\mathrm{dn}u + k' \mathrm{sn} u }{\mathrm{cn} u}.\label{primitive}
\end{equation*}
}
\begin{align*}
p(z)&=\exp\left(\int_0\sp{z} f_3(s)ds/2 \right)=\exp\left(k' K\int_0\sp{z} \frac{ds}{\mathrm{cn}\,{K}z} \right)
= \left[ \frac{ \mathrm{dn}\,{K}z + k'\mathrm{sn}\,{K}z}{ \mathrm{cn}\,{K}z}
\right]\sp{1/2}.
\end{align*}
{Now
\begin{align} \label{FGrels}
G^2(z)&=\frac12\left( \frac{\mathrm{dn} (Kz;k)}{\mathrm{cn}
(Kz;k)}-1\right)=\frac12\left( \frac{f_1}{K}-1\right),&\quad 2F(z)G(z)&=k'\frac{\mathrm{sn}(Kz;k)}{\mathrm{cn}(Kz;k)}=\frac{f_2}{K}, \\
F^2-G^2&=1,& F^2+G^2&=\frac{f_1}{K}.
 \nonumber
\end{align}

The Nahm data now follows from
$$\tilde T_3=\frac{\im}2 C(\dot h h\sp{-1} )C\sp{-1}= \frac{\im}2 f_3 \sigma_1,
\qquad
\beta= \tilde T_1+\im \tilde T_2= g\nu g\sp{-1}=-\frac{\im}2 \begin{pmatrix}f_1&f_2\\-f_2&-f_1 \end{pmatrix},
$$
which leads to
$$
\tilde T_1(z)=\frac{1}{2\im}f_1(z)\sigma_3,\quad
\tilde T_2(z)
=\frac{1}{2\im}f_2(z) \sigma_2,\quad 
\tilde T_3(z)=-\frac{1}{2\im} f_3(z) \sigma_1.
$$
Now $g$ and $C$ 
are only defined up to left multiplication by a constant unitary matrix.
Let $\mathcal{O}$ be the orthogonal matrix
\begin{equation}
\mathcal{O}=\frac1{\sqrt{2} }\begin{pmatrix}1&-1\\1&1\end{pmatrix}, 
\label{Odef}
\end{equation}
for which
\[  \mathcal{O}\sp{-1}\sigma_1 \mathcal{O}=\sigma_3,\qquad 
    \mathcal{O}\sp{-1}\sigma_2 \mathcal{O}=\sigma_2,\qquad 
    \mathcal{O}\sp{-1}\sigma_3 \mathcal{O}=-\sigma_1,
\]
With this we obtain
\begin{align}\label{standardnahm}
T_j(z) = \mathcal{O}\, \tilde T_j(z)  \mathcal{O}\sp{-1}=\frac{\sigma_j}{2\imath}\, f_j(z), \quad j=1,2,3.
\end{align}

\subsubsection{Expansions}
For later use we record
\begin{align}
f_1(1-\xi)&=\frac1{\xi}+\frac16\,{K}^{2} \left( {k}^{2}+1 \right) \xi +\mathcal{O}(\xi^3),
\nonumber\\
f_2(1-\xi)&=\frac1{\xi}-\frac16\,{K}^{2} \left(- {k}^{2}+2 \right) \xi +\mathcal{O}(\xi^3),
\label{expandfs}\\
f_3(1-\xi)&=\frac1{\xi}+\frac16\,{K}^{2} \left( -2{k}^{2}+1 \right)\xi +\mathcal{O}(\xi^3),
\nonumber
\intertext{and note that}
f_1(\xi-1)&=f_1(1-\xi)+\mathcal{O}(\xi^2),\nonumber\\
f_2(\xi-1)&=-f_2(1-\xi)+\mathcal{O}(\xi^2), \label{expandfsm}\\
 f_3(\xi-1)&=f_3(1-\xi)+\mathcal{O}(\xi^2).\nonumber
\end{align}
Then
\begin{align}
T_j(1-\xi)&\sim -\frac{\im}{2}\frac{\sigma_j}{\xi}
+\mathcal{O}(\xi),\quad j=1,2,3,\label{nahmexpandp}\\
T_j(-1+\xi)&\sim \begin{cases}  -\frac{\im}{2}\frac{\sigma_j}{\xi} +\mathcal{O}(\xi), & j=1,3,\\
\phantom{-}\frac{\im}{2}\frac{\sigma_j}{\xi}+\mathcal{O}(\xi),&j=2.
\end{cases}\label{nahmexpandm}
\end{align}

The  expansion of the entries of $F(z)$ and $G(z)$ in the vicinity of points $z=\pm 1$ may be obtained
as follows.
Taking into account the expressions for $F^2$ and $G^2$ we find that
\begin{align*}
F(\pm1\mp \xi)&=\pm \left( \frac{1}{ \sqrt{\pi} \theta_3(0) \sqrt{\xi}} +\frac14  \sqrt{\pi} \theta_3(0)\sqrt{\xi} +O(\xi^{3/2}) \right)\\
G(\pm1\mp \xi)&=\pm \left( \frac{1}{\sqrt{\pi} \theta_3(0)\sqrt{\xi}} -\frac14  \sqrt{\pi} \theta_3(0)\sqrt{\xi} +O(\xi^{3/2} \right)
\end{align*}
The final choice of sign follows from  the relation $2F(z)G(x)=k' \mathrm{tn}(K z;k)$ from which  it follows that the coefficients of $\xi^{-1/2}$ in $F$ and $G$ should be of the same sign at $z=1-\xi$ and of opposite sign at $z=-1+\xi$. Therefore we will fix the signs as follows
\begin{align}\begin{split}
F(1-\xi)&=  \frac{1}{ \sqrt{\pi} \theta_3(0) \sqrt{\xi}} +\frac14  \sqrt{\pi} \theta_3(0)\sqrt{\xi} +\mathcal{O}(\xi^{3/2}) ,\\
G(1- \xi)&=  \frac{1}{\sqrt{\pi} \theta_3(0)\sqrt{\xi}} -\frac14  \sqrt{\pi} \theta_3(0)\sqrt{\xi} +\mathcal{O}(\xi^{3/2}), \\
F(-1+\xi)&= - \frac{1}{ \sqrt{\pi} \theta_3(0) \sqrt{\xi}} -\frac14  \sqrt{\pi} \theta_3(0)\sqrt{\xi} +O(\xi^{3/2}) ,\\
G(-1+ \xi)&= \frac{1}{\sqrt{\pi} \theta_3(0)\sqrt{\xi}} -\frac14  \sqrt{\pi} \theta_3(0)\sqrt{\xi} +\mathcal{O}(\xi^{3/2}).
\end{split}\label{expansions}
\end{align}
These then yield
\begin{align}
C(1-\xi)&=\frac1{\sqrt {\xi} }\frac1{ \sqrt {2K}} \begin{pmatrix}
1+\xi\,K/2&1-\xi\,K/2
\\ \noalign{\medskip}1-\xi\,K/2 &1+\xi\,K/2
\end{pmatrix}
+\mathcal{O}(\xi^{3/2}),\label{expancp}
\\
C(\xi-1)&=\frac1{\sqrt {\xi} }\frac1{ \sqrt {2K}} \begin{pmatrix}
-1-\xi\,K/2&1-\xi\,K/2
\\ \noalign{\medskip}1-\xi\,K/2 &-1-\xi\,K/2
\end{pmatrix}
+\mathcal{O}(\xi^{3/2}).\label{expancm}
\end{align}

\subsection{Asymptotic Expansions}\label{asympsection}

We consider the expansion of the Weyl equation $\Delta\sp\dagger v=0$,
\begin{equation*}
\left[\frac{\mathrm{d}}{\mathrm{d} z}+\frac12 \left(\sum_{j=1}^3 \sigma_j\otimes \sigma_j f_j(z) \right)
-\left(\sum_{j=1}^3 \sigma_j \otimes x_j 1_2 \right) \right]
\boldsymbol{v}(z,\boldsymbol{x})=0,
\end{equation*}
in the vicinity of the pole 
$z=\pm1$.  First, with $z=1-\xi$ we find from (\ref{nahmexpandp}) and (\ref{expandfs})
the leading behaviour
\begin{equation}
\left[\frac{\mathrm{d}}{\mathrm{d} \xi}-\frac1{2\xi}\left(\sum_{j=1}^3 \sigma_j\otimes \sigma_j \right)+
\left(\sum_{j=1}^3 \sigma_j \otimes x_j 1_2 \right)  +\mathcal{O}(\xi)\right]\boldsymbol{v}_1(\xi,\boldsymbol{x})=0.\label{approxdelta1a}.
\end{equation}
Now $\frac1{2}\sum_{j=1}^3 \sigma_j\otimes \sigma_j $ has an eigenvector
$(0,1,-1,0)\sp{T}$  with eigenvalue $-3/2$ and  eigenvectors $
(0,0,0,1)\sp{T}$, 
$(0,1,1,0)\sp{T}$,
$(1,0,0,0)\sp{T}$ each with eigenvalue $1/2$. The singular solution at $z=1$ then behaves as
\begin{equation}\label{asymvp}
\boldsymbol{v}_1(\xi,\boldsymbol{x})=
\frac{1}{\xi\sp{3/2}}\begin{pmatrix}0\\ 1\\ -1\\ 0 \end{pmatrix}+
\frac{1}{\xi\sp{1/2}}\begin{pmatrix} \im x_2-x_1\\ x_3\\ x_3\\  \im x_2 +x_1\end{pmatrix}+
 \xi\sp{1/2}\begin{pmatrix}a\\ b-r^2/2\\ b+r^2/2\\ c \end{pmatrix}+
\mathcal{O}( \xi\sp{3/2}).
\end{equation}
The undetermined coefficients $a,b,c$ here reflect that we can get contributions from the regular solutions that begin at this order.

A similar analysis at $z=-1$ now using (\ref{nahmexpandm}) leads to consideration of the matrix
$\frac1{2}\left(\sum_{j=1,3}\sigma_j\otimes \sigma_j -\sigma_2\otimes \sigma_2 \right)$ which has the
eigenvector
$(1,0,0,1)\sp{T}$  with eigenvalue $3/2$ and eigenvectors $
(0,1,0,0)\sp{T}$, 
$(0,0,1,0)\sp{T}$,
$(1,0,0,-1)\sp{T}$ each with eigenvalue $-1/2$. We then obtain the expansion of the singular solution at
$z=-1$ to be
\begin{equation}\label{asymvm}
\boldsymbol{v}_{-1}(\xi,\boldsymbol{x})=
\frac{1}{\xi\sp{3/2}}\begin{pmatrix}1\\ 0\\ 0\\ 1\end{pmatrix}+
\frac{1}{\xi\sp{1/2}}\begin{pmatrix} -x_3\\ \im x_2-x_1\\  -\im x_2 -x_1\\  x_3\end{pmatrix}+
 \xi\sp{1/2}\begin{pmatrix}a'-r^2/2\\ b'\\ c'\\ -a'-r^2/2\end{pmatrix}+
\mathcal{O}( \xi\sp{3/2}),
\end{equation}
where again the coefficients $a',b',c'$  reflect that we can get contributions from the regular solutions at this order.

\section{Identities}\label{sectionidentities}
This section will gather a number of useful identities  to be used in the sequel.  Recall we have defined the Atiyah-Ward equation to be the quartic equation (in general, the degree $2n$ equation) obtained by substituting the mini-twistor relation  (\ref{AtiyahWard}) into the equation for the curve (\ref{curve}). Throughout we let $\{P_k=(\zeta_k,\eta_k)\}_{k=1}\sp{4}$ be the corresponding four solutions to this and denote their Abelian images by
\begin{equation}
\alpha_k =\frac14 \int_{\infty_1}^{P_k} \frac{\mathrm{d}  \zeta}{\eta} ,\qquad k=1,\ldots,4.
\end{equation} 
Throughout all our calculations have been checked numerically and some comment now may be helpful.
The general Abel image $\alpha_k$ is only defined up to a shift by a lattice point reflecting the ambiguity in choice of contour integral. Abel's theorem tells us that the degree zero divisor $\sum_i( \mathfrak{p}_i -\mathfrak{q}_i)$ corresponds to that of a function if and only if its Abel image is a lattice point. When constructing this function we typically specify sheets by choosing $\sum_i( \mathfrak{p}_i -\mathfrak{q}_i)=0$ (see Mumford \cite{mumford_theta}). The real structure of our curve will enable us to further specify our choice of contours and we shall see that we may take
\begin{equation}\label{abjacconst}
   \alpha_1+\alpha_2+\alpha_3+\alpha_4 = N \tau ,\quad   N\in \mathbb{Z},\end{equation}
where $N$ reflects a remaining choice of contours in the Abel-map.  Assume henceforth  that such a choice has been made and $N$ is then fixed for a given set of solutions to the Atiyah-Ward equation.
The points may generically be ordered by
\[ P_{i+2}= \mathfrak{I}(P_i).\]
We also introduce the important functions
\begin{equation}\label{defmu}
\mu_k=\beta_1(P_k)+\im \pi N -\im \left[(x_1-i x_2)\zeta_k-i x_3 -x_4\right]
\end{equation}
and derive a number of relations for them. This function combines the exponential term in Nahm's ansatz
(\ref{Nansatz}) and that coming from the Baker-Akhiezer function (\ref{bafull}); there is an additional phase
proportional to $N$ that comes from the choice of contours. 

The identities we describe are grouped as follows: those that follow from the mini-twistor correspondence
and Abel's theorem; those arising from the real structure of the curve; and those related to single points on the curve. Because of (\ref{abjacconst}) we may express theta functions with arguments depending
on three points to those involving a single point. Next we derive a number of identities for the functions $\mu$ defined by (\ref{defmu}).  Finally 
we derive a number of identities involving theta functions whose arguments have more than one Abel image: some of these hold true for arbitrary arguments and are based on the Weierstrass
trisecant identities which (together with other properties of the theta functions) are gathered together in
Appendix \ref{thetafunctidentapp}; others depend on properties peculiar to our curve.  All proofs unless given are deferred to
Appendix \ref{identitiesapp}.

Before turning to the identities we note that any three solutions to the mini-twistor correspondence may be solved to give the  monopole coordinates $(x_1,x_2,x_3)$, or equivalently  $x_{\pm}:=x_1\pm x_2$,  $x_3$.
Let  $i,j,k,l\in \{1,2,3,4\}$ be distinct solutions of
\begin{equation}
\eta_j = -\imath x_-\zeta_j^2 -2\zeta_j x_3 -\imath x_+, \quad j=1,\ldots,4. \label{mtc}
\end{equation}
Then given $\{i,j,k\}$
\begin{align}
x_- &= \frac{\imath \eta_i}{ (\zeta_i-\zeta_j)(\zeta_i-\zeta_k) }+ \text {cyclic permutations of}\;\; i,j,k,
\label{xm}\\
x_+ &= \frac{\imath \zeta_j\zeta_k  \eta_i}{ (\zeta_i-\zeta_j)(\zeta_i-\zeta_k) }+\text{cyclic permutations of}\;\; i,j,k,
\label{xp}\\
x_3 &=\frac12 \frac{( \zeta_j+\zeta_k)  \eta_i}{ (\zeta_i-\zeta_j)(\zeta_i-\zeta_k) }+\text{cyclic permutations of}\;\; i,j,k.
\label{x3}
\end{align}
The compatibility condition of all $4$ equations (\ref{mtc}) shows
\begin{equation}
\eta_l= \frac{\eta_i(\zeta_j-\zeta_l)(\zeta_k-\zeta_l)}{ (\zeta_i-\zeta_j)(\zeta_i-\zeta_k) }+\text{cyclic permutations of}\;\; i,j,k.
\label{compatibility}
\end{equation}
One can check that the permutations of $(i,j,k,l)$ in equation (\ref{compatibility}) leads to a solvable 
homogeneous system with respect to $\eta_1,\ldots,\eta_4$. 
By considering the two Atiyah-Ward equations with $\zeta=\zeta_i$ and $\zeta_j$ 
and eliminating the variable $k^2$ from the two equations  by computing resultant one also finds that
\begin{align}
x_3=\frac{\imath}{16} \frac{ (\zeta_i+\zeta_j) [
 \zeta_i^2\zeta_j^2(K^2-4 x_-^2) -K^2+4x_+^2  ]   }{\zeta_i\zeta_j(x_-\zeta_i\zeta_j - x_+ )}.
\label{x3eq}
\end{align}

We  have that
\begin{equation}   
\frac{x_-\zeta_i\zeta_j-x_+}{ \zeta_i+\zeta_j} +\frac{x_-\zeta_k\zeta_l-x_+}{ \zeta_k+\zeta_l} =0,\quad i\neq j \neq k \neq l \in \{1,2,3,4\}
 \end{equation}
or equivalently, 
\[ x_+(\zeta_1+\zeta_2+\zeta_3+\zeta_4)=x_-( 
\zeta_1\zeta_2\zeta_3 +
\zeta_1\zeta_2\zeta_4 +  
\zeta_1\zeta_3\zeta_4 +  
\zeta_2\zeta_3\zeta_4 ) \]
This relation follows from the Atiyah-Ward equation.

 \subsection{Derivatives}
In later calculations we will need various derivatives with respect to the spatial coordinates $(x_1,x_2,x_3)$.
Implicit differentiation of the  Atiyah-Ward equation gives us $\partial_i\zeta:=\partial\zeta/\partial x_i$; the differential of (\ref{AtiyahWard}) together with $\partial_i\zeta$ then yields $\partial_i\eta$. From
(\ref{gammainftdef}) we have
$$\partial_i\beta_1(P)=\frac{K^2}{4\eta}\left( \zeta^2 - \frac{2E-K}{K}  \right)\partial_i\zeta,
$$
from which we obtain $\partial_i\mu(P)$.

\subsection{Reflection}
We see that if $P=(\zeta,\eta)$ is a point on the spectral curve corresponding to $\boldsymbol{x}$ then
$P'=(\zeta,-\eta)$ corresponds to $-\boldsymbol{x}$. Now using $\gamma_\infty(P')=-\gamma_\infty(P)$
and that $P_0$ is a branch point we have
\begin{align*}\beta_1(P')-\frac{\imath \pi}{4}
&=\int_{P_0}\sp{P'}\gamma_\infty(P)=-\int_{P_0}\sp{P'}\gamma_\infty(P')
=-\int_{P'_0}\sp{P}\gamma_\infty(P)\\
&=-\left[\beta_1(P)-\frac{\imath \pi}{4}\right]+\int_{P_0}\sp{P'_0}\gamma_\infty(P)
=-\beta_1(P)+\frac{\imath \pi}{4}.
\end{align*}
Then
\begin{equation}\label{reflection}
\mu(P)\equiv -\mu(P') \pmod {\frac{i \pi}2}.
\end{equation}
Although we have used the reflected path to define $\beta_1(P)$ the path need not be given this way and there is an ambiguity of half $\mathfrak{b}$-periods.

\subsection{Abel-Jacobi Constraints}
As noted in section \ref{minitwistorabjac} the solutions to the Atiyah-Ward equation
satisfy a number of relations. 
\begin{proposition}\label{abjacprop}
 Let the four points $P_i=(\zeta_i,\eta_i)$ $i=1,\ldots,4$ solve 
(\ref{AtiyahWard}) for the curve (\ref{curve}). Then the following hold
\begin{equation}
\int_{\infty_1}^{P_1}\boldsymbol{v}+\int_{\infty_1}^{P_2}\boldsymbol{v}+\int_{\infty_1}^{P_3}\boldsymbol{v}+\int_{\infty_1}^{P_4}\boldsymbol{v}
=N\tau+M,\quad N,M\in \mathbb{Z} \label{abel}
\end{equation}
and for the choice of paths given by Lemma \ref{intgammainf}
\begin{equation}
\int_{P_0}^{P_1}\gamma_{\infty}+\int_{P_0}^{P_2}\gamma_{\infty}+\int_{P_0}^{P_3}\gamma_{\infty}
+\int_{P_0}^{P_4}\gamma_{\infty}=\frac{4K^2 x_3}{K^2-4x_-^2} -{\imath\pi(N+1)},
 \label{abel2}
\end{equation}
and
\begin{equation}\label{musuma}
\sum_k \mu_k=  3\imath \pi N.
\end{equation}
Here $\boldsymbol{v}$ in the normalized holomorphic, $\gamma_{\infty}$  the $\mathfrak{a}$-normalized differential of the second kind, and base point  $P_0= ( k'+\imath k,0_1)$.

\end{proposition}
The proposition is proven in Appendix \ref{proofabjacprop}. If we do not relate the paths $\int_{\infty_1}^{P_k}\boldsymbol{v}$ and $\int_{P_0}^{P_k}\gamma_{\infty}$ via Lemma \ref{intgammainf} then (\ref{abel2}, \ref{musuma}) are
only defined mod ${\imath\pi}$.

Using (\ref{abel}) and the periodicities of the theta functions (see Appendix \ref{thetafunctidentapp}) we have the further relations:
\begin{corollary}\label{abelprop}
Set $E_i:=e\sp{i\pi[ -N^2\tau +2N \int_{\infty_1}^{P_i}\boldsymbol{v} +Nz]}$. Then for distinct $i,j,k,l\in\{1,2,3,4\}$  if (\ref{abel}) holds we have 
\begin{align*}
\theta_1(P_j+P_k+P_l-z/2) &= (-1)\sp{N+M+1}\,E_i\,\theta_1(P_i+z/2),\\
\theta_2(P_j+P_k+P_l-z/2) &=\quad \quad(-1)\sp{M}\,E_i\,\theta_2(P_i+z/2),\\
\theta_3(P_j+P_k+P_l-z/2) &= \qquad\qquad\quad E_i\,\theta_3(P_i+z/2),\\
\theta_4(P_j+P_k+P_l-z/2) &= \quad \quad(-1)\sp{N}\,E_i\,\theta_4(P_i+z/2),
\intertext{and similarly}
\theta_r(P_i+P_j)\theta_r(P_i+P_k)e\sp{2i\pi N \int_{\infty_1}^{P_i}\boldsymbol{v}}&=
\theta_r(P_l+P_j)\theta_r(P_l+P_k)e\sp{2i\pi N \int_{\infty_1}^{P_l}\boldsymbol{v}}.
\end{align*}

\end{corollary}

\subsection{Conjugation}
We will need the complex conjugates of the Baker-Akhiezer functions to implement our strategy and we investigate this here. At the outset we note that choices of contours are implicit in the results stated; the proofs make these clear, but they are the natural ones: given a contour $\lambda$ between two points $P$ and $Q$, then we integrate between $\mathfrak{I}(P)$ and $\mathfrak{I}(P)$ along
$\mathfrak{I}_\ast(\lambda)$ and so forth.

\begin{proposition} Let $P_j$, $j=1,\ldots,4$ be the solutions of the Atiyah-Ward equation constraint. Then
\begin{align}\begin{split}
\overline{\int_{\infty_1}^{P_1} v}=-\int_{\infty_1}^{P_3} v-\frac{\tau}{2},\\
\overline{\int_{\infty_1}^{P_2} v}=-\int_{\infty_1}^{P_4} v-\frac{\tau}{2}.
\end{split} \label{conjugation}
\end{align}
\end{proposition}
\begin{proof} For the first of relations (\ref{conjugation}) we have
\begin{align*}
\frac14\overline{\int_{\infty_1}^{P_1} \frac{\mathrm{d}\zeta}{\eta}}&=
\frac14\int_{\infty_1}^{P_1} \overline{\frac{\mathrm{d}\zeta}{\eta}}=\frac14\int_{\infty_1}^{P_1} \frac{\mathrm{d}\overline{\zeta}}{\overline{\eta}}=-
\frac14\int_{\infty_1}^{P_1} \frac{\mathrm{d}\mathfrak{I}(\zeta)}{\mathfrak{I}(\eta)}=-
\frac14\int_{\infty_1}^{P_1} \mathfrak{I}\left(\frac{\mathrm{d}\zeta}{\eta}\right)\\
&=-\frac14\int_{\mathfrak{I}(\infty_1)}^{\mathfrak{I}(P_1)} \frac{\mathrm{d}\zeta}{\eta}=-\frac14\int_{0_1}^{P_3} \frac{\mathrm{d}\zeta}{\eta}=-\frac14\int_{\infty_1}^{P_3} \frac{\mathrm{d}\zeta}{\eta}-\frac14\int^{\infty_1}_{0_1} \frac{\mathrm{d}\zeta}{\eta}=-\int_{\infty_1}^{P_3} v-\frac{\tau}{2}.
\end{align*}
The second relation is proven analogously.
\end{proof}
Using the fact that $\tau$ is pure imaginary for our curve we have that
$$\sum_{k=1}\sp{4} {\int_{\infty_1}^{P_k} v}=
-\sum_{k=1}\sp{4} \overline{\int_{\infty_1}^{P_k} v}$$
whence
\begin{corollary} In Proposition \ref{abjacprop} we have $M=0$.
\end{corollary}

The conjugation rule induces the following conjugation rule of theta functions. Again the purely imaginary
period matrix for (\ref{curve}) yields that
\[   \overline{\theta_k(z)}= \theta_k( \overline{z}),\quad k=1,\ldots,4.   \]
The following relations are valid
\begin{align}\begin{split}
\overline{\vartheta_{1}\left(  \int_{\infty_1}^{P_{1,2}} v  \right)}=-\imath \vartheta_{4}\left(  \int_{\infty_1}^{P_{3,4}} v  \right)
\mathrm{exp}\left\{  -\imath \pi  \int_{\infty_1}^{P_{3,4}}v -\frac{\imath\pi\tau}{4}\right \},\\
\overline{\vartheta_{4}\left(  \int_{\infty_1}^{P_{1,2}} v  \right)}=\imath \vartheta_{1}\left(  \int_{\infty_1}^{P_{3,4}} v  \right)
\mathrm{exp}\left\{  -\imath \pi  \int_{\infty_1}^{P_{3,4}}v -\frac{\imath\pi\tau}{4}\right \},\\
\overline{\vartheta_{2,3}\left(  \int_{\infty_1}^{P_{1,2}} v  \right)}= \vartheta_{3,2}\left(  \int_{\infty_1}^{P_{3,4}} v  \right)
\mathrm{exp}\left\{  -\imath \pi  \int_{\infty_1}^{P_{3,4}}v -\frac{\imath\pi\tau}{4}\right \}.\end{split}
\end{align}

We will need also:

\begin{align}\begin{split}
\overline{\vartheta_{2}\left(  \int_{\infty_1}^{P_{1,2}} v-\frac{z}{2}  \right)}= \vartheta_{3}\left(  \int_{\infty_1}^{P_{3,4}} v +\frac{z}{2} \right)
\mathrm{exp}\left\{ - \imath \pi  \int_{\infty_1}^{P_{3,4}}v -\frac{\imath\pi z}{2}-\frac{\imath\pi\tau}{4}\right \},\\
\overline{\vartheta_{4}\left(  \int_{\infty_1}^{P_{1,2}} v-\frac{z}{2}  \right)}=\imath \vartheta_{1}\left(  \int_{\infty_1}^{P_{3,4}} v +\frac{z}{2} \right)
\mathrm{exp}\left\{ - \imath \pi  \int_{\infty_1}^{P_{3,4}}v -\frac{\imath\pi z}{2}-\frac{\imath\pi\tau}{4}\right \}.
\end{split}
\end{align}
 
\subsection{The curve} We shall now prove various properties of theta functions depending on one
and (via Corollary \ref{abelprop}) three distinct solutions to the Atiyah-Ward equation.
\begin{lemma}\label{curveidenta}
Let $P=(\zeta,\eta)\in \mathcal{C}$ for the curve (\ref{curve}). Then
\begin{align}
 \frac{\theta_1'[P]}{\theta_1[P]}-
 \frac{\theta_3'[P]}{\theta_3[P]}&=2\imath {K} \zeta,
 \label{curveident1} \\
\frac{\theta_1''[P]}{\theta_1[P]}-
 \frac{\theta_3''[P]}{\theta_3[P]}&=8\imath {K}\left (\zeta\,
\beta_1(P)+\eta\right), \label{curveident2}\\
\frac{d }{d\alpha(P)}\zeta&=4\eta,\label{curveident3}\\
\frac{d }{d\alpha(P)}\eta&=-2K^2\left(\zeta^3+(k^2-k'^2)\zeta\right),\label{curveident4}
\end{align}
where $\beta_1(P)$ was defined in (\ref{betadef2}).
\end{lemma}

\begin{corollary}\label{curveidentb}
 With $\alpha=\int_{\infty_1}^Pv $, where $P=(\zeta,\eta)\in \mathcal{C}$ for the curve (\ref{curve}),
we have
\begin{align}
 \frac{\theta_1'(\alpha)}{\theta_1(\alpha)}&=2\beta_1(P)+\imath K \zeta,\label{quot1a} \\
 \frac{\theta_3'(\alpha)}{\theta_3(\alpha)}&=2\beta_1(P)-\imath K \zeta, \quad \label{quot1b}\\
 \frac{\theta_1''(\alpha)}{ \theta_1(\alpha)}&=K^2\zeta^2-4EK+2K^2+4\beta_1(P)^2+4\imath K \zeta \beta_1(P) + 4\imath K \eta, \label{quot3a}\\
\frac{\theta_3''(\alpha)}{ \theta_3(\alpha)}&=K^2\zeta^2-4EK+2K^2+4\beta_1(P)^2-4\imath K \zeta \beta_1(P) - 4\imath K \eta, \label{quot3b}\\
\frac{\theta_1'''(\alpha)}{ \theta_1(\alpha)}&=
4K[ K\zeta+6\imath \beta_1(P)]\eta+6K^2\beta_1(P)\zeta^2
-24KE\beta_1(P)+12K^2\beta_1(P)\label{quot4a}\\
&\qquad  +8\beta_1(P)^3-3\imath K^3\zeta^3-2\imath K[ -8{k'}^2 K^2 +6KE +K^2 -6 \beta_1(P)^2]\zeta,
\nonumber\\
\frac{\theta_3'''(\alpha)}{ \theta_3(\alpha)}&=4K[K\zeta-6\imath \beta_1(P)]\eta+6K^2\beta_1(P)\zeta^2
-24KE\beta_1(P)+12K^2\beta_1(P)\label{quot4b}\\
&\qquad  +8\beta_1(P)^3+3\imath K^3\zeta^3+2\imath K[ -8{k'}^2 K^2 +6KE +K^2
 -6 \beta_1(P)^2]\zeta \nonumber
.
\end{align}
\end{corollary}

\begin{corollary}\label{curveidentbthreepts}
Let  $i,j,k,l\in \{1,2,3,4\}$ be distinct with $\alpha_i$ the Abel image of $P_i$ subject to
$\alpha_1+\alpha_2+\alpha_3+\alpha_4=N\tau$.
The following relations are valid:

\begin{align}
\frac{\theta_1'(\alpha_i+\alpha_j+\alpha_k)}{\theta_1(\alpha_i+\alpha_j+\alpha_k) } &= 
-2\beta_1(\alpha_l)-2\imath \pi N -\imath K \zeta_l ,\label{dfrac1a}\\
\frac{\theta_3'(\alpha_i+\alpha_j+\alpha_k)}{\theta_3(\alpha_i+\alpha_j+\alpha_k) } &=
-2\beta_1(\alpha_l)-2\imath \pi N +\imath K \zeta_l
,\label{dfrac1b}
\\
\frac{\theta_1''(\alpha_i+\alpha_j+\alpha_k)}{\theta_1(\alpha_i+\alpha_j+\alpha_k) } &= K^2\zeta_l^2-4\pi^2N^2+8\imath \beta_1(\alpha_l) \pi N  \label{dfrac2a} \\
&\qquad-2(2EK-K)+4\beta_l^2 - 4K(\pi N -\imath \beta_1(\alpha_l))\zeta_l + 4\imath K \eta_l,
\nonumber\\
\frac{\theta_3''(\alpha_i+\alpha_j+\alpha_k)}{\theta_3(\alpha_i+\alpha_j+\alpha_k) } &= K^2\zeta_l^2-4\pi^2N^2+8\imath    \beta_1(\alpha_l) \pi N \label{dfrac2b} \\
&\qquad -2(2EK-K)+4\beta_l^2 + 4K(\pi N -\imath \beta_1(\alpha_l))\zeta_l - 4\imath K \eta_l.
\nonumber
\end{align}
\end{corollary}

 \subsection{Conjugation and properties of \texorpdfstring{$\mu_k$}{muk}} We may now use the results of the previous subsections  to place useful constraints on the $\mu_k$'s.

\begin{proposition}\label{propmu13mu24}
With $\mu_k$ defined by (\ref{defmu}) the ordering $\mathcal{J}(P_i)=P_{i+2}$ and the conjugate
contours  then
the following relations for $\mu_k$ are valid for all  $(x_1,x_2,x_3)$ \begin{equation}
\mu_1+\overline{\mu}_3 =- \frac{\imath \pi}{2}, \quad \mu_2+\overline{\mu}_4 \equiv - \frac{ \imath \pi}{2} .
 \label{mu13mu24}
\end{equation}
\end{proposition}

\begin{proof}  Upon noting that $\mathcal{J}(P_1)=P_3$ we have
\begin{align*}
\mu_1+\overline{\mu}_3&=
\left[\beta_1(P_1) -x_3- (x_2+\imath x_1)\zeta_1\right]  + \overline{\left[\beta_1(P_3) 
-x_3- (x_2+\imath x_1)\zeta_3\right]}\\
&= \beta_1(P_1) +\overline{\beta_1(P_3)}
- (x_2+\imath x_1)\zeta_1-2 x_3+ (x_2-\imath x_1)\frac{1}{\zeta_1} \\
&=\beta_1(P_1) +\overline{\beta_1(P_3)} + \frac{\eta_1}{\zeta_1}.
\end{align*}
Taking into  account (\ref{gammatheta}) and (\ref{conjugation}) we find 
\begin{align*}
\mu_1+\overline{\mu}_3
&=\frac14 \left[  \frac{\theta_1'(P_1)}{\theta_1(P)}+ \frac{\theta_3'(P_1)}{\theta_3(P)} 
-   \frac{\theta_2'(P_1)}{\theta_2(P)}- \frac{\theta_4'(P_1)}{\theta_4(P)} \right]-\frac{\imath\pi}{2}
 + \frac{\eta_1}{\zeta_1}.
\end{align*}
Upon using the representation, 
\[  \zeta(P) =-\imath \frac{\theta_2(P)\theta_4(P) }{ \theta_1(P) \theta_3(P)} \]
this may be transformed into the form 
\begin{align*}
\mu_1+\overline{\mu}_3
&=-\frac14 \frac{\mathrm{d}}{ \mathrm{d} \alpha_1} \mathrm{\ln} \;\zeta(\alpha_1) 
+ \frac{\eta_1}{\zeta_1 } 
- \frac{\imath \pi}{2}.
\end{align*}
Finally, the first two terms cancel because of the relation (\ref{curveident3}).
Thus the first of the stated relation follows; the second is proven in analogous way. 
\end{proof}
Here we have chosen contours in a specified way; if we had chosen arbitrary contours then the
relations are only defined mod $\imath\pi$.
We also prove in  Appendix \ref{proofmux1x2starplane} that
\begin{proposition}\label{mux1x2starplane}
{\em Suppose that $\zeta_1+\zeta_*=0$,  $\zeta_2+\zeta_{*'}=0$. Then 
 for the $(x_1,x_2)$ plane, $x_3=0$ , }
\begin{equation}
\mu_1+\mu_* \equiv 0 \  \pmod {\imath \pi }, \quad \mu_2+\mu_{*'}
 \equiv 0 \  \pmod {\imath \pi } .   \label{mu1starmu2star}
\end{equation}
\end{proposition}

\subsection{Combined results} 
We next present a number of nontrivial identities that arise from the Weierstrass trisecant identities
or combining our expressions for the curve together with the Weierstrass trisecant identities.

\begin{lemma}\label{weierstrassthreept}
For arbitrary $\alpha_i  ,  \alpha_j  ,  \alpha_k$ the following relations are valid
\begin{align}
\begin{split}
&\left[\theta_1(\alpha_i+\alpha_j+\alpha_k)\theta_1(\alpha_i)\theta_1(\alpha_j)\theta_1(\alpha_k)+\theta_3(\alpha_i+\alpha_j+\alpha_k) \theta_3(\alpha_i)\theta_3(\alpha_j)\theta_3(\alpha_k) \right]\\&\hskip 3cm\times \theta_1(\alpha_i-\alpha_j)\theta_1(\alpha_j-\alpha_k)\theta_1(\alpha_i-\alpha_k)\\
&=\sum_{\mathrm{cyclic\  permutations} \; i,j,k}\theta_1(\alpha_i)\theta_3(\alpha_i)\theta_1(\alpha_j)\theta_3(\alpha_j)                       \theta_1(\alpha_i-\alpha_j)\theta_3(\alpha_i+\alpha_j)\theta_2(2\alpha_k),
\end{split}\label{relation1}
\\
\nonumber
\\
\begin{split}
&\theta_1(\alpha_i+\alpha_j+\alpha_k)\theta_1(\alpha_i)\theta_1(\alpha_j)\theta_1(\alpha_k)-\theta_3(\alpha_i+\alpha_j+\alpha_k) \theta_3(\alpha_i)\theta_3(\alpha_j)\theta_3(\alpha_k) \\
&\qquad\qquad= -\theta_3(0)  \theta_1(\alpha_i+\alpha_j)
\theta(\alpha_j+\alpha_k)\theta_1(\alpha_i+\alpha_k).
\end{split}\label{relation2a}
\end{align}
\end{lemma}

\begin{lemma}\label{Relations}
Let  $i,j,k,l\in \{1,2,3,4\}$ be distinct and $P_{i,j,k,l}$ be points on the curve corresponding to
solutions of the Atiyah-Ward equation. Let
 $\alpha_i$ be  the Abel image of $P_i$ subject to
$\alpha_1+\alpha_2+\alpha_3+\alpha_4=N\tau$.
Set $\beta_i:=\beta_1(P_i)$, $x_{\pm}=x_1\pm \imath x_2$ and 
$ \mu_i  =  \beta_i+\imath\pi N-(x_2+\imath x_1)\zeta_i-x_3 $.
The following relations are valid
\begin{align}
\begin{split}
&\pi\theta_3(0)\frac{\theta_1(\alpha_i+\alpha_j+\alpha_k)\theta_1(\alpha_i)\theta_1(\alpha_j)\theta_1(\alpha_k)}
{ \theta_3(\alpha_i+\alpha_j)
\theta_3(\alpha_i+\alpha_k)\theta_3(\alpha_j+\alpha_k)}=  2 x_--K,
\end{split}\label{sub1a}
\\
\nonumber
\\
\begin{split}
&\pi\theta_3(0)\frac{\theta_3(\alpha_i+\alpha_j+\alpha_k)\theta_3(\alpha_i)\theta_3(\alpha_j)\theta_3(\alpha_k)}
{ \theta_3(\alpha_i+\alpha_j)
\theta_3(\alpha_i+\alpha_k)\theta_3(\alpha_j+\alpha_k)}= 2 x_-+K,
\end{split}\label{sub1b}
\\
\nonumber
\\
\begin{split}
&\theta_1(\alpha_i+\alpha_j+\alpha_k)\theta_1(\alpha_i)\theta_1(\alpha_j)\theta_1(\alpha_k)+\theta_3(\alpha_i+\alpha_j+\alpha_k) \theta_3(\alpha_i)\theta_3(\alpha_j)\theta_3(\alpha_k) \\
&\qquad\qquad=4\frac{x_1-\imath x_2}{\pi \theta_3(0)}  \theta_3(\alpha_i+\alpha_j)
\theta_3(\alpha_j+\alpha_k)\theta_3(\alpha_i+\alpha_k),
\end{split}\label{relation2}
\\
\nonumber
\\
\begin{split}
&\frac{\theta_1'(\alpha_i+\alpha_j+\alpha_k)\theta_1(\alpha_i)\theta_1(\alpha_j)\theta_1(\alpha_k)}
{ \theta_3(\alpha_i+\alpha_j)
\theta_3(\alpha_i+\alpha_k)\theta_3(\alpha_j+\alpha_k)}=-( 2\beta_1(\alpha_l)+2\imath\pi N+\imath K \zeta_l ) 
\frac{2x_--K}{\pi \theta_3(0)} ,
\end{split}\label{sub2a}
\\
\nonumber
\\
\begin{split}
&\frac{\theta_3'(\alpha_i+\alpha_j+\alpha_k)\theta_3(\alpha_i)\theta_3(\alpha_j)\theta_3(\alpha_k)}
{ \theta_3(\alpha_i+\alpha_j)
\theta_3(\alpha_i+\alpha_k)\theta_3(\alpha_j+\alpha_k)}=-( 2\beta_1(\alpha_l)
+2\imath\pi N-\imath K \zeta_l ) \frac{2x_- +K}{\pi \theta_3(0)} ,\label{sub2b}
\end{split}
\\
\nonumber
\\
\begin{split}
&\theta_1'(\alpha_j+\alpha_k+\alpha_l)\theta_1(\alpha_j)\theta_1(\alpha_k)\theta_1(\alpha_l)-\theta_3'(\alpha_j+\alpha_k+\alpha_l) \theta_3(\alpha_j)\theta_3(\alpha_k)\theta_3(\alpha_l)\\
&\quad=2(\mu_i+x_3) \theta_3(0) \theta_3(\alpha_j+\alpha_k)\theta_3(\alpha_k+\alpha_l)\theta_3(\alpha_j+\alpha_l),
\end{split}\label{relation3a}
\\
\nonumber
\\
\begin{split}
&\frac{\theta_1''(\alpha_i+\alpha_j+\alpha_k)\theta_1(\alpha_i)\theta_1(\alpha_j)\theta_1(\alpha_k)}
{ \theta_3(\alpha_i+\alpha_j)
\theta_3(\alpha_i+\alpha_k)\theta_3(\alpha_j+\alpha_k)}=\left[ -4K(\pi N -\imath \beta_1(\alpha_l))+4\imath K \eta_l \right.\\
&\hskip2cm \left.+K^2\zeta_l^2-4\pi^2N^2+8\imath\pi\beta_1(\alpha_l) N-2K(2E-K)+4\beta_1(\alpha_l)^2 \right] \frac{2x_- -K}{\pi \theta_3(0)} ,
\end{split}\label{sub3a}
\\
\nonumber
\\
\begin{split}
&\frac{\theta_3''(\alpha_i+\alpha_j+\alpha_k)\theta_3(\alpha_i)\theta_3(\alpha_j)\theta_3(\alpha_k)}
{ \theta_3(\alpha_i+\alpha_j)
\theta_3(\alpha_i+\alpha_k)\theta_3(\alpha_j+\alpha_k)}=\left[4K(\pi N -\imath \beta_1(\alpha_l))-4\imath K \eta_l\right.\\
&\hskip2cm \left.+K^2\zeta_l^2-4\pi^2N^2+8\imath\pi\beta_1(\alpha_l) N-2K(2E-K)+4\beta_1(\alpha_l)^2\right]\frac{2x_-+K}{\pi \theta_3(0)},
\end{split}\label{sub3b}
\\
\nonumber
\\
\begin{split}
&\frac{
\theta_3''(\alpha_i+\alpha_j+\alpha_k)
\theta_3(\alpha_i)\theta_3(\alpha_j)\theta_3(\alpha_k)
-
\theta_1''(\alpha_i+\alpha_j+\alpha_k)
\theta_1(\alpha_i)\theta_1(\alpha_j)\theta_1(\alpha_k) 
}
{\theta_3(0)\theta_3(\alpha_i+\alpha_j) \theta_3(\alpha_i+\alpha_k) \theta_3(\alpha_k+\alpha_k)   }\\
&=K^2\zeta_l^2+8(\pi N-\imath\beta_1(\alpha_l))x_--4(\pi N -\imath\beta_1(\alpha_l))^2-8\imath \eta_l x_--2K(2E-K),
\end{split}
\label{relation4}
\\
\nonumber
\\
\begin{split}
\imath\pi \zeta_l\theta_3(0)^2& [\theta_1(\alpha_i+\alpha_j+\alpha_k)\theta_1(\alpha_i)\theta_1(\alpha_j)                  \theta_1(\alpha_k)+\theta_3(\alpha_i+\alpha_j+\alpha_k) \theta_3(\alpha_i)\theta_3(\alpha_j)\theta_3(\alpha_k)] \\
 =&-2 [\theta_1'(\alpha_i+\alpha_j+\alpha_k)\theta_1(\alpha_i)\theta_1(\alpha_j)\theta_1(\alpha_k)-\theta_3'(\alpha_i+\alpha_j+\alpha_k) \theta_3(\alpha_i)\theta_3(\alpha_j)\theta_3(\alpha_k)]\\
&+4 \theta_3(0)(\imath\pi  N+ \beta_l)\theta_3(\alpha_i+\alpha_j)\theta_3(\alpha_j+\alpha_k)\theta_3(\alpha_i+\alpha_k).
\end{split}\label{relation3}
\end{align}
\end{lemma}

We note that from (\ref{sub1a}, \ref{sub1b}) we may also express
\begin{align}
x_-=-\frac{\pi \theta_3^2(0)}{4}
\frac{\theta_1(\alpha_l)\theta_1(\alpha_i)
\theta_1(\alpha_j)\theta_1(\alpha_k)+
\theta_3(\alpha_l) \theta_3(\alpha_i)
\theta_3(\alpha_j)\theta_3(\alpha_k)}
{\theta_1(\alpha_l)\theta_1(\alpha_i)
\theta_1(\alpha_j)\theta_1(\alpha_k)-
\theta_3(\alpha_l) \theta_3(\alpha_i)\theta_3(\alpha_j)\theta_3(\alpha_k)} .\label{xm1}
\end{align}

Finally we consider expressions of the form  
\begin{align}
\frac{\theta_3'(\alpha_i+\alpha_j)}{ \theta_3(\alpha_i+\alpha_j) },\quad i,j, \in \{ 1,2, 3,4\},
\end{align}
involving the addition of two points.
Using that $\alpha_k+\alpha_l=N\tau - \alpha_i-\alpha_j$ we first note that 

\begin{align}\label{twotheta3sum}
\frac{\theta_3'(\alpha_i+\alpha_j)}{ \theta_3(\alpha_i+\alpha_j) }+\frac{\theta_3'(\alpha_k+\alpha_l)}
{ \theta_3(\alpha_k+\alpha_l) }=- 2\imath \pi N.
\end{align}

\begin{proposition}\label{twoptrelations}
Let  $i,j,k,l\in \{1,2,3,4\}$ be distinct and $P_{i,j,k,l}$ be points on the curve corresponding to
solutions of the Atiyah-Ward equation. Let
 $\alpha_i$ be  the Abel image of $P_i$ subject to
$\alpha_1+\alpha_2+\alpha_3+\alpha_4=N\tau$.
Then
\begin{align}
&\frac{\theta_3'(\alpha_i+\alpha_j)}{ \theta_3(\alpha_i+\alpha_j) } 
=2(\mu_i+\mu_j) \;\; \mathrm{mod} (2\imath \pi ) ,\label{2varlogdiff1}\\
&\frac{\theta_1'(\alpha_i+\alpha_j)}{ \theta_1(\alpha_i+\alpha_j) } 
=2(\mu_i+\mu_j)+ \frac{4\imath (x_-\zeta_i\zeta_j-x_+)}{ \zeta_i+\zeta_j} \;\;
 \mathrm{mod} (2\imath \pi ) ,\label{2varlogdiff2}\\
\begin{split}
&\frac{\theta_3''(\alpha_i+\alpha_j)}{ \theta_3(\alpha_i+\alpha_j) } 
= 4(2\imath x_3-x_-(\zeta_k+\zeta_l))^2+4(\imath \pi N-\mu_k-\mu_l)^2\\ &\hskip1.98cm -K^2(\zeta_k+\zeta_l)^2-4K(E-K{k'}^2)  \;\; \mathrm{mod} (2\imath \pi ) ,\end{split} \label{2varlogdiff3}\\
\begin{split}&\frac{\theta_1''(\alpha_i+\alpha_j)}{ \theta_1(\alpha_i+\alpha_j) } 
=\frac{\theta_3''(\alpha_i+\alpha_j)}{ \theta_3(\alpha_i+\alpha_j) }
+4K^2\zeta_i^2\zeta_j^2\\&\hskip3cm +
16\imath (\beta_i+\beta_j) \frac{x_-\zeta_i\zeta_j-x_+}
{\zeta_i+\zeta_j}  \;\; \mathrm{mod} (2\imath \pi ) .\end{split}
\label{2varlogdiff4}
\end{align}

\end{proposition}

\begin{corollary} Let $i,j,k,l\in \{1,2,3,4\}$ be distinct  and $P_{i,j,k,l}$ be points on the curve corresponding to solutions of the Atiyah-Ward equation. Let $\alpha_i$ be  the Abel image of $P_i$ subject to
$\alpha_1+\alpha_2+\alpha_3+\alpha_4=N\tau$.
Then with the paths specified in Lemma (\ref{intgammainf})
 the following relation is valid
\begin{equation}
\frac{\theta_3'(\alpha_i+\alpha_j)}{ \theta_3(\alpha_i+\alpha_j) } 
=-2(\mu_k+\mu_j) +2\imath \pi N. \label{musum}
\end{equation}
\end{corollary}
Observe the consistency between (\ref{musuma}, \ref{twotheta3sum}) and (\ref{musum}).

\section{Special Loci}\label{sectionloci}
In this section we shall describe those loci in $\mathbb{R}^3$ where the general analysis simplifies, in particular describing the corresponding transcendent $\mu$. These loci include each of the coordinate axes
and we shall give an alternate parameterization for these involving Jacobi elliptic functions that will facilitate comparison with results in the literature.

For charge $2$ the Atiyah-Ward equation is a general quartic and while this is solvable, there are several
simplifying loci for which it becomes a biquadratic. This occurs for the $x_3$-axis and the $x_3=0$ plane,
the latter including the $x_{1,2}$ axes.

\medskip
\noindent{$\boldsymbol{x_3}$-\textbf{axis.}} Here $\eta=-2x_3\zeta$. For $x_3<Kk'/2$ we may take
$$\zeta _{{1}}={\frac {i\sqrt {{K}^{2}{k}^{2}+4\,{x_{{3}}}^{2}}+\sqrt {{
K}^{2}{{ k'}}^{2}-4\,{x_{{3}}}^{2}}}{K}}.
$$
Then $|\zeta_1|=1$ and so with $P_1=(\zeta _{{1}},\eta _{{1}})$ our ordering yields $P_3
=(-1/\overline{\zeta_1}, -\overline{\eta_1}/\overline{\zeta_1}^2)=-(\zeta _{{1}},\eta _{{1}})$. We may then take
$P_2=(\overline{\zeta_1}, \overline{\eta_1})$ and $P_4=-(\overline{\zeta_1}, \overline{\eta_1})$.
For this range $|\zeta_i|=1$.

In general we may set $\beta=k^2-k'^2+8x_3^2/K^2$ and
$$\zeta _{{1}}=i\left[ \sqrt{\frac{\beta+1}2} -  \sqrt{\frac{\beta-1}2}\right]
$$
which reduces to the previous. For $x_3>Kk'/2$ now  $\zeta_1$ is purely imaginary and
$$\zeta _{{3}}=-i\left[ \sqrt{\frac{\beta+1}2} + \sqrt{\frac{\beta-1}2}\right],\ 
\zeta _{{2}}=-i\left[ \sqrt{\frac{\beta+1}2} -  \sqrt{\frac{\beta-1}2}\right]=-\zeta_1=\overline{\zeta}_1,\
\zeta _{{4}}=-\zeta_3.
$$

\medskip
\noindent{$\boldsymbol{x_3=0}$-\textbf{plane.}} We give the solutions of the general plane $x_3=0$  before
specialising to the simpler cases of the $x_1$ and $x_2$ axes.
Upon setting $ x_\pm=x_1\pm\imath x_2$ and
$ \eta = - \imath x_-\zeta^2-\imath x_+$  the Atiyah-Ward equation becomes
\begin{equation*}
\left( \frac14 K^2 - x_-^2  \right)\zeta^4
+\left(\frac12 K^2-K^2{k'}^2-2x_+x_-  \right)\zeta^2+\left( \frac14 K^2 - x_+^2  \right)=0.
\end{equation*}
This has solutions  $\tilde \zeta_i(\boldsymbol{x})$, $i\in \{1,\ldots,4\}$  \begin{align*}
\tilde \zeta_1(x_1,x_2) &= S_3^{-1}\sqrt{2K^2{k'}^2 +2KS_1-S_4^2},
&\tilde \zeta_2(x_1,x_2) =& S_3^{-1}\sqrt{2K^2{k'}^2 -2KS_1-S_4^2},\\
\tilde  \zeta_3(x_1,x_2) &=-\tilde \zeta_1(x_1,x_2),
&\tilde  \zeta_4(x_1,x_2) =&-\tilde \zeta_2(x_1,x_2),
\end{align*}
where we have set
\begin{align*}
S_1&=S_1(x_1,x_2)=\sqrt{ -{k'}^2(K^2k^2-4 x_+x_-)+(x_+-x_-)^2},
&S_2&= S_2(x_1,x_2)=\sqrt{K^2-4 x_+^2},\\
S_3&=S_3(x_1,x_2)=\sqrt{K^2-4 x_-^2}, & S_4&=S_4(x_1,x_2)=\sqrt{K^2-4 x_+x_-}.
\end{align*}
We need to order the roots.
We are free to choose $\zeta_1=\tilde \zeta_1$ once and for all, but the choice of $\zeta_3:=-1/
 \overline{\zeta}_1$ depends on $(x_1,x_2)$ according to whether $\pm S_1^2 >0$. Noting that
 $\tilde \zeta_1 \tilde \zeta_2 =S_2/S_3$ then
 $$
  \overline{\zeta}_1=
 \begin{cases}\tilde \zeta_2 \dfrac{S_3}{S_2} =\dfrac1{\tilde \zeta_1}=-\dfrac1{\tilde \zeta_3}
 &\mbox{if } S_1^2<0,\\
 \tilde \zeta_1 \dfrac{S_3}{S_2} =\dfrac1{\tilde \zeta_2}=-\dfrac1{\tilde \zeta_4}
 &\mbox{if }  S_1^2>0.
 \end{cases}
 $$
 Thus
 $$\zeta_3 =\tilde \zeta_3 \mbox{ if } S_1^2<0,\qquad \zeta_3 =\tilde \zeta_4 \mbox{ if } S_1^2>0.$$
 We observe that if $S_1^2<0$ we have four roots $\zeta_i$ each of modulus 1.
 We are  similarly free to choose $\zeta_2 =\tilde \zeta_2$, giving a double root $\zeta_1=\zeta_2$
 when  $S_1^2=0$. As we cross $S_1^2=0$ the ordering of $P_3$ and $P_4$ interchanges:
 \begin{align*}
 S_1^2<0:& \qquad \zeta_1 =\tilde  \zeta_1 ,\   \zeta_2 =\tilde  \zeta_2 ,\   \zeta_3 =\tilde  \zeta_3 ,\ 
  \zeta_4 =\tilde  \zeta_4 ,\\
  S_1^2>0:&\qquad  \zeta_1 =\tilde  \zeta_1 ,\   \zeta_2 =\tilde  \zeta_2 =1/\overline{\zeta}_1,\   \zeta_3 =\tilde  \zeta_4 ,\ 
  \zeta_4 =\tilde  \zeta_3.\\
 \end{align*}
 The solutions are ordered to satisfy 
\begin{align}\begin{split}
\zeta_1(\boldsymbol{0})& = k'+\imath k, \quad\hskip0.25cm  \zeta_2(\boldsymbol{0}) = k'-\imath k,\\
\zeta_3(\boldsymbol{0}) &= -k'-\imath k, \quad  \zeta_4(\boldsymbol{0}) =- k'+\imath k.
\end{split} \label{ordering}
\end{align}

When we restrict to the coordinate axes we can say more.

\medskip
\noindent{$\boldsymbol{x_2}$-\textbf{axis.}} On the $x_2$-axis $ S_1^2<0$ and so each $|\zeta_i|=1$.
Further, as $\eta =x_2(1-\zeta^2)$ is invariant under complex conjugation we have the four
points
\begin{align}\label{x2axisps}
P_1&=(\zeta, \eta),\  &P_3&=(-1/\overline{\zeta}, -\overline{\eta}/\overline{\zeta}^2)=(-\zeta, \eta),
&P_2&=(\overline{\zeta}, \overline{\eta}),
&P_4&=(-1/\zeta, -\eta/\zeta^2)=(-\overline{\zeta}, \overline{\eta}).
\end{align}
Explicitly we take
\begin{align}
\zeta_1=\frac{\imath Kk + \sqrt{K^2{k'}^2+4x_2^2}}{\sqrt{K^2+4x_2^2}},\quad 
 \label{x2roots}
\end{align}

\medskip
\noindent{$\boldsymbol{x_1}$-\textbf{axis.}} 
Finally consider the $x_1$-axis. Here $\eta =-\imath x_1(1+\zeta^2)$ is invariant under
$ (\zeta, \eta)\rightarrow (\pm\overline{\zeta}, -\overline{\eta})$.  
We have ordered the points
$P_1=(\zeta, \eta)$, $P_3=(-1/\overline{\zeta}, -\overline{\eta}/\overline{\zeta}^2)$, and with $P_4=\mathcal{J}(P_2)$ we are left to determine $P_2=(\zeta_2,\eta_2)$. 
We determine these depending on the sign of $ S_1^2$.
\begin{align}\label{x1axispssmall}
&x_1\mbox{-axis: }&
& P_1=(\zeta, \eta),&  &P_3=(-1/\overline{\zeta}, -\overline{\eta}/\overline{\zeta}^2)=(-\zeta, \eta), \\
&S_1^2<0& &P_2=(\overline{\zeta}, -\overline{\eta}),&
&P_4=(-1/\zeta, \eta/\zeta^2)=(-\overline{\zeta}, \overline{\eta}).\nonumber
\end{align}
We take here
\begin{equation}
\zeta_{1}=\zeta_1(x_1)=\frac{Kk'+\imath \sqrt{K^2k^2-4x_1^2}}{\sqrt{K^2-4x_1^2}}. \label{x1axis}
\end{equation}
For $ S_1^2>0$ we have chosen $\overline{\zeta_1}=1/\zeta_2$.
\begin{align}\label{x1axispsbig}
&x_1\mbox{-axis: }&
& P_1=(\zeta, \eta),&  &P_3=(-1/\overline{\zeta}, -\overline{\eta}/\overline{\zeta}^2), \\
&S_1^2>0& &P_2=(1/\overline{\zeta},
-\overline{\eta}/\overline{\zeta}^2),&
&P_4=(-\zeta, \eta).\nonumber
\end{align}
We remark that the solutions $\zeta_i(\boldsymbol{x})=\zeta_i(x_1,x_2) $ are singular on the $x_1$-axis
when
\begin{equation}
\zeta_i(K/2,0)=\infty, \quad i=1,\ldots,4.
\end{equation}
At this point the degree of the AWC is no longer a quartic. With  $0\leq k \leq 1$ we see that  $\zeta_1$ has a different analytic behaviour  for each of the following domains of  the $x_1$-axis
\begin{align*}\mathbf{I}: \hskip1cm& x_1\in\left[ - Kk/2 , \;  kK/2\right] ,\\
\mathbf{II}: \hskip1cm &x_1\in \left[ - K/2 , \; - kK/2 \right]   \cup  \left[  Kk/2 ,\;   K/2 \right] , \\
\mathbf{III}: \hskip1cm &x_1\in \left[ - \infty ,\;  - K/2 \right]   \cup  \left[  K/2 , \;  \infty\right] .
 \end{align*}
The interval {\bf I} corresponds to $S_1^2<0$ and we have $\overline{\zeta}_1=\zeta_2$.  
On interval {\bf II} we have that  $\zeta_1$ is real while on interval {\bf III}  it is purely imaginary. It follows from $\eta=-\imath x_1(1+\zeta^2)$ that $\eta$ is purely imaginary on intervals  {\bf II}  and  {\bf III}.

\subsection{\texorpdfstring{$\boldsymbol{\mu_i}$}{muis}}
We now calculate the transcendents $\mu$ for these loci.

\medskip
\noindent{$\boldsymbol{x_2}$-\textbf{axis.}} 
The invariance under complex conjugation gives
\begin{align}
\overline{\alpha(P_1)}&=\overline{\int_{\infty_1}^{P_1} v}
=\int_{\overline{\infty_1}}\sp{\overline{P_1}}v=
\int_{\overline{\infty_1}}\sp{{P_2}}v=
\int_{{\infty_1}}\sp{{P_2}}v+\int_{\overline{\infty_1}}\sp{{\infty_1}}v=
\alpha(P_2)+\int_{{\infty_2}}\sp{{\infty_1}}v\nonumber 
=
\alpha(P_2)+\frac{1+\tau}2 ,
\end{align}
where we have used our definition of sheets to give $\overline{\infty_1}={\infty_2}$. 
Therefore, from (\ref{betadef2}), we have that
$$\overline{\beta(P_1)}=\beta(P_2) -\frac{\imath\pi}2$$
and consequently from  (\ref{defmu}) and that $\overline{\zeta}_1=\zeta_2$ we see that
$$\overline{\mu}_1= \mu_2  - \frac{\imath\pi}2.$$
Combining this with Proposition (\ref{mux1x2starplane}) we  obtain

\begin{proposition} \label{propmux2axis}
The transcendents $\mu_i$ on the $x_2$-axis are given$\pmod{\imath\pi}$  by:
\begin{align}\label{muconjugate2}
 \mu_1=\lambda_2+\frac{\imath\pi}4, \
  \mu_2=\lambda_2+\frac{\imath\pi}4, \
 \mu_3=-\lambda_2-\frac{\imath\pi}4,\
 \mu_4=-\lambda_2-\frac{\imath\pi}4, \  \lambda_2\in \mathbf{ R}.
\end{align}
\end{proposition}
Thus on the $x_2$-axis there is only one transcendental function $\lambda_2=\lambda_2(x_2)$ to evaluate;
we shall identify this function shortly.
Now
\begin{align}\label{muorigin}
 \mu_1(\boldsymbol{0})=  \mu_2(\boldsymbol{0}) = \frac{\imath \pi}{4}, \quad  
\mu_3(\boldsymbol{0})=  \mu_4(\boldsymbol{0}) = \frac{3\imath \pi}{4}. 
\end{align}
The first equality is evident as $\zeta_1(\boldsymbol{0})=k'+\imath k=a$. 
To prove $\mu_2(\boldsymbol{0}) = {\imath \pi}/{4}$ we note 
\[ \frac{\imath K}{2} \int_{k'+\imath k}^{k'-\imath k} \frac{(z^2-c)\mathrm{d}z}{\sqrt{(x^2-a^2)(x^2-b^2)}}=0   \]
because of the normalization condition $\int_{\mathfrak{a}}\gamma_\infty=0$. The remaining equalities follow from the propositions. Thus we have that $\lambda_2(0)=0$. 

\medskip
\noindent{$\boldsymbol{x_1}$-\textbf{axis.}} 
The behaviour of $\mu$ depends on which interval $x_1$ belongs. A similar proof to the above
(given in Appendix \ref{proofmux1axis}) shows that:

\begin{proposition}\label{propmux1axis}
The transcendents $\mu_{1,2}$  on the $x_1$-axis behave$\pmod{\imath\pi}$ as
\begin{align}
{\bf I}&:
\mu_1(x_1,0) = \lambda_1 +\frac{ \imath \pi}{4}, \quad \mu_2(x_1,0) = -\lambda_1 +\frac{ \imath \pi}{4}, 
 \mu_3=-\lambda_1 -\frac{\imath\pi}4,\
 \mu_4=\lambda_1 -\frac{\imath\pi}4,\nonumber
\\
{\bf II}&: \mu_1=\imath\lambda'_1, \  \mu_2=-\imath\lambda'_1-\frac{\imath\pi}2, \
 \mu_3=\imath\lambda'_1+\frac{\imath\pi}2,\
 \mu_4=-\imath\lambda'_1,
 \label{conjugate1} \\
{\bf III}&:
\mu_1=\lambda_1'', \
 \mu_2=\lambda_1''-\frac{\imath\pi}2, \
 \mu_3=-\lambda_1''+\frac{\imath\pi}2,\
 \mu_4=-\lambda_1'',\nonumber
\end{align}
where $\lambda_1=\lambda_1(x_1)$, $\lambda'_1=\lambda'_1(x_1)$, $\lambda_1''=\lambda''(x_1) \in \mathbb{R}$ are such that $$0=\lambda_1(0)=\lambda_1(\pm Kk/2)=\lambda_1''(\pm K/2),\qquad \lambda'_1(\pm Kk/2)=\frac{\pi}4,$$
and $\lambda'_1(\pm K/2)=0\pmod{\pi}$.
\end{proposition}

\medskip
\noindent{$\boldsymbol{x_3=0}$-\textbf{plane.}} More generally the same symmetry arguments together
with (\ref{x2axisps}),  (\ref{x1axispssmall}) and  (\ref{x1axispsbig}) show that 
\begin{corollary}\label{muaxes}
 For the plane $x_3=0$ and the choices  above, we have that$\pmod {\imath \pi }$
 \begin{align*}
 S_1^2<0:&\quad
 \mu_1=\lambda+\frac{\imath\pi}4, \
 \mu_3=-\lambda-\frac{\imath\pi}4,\
  \mu_2=\lambda'+\frac{\imath\pi}4, \
 \mu_4=-\lambda'-\frac{\imath\pi}4, \\
   S_1^2>0:&\quad  
 \mu_1=\lambda''+\imath\alpha, \
 \mu_3=-\lambda''+\imath\alpha+\frac{\imath\pi}2,\
  \mu_2=\lambda''-\imath\alpha-\frac{\imath\pi}2, \
 \mu_4=-\lambda''-\imath\alpha,
 \end{align*}
where $\lambda, \lambda', \lambda'', \alpha\in \mathbb{R}$.
\end{corollary}
We observe that Proposition \ref{propmu13mu24} tells us that in the general case we have two complex functions to consider; Proposition \ref{mux1x2starplane} and Corollary  \ref{muaxes}
reduces this to two real functions. Although the transcendents $\mu_1(\boldsymbol{x})$ and  $\mu_2(\boldsymbol{x})$ are analytically independent they are related when we reduce to either of the axes $(x_1,0)$,  $(0,x_2)$.

We record:
\begin{proposition}\label{propmuderivatives}
The derivatives of $\mu_{1,2}$ are
\begin{align}\begin{split}
\frac{\partial}{\partial x_1} \mu_1(x_1,x_2) &= -\frac{2\imath (\imath (E-K)x_2-Ex_1)(-\imath S_1 x_2+K {k'}^2x_1)K\zeta_1}
{S_1(x_-\zeta_1^2+x_+)x_+S_3^2}\\
&-\frac{\imath(2EK-S_4^2)(K {k'}^2+S_1)x_1}{S_1\zeta_1x_+S_3^2} ,\\
\frac{\partial}{\partial x_1} \mu_2(x_1,x_2) &= \frac{2\imath (\imath (E-K)x_2-Ex_1)(\imath S_1 x_2+K {k'}^2x_1)K\zeta_2}
{S_1(x_-\zeta_2^2+x_+)x_+S_3^2}\\
&+\frac{\imath(2EK-S_4^2)(K {k'}^2-S_1)x_1}{S_1\zeta_2x_+S_3^2} ,\\
\frac{\partial}{\partial x_2} \mu_1(x_1,x_2) &= -\frac{2\imath (\imath (E-K)x_2-Ex_1)(\imath S_1 x_1-K {k}^2x_2)K\zeta_1}
{S_1(x_-\zeta_1^2+x_+)x_+S_3^2}\\
&-\frac{\imath(2EK-S_4^2)(-K {k}^2+S_1)x_2}{S_1\zeta_1x_+S_3^2} ,\\
\frac{\partial}{\partial x_2} \mu_2(x_1,x_2) &= -\frac{2\imath (\imath (E-K)x_2-Ex_1)(\imath S_1 x_1+K {k}^2x_2)K\zeta_2}
{S_1(x_-\zeta_2^2+x_+)x_+S_3^2}\\
&-\frac{\imath(2EK-S_4^2)(K {k}^2+S_1)x_2}{S_1\zeta_2x_+S_3^2}.
\end{split}\label{muderivatives}
\end{align}
\end{proposition}

\medskip
\noindent{$\boldsymbol{x_3}$-\textbf{axis.}} We show in Appendix \ref{proofpropmux3axis}
that

\begin{proposition}  \label{propmux3axis}
The transcendents $\mu_i$ on the $x_3$-axis are given$\pmod{\imath\pi}$  by:
\begin{align*}\label{muconjugate2}
x_3<\frac{Kk'}2&:\quad
 \mu_1=\imath \lambda&
  \mu_2&=- \imath\lambda-\frac{\imath\pi}2,&
 \mu_3&= \imath\lambda+\frac{\imath\pi}2,&
& \mu_4=-\imath\lambda, \qquad  \lambda\in  \mathbf{ R},\\
 x_3>\frac{Kk'}2&:\quad
  \mu_1=\lambda'_3-\frac{\imath\pi}2,& 
  \mu_2&=\lambda'_3,&
 \mu_3&=-\lambda'_3,&
& \mu_4=-\lambda'_3+\frac{\imath\pi}2, \quad  \lambda'_3\in \mathbf{ R}.
\end{align*}
\end{proposition}
Again there is only one transcendental function  to evaluate.

\subsection{Parameterizing of the axes in terms of Jacobi's Elliptic functions}\label{axesandjacobi}
To compare with existing results shall need in the sequel 
to parameterize the axes in terms of Jacobi's Elliptic functions. For reasons that will later be clearer we
take
\begin{align*}
x_1:&  &  \mathrm{sn}^2(t)&=\frac{4x_1^2}{k^2 K^2},
& \zeta^2&= \frac{k'^2- k^2  \mathrm{cn}^2(t)\pm 2\imath k k'  \mathrm{cn}(t)}{\mathrm{dn}^2(t)},
& \zeta&= \pm \frac{k' \pm \imath k\,  \mathrm{cn}(t)}{\mathrm{dn}(t)},
\\
x_2:& &  \mathrm{dn}^2(t)&=-\frac{4x_2^2}{K^2}
& \zeta^2&=1+ \frac{-2\pm 2 \, \mathrm{cn}(t)}{ \mathrm{sn}^2(t) },
& \zeta&= \pm \imath  \frac{1 \pm \mathrm{cn}(t)}{\mathrm{sn}(t)},
\\
x_3:&  &  \mathrm{cn}^2(t)&=-\frac{4x_3^2}{k^2 K^2},
& \zeta^2&= 2\, \mathrm{dn}^2(t)-1 \pm 2\imath k\, \mathrm{sn}(t) \mathrm{dn}(t),
& \zeta&= \pm (  \mathrm{dn}(t) \pm \imath k\, \mathrm{sn}(t) ).
\end{align*}
There are $4$ choices of $\zeta$ in each of the above corresponding to the $4$ signs, each point giving
a solution to the Atiyah-Ward equation.
Given one solution the other solutions are generated by
 $t\rightarrow t+2K$ and $t\rightarrow t+ 2\imath K'$
so only one solution must be found. Further we note that if
\begin{equation}
\label{transaxes}
t=t' +\imath K'=t''+K+\imath K' \quad\rm{then}\quad
\imath  \frac{1 + \mathrm{cn}(t)}{\mathrm{sn}(t)}= \mathrm{dn}(t') + \imath k\, \mathrm{sn}(t') 
=\frac{k'  + \imath k\,  \mathrm{cn}(t'')}{\mathrm{dn}(t'')}.
\end{equation}

Previously we have parameterised 
$\zeta=-i\,{\theta_2[P]\theta_4[P]}/{\theta_1[P]\theta_3[P]}$,
{and now we have, for example on the $x_2$-axis }
$$\zeta=\pm \imath  \frac{1 \pm \mathrm{cn}(t)}{\mathrm{sn}(t)}=
\pm \imath\, \frac{ \theta_2 \theta_4(z)\pm \theta_4 \theta_2(z)}{\theta_3 \theta_1(z)}
$$
where $z=t/(2 K)$. Can we relate the Abel images of $P$ to $z$? We show in Appendix \ref{proofjacobipar}
\begin{proposition}\label{jacobipar} 
The Jacobi parameterizations of the axes, given above, follow upon taking
$z=-2\,\alpha(P)-1/2-\tau/2$ for the $x_1$-axis, 
$z=-2\,\alpha(P)$ for the $x_2$-axis and
$z=-2\,\alpha(P)-\tau/2$ for the $x_3$-axis.
\end{proposition}

\subsection{Expressions for \texorpdfstring{$\mu$}{mu} in terms of the Jacobi Zeta function}\label{muandjacobi}
We have as yet to identify the functions $\lambda_{1,2,3}$ beyond their definition: we do this now.
Set
\begin{align}
\lambda_1&:= \frac12 K Z( \mathrm{sn}^{-1}\left( \frac{2x_1}{k K},k \right),k ),\label{deflambda1}\\
\lambda'_2 &:= \frac12 K Z( \mathrm{dn}^{-1}\left( \frac{2\imath x_2}{ K},k \right),k ),\label{deflambda2p}\\
\lambda_3& := \frac12 K Z( \mathrm{cn}^{-1}\left( \frac{2\imath x_3}{ K k},k \right),k ).\label{deflambda3}
\end{align}
where $Z(v)$ is the Jacobi Zeta function. Then

\begin{lemma}
\begin{align*}
\frac{\mathrm{d} \mu_1(x_1,0,0)}{\mathrm{d}  x_1} &=
\frac{\mathrm{d} \lambda_1}{\mathrm{d}  x_1} 
=-\frac{EK-K^2+4x_1^2}{\sqrt{K^2-4x_1^2}\sqrt{K^2{k}^2-4x_1^2}},
\\
\frac{\mathrm{d} \mu_1(0,x_2,0)}{\mathrm{d}  x_2} &=
\frac{\mathrm{d} \lambda'_2}{\mathrm{d}  x_2} =- \frac{EK+4x_2^2}{\sqrt{K^2+4x_2^2}\sqrt{K^2{k'}^2+4x_2^2}},
\\
\frac{\mathrm{d} \mu_1(0,0,x_3)}{\mathrm{d}  x_3} &=
\frac{\mathrm{d} \lambda_3}{\mathrm{d}  x_3}
=\frac{K^2 k^2+E K-K^2+4 x_3^2}{\sqrt{K^2 k^2+4x_3^2}\sqrt{-K^2{k'}^2+4x_3^2}}.
\end{align*}
\end{lemma}
This may be proven directly or via (\ref{muderivatives}). For example,
\begin{align*}
\frac{\mathrm{d} \mu_1}{\mathrm{d}  x_2}= \frac{  \mathrm{d} \beta_1}{ \mathrm{d} \alpha_1}
 \frac{  \mathrm{d} \alpha_1}{ \mathrm{d} \zeta_1} \frac{  \mathrm{d} \zeta_1}{ \mathrm{d} x_2}-\zeta_1
 - x_2\frac{\mathrm{d} \zeta_1}{ \mathrm{d} x_2} ,
\end{align*}
and $\mathrm{d} \beta_1/{ \mathrm{d} \alpha_1}$ was determined in (\ref{diffbetac}) while from 
(\ref{curveident3}) we have
$\mathrm{d} \alpha_1/{ \mathrm{d} \zeta_1} =1/(4\eta_1)$. Finally the ${\mathrm{d} \zeta_1}/{ \mathrm{d} x_2}$ term comes by implicit differentiation of the Atiyah-Ward equation. After simplification we obtain the derivative given above. 

It follows then that $\mu_1(0,x_2,0)=\lambda'_2+$constant and similarly for the other axes. The constant
may be determined by comparison at the origin. 
(There is a choice of square root here that may be appropriately chosen.) Now $\dn(K\pm i K')=0$ and
$ K Z(K \pm i K')/2 = \mp i\pi/4$ and we may identify $\mu_1(0,x_2,0)=\mp \lambda'_2(x_2)$. The choice of
sign ultimately makes no difference (it will correspond to the symmetry of the $x_2$ Higgs field about the origin) and we choose  $\mu_1(0,x_2,0)= \lambda'_2(x_2)$.
For the identification on the $x_3$ axis we recall that $\cn(K)=0$ and $ K Z(K )/2 = 0$.
Noting (\ref{muorigin}) we find\footnote{We remark that Maple's inbuilt function ${\tt InverseJacobiDN}$
has the unwanted behaviour 
$${\tt Re( InverseJacobiDN(JacobiDN(x, k), k))= |x|}$$
and must be used cautiously.
}
\begin{align}
\mu_1(x_1,0,0)&=\lambda_1+\frac{\imath \pi}4, \qquad\qquad (x_1\in {\bf I}),\nonumber\\
\quad  \mu_1(0,x_2,0)&=\lambda_2+\frac{\imath \pi}4= \lambda'_2(x_2) ,
\label{defmus}\\
\quad  \mu_1(0,0,x_3)&=\lambda_3+\frac{\imath \pi}4=\imath \lambda, \qquad (x_3 <Kk'/2).\nonumber
\end{align}
We plot these in Figures \ref{lambdax1}, \ref{lambdax2}, \ref{lambdax3} respectively and note that the points
of discontinuity on the $x_1$ and $x_3$ axes correspond to the vanishing of the Atiyah-Ward discriminant
corresponding to points of bitangency; these will be described in the next section.
\begin{figure}
\begin{center}
\includegraphics[width=0.4\textwidth]{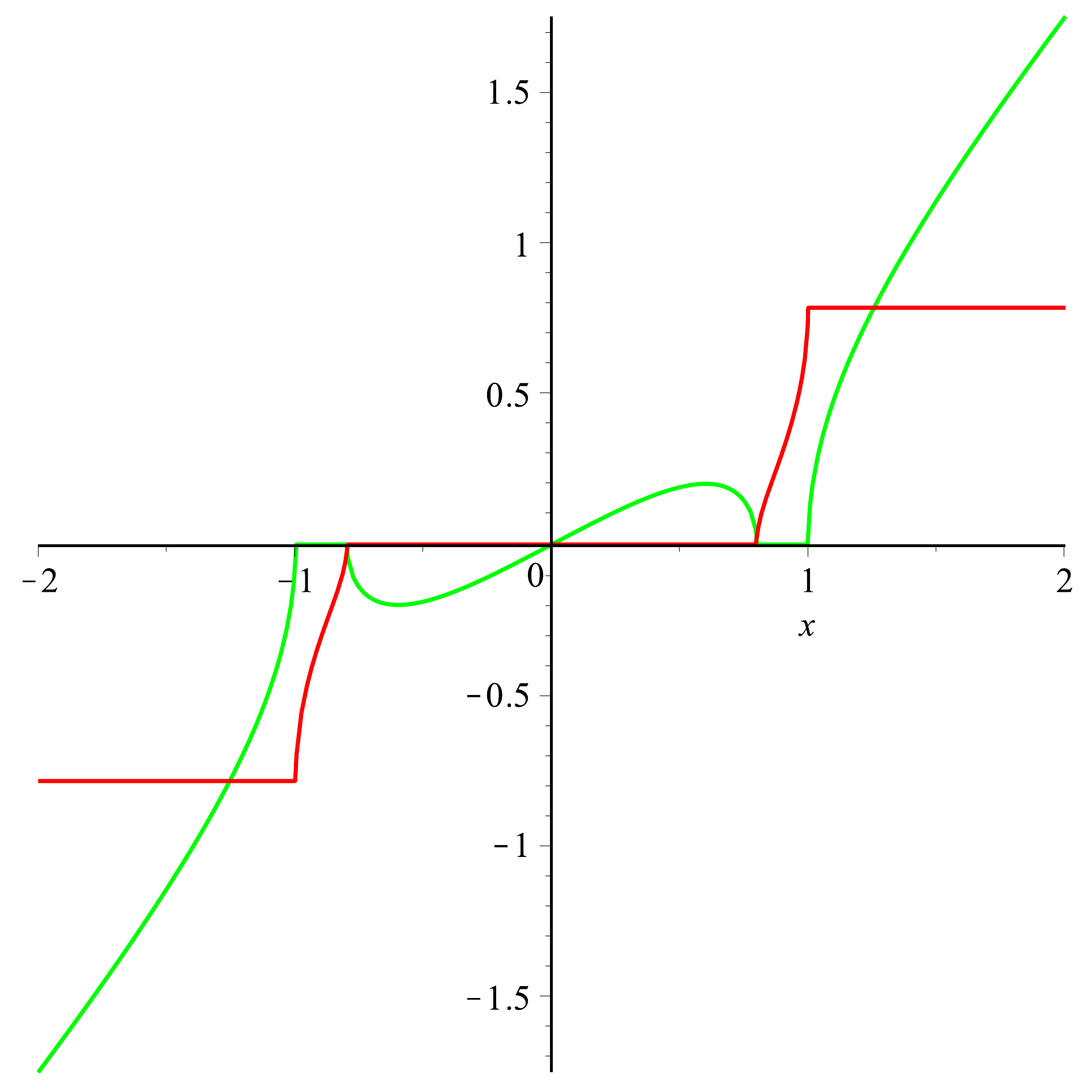}
\end{center}
\caption{The real (green) and imaginary (red) parts of $\lambda_1$ restricted to $x_1$-axis  $k=0.8$. 
The points of discontinuity correspond to  points of bitangency.}
\label{lambdax1}
\end{figure}
\begin{figure}
\begin{center}
\includegraphics[width=0.4\textwidth]{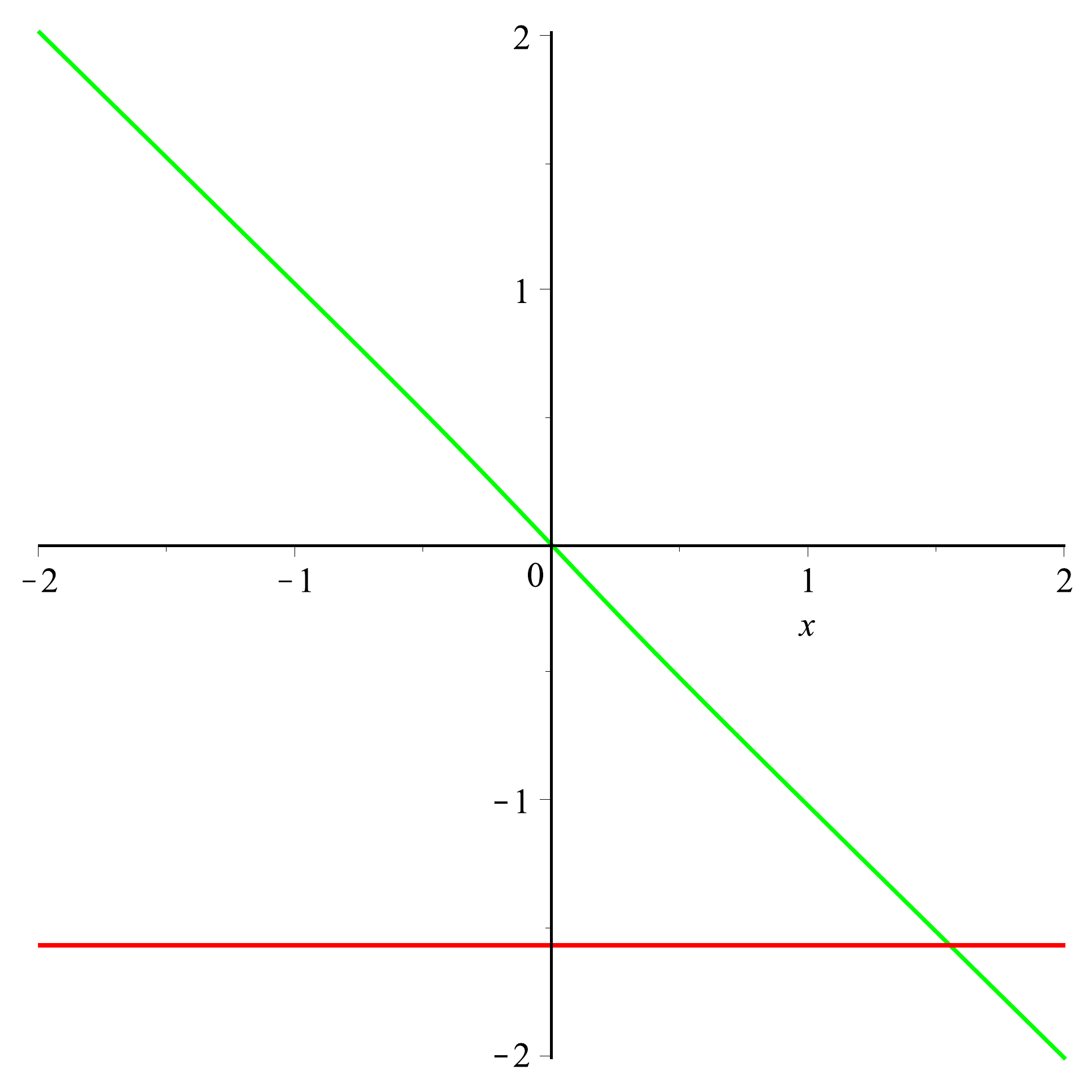}
\end{center}
\caption{The real (green) and imaginary (red) parts of $\lambda_2$ restricted to $x_2$-axis  $k=0.8$ }
\label{lambdax2}
\end{figure}
\begin{figure}
\begin{center}
\includegraphics[width=0.4\textwidth]{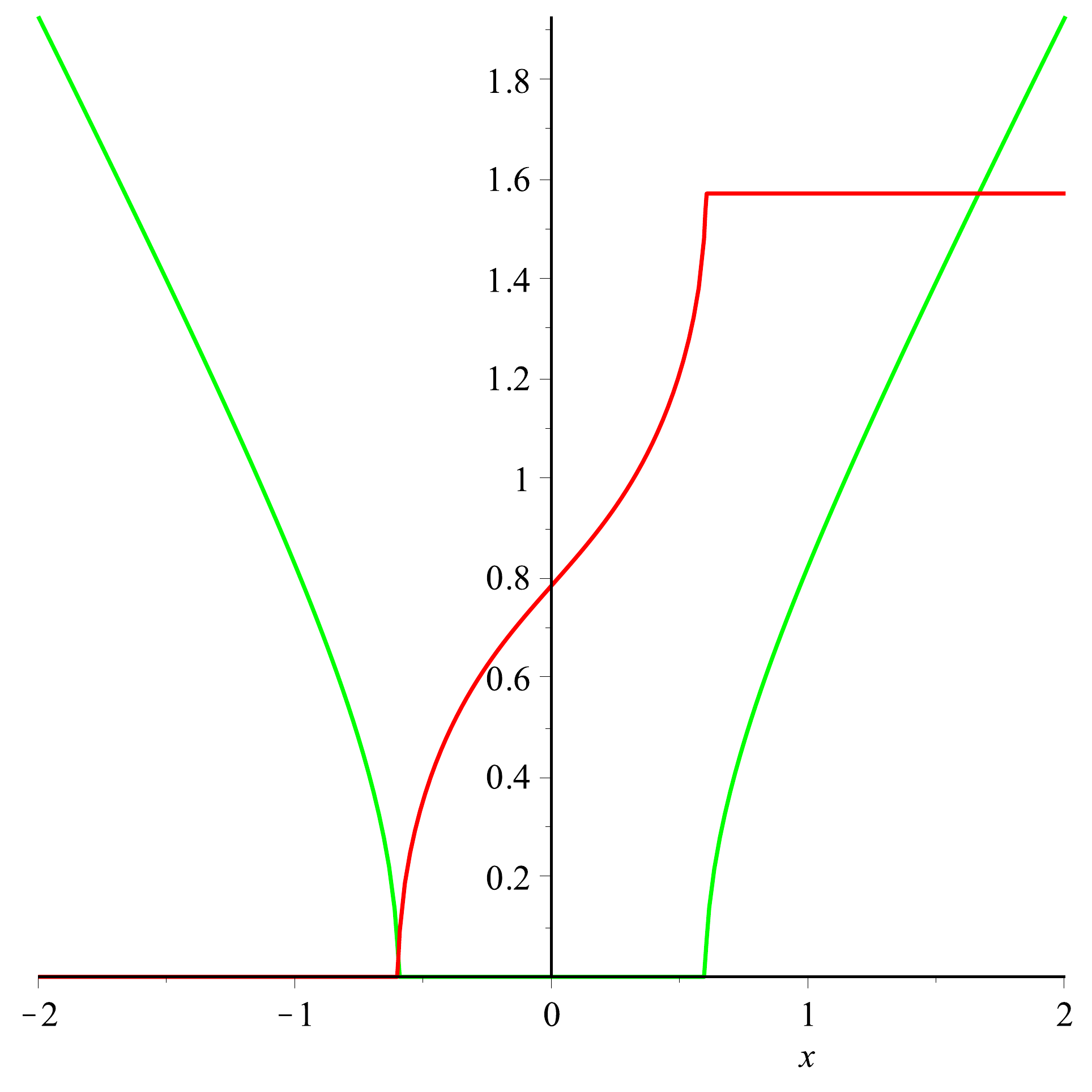}
\end{center}
\caption{The real (green) and imaginary (red) parts of $\mu_1$ restricted to $x_3$-axis  $k=0.8$.
The points of discontinuity correspond to  points of bitangency. }
\label{lambdax3}
\end{figure}

\section{The points of bitangency}\label{sectionbitangency}

Substituting the Atiyah-Ward constraint 
$$\eta=(x_2-\imath x_1)-2\zeta x_3-(x_2+\imath x_1)\zeta^2$$
into the equation of the spectral curve $P(\zeta,\eta)=0$ gives an equation of (in general) degree $2n$ in $\zeta$ and generically this has $2n$ solutions. There are however loci of points $\boldsymbol{x}\in \mathbb{R}^3$  for which we have multiple roots, the points of bitangency we have referred to;
calculating the discriminant of our equation describes the locus. This discriminant is of degree $2n$ in the $x_i$'s but because of the reality conditions it may be expressed as a polynomial of degree $n$ in the $x_i^2$'s with no odd power of $x_i$ appearing.
Focussing on the $n=2$ case  the discriminant $Q$ is a quartic in $X=x_1^2$, $Y=x_2^2$, $Z=x_3^2$. This discriminant is not very enlightening: we find  real loci
\begin{align}
0&={K}^{2}{{\it k'}}^{2}-4\,{{\it k'}}^{2}{x_{{1}}}^{2}-4\,{x_{{3}}}^{2}\label{bitangency1}
\\
{x_{{1}}}^{2}&=\frac14\,{\frac { \left( {K}^{2}{{\it k'}}^{2}+4\,{x_{{2}}}^
{2} \right) {k}^{2}}{{{\it k'}}^{2}}}.\label{bitangency2}
\end{align}
Thus on the $x_1$ axis there are the 4 solutions $\pm K/2$, $\pm Kk/2$, no solutions on the
$x_2$ axis and $\pm Kk'/2$ on the $x_3$ axis. These are the points we have previously noted.
We remark that in the case of (\ref{bitangency1}) we obtain double roots
with modulus less than one and greater than one while for (\ref{bitangency2}) we find the roots of the 
Atiyah-Ward constraint have modulus one. Hurtubise \cite{hurtubise85a} in his study of the asymptotics
of the Higgs field gave the first of these loci. In Appendix \ref{comparingcurves} we will relate  Hurtubise's curve and our own.

\section{The Matrix \texorpdfstring{${W}$}{W}}\label{sectionW}
In this section we will determine the matrix $W$  (\ref{nahmW}) and its inverse. We shall calculate $W\sp{-1}$ via cofactors and this will involve a number of the identities established in the previous section. We will also look at the behaviour of $W$ at the $z=\pm1$.

The $4\times4$ matrix
\[ W =\left(\boldsymbol{w}\sp{(k)}(z,x)\right)\]
is constructed from
\begin{align*}
\boldsymbol{w}\sp{(k)}(z,x)
&=\left(1_2\otimes C(z)\right) \left(
(1_2+\boldsymbol{\widehat{ u}}(\zeta)\cdot\boldsymbol{\sigma})\,
e\sp{-i z\left[(x_1-i x_2)\zeta-i x_3 -x_4\right]}|s>\otimes\,
\boldsymbol{\Phi}(z,P_k)
\right)
\end{align*}
where $P_k=(\zeta_k,\eta_k)$ are solutions to the mini-twistor constraint. From (\ref{bafull}, \ref{bachi})
the Baker-Akhiezer function takes the form
\[
\boldsymbol{\Phi}(z,P_k)= \begin{pmatrix}a_k \\ b_k \end{pmatrix}\,\mathcal{D}'_k
\]
where (again with $\alpha_k:=\alpha(P_k)$) we have
\begin{equation}\label{abbadef}
a_k=-\theta_3(\alpha_k) \theta_2(\alpha_k-z/2),\quad b_k=
\theta_1(\alpha_k) \theta_4(\alpha_k -z/2),
\end{equation}
and
$$
\mathcal{D}'_k= \frac{\theta_2(1/4)\theta_3(1/4)}{\theta_3(0)
\theta_1(\alpha_k-1/4)\theta_4(\alpha_k+1/4)}\,
\frac{e\sp{\beta_1(P_k) z}}{\theta_2(z/2)}
.
$$
Now the kernel of
$$1_2+\boldsymbol{\widehat{ u}}(\zeta)\cdot\boldsymbol{\sigma}=
\frac{2}{1+|\zeta|^2}\, \begin{pmatrix}1&0\\ 0&\im\zeta \end{pmatrix}
 \begin{pmatrix} 1& -\im {\bar \zeta}\\  1& -\im {\bar \zeta}\end{pmatrix}
$$
has basis $\begin{pmatrix}\im {\bar \zeta}\\ 1\end{pmatrix}$ for finite $\zeta$ and 
 $\begin{pmatrix}1\\0 \end{pmatrix}$ for the infinite case. Thus we may take for our construction
\[|s>=\begin{pmatrix}\frac12\\ 0\end{pmatrix}\]
for all directions apart from $\zeta=\infty$ which gives us 
\[
(1_2+\boldsymbol{\widehat{ u}}(\zeta)\cdot\boldsymbol{\sigma})\,
[|s>=\frac{1}{1+|\zeta|^2}\begin{pmatrix}1\\ i\zeta\end{pmatrix}.
\]
We may therefore write
\begin{equation}\label{formW}
W=\left(1_2\otimes C(z)\right) \Psi\, \mathcal{D},\qquad \mathcal{D}=\diag(\mathcal{D}_k),
\end{equation}
where
\begin{equation}\label{formD}
\quad
\mathcal{D}_k=\frac{1}{1+|\zeta_k|^2}\, \frac{\theta_2(1/4)\theta_3(1/4)}{\theta_3(0)
\theta_1(\alpha_k-1/4)\theta_4(\alpha_k+1/4)}\,
\frac{e\sp{\beta_1(P_k) z-i z\left[(x_1-i x_2)\zeta_k-i x_3 -x_4\right]}}{\theta_2(z/2)}
\end{equation}
and
\begin{equation}\label{formphi}
\Psi=\left( \begin{pmatrix}1\\ i\zeta_k\end{pmatrix}\otimes
\begin{pmatrix}a_k \\ b_k \end{pmatrix}\right)
=\begin{pmatrix} a_{{1}}&a_{{2}}&a_{{3}}&a_{{4}}
\\ \noalign{\medskip}b_{{1}}&b_{{2}}&b_{{3}}&b_{{4}}
\\ \noalign{\medskip}i\zeta_{{1}}a_{{1}}&i\zeta_{{2}}a_{{2}}&i\zeta_{{3}}
a_{{3}}&i\zeta_{{4}}a_{{4}}\\ \noalign{\medskip}i\zeta_{{1}}b_{{1}}&i
\zeta_{{2}}b_{{2}}&i\zeta_{{3}}b_{{3}}&i\zeta_{{4}}b_{{4}}\end {pmatrix}.
\end{equation}

We remark that if we work with with the standard Nahm matrices (\ref{standardnahm}) then
the corresponding conjugation by $\mathcal{O}$ replaces 
${\widehat{ \boldsymbol{w}}(z)}=C(z)\boldsymbol{\Psi}(z,\zeta,\eta)$ by $\mathcal{O}\,C(z)\boldsymbol{\Psi}(z,\zeta,\eta)$ and consequently we would have
$$W=\left(1_2\otimes \mathcal{O}\,C(z)\right) \Psi\, \mathcal{D}.$$

We shall calculate $W\sp{-1}$ via cofactors. Upon noting that
\[ |W|=|C(z)|\, | \Psi|\, |\mathcal{D}|= | \Psi|\, |\mathcal{D}|\]
in the following subsections we shall calculate $|\mathcal{D}|$, ${|{\Psi}|}$ and finally
the adjoint $\adj\Psi$ but before turning to this a helpful check is to look at the pole structure of $W$
which provides a nontrivial check of the solution.

\subsection{The pole structure of \texorpdfstring{$\boldsymbol{W}$}{W}} 
We have seen in Lemma \ref{bafunction} that for the case at hand the Baker-Akhiezer function
only has simple poles at $z=\pm1$. (For general $n$ one has that $\dot h h\sp{-1}$ only has simple
poles \cite{ercolani_sinha_89, Braden2010d} though the the Baker-Akhiezer function has
higher order poles  \cite{Braden2010d}.) Using 
\begin{equation}\label{theta2expanp}
\theta_2\left( (1-\xi)/2\right)=\theta_1\left( \xi/2\right)=
\frac{\xi}{2}\,\theta_1'(0)+\frac{\xi^3}{48}\,\theta_1'''(0)+\mathcal{O}(\xi^5),\qquad
\theta_1'(0)
=\pi \,\theta_{{2}} \theta_{{3}} \theta_{{4}}, \end{equation}
we obtain (in the following $c$, $c'$ etc are constants)
\begin{align*}
\boldsymbol{\Phi}(1-\xi,P_1) &=\frac{2 c}{\xi} \begin{pmatrix}
-{\theta_{{3}} \left( P_{{1}}
 \right) \theta_{{2}} \left( P_{{1}}-1/2+\xi /2 \right) {{\rm e}^{
\beta \left( P_{{1}} \right)  \left( 1-\xi \right) }}}\\ 
{\theta_{{1}}
 \left( P_{{1}} \right) \theta_{{4}} \left( P_{{1}}-1/2+\xi /2
 \right) {{\rm e}^{\beta \left( P_{{1}} \right)  \left( 1-\xi \right) 
}}}\end {pmatrix}  +\mathcal{O}(\xi) \\
&=
\frac{ c}{\xi} \begin{pmatrix}
-2\,\theta_{{3}} \left( P_{{1}} \right) 
\theta_{{1}} \left( P_{{1}} \right)  +
\theta_{{3}} \left( P_{{1}} \right) \xi\, \left( 2\,\beta \left( P_{{1
}} \right) \theta_{{1}} \left( P_{{1}} \right) -
\theta_{{1}}'  \left( P_{{1}} \right)  \right) 
\\ 
\noalign{\medskip}2\,\theta_{{3}} \left( P_{{1}}
 \right) \theta_{{1}} \left( P_{{1}} \right)
 -\theta_{{1}} \left( P_{{1}} \right) \xi\, \left( 2\,\beta \left( P_{{
1}} \right) \theta_{{3}} \left( P_{{1}} \right) -
\theta_{{3}}'  \left( P_{{1}} \right)  \right) 
 \end {pmatrix} {{\rm e}^{\beta \left( P_{{1}}
 \right) }}
 +\mathcal{O}(\xi)
\end{align*}
and we see this simple pole behaviour. Now from (\ref{expancp}) and (\ref{Odef}) we have
\begin{align*}
\mathcal{O}\,C(1-\xi)&=\frac1{\sqrt {\xi} }\frac1{ \sqrt {K}} \begin{pmatrix}
\xi\,K/2&-\xi\,K/2
\\ \noalign{\medskip}1 &1
\end{pmatrix}
+\mathcal{O}(\xi^{3/2}),
\intertext{and consequently that}
\mathcal{O}\,C(1-\xi) \boldsymbol{\Psi}(1-\xi,P_1) &=
\frac{c'}{\sqrt {\xi} }\begin{pmatrix}
-2\,\theta_{{3}} \left( P_{{1}} \right) 
\theta_{{1}} \left( P_{{1}} \right)  K
 \\
\theta_{{1}} \left( P_{{1}} \right) \theta_{{3}}'  \left( P_{{1}} \right) -\theta_{{3}} \left( P_{{1}} \right) 
\theta_{{1}}' \left( P_{{1}} \right) 
\end{pmatrix}
+\mathcal{O}(\xi^{1/2}).
\intertext{Upon making use of (\ref{curveident1}) this takes the form}
&=
\frac{c_1}{\sqrt {\xi} }\begin{pmatrix}
1\\ \im \zeta_1 \end{pmatrix} +\mathcal{O}(\xi^{1/2}).
\end{align*}
Therefore the pole structure of the first column (and similarly the remaining columns) of $W$ takes the form
\[W_1(1-\xi)=\frac{c_1}{\sqrt {\xi} }\begin{pmatrix}
1\\ \im \zeta_1 \end{pmatrix}\otimes \begin{pmatrix}
1\\ \im \zeta_1 \end{pmatrix} +\mathcal{O}(\xi^{1/2})=
\frac{c_1}{\sqrt {\xi} }\begin{pmatrix}
1\\ \im \zeta_1 
\\ \im \zeta_1 \\ -\zeta_1^2 \end{pmatrix} +\mathcal{O}(\xi^{1/2}).
\]
Linear combinations of the columns of $W$ therefore give us three solutions at $z=1$ with singular behaviour 
$1/\xi^{1/2}$ proportional to the vectors $(1,0,0,0)\sp{T}$, $(0,0,0,1)\sp{T}$ and $(0,1,1,0)\sp{T}$. This
is what the discussion of section \ref{asympsection} requires, with the orthogonal direction, spanned by
$(0,1,-1,0)\sp{T}$, corresponding to the solution with $\xi^{3/2}$ behaviour. We note that to get this behaviour requires the use of the identity (\ref{curveident1}). We shall see
significantly more complicated identities are required when we examine the pole behaviour of $V$.
Using the $\theta$-constant relations
\begin{equation}
\frac{\theta_i''(0)}{\theta_i(0)}- \frac{\theta_j''(0)}{\theta_j(0)}=\pi^2 \theta_k^4(0)
\end{equation}
for $(ijk)\in\{ (4,3,2), (4,2,3), (3,2,4)\}$,
we find the analogous expansion near $z=-1$,
\begin{align*}
W_1(-1+\xi) &=  
\frac{c_1}{\sqrt{\xi}}\, 
 \left( \begin{array}{c}  \imath \zeta_1\\ 
-1 \\ - \zeta_1^2\\ - \imath \zeta_1 \end{array} \right)
\, e^{-2\left( {\beta_1(P_1) -i \left[(x_1-i x_2)\zeta_1-i x_3 -x_4\right]}  \right)}
+ \mathcal{O}(\xi^{1/2})
\\
&= \frac{c_1}{\sqrt{\xi}}\, 
\begin{pmatrix}
0&1&0&0\\ -1&0&0&0
\\ 0&0&0&1\\ 0&0&-1&0\end{pmatrix}
\begin{pmatrix}
1\\ \im \zeta_1 
\\ \im \zeta_1 \\ -\zeta_1^2 \end{pmatrix}
\, e^{-2\left( {\beta_1(P_1) -i \left[(x_1-i x_2)\zeta_1-i x_3 -x_4\right]}  \right)}
 +\mathcal{O}(\xi^{1/2}).
\end{align*}
The reason for our writing the expansion in the second form will be made clearer in due course.

\subsubsection{Higher expansion terms of $W$}
The identities of the last section yield further terms in the  expansion  of the matrix ${W}(\mp1\pm \xi)$, where
$W(z, \boldsymbol{x}) = (\boldsymbol{W}_1(z, \boldsymbol{x}), \ldots, \boldsymbol{W}_4(z, \boldsymbol{x})   )$.
Introduce the 2-vectors
\begin{align}
 \boldsymbol{w}_{0,k}&= \left(  \begin{array}{c}  1 \\ \imath \zeta_k \end{array}\right),
\qquad  \boldsymbol{w}_{1,k}=\left( \begin{array} {c} \imath x_-\zeta_k+x_3\\\\
 x_+ - \imath \zeta_k x_3 \end{array}\right),\\  \nonumber \\
 \boldsymbol{w}_{2,k}&=\left(\begin{array}{c}\frac18 \left(  K^2 -4x_-^2\right)\zeta_k^2 +\imath x_-x_3\zeta_k -\frac{1}{24}(2{k'}^2-1)K^2+\frac12x_3^2\\ \\
\left( -\frac{\imath}{24}( 2{k'}^2 -1)K^2 + \frac12 \imath x_3^2 \right)\zeta_k - x_+x_3 + 
\frac{\imath}{8\zeta_k} \left(   K^2 -4x_+^2  \right) 
\end{array}
\right),
\end{align}
and the 4-vectors
\begin{align}
\boldsymbol{W}_{0,k} = \boldsymbol{ w}_{0,k} \otimes  \boldsymbol{w}_{0,k}, \quad 
\boldsymbol{W}_{1,k}= 
 \boldsymbol{w}_{0,k} \otimes \boldsymbol{ w}_{1,k}, \quad  
\boldsymbol{W}_{2,k}= \boldsymbol{w}_{0,k} \otimes \boldsymbol{w}_{2,k}.
\end{align}
Then we find
\begin{align}\label{expanwk}\begin{split}
\boldsymbol{W}_k(1-\xi) &= c_k  \left\{  \frac{1}{\sqrt{\xi}} \boldsymbol{W}_{0,k}+ 
\sqrt{\xi}\boldsymbol{W}_{1,k}+ \xi^{3/2}\boldsymbol{W}_{2,k} +\mathcal{O}(\xi^{5/2})\right\} ,\\
\boldsymbol{W}_k(-1+\xi) &=U  c_k  \left\{ \frac{1}{\sqrt{\xi}} \boldsymbol{W}_{0,k}- 
\sqrt{\xi}\boldsymbol{W}_{1,k}+ \xi^{3/2}\boldsymbol{W}_{2,k} 
+\mathcal{O}(\xi^{5/2})\right\} \mathrm{e}^{-2\mu_k},
\end{split}
\end{align}
where $U=1_2\otimes \begin{pmatrix}0&1\\ -1& 0\end{pmatrix}$ is the matrix already encountered and
$c_k$ are constants that we need not specify.

\subsection{The determinant \texorpdfstring{$\boldsymbol{|\mathcal{D}|}$}{|D|}}

From (\ref{formD}) we obtain
\begin{align*}
|\mathcal{D}|&=\prod_{k=1}\sp4 \mathcal{D}_k
=\frac1{\theta_2^4(z/2)}\,\frac{\theta_2^4(1/4)\theta_3^4(1/4)}{\theta_3^4(0)}
\prod_{k=1}\sp4 \frac{1}{1+|\zeta_k|^2}\,
\frac{e\sp{\beta_1(P_k) z-i z\left[(x_1-i x_2)\zeta_k-i x_3 -x_4\right]}}{\theta_1(\alpha_k-1/4)\theta_4(\alpha_k+1/4)} .
\end{align*}
Recalling that we have ordered the roots so that
\[ (\zeta_1,\zeta_2,\zeta_3,\zeta_4)=(\zeta_1,\zeta_2,-1/\bar{\zeta_1},-1/\bar{\zeta_2})
\]
we consequently find that
\begin{align*}
\prod_{k=1}\sp4 \frac{1}{1+|\zeta_k|^2}&=\frac{ \zeta_1\zeta_2\zeta_3\zeta_4}{(\zeta_1-\zeta_3)^2(\zeta_2-\zeta_4)^2}\\
&=\frac1{\theta_2^2(0)\theta_4^2(0)}\,
\frac{ \prod_{k=1}\sp4 \theta_1(\alpha_k)\theta_2(\alpha_k)\theta_3(\alpha_k)\theta_4(\alpha_k)}
{\theta_1^2(\alpha_1-\alpha_3)\theta_3^2(\alpha_1+\alpha_3)\theta_1^2(\alpha_2-\alpha_4)\theta_3^2(\alpha_2+\alpha_4)}.
\end{align*}

\subsection{The determinant \texorpdfstring{$\boldsymbol{|{\Psi}|}$}|Psi}}
As the identities needed to evaluate  $\boldsymbol{|{\Psi}|}$ are illustrative of the more complicated
identities that are employed in calculating $\Adj\Psi$ we shall describe these in the text, and leave
the latter to  Appendix  \ref{proofthmadjpsi}.
We begin by observing that
\begin{align*}
|\Psi|&=\begin{vmatrix} a_{{1}}&a_{{2}}&a_{{3}}&a_{{4}}
\\ \noalign{\medskip}b_{{1}}&b_{{2}}&b_{{3}}&b_{{4}}
\\ \noalign{\medskip}\zeta_{{1}}a_{{1}}&\zeta_{{2}}a_{{2}}&\zeta_{{3}}
a_{{3}}&\zeta_{{4}}a_{{4}}\\ \noalign{\medskip}\zeta_{{1}}b_{{1}}&
\zeta_{{2}}b_{{2}}&\zeta_{{3}}b_{{3}}&\zeta_{{4}}b_{{4}}\end {vmatrix}\\
&=
\begin{vmatrix} a_{{1}}&a_{{2}}\\b_{{1}}&b_{{2}}\end {vmatrix}
\begin{vmatrix} a_{{3}}&a_{{4}}\\b_{{3}}&b_{{4}}\end {vmatrix}
(\zeta_1 \zeta_2+\zeta_3\zeta_4)
-
\begin{vmatrix} a_{{1}}&a_{{3}}\\b_{{1}}&b_{{3}}\end {vmatrix}
\begin{vmatrix} a_{{2}}&a_{{4}}\\b_{{2}}&b_{{4}}\end {vmatrix}
(\zeta_1 \zeta_3+\zeta_2\zeta_4)\\
&\qquad +
\begin{vmatrix} a_{{1}}&a_{{4}}\\b_{{1}}&b_{{4}}\end {vmatrix}
\begin{vmatrix} a_{{2}}&a_{{3}}\\b_{{2}}&b_{{3}}\end {vmatrix}
(\zeta_1 \zeta_4+\zeta_2\zeta_3).
\end{align*}
Noting (\ref{abbadef}) we may may show
\[\begin{vmatrix} a_{{i}}&a_{{j}}\\b_{{i}}&b_{{j}}\end {vmatrix}
=\theta_3(0) \theta_2(z/2) \theta_1(\alpha_i-\alpha_j) \theta_4(\alpha_i+\alpha_j-z/2). 
\]
Now it is clear that the determinant vanishes when $\zeta_1=\zeta_2=\zeta_3=\zeta_4$
and so
\[0=
\begin{vmatrix} a_{{1}}&a_{{2}}\\b_{{1}}&b_{{2}}\end {vmatrix}
\begin{vmatrix} a_{{3}}&a_{{4}}\\b_{{3}}&b_{{4}}\end {vmatrix}
-
\begin{vmatrix} a_{{1}}&a_{{3}}\\b_{{1}}&b_{{3}}\end {vmatrix}
\begin{vmatrix} a_{{2}}&a_{{4}}\\b_{{2}}&b_{{4}}\end {vmatrix}
+
\begin{vmatrix} a_{{1}}&a_{{4}}\\b_{{1}}&b_{{4}}\end {vmatrix}
\begin{vmatrix} a_{{2}}&a_{{3}}\\b_{{2}}&b_{{3}}\end {vmatrix}
\]
whence
\begin{align*}
|\Psi|&=
\begin{vmatrix} a_{{1}}&a_{{2}}\\b_{{1}}&b_{{2}}\end {vmatrix}
\begin{vmatrix} a_{{3}}&a_{{4}}\\b_{{3}}&b_{{4}}\end {vmatrix}
\left[(\zeta_1 \zeta_2+\zeta_3\zeta_4)-(\zeta_1 \zeta_4+\zeta_2\zeta_3)\right]\\
&\quad -
\begin{vmatrix} a_{{1}}&a_{{3}}\\b_{{1}}&b_{{3}}\end {vmatrix}
\begin{vmatrix} a_{{2}}&a_{{4}}\\b_{{2}}&b_{{4}}\end {vmatrix}
\left[(\zeta_1 \zeta_3+\zeta_2\zeta_4)-(\zeta_1 \zeta_4+\zeta_2\zeta_3)\right]\\
&=
\begin{vmatrix} a_{{1}}&a_{{2}}\\b_{{1}}&b_{{2}}\end {vmatrix}
\begin{vmatrix} a_{{3}}&a_{{4}}\\b_{{3}}&b_{{4}}\end {vmatrix}
(\zeta_1-\zeta_3)(\zeta_2-\zeta_4)
-
\begin{vmatrix} a_{{1}}&a_{{3}}\\b_{{1}}&b_{{3}}\end {vmatrix}
\begin{vmatrix} a_{{2}}&a_{{4}}\\b_{{2}}&b_{{4}}\end {vmatrix}
(\zeta_1-\zeta_2)(\zeta_3-\zeta_4).
\end{align*}
Now
\begin{align*}
i\left[\zeta_j-\zeta_k\right]&=
\frac{\theta_2(\alpha_j)\theta_4(\alpha_j)}{\theta_1(\alpha_j)\theta_3(\alpha_j)}-
\frac{\theta_2(\alpha_k)\theta_4(\alpha_k)}{\theta_1(\alpha_k)\theta_3(\alpha_k)}\\
&=
\frac{\theta_2(\alpha_j)\theta_4(\alpha_j)\theta_1(\alpha_k)\theta_3(\alpha_k) -
\theta_1(\alpha_j)\theta_3(\alpha_j)\theta_2(\alpha_k)\theta_4(\alpha_k)}
{\theta_1(\alpha_j)\theta_3(\alpha_j)\theta_1(\alpha_k)\theta_3(\alpha_k)}
\end{align*}
and upon using (W3, Appendix \ref{thetafunctidentapp}) with $\boldsymbol{\alpha}=(\alpha_j,\alpha_k,\alpha_k,\alpha_j)$
we obtain
\begin{equation}\label{zetadiff}
\zeta_j-\zeta_k = i\, \theta_2(0)\theta_4(0)\,
\frac{\theta_1(\alpha_j-\alpha_k)\theta_3(\alpha_j+\alpha_k)}
{\theta_1(\alpha_j)\theta_3(\alpha_j)\theta_1(\alpha_k)\theta_3(\alpha_k)}.
\end{equation}
Thus
\begin{align*}
|\Psi|&=-\frac{\theta_2^2(z/2) \theta_2^2(0)\theta_3^2(0)\theta_4^2(0)}
{\prod_{j=1}\sp{4}\theta_1(\alpha_j)\theta_3(\alpha_j)}\,
\theta_1(\alpha_1-\alpha_2) \theta_1(\alpha_3-\alpha_4)
\theta_1(\alpha_1-\alpha_3) \theta_1(\alpha_2-\alpha_4)\\
&\quad\times\left[
\theta_3(\alpha_1+\alpha_3) \theta_3(\alpha_2+\alpha_4)
\theta_4(\alpha_1+\alpha_2-z/2) \theta_4(\alpha_3+\alpha_4-z/2)\right. \\
&\qquad -\left.
\theta_3(\alpha_1+\alpha_2) \theta_3(\alpha_3+\alpha_4)
\theta_4(\alpha_1+\alpha_3-z/2) \theta_4(\alpha_2+\alpha_4-z/2)
\right].
\end{align*}
Next, using (W4, Appendix \ref{thetafunctidentapp}) with $\boldsymbol{\alpha}=(\alpha_1+\alpha_2-z/2,\alpha_3+\alpha_4-z/2,
\alpha_1+\alpha_3,\alpha_2+\alpha_4)$ we find
\begin{align*}
&\theta_3(\alpha_1+\alpha_3) \theta_3(\alpha_2+\alpha_4)
\theta_4(\alpha_1+\alpha_2-z/2) \theta_4(\alpha_3+\alpha_4-z/2)\\
&\quad-
\theta_3(\alpha_1+\alpha_2) \theta_3(\alpha_3+\alpha_4)
\theta_4(\alpha_1+\alpha_3-z/2) \theta_4(\alpha_2+\alpha_4-z/2)\\
&= 
 \theta_2 (z/2) \theta_2\left(\sum_{k=1}\sp4 \alpha_k-z/2\right)\,
 \theta_1(\alpha_1-\alpha_4)\theta_1(\alpha_2-\alpha_3)
\end{align*}
where the sign is determined for example by considering $\alpha_1=-\alpha_3$, 
$\alpha_2=-\alpha_4$. Then
\begin{equation*}
|\Psi|= -\theta_2^3 (z/2) \theta_2\left(\sum_{k=1}\sp4 \alpha_k-z/2\right)\,
{\theta_2^2(0)\theta_3^2(0)\theta_4^2(0)}\,
\frac{\prod_{i<j}\theta_1(\alpha_i-\alpha_j)}{\prod_{j=1}\sp{4}\theta_1(\alpha_j)\theta_3(\alpha_j)}.
\end{equation*}

Upon using (\ref{abel}) 
\[
\theta_2\left(\sum_{k=1}\sp4 \alpha_k-z/2\right)=\theta_2\left(N\tau+M-z/2\right)=
\theta_2\left(z/2\right)\,e\sp{-i\pi[M+N^2\tau-Nz]}
\]
giving finally that
\begin{equation}\label{detphi}
|\Psi|= -\theta_2^4 (z/2) \,
{\theta_2^2(0)\theta_3^2(0)\theta_4^2(0)}\,
\frac{\prod_{i<j}\theta_1(\alpha_i-\alpha_j)}{\prod_{j=1}\sp{4}\theta_1(\alpha_j)\theta_3(\alpha_j)}
\, e\sp{-i\pi[M+N^2\tau-Nz]} .
\end{equation}

\subsection{The Adjoint  of \texorpdfstring{$\Psi$}{Psi}}
Developing the method just employed we show in Appendix  \ref{proofthmadjpsi} that
\begin{theorem}\label{thmadjpsi}
The $i$-th column of $\Adj(\Psi)\sp{T}$ takes the form 
(for $i,j,k,l$ distinct)
\begin{equation}
\begin{pmatrix}\label{STPHI}
i \zeta_j\,\frac{\theta_2(z/2)\,\theta_2(\sum_{s\ne i}\alpha_s-z/2)}{\prod_{s\ne i}\theta_3(\alpha_s)}
-\frac{ \theta_2(0)\theta_4(0) \theta_3(\alpha_j+\alpha_k) \theta_3(\alpha_j+\alpha_l) \theta_4(\alpha_j-z/2) 
\theta_4(\alpha_k+\alpha_l-z/2) }{\theta_3(\alpha_j)\,\prod_{s\ne i}\theta_1(\alpha_s)\theta_3(\alpha_s)}
\\ \\
i \zeta_j\,\frac{\theta_2(z/2)\,\theta_4(\sum_{s\ne i}\alpha_s-z/2)}{\prod_{s\ne i}\theta_1(\alpha_s)}
-\frac{ \theta_2(0)\theta_4(0) \theta_3(\alpha_j+\alpha_k) \theta_3(\alpha_j+\alpha_l) \theta_2(\alpha_j-z/2) 
\theta_4(\alpha_k+\alpha_l-z/2) }{\theta_1(\alpha_j)\,\prod_{s\ne i}\theta_1(\alpha_s)\theta_3(\alpha_s)}
\\ \\
- \frac{\theta_2(z/2)\,\theta_2(\sum_{s\ne i}\alpha_s-z/2)}{\prod_{s\ne i}\theta_3(\alpha_s)}\\ \\
- \frac{\theta_2(z/2)\,\theta_4(\sum_{s\ne i}\alpha_s-z/2)}{\prod_{s\ne i}\theta_1(\alpha_s)}
\end{pmatrix}\, d_i
\end{equation}
where
\begin{equation}
\label{NDi} d_i = -\epsilon_{ijkl}\,\theta_2(0)\theta_3(0)\theta_4(0) \theta_1(\alpha_j-\alpha_k) \theta_1(\alpha_j-\alpha_l) \theta_1(\alpha_k-\alpha_l)\,\theta_2(z/2).
\end{equation}
\end{theorem}

\section{The matrix \texorpdfstring{$\overline{V}$}{barV}}\label{sectionV}
Recall that
\[ V=W\sp{\dagger\,-1}=\left(\mathcal{D}\sp{-1}\,\Psi\sp{-1}\, \left(1_2\otimes C\sp{-1}(z)\right)\right)\sp{\dagger}=\left(1_2\otimes C\sp{-1}(z)\right)\,\Psi\sp{\dagger\,-1}\,{\bar{\mathcal{D}}}\sp{-1}
\]
and so
\[
{\overline V}=\left(1_2\otimes C\sp{-1}(z)\right)\,\Adj(\Psi)\sp{T}\,\mathcal{D}\sp{-1} |\Psi|\sp{-1},
\]
where we have made use of the fact that $C(z)$ is real and symmetric. 
Here $\overline V$ is the complex conjugate of the matrix $V$ and it will be convenient to deal with this for the time being. If we work with the standard Nahm basis (\ref{standardnahm}) we have instead
\[
{\overline V}=\left(1_2\otimes \mathcal{O}\,C\sp{-1}(z)\right)\,\Adj(\Psi)\sp{T}\,\mathcal{D}\sp{-1} |\Psi|\sp{-1}
\]
as $\mathcal{O}$ is orthogonal.
Collecting the $z$ dependent factors together and utilising (\ref{abel}) this may be rewritten as
\[
{\overline V}=\left(1_2\otimes \mathcal{O}\,C\sp{-1}(z)\right)\,\frac1{\theta_2^2(z/2)}{\overline \Lambda}\,\tilde{\mathcal{D}}
\]
where $\tilde{\mathcal{D}}$ is a $z$-independent diagonal matrix that we shall not need and
\begin{equation}
{\overline \Lambda}_i
=
\begin{pmatrix}\label{LAMBDABAR}
-i \zeta_j\,\frac{\theta_2(z/2)\,\theta_2(\sum_{r\ne i}\alpha_r-z/2)}{\prod_{s\ne i}\theta_3(\alpha_s)}
+\frac{ \theta_2(0)\theta_4(0) \theta_3(\alpha_j+\alpha_k) \theta_3(\alpha_j+\alpha_l) \theta_4(\alpha_j-z/2) 
\theta_4(\alpha_k+\alpha_l-z/2) }{\theta_3(\alpha_j)\,\prod_{s\ne i}\theta_1(\alpha_s)\theta_3(\alpha_s)}
\\ \\
-i \zeta_j\,\frac{\theta_2(z/2)\,\theta_4(\sum_{r\ne i}\alpha_r-z/2)}{\prod_{s\ne i}\theta_1(\alpha_s)}
+\frac{ \theta_2(0)\theta_4(0) \theta_3(\alpha_j+\alpha_k) \theta_3(\alpha_j+\alpha_l) \theta_2(\alpha_j-z/2) 
\theta_4(\alpha_k+\alpha_l-z/2) }{\theta_1(\alpha_j)\,\prod_{s\ne i}\theta_1(\alpha_s)\theta_3(\alpha_s)}
\\ \\
\frac{\theta_2(z/2)\,\theta_2(\sum_{r\ne i}\alpha_r-z/2)}{\prod_{s\ne i}\theta_3(\alpha_s)}\\ \\
\frac{\theta_2(z/2)\,\theta_4(\sum_{r\ne i}\alpha_r-z/2)}{\prod_{s\ne i}\theta_1(\alpha_s)}
\end{pmatrix}\, 
e\sp{-z\mu_i}
\end{equation}
where we may choose $(i,j,k,l)$ as a cyclic permutation of $(1,2,3,4)$. Here we encounter $\mu_k$
defined in (\ref{defmu}).

\subsection{The pole structure of  \texorpdfstring{$\overline V$}{barV}} We need the expansion (\ref{LAMBDABAR})
to determine the projector and for calculating the Higgs field via the Panagopoulos formulae.
Set
\begin{equation}\label{defbarvi}
\boldsymbol{\overline v}_i(z)=\left(1_2\otimes \mathcal{O}\,C\sp{-1}(z)\right)\,\frac1{\theta_2^2(z/2)}\,
{\overline \Lambda}_i = \sum_{s\ge 0} \boldsymbol{\overline v}_{i, s}\, \xi^{s-5/2},
\qquad \text{where}\ z=1-\xi,
\end{equation}
and
\begin{equation}
{\overline \Lambda}_i = 
\begin{pmatrix} -\im \zeta_j  \\   1\end{pmatrix} \otimes \begin{pmatrix} A  \\  B \end{pmatrix}
+ 
 \begin{pmatrix} 1 \\   0\end{pmatrix} \otimes \begin{pmatrix} \alpha \\   \beta \end{pmatrix}.
\end{equation}
Here
\begin{align*}
A&= {\theta_2(z/2)}\, \frac{\theta_2(\sum_{r\ne i}\alpha_r-z/2)}{\prod_{s\ne i}\theta_3(\alpha_s)}\,
e\sp{-z\mu_i} ,
\\
B&=  {\theta_2(z/2)}\,\frac{\theta_4(\sum_{r\ne i}\alpha_r-z/2)}{\prod_{s\ne i}\theta_1(\alpha_s)}\,
e\sp{-z\mu_i},\\
\alpha&=\frac{ \theta_2(0)\theta_4(0) \theta_3(\alpha_j+\alpha_k) \theta_3(\alpha_j+\alpha_l) \theta_4(\alpha_j-z/2) 
\theta_4(\alpha_k+\alpha_l-z/2) }{\theta_3(\alpha_j)\,\prod_{s\ne i}\theta_1(\alpha_s)\theta_3(\alpha_s)}
\,e\sp{-z\mu_i},
\\
 \beta&=
\frac{ \theta_2(0)\theta_4(0) \theta_3(\alpha_j+\alpha_k) \theta_3(\alpha_j+\alpha_l) \theta_2(\alpha_j-z/2) 
\theta_4(\alpha_k+\alpha_l-z/2) }{\theta_1(\alpha_j)\,\prod_{s\ne i}\theta_1(\alpha_s)\theta_3(\alpha_s)}
\, e\sp{-z\mu_i}.
\end{align*}

A lengthy calculation given in Appendix  \ref{appendixfiniteterm} shows that

\begin{theorem} \label{finiteterm} 
Each column of $\boldsymbol{\overline v}_i$ has expansion at $z=1-\xi$
\begin{equation}\label{barviexp}
\boldsymbol{\overline v}_i =
N_i \left(   
\frac{1}{\xi^{3/2}} \begin{pmatrix} 0\\1\\-1\\ 0\end{pmatrix}+
\frac{1}{\xi^{1/2}} \begin{pmatrix} -x_1 -\im x_2 \\x_3\\ x_3 \\ x_1 -\im x_2\end{pmatrix}+
\frac{\boldsymbol{\overline v}_{i, 3}}{N_i}\,\xi^{1/2}
+\mathcal{O}
(\xi^{3/2})
\right)
\end{equation}
where
\[
N_i:= c \sqrt {K}\,\theta_{{2}}  \theta_{{4}}\, \frac{\prod_{j<k \  j,k\ne i}\theta_3(P_j+P_k)}
{\prod_{r\ne i}\theta_1(P_r)\theta_3(P_r)}\, {{\rm e}^{-\mu_{{i}}}}
\]
and the finite term ${\boldsymbol{\overline v}_{i, 3}}/{N_i}$ has the equivalent expansions
(for $j\ne i$)
\begin{equation}
\frac{\boldsymbol{\overline v}_{i, 3}}{N_i}= 
\left(\begin{array}{c}
\dfrac\imath4\,(\zeta_i^2+\zeta_i\zeta_j+\zeta_j^2)\zeta_j\,X-x_+x_3
-2\imath\,\left( r^2 -3 x_3^2+ 2 \lambda \right)\zeta_j-4\,x_-x_3\,\zeta_j(\zeta_i+\zeta_j)\\
\\
-\dfrac18 X \,\zeta_i^2-2\imath \, x_-x_3\,\zeta_i+\lambda-\dfrac32x_3^2\\
\\
-\dfrac18 X\, \zeta_i^2-2\imath\, x_-x_3\,\zeta_i+\lambda+x_+x_-  - \dfrac12x_3^2\\
\\
\dfrac\imath4\, X\,\zeta_i - x_- x_3
\end{array}\right),
\end{equation}
\begin{equation}
\frac{\boldsymbol{\overline v}_{i, 3}}{N_i}= 
\left(\begin{array}{c}
-\dfrac\imath4\,X\,\zeta_i^3  +4\, x_-x_3\,\zeta_i^2
+2\imath\,\left( r^2 -3 x_3^2+ 2 \lambda \right)\zeta_i+3\,x_+x_3\\
\\
-\dfrac18 X \,\zeta_i^2-2\imath \, x_-x_3\,\zeta_i+\lambda-\dfrac32x_3^2\\
\\
-\dfrac18 X\, \zeta_i^2-2\imath\, x_-x_3\,\zeta_i+\lambda+x_+x_-  - \dfrac12x_3^2\\
\\
\dfrac\imath4\, X\,\zeta_i - x_- x_3
\end{array}\right),
\end{equation}
\begin{equation}
\frac{\boldsymbol{\overline v}_{i, 3}}{N_i}= 
\left(\begin{array}{c}
\dfrac\imath4\dfrac{(K^2-4x_+^2)}{\zeta_i}-x_+x_3
\\ 
\\
-\dfrac18 X \,\zeta_i^2-2\imath \, x_-x_3\,\zeta_i+\lambda-\dfrac32x_3^2\\
\\
-\dfrac18 X\, \zeta_i^2-2\imath\, x_-x_3\,\zeta_i+\lambda+x_+x_-  - \dfrac12x_3^2\\
\\
\dfrac\imath4\, X\,\zeta_i - x_- x_3
\end{array}\right),
\end{equation}
with
\[ x_\pm= x_1\pm\imath x_2,\qquad \lambda = \frac18\, K^2(1-2k^2), \qquad  X=  K^2-4 x_-^2. \]
\end{theorem}

Several observations are in order. First, 
up to normalisation, this takes the form of (the complex conjugate of) (\ref{asymvp}) with $\boldsymbol{\overline v}_{i, 3}'=(a,b-r^2/2,b+r^2/2,c)\sp{T}$. The common pole structure means that
we may determine a projector (see later) onto the normalisable solution.
Next, using the second representation we find that
$$ |\boldsymbol{\overline v}_{1, 3},\, \boldsymbol{\overline v}_{2, 3} ,\, 
\boldsymbol{\overline v}_{3, 3},\, \boldsymbol{\overline v}_{4,3}|=
{\frac {r^2}{128}}\, N_1N_2N_3N_4 (K-2x_-)^3 (K+2x_-)^3
\prod_{i<j}(\zeta_i-\zeta_j).
 $$
We have already noted that points of bitangency of the spectral curve yield solutions with multiplicity to the
mini-twistor constraint and at these nongeneric points 
we need to take further terms in our expansions to get a basis for solutions to
$\Delta W=0$ and $\Delta\sp\dagger V=0$; this occurs when $\prod_{i<j}(\zeta_i-\zeta_j)=0$.
It naively appears that the solutions are also linearly dependent whenever $K=\pm 2x_-$: what is happening here? The reality of $K$ means this
only occurs for the lines  $x_2=0$, $x_1=\pm K/2$ and $x_3$ arbitrary; this means one of the roots of the mini-twistor constraint is $\zeta=0$. Indeed if we call this vanishing point $P_1$ we have
\begin{itemize}
\item $x_1=K/2$, $x_2=0$,  $x_3$ arbitrary: $P_1=0_1$, $P_3=\infty_1$; $P_2,P_4$ determined by $x_3$.

\item $x_1=-K/2$, $x_2=0$,  $x_3$ arbitrary: $P_1=0_2$, $P_3=\infty_2$; $P_2,P_4$ determined by $x_3$.

\end{itemize}
Upon noting that the normalizations $N_i$ ($i=2,3,4$) have a denominator $\theta_1(P_1)$ there is a precise cancellation between numerator and denominators and the determinant does not vanish along these lines. Thus we only need to consider further terms in the expansion at (nongeneric) points of bitangency.

\subsection{The behaviour at \texorpdfstring{$z=-1$}{z-1} and Monodromy}\label{section_monodromy}
Although we must take the expansions of the solutions at both $z=\pm1$ these are related. We establish
in Appendix \ref{proofsmonodromy}
\begin{theorem} \label{smonodromy}
Let $\boldsymbol{\overline v}_i(1-\xi)=
 \sum_{s\ge 0} \boldsymbol{\overline v}_{i, s}\, \xi^{s-5/2}$. Then
 $\boldsymbol{\overline v}_i(-1+\xi)=\pm \sum_{s\ge 0} \boldsymbol{\overline v}_{i, s}'\, \xi^{s-5/2}$
where
\[  \boldsymbol{\overline v}_{i, s}'=
(-1)^s \left( 1_2\otimes \begin{pmatrix}0&1\\ -1& 0\end{pmatrix} e\sp{2\mu_i}\right)
 \boldsymbol{\overline v}_{i, s}.
\]
\end{theorem}
This explains the origin of the matrix $U=1_2\otimes \begin{pmatrix}0&1\\ -1& 0\end{pmatrix}$ 
we encountered earlier relating solutions at the endpoints.
In particular the expansion at 
 $z=-1+\xi$ then takes the form
\begin{equation}\label{barviexm}
\boldsymbol{\overline v}_i =
N_i \, e^{2\mu_i}\,\left(   
\frac{1}{\xi\sp{3/2}}\begin{pmatrix}1\\ 0\\ 0\\ 1\end{pmatrix}+
\frac{1}{\xi\sp{1/2}}\begin{pmatrix} -x_3\\  -\im x_2-x_1\\  \im x_2 -x_1\\  x_3\end{pmatrix}
+\mathcal{O}(\xi^{1/2})
\right).
\end{equation}

\subsection{A convenient normalization}

We have just established that up to the normalisation factors $N_i$ the pole terms of $\bar V$ have
a common form; these normalisations may be removed by left multiplication by a constant matrix
and is convenient to define
\begin{align}\label{EVBAR}
\bar{\mathcal{V}}:&=\bar {V}\diag(1/N_1,1/N_2,1/N_3,1/N_4).\\
\intertext{Then}
\bar{\mathcal{V}}_i(1-\xi)&=\frac{1}{\xi\sp{3/2}}\boldsymbol{\bar{\mathfrak{v}}}_0+
\frac{1}{\xi\sp{1/2}}\boldsymbol{\bar{\mathfrak{v}}}_1+{\xi\sp{1/2}}\boldsymbol{\bar{\mathfrak{v}}}_{i,2}+
\mathcal{O}(\xi^{3/2})
\intertext{where}
\boldsymbol{\bar{\mathfrak{v}}}_0:&=
\begin{pmatrix} 0\\1\\-1\\ 0\end{pmatrix}, \quad
\boldsymbol{\bar{\mathfrak{v}}}_1:=
\begin{pmatrix} -x_1 -\im x_2 \\x_3\\ x_3 \\ x_1 -\im x_2\end{pmatrix},
\quad
\boldsymbol{\bar{\mathfrak{v}}}_{i,2}:=\boldsymbol{\overline v}_{i, 3}/N_i .
\label{defcalv}
\end{align}
Similarly we have the non-conjugated quantities
$$
{\mathcal{V}}_i(1-\xi)=\frac{1}{\xi\sp{3/2}}\boldsymbol{{\mathfrak{v}}}_0+
\frac{1}{\xi\sp{1/2}}\boldsymbol{{\mathfrak{v}}}_1+{\xi\sp{1/2}}\boldsymbol{{\mathfrak{v}}}_{i,2}+
\mathcal{O}(\xi^{3/2}),\quad\text{where}\quad
\boldsymbol{{\mathfrak{v}}}_1:=
\begin{pmatrix} -x_1 +\im x_2 \\x_3\\ x_3 \\ x_1 +\im x_2\end{pmatrix}
$$
and so forth.
Now we have shown that
$$\bar{\mathcal{V}}_i(-1+\xi)= U\left( \frac1{\xi^{3/2}}\boldsymbol{\bar{\mathfrak{v}}}_0-
\frac1{\xi^{1/2}}\boldsymbol{\bar{\mathfrak{v}}}_1+\xi^{1/2}\boldsymbol{\bar{\mathfrak{v}}}_{i,2}+\mathcal{O}(\xi^{3/2})
\right) \exp(2\mu_i)
$$
where
\begin{equation}\label{defmonU}
{U}:= 1_2\otimes \begin{pmatrix}0&1\\ -1& 0\end{pmatrix}=\begin{pmatrix}
0&1&0&0\\ -1&0&0&0
\\ 0&0&0&1\\ 0&0&-1&0\end{pmatrix}
\end{equation}
Thus
\[
\bar{\mathcal{V}}(-1+\xi)=U\left[ \bar{\mathcal{V}}_+(1-\xi) -  \bar{\mathcal{V}}_-(1-\xi)\right]\mathcal{M}
\]
where $ \bar{\mathcal{V}}_\pm$ are the even (odd) terms $\boldsymbol{\bar{\mathfrak{v}}}_{l}$ and
\begin{equation}\label{defmonM}
\mathcal{M}:=\diag(e^{2\mu_1},e^{2\mu_2},e^{2\mu_3},e^{2\mu_4}).
\end{equation}

With this normalisation $\mathcal{W}:=\mathcal{V}\sp{\dagger\, -1}$ has columns
\begin{equation}\label{SMATW}
\mathcal{W}_k:=
\frac1{\theta_2(z/2)}\,
\left( 1_2\otimes \begin{pmatrix} 1/p(z)&0 \\  0& p(z) \end{pmatrix} \mathcal{O} \right)
\left( \begin{pmatrix}1\\ i\zeta_k\end{pmatrix}\otimes
\begin{pmatrix}a_k \\ b_k \end{pmatrix}\right)\,
e^{z \left[\mu_k -\im\pi N\right]} \, d_k
\end{equation}
where
\begin{equation}
d_k:=- c\,
\sqrt{\frac{\pi}2} \, \theta_1[ P_k] \theta_3[P_k]\, e\sp{-\mu_k}\,
\left( \frac{ \theta_2(0)}{\theta_2[\sum_{j=1}\sp{4} P_j]}
\right)
\frac{\prod_{\substack{r<s \\  r,s\ne k}}\theta_3[P_r+P_s]}{\prod_{l\ne k}\theta_1[P_l-P_k]}.
\end{equation}
From our expansion (\ref{expanwk}) we find that
\begin{align*}
\mathcal{W}_k(1-\xi) &= \mathfrak{d} _k  \left\{  \frac{1}{\sqrt{\xi}} \boldsymbol{W}_{0,k}+ 
\sqrt{\xi}\boldsymbol{W}_{1,k}+ \xi^{3/2}\boldsymbol{W}_{2,k} +\mathcal{O}(\xi^{5/2})\right\}\\
&:= \frac{1}{\sqrt{\xi}} \mathfrak{w}_{0,k}+ 
\sqrt{\xi}\mathfrak{w}_{1,k}+ \xi^{3/2}\mathfrak{w}_{2,k} +\mathcal{O}(\xi^{5/2})
\end{align*}
where
\begin{align*}
\mathfrak{d} _k =   -d_k \sqrt{ \frac{2}{\pi}}  \frac{\theta_1(\alpha_k) \theta_3(\alpha_k)}{\theta_2(0)\theta_4(0)} \mathrm{e}^{\mu_k-\imath\pi N} .
\end{align*}

We remark that from  $\mathcal{W}\sp{T}(z)\overline{\mathcal{V}}(z)=1_4$ the expansion at $z=1$ yields the consistency relations
\begin{align}\begin{split}
0&=\mathfrak{w}_0^T.\boldsymbol{\bar{\mathfrak{v}}}_0,\\
0&=\mathfrak{w}_1^T.\boldsymbol{\bar{\mathfrak{v}}}_0+\mathfrak{w}_0^T.\boldsymbol{\bar{\mathfrak{v}}}_1,\\
1_4 &=\mathfrak{w}_0^T.\boldsymbol{\bar{\mathfrak{v}}}_2+\mathfrak{w}_1^T.\boldsymbol{\bar{\mathfrak{v}}}_1+\mathfrak{w}_2^T\boldsymbol{\bar{\mathfrak{v}}}_0.\end{split}\label{ort}
\end{align}
The first two of these are easily seen to be true while, upon noting
$\mathfrak{w}_1^T.\boldsymbol{\bar{\mathfrak{v}}}_1=0$, the third simplifies to
\begin{equation}\label{VWrel}
1_4 =\mathfrak{w}_0^T.\boldsymbol{\bar{\mathfrak{v}}}_2+\mathfrak{w}_2^T\boldsymbol{\bar{\mathfrak{v}}}_0.
\end{equation}
This then follows from the relation
(for  each $j\in\{1,\ldots,4\}$)
\begin{align}
\mathfrak{D}_j&=1/{\mathfrak{d}_j },\label{NEW3}
\intertext{where}
\mathfrak{D}_j&: = -4x_-x_3\zeta_j^2-\imath ( (2{k'}^2 -1) K^2 +4 r^2 -12 x_3^2 )\zeta_j 
-12 x_+x_3+\frac{\imath S_+^2}{\zeta_j},\label{frakD}
\end{align}
which holds for  $\zeta_j,\eta_j$ satisfying the Atiyah-Ward equation. Theorem \ref{smonodromy} then means that the expansion at $z=-1$ also holds true.

\subsection{Derivation of \texorpdfstring{${\mathfrak{v}}_2$}{v2} from \texorpdfstring{$\mathcal{W}$}{W}-data and vice versa}
In the preceding sections we have derived  ${\mathfrak{v}}_2$ by calculating the general inverse of $W$ and taking its expansion. In this subsection we show that the same result can be obtained from the expansion of $W$ itself; this is a useful check.
We shall show that using (\ref{ort}) we can in fact determine the ${\mathfrak{v}}_2$ term of the expansion
of $\mathcal{V}$ from that of  $\mathcal{W}$ and vice versa. 

Recall from (\ref{asymvp}) that we
have an expansion for $\overline{\mathcal{V}}$ of the form
\begin{equation}\boldsymbol{\bar{\mathfrak{v}}}_{0,j}=\frac{1}{\xi^{3/2}}\left( \begin{array}{r} 0\\ 1\\-1\\0 \end{array}\right), \quad 
\boldsymbol{\bar{\mathfrak{v}}}_{1,j}=\frac{1}{\xi^{1/2}}\left( \begin{array}{r} -x_+\\ x_3\\x_3\\x_- \end{array}\right), \quad
\widetilde{\boldsymbol{\bar{\mathfrak{v}}}}_{2,j} =\xi^{1/2} \left( \begin{array} {c}   a_j\\ b_j- r^2/2 \\ b_j+r^2/2\\ c_j \end{array}  \right) ,
\end{equation}
while that of  $\mathcal{W}$  is (up to the constants $\mathfrak{d}_j$) in terms of 
${\boldsymbol{W}}_{0,j}=\boldsymbol{w}_{0,j}\otimes\boldsymbol{w}_{0,j}$,
${\boldsymbol{W}}_{1,j}=\boldsymbol{w}_{0,j}\otimes\boldsymbol{w}_{1,j}$ and
${\boldsymbol{W}}_{2,j}=\boldsymbol{w}_{0,j}\otimes\boldsymbol{w}_{2,j}$.
We choose to rewrite
$$\boldsymbol{w}_{2,j}= \left( \begin{array}{c} \frac18 S_-^2\zeta_j^2+\imath x_-x_3\zeta_j- \frac{1}{24}(2{k'}^2-1)K^2+\frac12 x_3^3\\ \\ \imath\left( -\frac{1}{24}(2{k'}^2-1) K^2 + \frac12 x_3^2\right)\zeta_j-x_-x_3+\frac{1}{8\zeta_j}S_+^2 \end{array}\right)$$
with $S_{\pm}=\sqrt{  K^2-4 x_{\pm}^2}$, $x_{\pm}=x_1\pm\imath x_2$. 
The first two identities of (\ref{ort}) hold and we find

\begin{proposition}\label{Dream} 
The matrix  $ \mathfrak{D}$ in  the following relation
\begin{align}
{\boldsymbol{W}}_2^T . \boldsymbol{\bar{\mathfrak{v}}}_0+{\boldsymbol{W}}_1^T . \boldsymbol{\bar{\mathfrak{v}}}_1+{\boldsymbol{W}}_0^T . \widetilde{\boldsymbol{\bar{\mathfrak{v}}}}_2=
 \mathfrak{D}\label{ort1}
\end{align}  
is diagonal if and only if the quatities $a_j,b_j,c_j$ are given by the formulae 
\begin{align}\begin{split}
a_j&=\frac{\imath}{4} \frac{S_+^2}{\zeta_j}-x_+x_3,\\
b_j&=-\frac18 S_-^2\zeta_j^2-2\imath x_-x_3\zeta_j+\frac18K^2(1-2k^2)+\frac12x_+x_--x_3^2,\\
c_j&=\frac{\imath}{4}S_-^2\zeta_j - x_-x_3,\end{split}\label{abc}
\end{align}
where $\zeta_j$, $j=1,\ldots,4$  are solutions of the Atiyah-Ward equation and 
the $j$-th diagonal element $\mathfrak{D}_j$ is given by (\ref{frakD}).
\end{proposition}

Analogously, given $\boldsymbol{{\mathfrak{v}}}_2$ one can construct $\boldsymbol{W}_2$.

\section{The Projector}\label{sectionprojector}

The common pole structure at $z=\pm1$ means that the construction of the projection
matrix becomes algebraic. It is clear from (\ref{defcalv}) that
\begin{equation}\label{proj1}
\bar{\mathcal{V}}\begin{pmatrix} 1 &1& 1 \\
-1 & 0& 0 \\ 0& -1& 0 \\ 0&0& -1
\end{pmatrix}
\end{equation}
gives three vectors vanishing at $z=1$ and with behaviour at $z=-1$ going as
\begin{equation}\label{proj2}
\left(e^{2\mu_1} - e^{2\mu_r}  \right)\,\left(   
\frac{1}{\xi\sp{3/2}}\begin{pmatrix}1\\ 0\\ 0\\ 1\end{pmatrix}+
\frac{1}{\xi\sp{1/2}}\begin{pmatrix} -x_3\\  -\im x_2-x_1\\  \im x_2 -x_1\\  x_3\end{pmatrix}
+\mathcal{O}(\xi^{1/2})
\right)
\end{equation}
for $r=2,3,4$. While it is possible for $e^{2\mu_1} = e^{2\mu_2} $ for certain $\boldsymbol{x}$ (see the
$x_2$-axis below) a consequence of proposition (\ref{propmu13mu24}) is that   
$$e^{2\mu_1} -e^{2\mu_3}= e^{2\mu_1} +e^{-2\bar\mu_1}\ne0.$$
Therefore the rank of
$$
\begin{pmatrix}  e^{2\mu_1} - e^{2\mu_3}   &0 \\
e^{2\mu_2} - e^{2\mu_1}&e^{2\mu_1} - e^{2\mu_4} 
 \\0& e^{2\mu_3} - e^{2\mu_1}\end{pmatrix}.
$$
is always 2. Thus we may construct from these the two required normalisable solutions by taking
\[
\bar{\mathcal{V}}\begin{pmatrix} 1 &1& 1 \\
-1 & 0& 0 \\ 0& -1& 0 \\ 0&0& -1
\end{pmatrix}
\begin{pmatrix}  e^{2\mu_1} - e^{2\mu_3}   &0 \\
e^{2\mu_2} - e^{2\mu_1}&e^{2\mu_1} - e^{2\mu_4} 
 \\0& e^{2\mu_3} - e^{2\mu_1}\end{pmatrix}.
\]
We observe that the projector is $z$ independent, as required.\footnote{If we had used the 
projector
\[
\begin{pmatrix}  e^{2\mu_1} - e^{2\mu_3}   &e^{2\mu_1} - e^{2\mu_4}  \\
e^{2\mu_2} - e^{2\mu_1}&0 \\0& e^{2\mu_2} - e^{2\mu_1}\end{pmatrix}
\]
instead, then on the $x_2$ axis the first vector $\bar{\mathcal{V}}_1 -\bar{\mathcal{V}}_2$ is already normalizable at both ends and we construct the remaining vector as an appropriate linear combination of the final two columns.
}
Thus we have the projector
\begin{equation}\label{defmunorm}
\overline{\mu}:=\begin{pmatrix} 1 &1& 1 \\
-1 & 0& 0 \\ 0& -1& 0 \\ 0&0& -1
\end{pmatrix}
\begin{pmatrix}  e^{2\mu_1} - e^{2\mu_3}   &0 \\
e^{2\mu_2} - e^{2\mu_1}&e^{2\mu_1} - e^{2\mu_4} 
 \\0& e^{2\mu_3} - e^{2\mu_1}\end{pmatrix}.
 \end{equation}
which is such that
$$(1,1,1,1)\overline{\mu}=(0,0),\qquad (1,1,1,1)\mathcal{M}\overline{\mu}=(0,0),$$
the latter showing this holds for $z=-1$ as well.

\section{An Example: the \texorpdfstring{$x_2$}{x2} axis}\label{sectionx2axis}
Before turning to the general formulae its helpful to see an example
of our formalism to reconstruct the Higgs field on the $x_2$-axis.  
We shall first construct the two normalisable solutions, then use the
Panagopoulos formalism to calculate their normalisation and then finally calculate the Higgs field. Aleady at this stage we obtain a new analytic result for the depth of the well.

\subsection{The normalizable solutions}

Utilising Theorem \ref{finiteterm} and restricting to the $x_2$ direction we find that  at $z=1-\xi$:

\begin{proposition} Let $x_1=x_3=0$ and the points $P_i$ be given by  (\ref{x2roots}). Then the first column of the expansion of the fundamental solution takes the form
\begin{align*}
\left.\overline{\boldsymbol{v}}_1(1-\xi)\right|_{\xi\sim 0}= N_1 &\left\{ 
\xi^{-3/2}  \left(\begin{array}{r} 0\\1\\-1\\ 0   \end{array}\right) + 
\xi^{-1/2}\left(\begin{array}{c} \imath x_2\\0\\0\\ \imath x_2   \end{array}\right)  \right.
      \\
 &\qquad+ \left.
  \xi^{1/2} \, \frac{Kk}{4} \left(
\begin{array}{c}
\sqrt{K^2+4x_2^2}\left(1 + \imath \dfrac1{Kk}\sqrt{K^2{k'}^2+4x_2^2} \right)   \\   
-\imath \sqrt{K^2{k'}^2 +4 x_2^2}-\dfrac{2 x_2^2}{K k}              \\ \\
- \imath \sqrt{K^2{k'}^2 +4 x_2^2}+\dfrac{2x_2^2}{K k}            \\ 
- \sqrt{K^2-4x_2^2}\left(1 + \imath \dfrac1{Kk}\sqrt{K^2{k'}^2+4x_2^2} \right)
      \end{array}  
      \right) \right\}.
\end{align*}
\end{proposition}

It is convenient to set
\begin{equation} 
p= \sqrt{K^2+4x_2^2}, \quad q= \sqrt{K^2{k'}^2+4x_2^2} \label{notations}
 \end{equation}
in terms of which
the whole $\xi^{1/2}$ entry to the expansion of the fundamental solution (for the ordering of roots (\ref{x2axisps}))then reads 
\begin{align*}
\frac{Kk}{4} \left(    \begin{array}   {cccc} 
p+\imath pq/Kk &  -p+\imath pq/Kk & -p-\imath pq/Kk & p-\imath pq/Kk   \\
-\imath q-\frac{2 x_2^2}{K k} &   \imath q-\frac{2 x_2^2}{K k} & - \imath q-\frac{2 x_2^2}{K k} &  \imath q-\frac{2 x_2^2}{K k}\\
-\imath q+\frac{2 x_2^2}{K k} &   \imath q+\frac{2 x_2^2}{K k}&  -\imath q+\frac{2 x_2^2}{K k} & 
 \imath q+\frac{2 x_2^2}{K k}\\
- p+\imath pq/Kk  & p+\imath pq/Kk & p-\imath pq/Kk & -p-\imath pq/Kk 
  \end{array} \right) \mathrm{Diag} ( N_1,\ldots, N_4).
\end{align*}
 The expansion of the vector $ \boldsymbol{v}_i(z)$ near the point $z=-1+\xi$ is then given by
 Theorem \ref{smonodromy},
 \begin{align*}
\overline{\boldsymbol{v}}_i(1-\xi)=N_i \left( \begin{array}{c}a_i\\   b_i-x_2^2/2\\  b_i+x_2^2/2\\ - \overline{a}_i  \end{array}         \right)
\Longrightarrow \overline{\boldsymbol{v}}_i(-1+\xi) =N_i \mathrm{e}^{2\mu_i} \left( \begin{array}{c}   b_i-x_2^2/2 \\ -a_i\\ \overline{a}_i\\  b_i+x_2^2/2    \end{array} \right).
\end{align*}

Now acting by the projector (\ref{proj1}) yields three normalizable vectors at $z=1$
\[
\xi^{1/2}\,  \frac{Kk}{2}\left( \begin{array}   { ccc}    
p& p+\imath pq/Kk& \imath pq/Kk  \\ 
-\imath q&0&-\imath q\\ 
-\imath q&0&-\imath q \\  
-p&- p+\imath pq/Kk&\imath pq/Kk
  \end{array}\right).
\]
From (\ref{proj2}) and (\ref{muconjugate2}) the first column here is also finite at $z=-1$, and because
\[e^{2\mu_3}=e^{2\mu_4}=e^{-2\mu_1}=e^{-2\mu_2}\]
as a result of (\ref{muconjugate2}) the remaining two vectors have the same poles and consequently their difference is then finite 
at $z=-1$. We have then

\begin{proposition}\label{weylasmp}
The Weyl equation $\Delta^{\dagger}  \boldsymbol{v}   =0$ admits precisely two normalizable solutions 
$ \boldsymbol{v}_1(z;\boldsymbol{x})$, $\boldsymbol{v}_2(z;\boldsymbol{x})$, for $ z\in [-1,1] $ which 
for $\boldsymbol{x}=(0,x_2,0)$ vanish at the end points as 
\begin{align*}\begin{split}
{\boldsymbol{v}}_1(1-\xi; \boldsymbol{x})&=
 \frac{Kk}{2} \left( \begin{array}{c}  p\\ \imath q \\ \imath q \\ -p \end{array}  \right)\sqrt{\xi}+O(\xi^{3/2}), \quad
{\boldsymbol{v}}_1(-1+\xi; \boldsymbol{x})= 
\frac{Kk}{2}\mathrm{e}^{2 \lambda_2} \left( \begin{array}{c}
 q \\  \imath p\\  \imath p\\ - q  \end{array}  \right)\sqrt{\xi}+O(\xi^{3/2}),
    \\
{\boldsymbol{v}}_2(1-\xi; \boldsymbol{x})&= 
 \frac{Kk}{2} \left( \begin{array}{c}  p\\ -\imath q \\ - \imath q \\ -p  \end{array}  \right)\sqrt{\xi}+O(\xi^{3/2}),\quad 
 {\boldsymbol{v}}_2(-1+\xi; \boldsymbol{x})= 
\frac{Kk}{2}\mathrm{e}^{-2\lambda_2} \left( \begin{array}{c}
  q \\  -\imath p\\ - \imath  p\\ -q  \end{array}  \right)\sqrt{\xi}+O(\xi^{3/2}),
\end{split}
\end{align*} 
where $p= \sqrt{K^2+4x_2^2}$ and $q= \sqrt{K^2{k'}^2+4x_2^2}$.
\end{proposition}

\begin{proof} Recall that we have been working with the matrix $\overline{V}$ and so the complex conjugate of the required solutions of the Weyl equation. The normalisable solutions we have constructed vanish at the end points as 
\begin{align*}\begin{split}
\overline{\boldsymbol{v}}_1(1-\xi; \boldsymbol{x})&=
 \frac{Kk}{2} \left( \begin{array}{c}  p\\- \imath q \\ -\imath q \\ -p  \end{array}  \right)\sqrt{\xi}+O(\xi^{3/2}), \quad
\overline{\boldsymbol{v}}_1(-1+\xi; \boldsymbol{x})= -
\frac{Kk}{2}\mathrm{e}^{2\mu_1} \left( \begin{array}{c}
   \imath q \\p\\p\\- \imath q  \end{array}  \right)\sqrt{\xi}+O(\xi^{3/2}),
    \\
\overline{\boldsymbol{v}}_2(1-\xi; \boldsymbol{x})&= 
 \frac{Kk}{2} \left( \begin{array}{c}  p\\ \imath q \\ \imath q \\ - p  \end{array}  \right)\sqrt{\xi}+O(\xi^{3/2}),\quad 
\overline{\boldsymbol{v}}_2(-1+\xi; \boldsymbol{x})= 
\frac{Kk}{2}\mathrm{e}^{-2\mu_1} \left( \begin{array}{c} \imath q \\- p\\ -p\\ -\imath q  \end{array}  \right)\sqrt{\xi}+O(\xi^{3/2}).
\end{split}
\end{align*} 
Both $p$, $q$ are real in our setting and so the behaviour at $z=1$ follows. For $z=-1$ we use  (\ref{muconjugate2})
and so
\[\overline{\boldsymbol{v}}_1(-1+\xi; \boldsymbol{x})= 
\frac{Kk}{2}\mathrm{e}^{2\lambda_2} \left( \begin{array}{c}
  q \\  -\imath p\\ - \imath  p\\ - q  \end{array}  \right)\sqrt{\xi}+O(\xi^{3/2}),
 \;
\overline{\boldsymbol{v}}_2(-1+\xi; \boldsymbol{x})= 
\frac{Kk}{2}\mathrm{e}^{-2\lambda_2} \left( \begin{array}{c}
  q \\  \imath p\\  \imath  p\\ -q  \end{array}  \right)\sqrt{\xi}+O(\xi^{3/2}),
\]
with the proposition following.
\end{proof}

\subsection{Orthogonalization}
Given our two normalisable solutions we must now normalise them. To do this we shall
calculate the inner products using Panagopoulos's formulae,

\begin{align}\nonumber
\int_{-1}^1  \boldsymbol{v}_i^{\dagger}(z,\boldsymbol{x}) \boldsymbol{v}_j(z,\boldsymbol{x})\mathrm{d} z
&=
 \mathrm{Lim}_{\xi\to 0}  \frac{1}{\xi}\boldsymbol{v}_i^{\dagger}(1-\xi,\boldsymbol{x})\;\mathrm{Res}_{\xi=0} \mathcal{Q}(1-\xi)^{-1}  \boldsymbol{v}_j(1-\xi,\boldsymbol{x})\\
&\quad - \mathrm{Lim}_{\xi\to 0}  \frac{1}{\xi}\boldsymbol{v}_i^{\dagger}(-1+\xi,\boldsymbol{x})\;\mathrm{Res}_{\xi=0} \mathcal{Q}(-1+\xi)^{-1}  \boldsymbol{v}_j(-1+\xi,\boldsymbol{x}).
\end{align}
Direct calculation shows that for $\boldsymbol{x}=(0,x_2,0)$
\begin{align*}\mathrm{Res}_{\xi=0} \mathcal{Q}(1-\xi)^{-1}=  \mathrm{Res}_{\xi=0} \mathcal{Q}(-1+\xi)^{-1}
=\frac{1}{K^2 k^2}\left(\begin{array}{rrrr}  1&0&0&-1\\
0&-1&-1&0\\0&-1&-1&0\\-1&0&0&1
\end{array}\right)\end{align*}
and consequently we find
\begin{proposition} 
The Gram matrix built from the vectors $\boldsymbol{v}_i(z,\boldsymbol{x})$, $i=1,2$ is diagonal, 
\begin{align}
\left(\int_{-1}^1  \boldsymbol{v}_i^{\dagger}(z,\boldsymbol{x}) \boldsymbol{v}_i(z,\boldsymbol{x})\mathrm{d} z\right)_{i,j=1,2}
= \left(  \begin{array}{cc}  K^2k^2(1+ \mathrm{e}^{4\lambda_2})&0  \\ 0&   K^2k^2(1+ \mathrm{e}^{-4\lambda_2})  \end{array}  \right)\end{align} 
Therefore the vectors $\boldsymbol{v}_i(z,\boldsymbol{x})$, $i=1,2$ are orthogonal with norms
\begin{equation}\mathcal{N}_1= ||\boldsymbol{v}_1(z,\boldsymbol{x}) || = Kk\sqrt{1+ \mathrm{e}^{4\lambda_2}},\qquad\mathcal{N}_2=  ||\boldsymbol{v}_2(z,\boldsymbol{x}) || = Kk\sqrt{1+\mathrm{e}^{-4\lambda_2}} .\end{equation}
\end{proposition}

In what follows we shall denote the orthonormal vectors used in the ADHM construction by
\begin{align*}
{\boldsymbol{V}}_1(z,\boldsymbol{x})=  \frac{1}{\mathcal{N}_1}\boldsymbol{v}_1(z,\boldsymbol{x}), \quad {\boldsymbol{V}}_2(z,\boldsymbol{x})=  \frac{1}{\mathcal{N}_2}\boldsymbol{v}_2(z,\boldsymbol{x}).
\end{align*}

\subsection{The Higgs field}

We now compute the Higgs field using (\ref{panhiggs}) to evaluate the integrals. For the $x_2$-axis
this takes the form
\begin{align} 
-\im \Phi_{ij}&=\int_{-1}^1 z \boldsymbol{V}_i^{\dagger}(z,\boldsymbol{x}) \boldsymbol{V}_j(z,\boldsymbol{x})\mathrm{d} z  \nonumber
\\
&=
 \mathrm{Lim}_{\xi\to 0}  \frac{1}{\xi}\boldsymbol{V}_i^{\dagger}(1-\xi,\boldsymbol{x})\;\mathrm{Res}_{\xi=0} \mathcal{Q}(1-\xi)^{-1}  \boldsymbol{V}_j(1-\xi,\boldsymbol{x})  \nonumber  \\
&\quad+ \mathrm{Lim}_{\xi\to 0}  \frac{1}{\xi}\boldsymbol{V}_i^{\dagger}(-1+\xi,\boldsymbol{x})\;\mathrm{Res}_{\xi=0} \mathcal{Q}(-1+\xi)^{-1}  \boldsymbol{V}_j(-1+\xi,\boldsymbol{x})
\label{integral1} \\
&\quad +
 \mathrm{Lim}_{\xi\to 0}  \frac{1}{\xi}\boldsymbol{V}_i^{\dagger}(1-\xi,\boldsymbol{x})\;\mathrm{Res}_{\xi=0} \mathcal{Q}(1-\xi)^{-1} H_0 \frac{\mathrm{d}}{\mathrm{d} x_2} \boldsymbol{V}_j(1-\xi,\boldsymbol{x})
 \nonumber \\
&\quad - \mathrm{Lim}_{\xi\to 0}  \frac{1}{\xi}\boldsymbol{V}_i^{\dagger}(-1+\xi,\boldsymbol{x})\;\mathrm{Res}_{\xi=0} \mathcal{Q}(-1+\xi)^{-1} H_0 \frac{\mathrm{d}}{\mathrm{d} x_2} \boldsymbol{V}_j(-1+\xi,\boldsymbol{x})
\nonumber
\end{align}
where
\[ H_0=\imath x_2\left(\begin{array}{cc} 0&1_2\\-1_2&0 \end{array}  \right). \]
We remark in passing that  while the Panagopoulos formula used to establish the norms was insensitive to the interchange of $V_i$ and $\overline{V}_i$, this is no longer the case for the above formula as it has
been derived assuming $\Delta\sp\dagger V=0$. 

Now with  $i=j=1$, the first two lines of (\ref{integral1}) give  in this case
\begin{equation*}
\frac{1- \mathrm{e}^{4\lambda_2 }}{1+ \mathrm{e}^{4\lambda_2}}=\tanh (2\lambda_2)
\end{equation*}
whilst the next two lines reduce after simplifications to 
\[  - \frac{4x_2}{\sqrt{K^2+4x_2^2} \sqrt{K^2{k'}^2+4x_2^2} }.  \] 
Therefore, 
\begin{equation}\label{bppx2japan}
\Phi_{1,1} = \imath \left( \tanh(2\lambda_2)- \frac{4x_2}{\sqrt{K^2+4x_2^2} \sqrt{K^2{k'}^2+4x_2^2} }\right).
\end{equation}
In analogous way we compute 
\begin{equation}
\Phi_{2,2} = \imath \left( - \tanh(2\lambda_2)+ \frac{4x_2}{\sqrt{K^2+4x_2^2} \sqrt{K^2{k'}^2+4x_2^2} }\right)
=-\Phi_{1,1}.
\end{equation}
Similar calculations leads to 
\begin{align}
\Phi_{1,2}=\Phi_{2,1} = -\imath\,
\frac{K k^2+2E-2K}{\cosh(2 \lambda_2)Kk^2}.
\end{align}
Recalling (\ref{defH}) then
$H= \sqrt{  -\frac12 \Phi_{1,1}^2-\frac12 \Phi_{2,2}^2 - \Phi_{1,2} \Phi_{2,1}}$ and so
\begin{equation}
H^2(0,x_2,0) = \left(\tanh 2\lambda_2+\frac{4 x_2}{W_2}\right)^2 + 
\frac{(K {k'}^2-2E+K)^2}{ K^2 k^4 \cosh^2 2\lambda_2} \label{h2}
\end{equation}
with 
\[  \lambda_2 = \mu_1-\frac{\imath \pi}{4}, \quad W_2=\sqrt{ (K^2+4x_2^2)(K^2 {k'}^2+4 x_2^2) } . \]
We find that as $x_2\to \infty$ that $H$ approaches to 1 as
\begin{equation}H   = 1 - \frac{1}{|x_2|} +\frac18\frac{K^2(1+{k'}^2)}{|x_2|^3} +O\left(\frac{1}{|x_2|^4}\right) .\end{equation}

The depth of the well ($x_2=0, \lambda_2=0$) is found to be
\begin{equation}\label{welldepth}
H(0) = \frac{  K(1+{k'}^2)-2E}{K k^2},
\end{equation}
reproducing (\ref{higgsorigin}) found by Brown \emph{et. al.} \cite[see their equation 7.2;  Appendix \ref{applame}
compares notation]{bpp82}; the work of \cite[(7.1)]{ors82} presented this in terms of an infinite series together with an undetermined integral.

\section{Formulae for the Higgs Field and Energy Density}\label{sectionHiggs}
The aim of this section is to evaluate the formulae (\ref{pannorm}, \ref{panhiggs})  determining the
Higgs field (at a generic point $\boldsymbol{x}$). We have all of the necessary components with the exception of
the quantity $\mathcal{Q}^{-1}$ which we will evaluate in the first subsection. We have already noted
that the combinations
$\left(V\sp\dagger \mathcal{Q}^{-1}\mathcal{H} V\right)(z)$
and
$\left(W\sp\dagger \mathcal{Q}\mathcal{H} W\right)(z)$ are constant. These structured matrices help us to simplify our results and we next evaluate these. The final subsection then combines preceding results. Again some proofs are relegated to an Appendix.

\subsection{\texorpdfstring{$\boldsymbol{\mathcal{Q}^{-1}}$}{Qinv} and related quantities}
We need  $\mathcal{Q}^{-1}$ at the endpoints. Here
 the Hermitian $\mathcal{Q}$, defined in (\ref{pandefs}), takes the form
$$
\mathcal{Q}=
\begin{pmatrix} -{\frac { \left( {r}^{2}+{x_{{1}}}^{
2}+{x_{{2}}}^{2}-{x_{{3}}}^{2} \right) f_{{3}}}{2{r}^{2}}}
&-{\frac {x_{
{3}} \left( if_{{2}}x_{{2}}-f_{{1}}x_{{1}} \right) }{{r}^{2}}}
&-{
\frac {x_{{3}}f_{{3}} \left( -x_{{1}}+ix_{{2}} \right) }{{r}^{2}}}
& *
\\ 
{\frac {x_{{3}} \left( if_{{2}
}x_{{2}}+f_{{1}}x_{{1}} \right) }{{r}^{2}}}
&{\frac { \left( {r}^{
2}+{x_{{1}}}^{2}+{x_{{2}}}^{2}-{x_{{3}}}^{2} \right) f_{{3}}}{2{r}^{2}}
}
& *
&{\frac {x_{{3}}f_{{3}} \left( -x_{{1}}+ix_{{2
}} \right) }{{r}^{2}}}
\\ 
{\frac { \left( x_{{1}}+ix_
{{2}} \right) f_{{3}}x_{{3}}}{{r}^{2}}}
& *
&{
\frac { \left( {r}^{2}+{x_{{1}}}^{2}+{x_{{2}}}^{2}-{x_{{3}}}^{2}
 \right) f_{{3}}}{2{r}^{2}}}
&{\frac {x_{{3}} \left( if_{{2}}x_{{2}}-f_{
{1}}x_{{1}} \right) }{{r}^{2}}}
\\ 
*
&-{\frac { \left( x_{{1}}+ix_{{2}} \right) f_{{3}}x_{{3}}}{{r}^{2}
}}
&-{\frac {x_{{3}} \left( if_{{2}}x_{{2}}+f_{{1}}x_{{1}} \right) }{{r
}^{2}}}
&-{\frac { \left( {r}^{2}+{x_{{1}}}^{2}+{x_{{2}}}^{2}-{x_{
{3}}}^{2} \right) f_{{3}}}{2{r}^{2}}}
\end {pmatrix}
$$
where
\begin{align*}
\mathcal{Q}_{14}&={\scriptstyle
-{\frac {f_{{1}}{x_{{2}}}^{2}+f_{{2}}{x_{{2}}}^{2}+2\,if_{{1}}x_{{1}}x_{{2}}+2\,if_{{2}}x_{{1}}x_{{2}}+f_{{1}}{r}^{2}-f_{{2}}{r}^{2}-f_{{1}}{x_{{1}}}^{2}+{x_{{3}}}^{2}f_{{1}}-f_{{2}}{x_{{1}}}^{2}-{x_{{3}}}^{2}f_{{2}}}{2{r}^{2}}}  }
=\bar{ \mathcal{Q}}_{41}
\\ 
\mathcal{Q}_{23}&={\scriptstyle
{\frac {-f_{{1}}{x_{{2}}}^{2}+f_{{2}}{x_{{2}}}^{2}+2\,if_{{1}}x_{{1}}x_{{2}}-2\,if_{{2}}x_{{1}}
x_{{2}}-f_{{1}}{r}^{2}-f_{{2}}{r}^{2}+f_{{1}}{x_{{1}}}^{2}-{x_{{3}}}^{2}f_{{1}}-f_{{2}}{x_{{1}}}^{2}-{x_{{3}}}^{2}f_{{2}}}{2{r}^{2}}}
}
=\bar{ \mathcal{Q}}_{32}.
\end{align*}
Thus $\mathcal{Q}$ has first order poles at $z=\pm1$. One finds upon use of elliptic function identities that the determinant of $\mathcal{Q}$ is constant,
$|\mathcal{Q}|=D\times  {{K}^{2}}/{{r}^{4}}$ with
$D$ given below.
Indeed one finds that  the entries of  $\adj \mathcal{Q}$ are again linear in the $f_i$'s and consequently 
 $\mathcal{Q}^{-1}$ has only first order poles at $z=\pm1$ (the possible third order poles cancelling) and
 vanishing constant term.
One finds that
\begin{align}
\mathrm{Res}_{z=1} \mathcal{Q}^{-1}(z,\boldsymbol{x})&
=\begin{pmatrix} A&B&B&C\\
                                        \overline{B}&-A&-A&-B\\
 \overline{B}&-A&-A&-B\\
\overline{C}&-\overline{B}&-\overline{B}&A
  \end{pmatrix}, \nonumber
  \\
\mathrm{Res}_{z=-1} \mathcal{Q}^{-1}(z,\boldsymbol{x})&=
\begin{pmatrix}A&\overline{B}&B&-A\\
                                        B&-A&C&-B\\
 \overline{B}&\overline{C}&-A&-\overline{B}\\
-A&-\overline{B}&B&A
  \end{pmatrix}
  = U \left( \mathrm{Res}_{z=1} \mathcal{Q}^{-1}(z,\boldsymbol{x}) \right) U,
  \label{IQmon}
\end{align}
where
\begin{align*}
A
&=(-  {k'}^2 x_+^2x_-^2+x_2^2\, x_+x_- -x_3^2(x_1^2-x_2^2))/D,\\
B
&=x_3 \left( x_+({k'}^2 x_-^2+x_3^2) + \imath (x_1+x_+)x_-x_2   \right)/D,\\
C
&=\left( -(x_+x_-+2\,x_3^2) x_-^2{k'}^2-2x_3^4-(x_+ +2\imath x_2)x_-x_3^2+x_-^2 x_2^2\right)/D,
\\
D
&=-{K}^{2} \left( 2\,ikx_{{2}}x_{{3}}-{k}^{2}{x_{{1}}}^{2}-{k}^{2}{x_{{2
}}}^{2}+{x_{{1}}}^{2}+{x_{{3}}}^{2} \right)  \left( {k}^{2}{x_{{1}}}^{
2}+{k}^{2}{x_{{2}}}^{2}+2\,ikx_{{2}}x_{{3}}-{x_{{1}}}^{2}-{x_{{3}}}^{2
} \right).
\end{align*}

Towards evaluating $V\sp\dagger \mathcal{Q}^{-1}$ and $V\sp\dagger \mathcal{Q}^{-1}\mathcal{H}$ at the end points  we record that (recall  $\mathcal{H}^{-1}=\mathcal{H}/r^2$ )
\begin{equation}
 \mathcal{H}\boldsymbol{{\mathfrak{v}}}_1=-r^2\,\boldsymbol{{\mathfrak{v}}}_0,
 \quad
  \mathcal{H}\boldsymbol{{\mathfrak{v}}}_0=-\boldsymbol{{\mathfrak{v}}}_1.
 \end{equation}
 These, together with $\mathcal{Q}^{-1}\mathcal{H}=-\mathcal{H}\mathcal{Q}^{-1}$,
yield that
\begin{equation}\label{projgeneral00}
 \boldsymbol{{\mathfrak{v}}}_0\sp{\dagger}\,\mathcal{Q}^{-1}\mathcal{H}\,
\boldsymbol{{\mathfrak{v}}}_0
= -\frac1{r^2}
\boldsymbol{{\mathfrak{v}}}_1\sp{\dagger}\,\mathcal{Q}^{-1}\mathcal{H}\,
\boldsymbol{{\mathfrak{v}}}_1, \quad
\boldsymbol{{\mathfrak{v}}}_0\sp{\dagger}\,\mathcal{Q}^{-1}\mathcal{H}\,
\boldsymbol{{\mathfrak{v}}}_1+
\boldsymbol{{\mathfrak{v}}}_1\sp{\dagger}\,\mathcal{Q}^{-1}\mathcal{H}\,
\boldsymbol{{\mathfrak{v}}}_0=0.
 \end{equation}
Further, calculations show that
  \begin{align}
\boldsymbol{{\mathfrak{v}}}_0\sp{\dagger}\,\mathcal{Q}^{-1}\mathcal{H}\,
\boldsymbol{{\mathfrak{v}}}_0
&=4 \im r^2 x_1 x_2 x_3 (k^2 f_1 -k^2 f_2 -f_1 +f_3)/D, \label{qproj00} \\ \nonumber
\boldsymbol{{\mathfrak{v}}}_0\sp{\dagger}\,\mathrm{Res}_{z=1} \mathcal{Q}^{-1}&=
\boldsymbol{{\mathfrak{v}}}_1\sp{\dagger}\,\mathrm{Res}_{z=1} \mathcal{Q}^{-1}=(0,0,0,0).
\end{align}

\subsection{The matrices \texorpdfstring{$\boldsymbol{V\sp\dagger \mathcal{Q}^{-1}\mathcal{H} V}$,
$\boldsymbol{{W\sp\dagger \mathcal{Q}\mathcal{H} W}$ and
$\boldsymbol{\mathcal{V}}\sp{\dagger}
 \mathcal{Q}^{-1} {\mathcal{V}}}$}{VQinvV} }
The constancy of the matrix 
$\left( {\mathcal{V}}\sp{\dagger}\,\mathcal{Q}^{-1}\mathcal{H}\,{\mathcal{V}}  \right)(z)$ (and
similarly for $W$) means that the possible poles at the end points must occur in vanishing
combinations. Thus, for example, the leading pole 
$(\boldsymbol{{\mathfrak{v}}}_0\sp{\dagger}\,\mathcal{Q}^{-1}\mathcal{H}\,
\boldsymbol{{\mathfrak{v}}}_0)/\xi^3$ (for $z=1-\xi$) together with (\ref{qproj00}) and
 $$(k^2 f_1 -k^2 f_2 -f_1 +f_3)(1-\xi)=\frac18 k^2 k'^2 K^4 \xi^3 +\mathcal{O}(\xi^5 )$$
in fact gives a finite contribution.
The results of the previous section show that
\begin{align*}
\left( {\mathcal{V}}\sp{\dagger}\,\mathcal{Q}^{-1}\mathcal{H}\,{\mathcal{V}}  \right)_{ij}&=
\frac{\im  r^2 x_1 x_2 x_3\, k^2 k'^2 K^4}{2 D} + 
\boldsymbol{{\mathfrak{v}}}_{i,2}\sp{\dagger}\left(
\mathrm{Res}_{z=1} \mathcal{Q}^{-1}(z,\boldsymbol{x})\right)
\mathcal{H}\,
\boldsymbol{{\mathfrak{v}}}_{j,2}\\
&
\quad +
\boldsymbol{{\mathfrak{v}}}_{i,2}\sp{\dagger}\left(
\mathrm{Res}_{z=1} \frac{\mathcal{Q}^{-1}(z,\boldsymbol{x})}{(1-z)^2}\right)
\mathcal{H}\,
\boldsymbol{{\mathfrak{v}}}_{j,0}
 +
\boldsymbol{{\mathfrak{v}}}_{i,0}\sp{\dagger}\left(
\mathrm{Res}_{z=1} \frac{\mathcal{Q}^{-1}(z,\boldsymbol{x})}{(1-z)^2}\right)
\mathcal{H}\,
\boldsymbol{{\mathfrak{v}}}_{j,2}.
\end{align*}

We can in fact say more about the structure of this matrix. Its constancy means that
\begin{align*}
\lim_{z\rightarrow 1}
\left( {\mathcal{V}}\sp{\dagger}\,\mathcal{Q}^{-1}\mathcal{H}\,{\mathcal{V} }\right)_{jk}&=
\lim_{z\rightarrow -1}
\left( {\mathcal{V}}\sp{\dagger}\,\mathcal{Q}^{-1}\mathcal{H}\,{\mathcal{V} }\right)_{jk}\\
&=\mathcal{M}_{jj} \left(\lim_{z\rightarrow 1}{\mathcal{V}}\sp{\dagger}U\sp\dagger U
\,\mathcal{Q}^{-1}\mathcal{H}\, U U{\mathcal{V} }\right)_{jk}
\overline{\mathcal{M}}_{kk}\\
&= -\exp(2\mu_j+ 2\overline{\mu}_k) 
\lim_{z\rightarrow 1}
\left( {\mathcal{V}}\sp{\dagger}\,\mathcal{Q}^{-1}\mathcal{H}\,{\mathcal{V} }\right)_{jk}
\end{align*}
upon using (\ref{IQmon}) and that $U\sp{T}U=1_4$, $U^2=-1_4$, $U\mathcal{H}=\mathcal{H}U$. 
Thus the $(j,k)$-element is non-vanishing
only if $\exp(2\mu_j+ 2\overline{\mu}_k) =-1$. We have seen that this is always the case for the $(1,3)$,
$(3,1)$, $(2,4)$ and $(4,2)$ elements. We prove in Appendix \ref{proofconstantmatrix}:

\begin{theorem}\label{constantmatrix} 
With our ordering $\mathcal{J}(P_1)=P_3$,  $\mathcal{J}(P_2)=P_4$, then
\begin{equation}\label{constancystructurew}
{\mathcal{W}}\sp{\dagger}\,\mathcal{Q}\mathcal{H}\,{\mathcal{W} }=
\begin{pmatrix} 0&0& \mathfrak{f}_3 &0\\ 0&0&0&\mathfrak{f}_4 \\  \mathfrak{f}_10&0&0\\ 0& \mathfrak{f}_2&0&0
\end{pmatrix}
\end{equation}
and the constant matrix
\begin{equation}\label{constancystructure}
\mathfrak{C}:=
{\mathcal{V}}\sp{\dagger}\,\mathcal{Q}^{-1}\mathcal{H}\,{\mathcal{V} }=-r^2
\begin{pmatrix} 0&0& 1/\mathfrak{f}_1 &0\\ 0&0&0&1/\mathfrak{f}_2 \\ 1/ \mathfrak{f}_3& 0&0&0\\ 0& 1/\mathfrak{f}_4&0&0
\end{pmatrix}
\end{equation}
satisfies $\mathcal{M}\mathfrak{C}=-\mathfrak{C}\overline{\mathcal{M}}^{-1}$. 
Here
\begin{align}\label{constancystructuref}
\mathfrak{f}_k:&=c^2\,\frac{\pi^2}2 \prod_{s=1}\sp{4} \theta_s[P_k]\, 
\left( \frac{ \theta_2(0)}{\theta_2[\sum_{j=1}\sp{4} P_j]}
\right)^2
\frac{\prod_{\substack{r<s \\  r,s\ne k}}\theta_3[P_r+P_s]^2}{\prod_{l\ne k}\theta_1[P_l-P_k]^2}
\\
&\quad\times
\left(
 \im \theta_2[2 P_k] \theta_2(0)^3\, x_2 +  \theta_4[2 P_k] \theta_4(0)^3 \,x_1
-\frac{  \theta_2[2 P_k] \theta_4[2 P_k]}
{   \theta_1[2 P_k]   }\theta_3(0)^3\, x_3
\right)\nonumber\\
&=
{\frac {{K}^{2}{\zeta _{{k}}}^{2} \left( {R_-}\,{\zeta _{{k}}}^{3}+
2\,Rx_{{3}}{\zeta _{{k}}}^{2}-{ R_+}\,\zeta _{{k}}-x_{{3}} \left( {{
 S_-}}^{2}+{{S_+}}^{2} \right)  \right) }{ S_-^{2} \left( 4
\,ix_{{3}}{ x_-}\,{\zeta _{{k}}}^{3}-R{\zeta _{{k}}}^{2}+12\,i{
x_+}\,x_{{3}}\zeta _{{k}}+{{S_+}}^{2} \right) ^{2}}}
\end{align}
satisfies $\overline{\mathfrak{f}_k}=-\mathfrak{f}_{\mathcal{J}(k)}$, where
\begin{align*}
 R_{\mp}&=\imath S_-^2( x_{\mp}(2{k'}^2-1)+x_{\pm} )\mp 16\imath x_{\mp}x_3^2,\quad
R=2K^2{k'}^2-S^2-8x_3^2,\\
 S_{\pm}&=\sqrt{K^2-4x_{\pm}^2},\quad    S=\sqrt{K^2-4x_+x_-}.
\end{align*}

\end{theorem}

We may now use the matrix $\mathfrak{C}$ in evaluating
$\mathcal{V}\sp{\dagger} \mathcal{Q}^{-1} {\mathcal{V}}$.  By inserting $1_4={\mathcal W}\sp{\dagger}{\mathcal{V}},$ 
we have
$$
{\mathcal{V}}\sp{\dagger} 
 \mathcal{Q}^{-1} 
 {\mathcal{V}}
 = \frac1{r^2}\,\mathfrak{C} \left(\mathcal{W}\sp\dagger \mathcal{H} 
 {\mathcal{V}}\right)
 =
-\frac1{r^2}\left(\mathcal{V}\sp\dagger \mathcal{H} 
 {\mathcal{W}} (z)\right) \mathfrak{C}  
 $$
 and so 
$\mathfrak{C} \left(\mathcal{W}\sp\dagger \mathcal{H} 
 {\mathcal{V}}(z)\right)$ is Hermitian for general $z$. Although $\mathcal{W}\sp\dagger \mathcal{H} 
 {\mathcal{V}}(z)$ diverges as $z\rightarrow\pm$ the projector removes these. Thus
 $$\lim_{z\rightarrow 1}\mu\sp\dagger\mathfrak{C} \left(\mathcal{W}\sp\dagger \mathcal{H} 
 {\mathcal{V}}(z)\right)\mu
 =
\mu\sp\dagger
  \mathfrak{C}
 \boldsymbol{{\mathfrak{w}}}_{0}\sp\dagger 
\mathcal{H} \boldsymbol{{\mathfrak{v}}}_{2}
 \mu
 $$
 is Hermitian, whence
$$\left. \mu\sp\dagger \left(
{\mathcal{V}}\sp{\dagger} 
 \mathcal{Q}^{-1} {\mathcal{V}} 
\right)
\mu\right|_{z=1}
 =
 \left.
 \frac1{r^2}\,\mu\sp\dagger \mathfrak{C} \left(\mathcal{W}\sp\dagger \mathcal{H} 
 {\mathcal{V}}\right)
 \mu\right|_{z=1}
 =
  \frac1{r^2}\,\mu\sp\dagger 
 \mathfrak{C} \left(
 \boldsymbol{{\mathfrak{w}}}_{0}\sp\dagger 
\mathcal{H} \boldsymbol{{\mathfrak{v}}}_{2}\right)
 \mu .
 $$
 Therefore
 \begin{align}\nonumber
 \mu_{aj}\sp\dagger \left( \left.{\mathcal{V}}\sp{\dagger}
 \mathcal{Q}^{-1} {\mathcal{V}} \right|_{z=-1}\sp{z=1}   \right)_{jk}\mu_{k b}
 &=
 \frac1{r^2}\, \mu_{aj}\sp\dagger
  \mathfrak{C}_{jl} \left(\left. \mathcal{W}\sp\dagger \mathcal{H} 
 {\mathcal{V}}   \right|_{z=-1}\sp{z=1}   \right)_{lk}\mu_{k b} ,\\
 &=
 \frac1{r^2}\, \mu_{aj}\sp\dagger
  \mathfrak{C}_{jl} 
   \left(
 \boldsymbol{{\mathfrak{w}}}_{0}\sp\dagger 
\mathcal{H} \boldsymbol{{\mathfrak{v}}}_{2}\right)_{lk}
\left[1-\exp(-2\overline{\mu}_l+ 2\overline{\mu}_k) \right]
 \mu_{k b} ,
 \label{normexpa}
 \\
 \mu_{aj}\sp\dagger \left( \left.  z {\mathcal{V}}\sp{\dagger}
 \mathcal{Q}^{-1} {\mathcal{V}} \right|_{z=-1}\sp{z=1}   \right)_{jk}\mu_{k b}
 &=
\frac1{r^2}\, \mu_{aj}\sp\dagger
  \mathfrak{C}_{jl} 
   \left(
 \boldsymbol{{\mathfrak{w}}}_{0}\sp\dagger 
\mathcal{H} \boldsymbol{{\mathfrak{v}}}_{2}\right)_{lk}
\left[1+\exp(-2\overline{\mu}_l+ 2\overline{\mu}_k) \right]
 \mu_{k b}.
 \end{align}

\subsection{Evaluating \texorpdfstring{$\boldsymbol{\int_{-1}\sp{1} \mathrm{d}z\, \boldsymbol{{v}}_a\sp\dagger \boldsymbol{{v}}_b }$}{intzvbarv}}
 
 We have
 \begin{align*}
\int \mathrm{d}z\, \boldsymbol{{v}}_a\sp\dagger \boldsymbol{{v}}_b
&=\mu_{aj}\sp\dagger\left( \int \mathrm{d}z\, {\mathcal{V}}\sp{\dagger} {\mathcal{V}}\right)_{jk}\mu_{k b}
= \mu_{aj}\sp\dagger \left({\mathcal{V}}\sp{\dagger}
 \mathcal{Q}^{-1} {\mathcal{V}}\right)_{jk}\mu_{k b}.
\end{align*}
Thus
 \begin{align*}
\int_{-1}\sp{1} \mathrm{d}z\, \boldsymbol{{v}}_a\sp\dagger \boldsymbol{{v}}_b
&
= \mu_{aj}\sp\dagger \left(
\boldsymbol{{\mathfrak{v}}}_{j,2}\sp{\dagger}
\mathrm{Res}_{z=1} \mathcal{Q}^{-1}(z,\boldsymbol{x})
\boldsymbol{{\mathfrak{v}}}_{k,2}
-
\mathcal{M}_{jj}
\boldsymbol{{\mathfrak{v}}}_{j,2}\sp{\dagger} U\sp{T}
\mathrm{Res}_{z=-1} \mathcal{Q}^{-1}(z,\boldsymbol{x}) U
\boldsymbol{{\mathfrak{v}}}_{k,2}
{\overline{\mathcal{M}}}_{kk}
 \right)\mu_{k b}\\
 &=
 \mu_{aj}\sp\dagger \left(
\boldsymbol{{\mathfrak{v}}}_{j,2}\sp{\dagger}
\mathrm{Res}_{z=1} \mathcal{Q}^{-1}(z,\boldsymbol{x})
\boldsymbol{{\mathfrak{v}}}_{k,2}
\left[1+\exp(2\mu_j+ 2\overline{\mu}_k) \right]
 \right)\mu_{k b}
\end{align*}
where we have used (\ref{IQmon}) together with $U\sp{T}U=1_4$ and $U^2=-1_4$.
Combining this with (\ref{normexpa}) and Theorem \ref{constantmatrix} yields:

\begin{theorem}\label{normalizationthm}
 We have the formulae for the orthogonalisation
  \begin{align}\nonumber
\int_{-1}\sp{1} \mathrm{d}z\, \boldsymbol{{v}}_a\sp\dagger \boldsymbol{{v}}_b
 &=
 \mu_{aj}\sp\dagger \left(
\boldsymbol{{\mathfrak{v}}}_{j,2}\sp{\dagger}
\mathrm{Res}_{z=1} \mathcal{Q}^{-1}(z,\boldsymbol{x})
\boldsymbol{{\mathfrak{v}}}_{k,2}
\left[1+\exp(2\mu_j+ 2\overline{\mu}_k) \right]
 \right)\mu_{k b} ,\\
 &
 =
 \frac1{r^2}\, \mu_{aj}\sp\dagger
  \mathfrak{C}_{jl} 
   \left(
 \boldsymbol{{\mathfrak{w}}}_{0}\sp\dagger 
\mathcal{H} \boldsymbol{{\mathfrak{v}}}_{2}\right)_{lk}
\left[1-\exp(-2\overline{\mu}_l+ 2\overline{\mu}_k) \right]
 \mu_{k b}
, \\ \nonumber
  &=
 \mu_{aj}\sp\dagger \mathfrak{F}_{jl}
 \left(
 \boldsymbol{W}_{0}\sp\dagger 
\mathcal{H} \boldsymbol{{\mathfrak{v}}}_{2}\right)_{lk}
\left[1-\exp(-2\overline{\mu}_l+ 2\overline{\mu}_k) \right]
 \mu_{k b},
\end{align}
where
$$
\mathfrak{F}= \frac1{r^2}\,  \mathfrak{C}\mathfrak{d}\sp\dagger
=
\begin{pmatrix} 0&0& {\overline{\mathfrak{d}_3 /\mathfrak{f}_3}} &0\\ 
0&0&0&{\overline{\mathfrak{d}_4/\mathfrak{f}_4}}\\ {\overline{\mathfrak{d}_1/ \mathfrak{f}_1}}& 0&0&0\\ 
0& {\overline{\mathfrak{d}_2 /\mathfrak{f}_4}}&0&0
\end{pmatrix} = \mathfrak{F}\sp\dagger
$$
and
 \begin{align*}
 \mathfrak{f}_j /\mathfrak{d}_j=\mathfrak{D}_j \mathfrak{f}_j=
 \imath 
\frac{\zeta_j K^2 ( R_-\zeta_j^3 +2R x_3 \zeta_j^2 -
  R_+\zeta_j-x_3(S_+^2+S_-^2)   )}
{(4\imath x_3 x_-\zeta_j^3-R\zeta_j^2+12 \imath x_+x_3\zeta_j+S_+^2)\,  S_-^2
}.
\end{align*}
Here we have defined
\begin{align*}
 R_{\mp}&=\imath S_-^2( x_{\mp}(2{k'}^2-1)+x_{\pm} )\mp 16\imath x_{\mp}x_3^2,\quad
R=2K^2{k'}^2-S^2-8x_3^2 ,\\
 S_{\pm}&=\sqrt{K^2-4x_{\pm}^2},\quad    S=\sqrt{K^2-4x_+x_-}.
\end{align*}
Conjugation acts at these quantities by
\begin{align}
\overline{S}&=S,\quad \overline{S_{\pm}} = S_{\mp},\quad
\overline{R}=R,\quad\overline{R_{\pm}}=-\imath S_+^2(x_{\mp}(2{k'}^2-1) + x_{\pm})
\mp 16\imath x_{\mp}x_3^2
\end{align} 
and our convention is that
\begin{equation} \overline{(\zeta_1,\zeta_2,\zeta_3, \zeta_4)} = (-1/\zeta_3,-1/\zeta_4,-1/\zeta_1, -1/\zeta_2) .
 \end{equation}
\end{theorem}

\subsection{Evaluating \texorpdfstring{$\boldsymbol{\int_{-1}\sp{1} \mathrm{d}z\,z \boldsymbol{{v}}_a\sp\dagger \boldsymbol{{v}}_b }$}{intzzvbarv}}
Now we must simplify
\begin{equation}
\int \mathrm{d}z\, z \boldsymbol{{v}}_a\sp\dagger \boldsymbol{{v}}_b
=\mu_{aj}\sp\dagger\left( \int \mathrm{d}z\,z {\mathcal{V}}\sp{\dagger} {\mathcal{V}}\right)_{jk}\mu_{k b}
=\mu_{aj}\sp\dagger \left( {\mathcal{V}}\sp{\dagger} 
 \mathcal{Q}^{-1} \left[ z+\mathcal{H}\,\frac{x\sp{i}}{r^2}\frac{\partial}{\partial x\sp{i}}   \right]
 {\mathcal{V}}
\right)_{jk}
\mu_{k b}.
 \end{equation}
The first term here is treated as before and we have
\begin{align*}
\int_{-1}\sp{1}\mathrm{d}z\, z \boldsymbol{{v}}_a\sp\dagger \boldsymbol{{v}}_b
&=\mu_{aj}\sp\dagger \left(
\boldsymbol{{\mathfrak{v}}}_{j,2}\sp{\dagger}
\mathrm{Res}_{z=1} \mathcal{Q}^{-1}(z,\boldsymbol{x})
\boldsymbol{{\mathfrak{v}}}_{k,2}
\left[1-\exp(2\mu_j+ 2\overline{\mu}_k) \right]
 \right)\mu_{k b}\\
&\qquad +\  \mu_{aj}\sp\dagger \left( {\mathcal{V}}\sp{\dagger} 
 \mathcal{Q}^{-1} \left. \mathcal{H}\,\frac{x\sp{i}}{r^2}\frac{\partial}{\partial x\sp{i}}  
 {\mathcal{V}} \right|_{z=-1}\sp{z=1}
\right)_{jk}
\mu_{k b}.
 \end{align*}
The constancy of $\mathfrak{C}={\mathcal{V}}\sp{\dagger}  \mathcal{Q}^{-1} \mathcal{H}  {\mathcal{V}} $  means that the derivative acts only on the ${\mathcal{V}} $.
Let us further consider the final term. Writing
$$
{\mathcal{V}}\sp{\dagger} 
 \mathcal{Q}^{-1} \mathcal{H}\,\frac{x\sp{i}}{r^2}\frac{\partial}{\partial x\sp{i}}  
 {\mathcal{V}}
 =\mathfrak{C} \left(\mathcal{W}\sp\dagger\frac{x\sp{i}}{r^2}\frac{\partial}{\partial x\sp{i}}  
 {\mathcal{V}}\right)
 $$
then
$$\mu\sp\dagger \left(
{\mathcal{V}}\sp{\dagger} 
 \mathcal{Q}^{-1}  \left. \mathcal{H}\,\frac{x\sp{i}}{r^2}\frac{\partial}{\partial x\sp{i}}  
 {\mathcal{V}} \right|_{z=-1}\sp{z=1}
\right)
\mu
 =\mu\sp\dagger 
 \mathfrak{C} 
  \left(\mathcal{W}\sp\dagger  \left. \frac{x\sp{i}}{r^2}\frac{\partial}{\partial x\sp{i}}  
 {\mathcal{V}}   \right|_{z=-1}\sp{z=1} \right)\mu
 $$
 with the only $z$ dependence in $\mathcal{W}\sp\dagger {\mathcal{V}}'$, where $'$ abbreviates
 $ \frac{x\sp{i}}{r^2}\frac{\partial}{\partial x\sp{i}}  $.
Noting that $x\sp{i}\partial_i \boldsymbol{{\mathfrak{v}}}_{k,0}=0$ and that $x\sp{i}\partial_i \boldsymbol{{\mathfrak{v}}}_{k,1}
=\boldsymbol{{\mathfrak{v}}}_{k,1}$  is annihilated by the projector then
\begin{align*}
\left(\mathcal{W}\sp\dagger \left. 
 \frac{x\sp{i}}{r^2}\frac{\partial}{\partial x\sp{i}}  
 {\mathcal{V}}\right|_{z=1}\right)\mu
 &=
 \boldsymbol{{\mathfrak{w}}}_{0}\sp\dagger 
 \left(\frac{x\sp{i}}{r^2}\frac{\partial}{\partial x\sp{i}} \boldsymbol{{\mathfrak{v}}}_{2}\right) \mu
 \intertext{while}
 \left(\mathcal{W}\sp\dagger \left. 
 \frac{x\sp{i}}{r^2}\frac{\partial}{\partial x\sp{i}}  
 {\mathcal{V}}\right|_{z=-1}\right)\mu
 &=
 \lim_{\xi\rightarrow0}\, 
 \overline{\mathcal{M}}^{-1}\left(
  \frac{1}{\sqrt{\xi}}\boldsymbol{ \mathfrak{w}_{0}}-
\sqrt{\xi}\boldsymbol{\mathfrak{w}_{1}}+ \xi^{3/2}\boldsymbol{\mathfrak{w}_{2}}
 \right)\sp\dagger U\sp\dagger U\\
 &\qquad\times
\left(
\frac{x\sp{i}}{r^2}\frac{\partial}{\partial x\sp{i}}
\left[\left(
 \frac{1}{\xi\sp{3/2}}\boldsymbol{{\mathfrak{v}}}_0-
\frac{1}{\xi\sp{1/2}}\boldsymbol{{\mathfrak{v}}}_1+{\xi\sp{1/2}}\boldsymbol{{\mathfrak{v}}}_{2}
\right) \overline{\mathcal{M}}\right]\right)\mu
\\
&= \overline{\mathcal{M}}^{-1}
 \boldsymbol{{\mathfrak{w}}}_{0}\sp\dagger 
 \left(\frac{x\sp{i}}{r^2}\frac{\partial}{\partial x\sp{i}} \boldsymbol{{\mathfrak{v}}}_{2}\right)
 \overline{\mathcal{M}}\mu
 + \left( \overline{\mathcal{M}}^{-1}
 \frac{x\sp{i}}{r^2}\frac{\partial}{\partial x\sp{i}} 
 \overline{\mathcal{M}}\right)\mu
 \end{align*}
where we have used (\ref{ort}) and (\ref{VWrel}).
Thus
 \begin{align*}\mu_{aj}\sp\dagger \left( {\mathcal{V}}\sp{\dagger} 
 \mathcal{Q}^{-1} \left. \mathcal{H}\,\frac{x\sp{i}}{r^2}\frac{\partial}{\partial x\sp{i}}  
 {\mathcal{V}} \right|_{z=-1}\sp{z=1}
\right)_{jk}
\mu_{k b}
&=\mu_{aj}\sp\dagger \mathfrak{C}_{jl}
 \left(\boldsymbol{{\mathfrak{w}}}_{0}\sp\dagger 
 \left[\frac{x\sp{i}}{r^2}\frac{\partial}{\partial x\sp{i}} \boldsymbol{{\mathfrak{v}}}_{2}\right]\right)_{lk}
 \left[1-\exp(-2\overline{\mu}_l+ 2\overline{\mu}_k) \right]\mu_{k b}\\
 &\qquad
 - \mu_{aj}\sp\dagger \mathfrak{C}_{jk}\left(
 \frac{2 x\sp{i}}{r^2}\frac{\partial}{\partial x\sp{i}}\overline{\mu}_k
 \right)\mu_{k b}
  \end{align*}
and consequently
 \begin{align}
\int_{-1}\sp{1}\mathrm{d}z\, z \boldsymbol{{v}}_a\sp\dagger \boldsymbol{{v}}_b
&=\mu_{aj}\sp\dagger \left(
\boldsymbol{{\mathfrak{v}}}_{j,2}\sp{\dagger}
\mathrm{Res}_{z=1} \mathcal{Q}^{-1}(z,\boldsymbol{x})
\boldsymbol{{\mathfrak{v}}}_{k,2}
\left[1-\exp(2\mu_j+ 2\overline{\mu}_k) \right]
 \right)\mu_{k b} \nonumber\\
&
\label{inthiggs}
\qquad +
\mu_{aj}\sp\dagger \mathfrak{C}_{jl}
 \left(\boldsymbol{{\mathfrak{w}}}_{0}\sp\dagger 
 \left[\frac{x\sp{i}}{r^2}\frac{\partial}{\partial x\sp{i}} \boldsymbol{{\mathfrak{v}}}_{2}\right]\right)_{lk}
 \left[1-\exp(-2\overline{\mu}_l+ 2\overline{\mu}_k) \right]\mu_{k b}\\
 &\qquad
 - \mu_{aj}\sp\dagger \mathfrak{C}_{jk}\left(
 \frac{2 x\sp{i}}{r^2}\frac{\partial}{\partial x\sp{i}}\overline{\mu}_k
 \right)\mu_{k b} . \nonumber
 \end{align}
The derivatives 
 $ \boldsymbol{{\mathfrak{w}}}_{0}\sp\dagger \left({\partial _{i}}  \boldsymbol{{\mathfrak{v}}}_{2}\right) \mu$
 appearing here may be combined in various ways.
   From (\ref{VWrel}) and that $\boldsymbol{{\mathfrak{v}}}_{k,0}$ is constant we have that
 \begin{align*}
 \boldsymbol{{\mathfrak{w}}}_{0}\sp\dagger  \boldsymbol{{\mathfrak{v}}}_{2} \mu=\mu,\qquad
 \boldsymbol{{\mathfrak{w}}}_{0}\sp\dagger 
\left({\partial _{i}}  \boldsymbol{{\mathfrak{v}}}_{2}\right) \mu&=
 -\left({\partial _{i}}   \boldsymbol{{\mathfrak{w}}}_{0}\sp\dagger\right)  \boldsymbol{{\mathfrak{v}}}_{2} \mu 
 .
 \end{align*}
 Also
 \begin{align*}
 {\partial _{i}}   \boldsymbol{\mathfrak{w}}_{k,0}&=
 {\partial _{i}}  \left(   {\mathfrak{d}}_{k} \boldsymbol{W}_{k,0}  \right)=
 \left(  {\partial _{i}}   {\mathfrak{d}}_{k} \right) \boldsymbol{W}_{k,0}  +
 {\mathfrak{d}}_{k}  {\partial _{i}}  
 \begin{pmatrix}1\\ \imath \zeta_k\\  \imath \zeta_k\\-\zeta_k^2 \end{pmatrix},
 \\
& =\left[\frac{  {\partial _{i}}   {\mathfrak{d}}_{k} }{  {\mathfrak{d}}_{k} } +
 \frac{  {\partial _{i}}\zeta_k}{\zeta_k}\right]  \boldsymbol{\mathfrak{w}}_{k,0}
 - \begin{pmatrix}1\\ 0\\0\\ \zeta_k^2 \end{pmatrix}
  {\mathfrak{d}}_{k} \,  \frac{  {\partial _{i}}\zeta_k}{\zeta_k} ,
\intertext{giving}
\boldsymbol{{\mathfrak{w}}}_{0}\sp\dagger 
\left({\partial _{i}}  \boldsymbol{{\mathfrak{v}}}_{2}\right) \mu
&=-\diag \left[\frac{  {\partial _{i}}   {\mathfrak{d}}_{k} }{  {\mathfrak{d}}_{k} } +
 \frac{  {\partial _{i}}\zeta_k}{\zeta_k}\right] \sp\dagger\mu
 +
 \diag\left[ {\mathfrak{d}}_{k}  \frac{  {\partial _{i}}\zeta_k}{\zeta_k}  \right] \sp\dagger
 \begin{pmatrix}
 1&\ldots&1& \\ 0&&0\\0&&0\\ \zeta_1^2 &\ldots&\zeta_4^2
 \end{pmatrix} \sp\dagger
  \boldsymbol{{\mathfrak{v}}}_{2} \mu .
 \end{align*}
Similarly
\begin{align*}
 \boldsymbol{{\mathfrak{w}}}_{0}\sp\dagger  \boldsymbol{{\mathfrak{v}}}_{2}  \overline{\mathcal{M}}\mu= \overline{\mathcal{M}}\mu,
\qquad
 \boldsymbol{{\mathfrak{w}}}_{0}\sp\dagger 
\left({\partial _{i}}  \boldsymbol{{\mathfrak{v}}}_{2}\right) \overline{\mathcal{M}} \mu=
 -\left({\partial _{i}}   \boldsymbol{{\mathfrak{w}}}_{0}\sp\dagger\right)  \boldsymbol{{\mathfrak{v}}}_{2}
  \overline{\mathcal{M}} \mu,
 \end{align*}
and we can express the derivative similarly at $z=-1$.

Bringing these results together gives
\begin{theorem}\label{higgsthm}
\begin{align*}
\nonumber
\int_{-1}\sp{1}\mathrm{d}z\, z \boldsymbol{{v}}_a\sp\dagger \boldsymbol{{v}}_b
&=\mu_{aj}\sp\dagger \left(
\boldsymbol{{\mathfrak{v}}}_{j,2}\sp{\dagger}
\mathrm{Res}_{z=1} \mathcal{Q}^{-1}(z,\boldsymbol{x})
\boldsymbol{{\mathfrak{v}}}_{k,2}
\left[1-\exp(2\mu_j+ 2\overline{\mu}_k) \right]
 \right)\mu_{k b}\\
&\qquad +\  \mu_{aj}\sp\dagger \left( {\mathcal{V}}\sp{\dagger} 
 \mathcal{Q}^{-1} \left. \mathcal{H}\,\frac{x\sp{i}}{r^2}\frac{\partial}{\partial x\sp{i}}  
 {\mathcal{V}} \right|_{z=-1}\sp{z=1}
\right)_{jk}
\mu_{k b}
\\
\nonumber
&=\mu_{aj}\sp\dagger \left(
\boldsymbol{{\mathfrak{v}}}_{j,2}\sp{\dagger}
\mathrm{Res}_{z=1} \mathcal{Q}^{-1}(z,\boldsymbol{x})
\boldsymbol{{\mathfrak{v}}}_{k,2}
\left[1-\exp(2\mu_j+ 2\overline{\mu}_k) \right]
 \right)\mu_{k b} \nonumber\\
&
\qquad +
\mu_{aj}\sp\dagger \mathfrak{C}_{jl}
 \left(\boldsymbol{{\mathfrak{w}}}_{0}\sp\dagger 
 \left[\frac{x\sp{i}}{r^2}\frac{\partial}{\partial x\sp{i}} \boldsymbol{{\mathfrak{v}}}_{2}\right]\right)_{lk}
 \left[1-\exp(-2\overline{\mu}_l+ 2\overline{\mu}_k) \right]\mu_{k b}\\
 &\qquad
 - \mu_{aj}\sp\dagger \mathfrak{C}_{jk}\left(
 \frac{2 x\sp{i}}{r^2}\frac{\partial}{\partial x\sp{i}}\overline{\mu}_k
 \right)\mu_{k b} 
 \nonumber
 \\
  &=
\frac1{r^2}\, \mu_{aj}\sp\dagger
  \mathfrak{C}_{jl} 
   \left(
 \boldsymbol{{\mathfrak{w}}}_{0}\sp\dagger 
\mathcal{H} \boldsymbol{{\mathfrak{v}}}_{2}\right)_{lk}
\left[1+\exp(-2\overline{\mu}_l+ 2\overline{\mu}_k) \right]
 \mu_{k b}
 \\
  &\qquad +
  \mu_{aj}\sp\dagger \mathfrak{C}_{jl}
 \left(\boldsymbol{{\mathfrak{w}}}_{0}\sp\dagger 
 \left[\frac{x\sp{i}}{r^2}\frac{\partial}{\partial x\sp{i}} \boldsymbol{{\mathfrak{v}}}_{2}\right]\right)_{lk}
 \left[1-\exp(-2\overline{\mu}_l+ 2\overline{\mu}_k) \right]\mu_{k b}
   \nonumber
   \\
 &\qquad
 - \mu_{aj}\sp\dagger \mathfrak{C}_{jk}\left(
 \frac{2 x\sp{i}}{r^2}\frac{\partial}{\partial x\sp{i}}\overline{\mu}_k
 \right)\mu_{k b} \\
   \nonumber
  &=
   \mu_{aj}\sp\dagger
  \mathfrak{F}_{jl} 
   \left(
 \boldsymbol{{{W}}}_{0}\sp\dagger 
\mathcal{H} \boldsymbol{{\mathfrak{v}}}_{2}\right)_{lk}
\left[1+\exp(-2\overline{\mu}_l+ 2\overline{\mu}_k) \right]
 \mu_{k b}
 \\
  &\qquad +
  \mu_{aj}\sp\dagger \mathfrak{C}_{jl}
 \left(\boldsymbol{{\mathfrak{w}}}_{0}\sp\dagger 
 \left[\frac{x\sp{i}}{r^2}\frac{\partial}{\partial x\sp{i}} \boldsymbol{{\mathfrak{v}}}_{2}\right]\right)_{lk}
 \left[1-\exp(-2\overline{\mu}_l+ 2\overline{\mu}_k) \right]\mu_{k b}\\
 &\qquad
 - \mu_{aj}\sp\dagger \mathfrak{C}_{jk}\left(
 \frac{2 x\sp{i}}{r^2}\frac{\partial}{\partial x\sp{i}}\overline{\mu}_k
 \right)\mu_{k b} 
 \\
 &=
   \mu_{aj}\sp\dagger
  \mathfrak{F}_{jl} 
   \left(
 \boldsymbol{{{W}}}_{0}\sp\dagger 
\mathcal{H} \boldsymbol{{\mathfrak{v}}}_{2}\right)_{lk}
\left[1+\exp(-2\overline{\mu}_l+ 2\overline{\mu}_k) \right]
 \mu_{k b}
 \\
  &\qquad +
  \mu_{aj}\sp\dagger \mathfrak{F}_{jl}
 \left(\boldsymbol{W_{0}}\sp\dagger 
 \left[x\sp{i} \frac{\partial}{\partial x\sp{i}} \boldsymbol{{\mathfrak{v}}}_{2}\right]\right)_{lk}
 \left[1-\exp(-2\overline{\mu}_l+ 2\overline{\mu}_k) \right]\mu_{k b}\\
 &\qquad
 - \mu_{aj}\sp\dagger \mathfrak{C}_{jk}\left(
 \frac{2 x\sp{i}}{r^2}\frac{\partial}{\partial x\sp{i}}\overline{\mu}_k
 \right)\mu_{k b} 
 \\
 &=
   \mu_{aj}\sp\dagger
  \mathfrak{F}_{jl} 
   \left(
 \boldsymbol{{{W}}}_{0}\sp\dagger 
\mathcal{H} \boldsymbol{{\mathfrak{v}}}_{2}\right)_{lk}
\left[1+\exp(-2\overline{\mu}_l+ 2\overline{\mu}_k) \right]
 \mu_{k b}
 \nonumber
 \\
  &\qquad +
  \mu_{aj}\sp\dagger \mathfrak{F}_{jl}
   \diag\left[  \frac{  {x\sp{i}\,\partial _{i}}\zeta_k}{\zeta_k}  \right]_{ll} \sp\dagger
 \left(
 \begin{pmatrix}
 1&\ldots&1& \\ 0&&0\\0&&0\\ \zeta_1^2 &\ldots&\zeta_4^2
 \end{pmatrix} \sp\dagger
  \boldsymbol{{\mathfrak{v}}}_{2}
  \right)_{lk}
 \left[1-\exp(-2\overline{\mu}_l+ 2\overline{\mu}_k) \right]\mu_{k b}\\
 &\qquad
 - \mu_{aj}\sp\dagger \mathfrak{C}_{jk}\left(
 \frac{2 x\sp{i}}{r^2}\frac{\partial}{\partial x\sp{i}}\overline{\mu}_k
 \right)\mu_{k b} .
 \nonumber
 \end{align*}

\end{theorem}

We conclude with some comments on the Hermiticity of these expressions.
From
 \begin{align*}
  \left( {\mathcal{V}}\sp{\dagger} 
 \mathcal{Q}^{-1} \mathcal{H}\,\frac{x\sp{i}}{r^2}\frac{\partial}{\partial x\sp{i}}  
 {\mathcal{V}} 
\right)\sp\dagger
&=\frac{x\sp{i}}{r^2}\frac{\partial}{\partial x\sp{i}}\left( {\mathcal{V}}\sp{\dagger} 
\mathcal{H}  \mathcal{Q}^{-1}  
 {\mathcal{V}} 
\right)
-
{\mathcal{V}}\sp{\dagger} \left(
\frac{x\sp{i}}{r^2}\frac{\partial}{\partial x\sp{i}}\left[ \mathcal{H} \mathcal{Q}^{-1}\right]\right)
 {\mathcal{V}} 
-
{\mathcal{V}}\sp{\dagger} 
 \mathcal{H} \mathcal{Q}^{-1}\,\frac{x\sp{i}}{r^2}\frac{\partial}{\partial x\sp{i}}  
 {\mathcal{V}} 
 \intertext{together with the homogeneity $ x\sp{i} \partial_i\left(\mathcal{Q}^{-1}
 \mathcal{H} \right)= \mathcal{Q}^{-1}
 \mathcal{H} $ (which easily follows from the simpler
 $ x\sp{i} \partial_i\left(\mathcal{Q}
 \mathcal{H} \right)= \mathcal{Q}
 \mathcal{H} $)
this becomes }
&=\frac{x\sp{i}}{r^2}\frac{\partial}{\partial x\sp{i}}\left( {\mathcal{V}}\sp{\dagger} 
\mathcal{H}  \mathcal{Q}^{-1}  
 {\mathcal{V}} 
\right)
+
\frac1{r^2}\,
{\mathcal{V}}\sp{\dagger} \mathcal{Q}^{-1} \mathcal{H} 
 {\mathcal{V}} 
+
{\mathcal{V}}\sp{\dagger} 
 \mathcal{Q}^{-1} \mathcal{H}\,\frac{x\sp{i}}{r^2}\frac{\partial}{\partial x\sp{i}}  
 {\mathcal{V}} .
  \end{align*}
Then the constancy of $\mathcal{V}\sp{\dagger} 
 \mathcal{Q}^{-1} \mathcal{H}\mathcal{V}$  then shows that
 $$
  \left( {\mathcal{V}}\sp{\dagger} 
 \mathcal{Q}^{-1} \left. \mathcal{H}\,\frac{x\sp{i}}{r^2}\frac{\partial}{\partial x\sp{i}}  
 {\mathcal{V}} \right|_{z=-1}\sp{z=1}
\right)\sp\dagger
=
 \left( {\mathcal{V}}\sp{\dagger} 
 \mathcal{Q}^{-1} \left. \mathcal{H}\,\frac{x\sp{i}}{r^2}\frac{\partial}{\partial x\sp{i}}  
 {\mathcal{V}} \right|_{z=-1}\sp{z=1}
\right)
$$
 is Hermitian. If we consider the hermiticity of (\ref{inthiggs})
 the first term on the right-hand side is manifestly Hermitian while the property $\overline{\mathfrak{f}_k}=-\mathfrak{f}_{\mathcal{J}(k)}$ together with (\ref{mu13mu24})  show that the final term is Hermitian.
Indeed the matrix $ \mathfrak{C}\Lambda$ is Hermitian for any diagonal matrix $\Lambda=
\diag(\lambda_1,\ldots,\lambda_4)$ provided $\overline{\lambda}_j=-\lambda_j$.
By rewriting the middle term 
\begin{align*}
\mu_{aj}\sp\dagger &  \mathfrak{C}_{jl}
 \left(\boldsymbol{{\mathfrak{w}}}_{0}\sp\dagger 
 \left[\frac{x\sp{i}}{r^2}\frac{\partial}{\partial x\sp{i}} \boldsymbol{{\mathfrak{v}}}_{2}\right]\right)_{lk}
 \left[1-\exp(-2\overline{\mu}_l+ 2\overline{\mu}_k) \right]\mu_{k b}\\
 &=
 \mu_{aj}\sp\dagger \left(
\boldsymbol{{\mathfrak{v}}}_{j,2}\sp{\dagger}
\mathrm{Res}_{z=1} \mathcal{Q}^{-1}\mathcal{H}\,
\left[\frac{x\sp{i}}{r^2}\frac{\partial}{\partial x\sp{i}} \boldsymbol{{\mathfrak{v}}}_{k,2}\right]
\left[1+\exp(2\mu_j+ 2\overline{\mu}_k) \right]
 \right)\mu_{k b}
\end{align*}
we see that it also is Hermitian. 
 
\subsection{Calculating the Higgs Field}
We now describe how to calculate the Higgs field $\Phi$ and the more important gauge invariant quantity
\begin{equation}\label{defH}
H(\boldsymbol{x}):=
 \sqrt{-\frac12 \Tr\, \Phi^2}=1-\frac1{r}+O(r\sp{-2}).
 \end{equation}
Define
\begin{align}\label{defGH}
\mathcal{G}&:=\text{Gram}=
   \left( \int_{-1}^1 dz\, \boldsymbol{v}^{\dagger}_a . \boldsymbol{v}_b  \right)_{a,b=1,2},\qquad
\mathcal{H}:=\text{Higgs}'=   \left( \int_{-1}^1 dz\,  z \boldsymbol{v}^{\dagger}_a . \boldsymbol{v}_b  \right)_{a,b=1,2}.
\end{align} 
The (Hermitian and positive definite) Gram matrix $\mathcal{G}$ may be diagonalized and written
as
$$\mathcal{G}= U\sp\dagger\diag U = N\sp\dagger N,\qquad  N:=\sqrt{\diag}\ U.$$
The Higgs field is then (in terms of the unnormalized Higgs$'$ expressions)
$$\Phi= N\sp{\dagger-1}\mathcal{H} N\sp{-1}.$$
When we calculate the Higgs field we will need to calculate this factorization of $\mathcal{G}$ but this
is not necessary to calculate the gauge invariant quantity
$$
H^2=
 -\frac12 \Tr\, \Phi^2=-\frac12 \Tr\, \left(N\sp{\dagger-1}\mathcal{H} N\sp{-1}\right)
 \left(N\sp{\dagger-1}\mathcal{H} N\sp{-1}\right)
 = -\frac12 \Tr\, \mathcal{H}\mathcal{G}\sp{-1} \mathcal{H} \mathcal{G}\sp{-1}  .
 $$
At this stage we have all the needed formulae to evaluate $H^2(\boldsymbol{x})$ (for generic space time
points) and, upon solving the diagonalization, the Higgs field $\Phi$.

Our strategy is then to calculate
\begin{align*}
\text{Gram}&=
\mu_{aj}\sp\dagger \mathfrak{F}_{jl}
 \left(
 \boldsymbol{W}_{0}\sp\dagger 
\mathcal{H} \boldsymbol{{\mathfrak{v}}}_{2}\right)_{lk}
\left[1-\exp(-2\overline{\mu}_l+ 2\overline{\mu}_k) \right]
 \mu_{k b}
 ,\\
\text{Higgs}'_1&= i
 \mu_{aj}\sp\dagger
  \mathfrak{F}_{jl} 
   \left(
 \boldsymbol{{{W}}}_{0}\sp\dagger 
\mathcal{H} \boldsymbol{{\mathfrak{v}}}_{2}\right)_{lk}
\left[1+\exp(-2\overline{\mu}_l+ 2\overline{\mu}_k) \right]
 \mu_{k b}
,\\
\text{Higgs}'_2&= i
\mu_{aj}\sp\dagger \mathfrak{F}_{jl}
   \diag\left[  \frac{  {x\sp{i}\,\partial _{i}}\zeta_k}{\zeta_k}  \right]_{ll} \sp\dagger
 \left(
 \begin{pmatrix}
 1&\ldots&1& \\ 0&&0\\0&&0\\ \zeta_1^2 &\ldots&\zeta_4^2
 \end{pmatrix} \sp\dagger
  \boldsymbol{{\mathfrak{v}}}_{2}
  \right)_{lk}
 \left[1-\exp(-2\overline{\mu}_l+ 2\overline{\mu}_k) \right]\mu_{k b}
,\\
\text{Higgs}'_3&=
 -  i \mu_{aj}\sp\dagger \mathfrak{C}_{jk}\left(
 \frac{2 x\sp{i}}{r^2}\frac{\partial}{\partial x\sp{i}}\overline{\mu}_k
 \right)\mu_{k b} 
\end{align*}
where $\text{Higgs}'_{1,2,3}$ simply correspond to the three terms arising in the evaluation of the integral 
in Theorem \ref{higgsthm} and Gram the terms of Theorem \ref{normalizationthm}.

\subsection{Calculating the Energy Density}
Although one could calculate the gauge fields via (\ref{pangauge}) and from these the energy density 
$\mathcal{E}$,
the easiest way to calculate the energy density is using a formula of Ward \cite{Ward1981b}
\begin{equation}\label{wardenerydensity}
\mathcal{E}(\boldsymbol{x})=-\frac12 \nabla^2 \Tr \Phi^2,
\end{equation}
which is normalized\footnote{
This normalization varies: \cite{manton_sutcliffe_book} chooses $E=\int_{\mathbb{R}^3}\mathcal{E}(\boldsymbol{x})\, d^3x=2\pi n$ while \cite{sutcliffe97} has
$E=\int_{\mathbb{R}^3}\mathcal{E}(\boldsymbol{x})\, d^3x=4\pi n$.
For the charge $n=1$ monopole we have $H(\boldsymbol{x})=\coth(2r)-\frac{1}{2r}$ and
$\mathcal{E}(0)=8/3$.
} such that $E=\int_{\mathbb{R}^3}\mathcal{E}(\boldsymbol{x})\, d^3x=4\pi n$ for the
charge $n$ monopole. Then

\begin{lemma}\label{energydensitylemma}
\begin{align}
-\mathcal{E}(\boldsymbol{x})
&=\,\mathrm{Trace}\left( \left[\frac{\partial \mathcal{H}}{\partial x_i} . \mathcal{G}^{-1}
-\mathcal{H} . \mathcal{G}_{1,i}\right]^2\right) \nonumber\\
&+\,\mathrm{Trace}\left(\left\{\frac{\partial^2 \mathcal{H}}{\partial x_i^2}. \mathcal{G}^{-1}-2\frac{\partial \mathcal{H}}{\partial x_i} . \mathcal{G}_{1,i} + \mathcal{H} . \left[     \mathcal{G}_{1,i}. \frac{\partial \mathcal{G}}{\partial x_i}. \mathcal{G}^{-1}
- \mathcal{G}_{2,i}+ \mathcal{G}^{-1}\,. \frac{\partial \mathcal{G}}{\partial x_i} \,.\mathcal{G}_{1,i}   \right]\right\}.  \mathcal{H}.\mathcal{G}^{-1}\right)
 \label{energy}
\end{align}
where
$$
\mathcal{G}_{1,i} = \mathcal{G}^{-1}\,. \frac{\partial \mathcal{G}}{\partial x_i} \,. \mathcal{G}^{-1},\qquad \mathcal{G}_{2,i} = \mathcal{G}^{-1}\,.
 \frac{\partial^2 \mathcal{G}}{\partial x_i^2} \,. \mathcal{G}^{-1}.
$$
\end{lemma}
Using (\ref{normalizationthm},\ref{higgsthm}) all of the derivatives here involve expressions such as
$\partial_i\zeta$ and $\partial_i\mu$ and we have described earlier how these are to evaluated.
Thus the energy density may be calculated analytically; the large number of terms mean this is best
done with computer algebra. As an example we establish in Appendix \ref{appendix_energy_density_origin}

 \begin{proposition}\label{energy_density_origin}
 The Energy density at the origin is given by
\begin{align} 
\mathcal{E}_{\boldsymbol{x}=0}(k)= \frac{32}{k^8{k'}^2K^4} \left[ k^2(K^2{k'}^2 +E^2-4EK+2K^2+k^2)
-2(E-K)^2    \right]^2 \label{energy_density0}  
\end{align}
The limiting values of $\mathcal{E}_{\boldsymbol{x}=0}(k)$ at $k=0$ and $k=1$ are
\begin{align}
\mathcal{E}_{\boldsymbol{x}=0}(0) =\frac{8}{\pi^4}(\pi^2-8)^2,
\qquad \mathcal{E}_{\boldsymbol{x}=0}(1)=0 .\label{energy_density_limits}
\end{align}
\end{proposition}

We plot $\mathcal{E}_{\boldsymbol{x}=0}(k)$ in Figure \ref{ed_origin}.

\begin{figure}[h]
\centering
\includegraphics[width=0.5\textwidth]{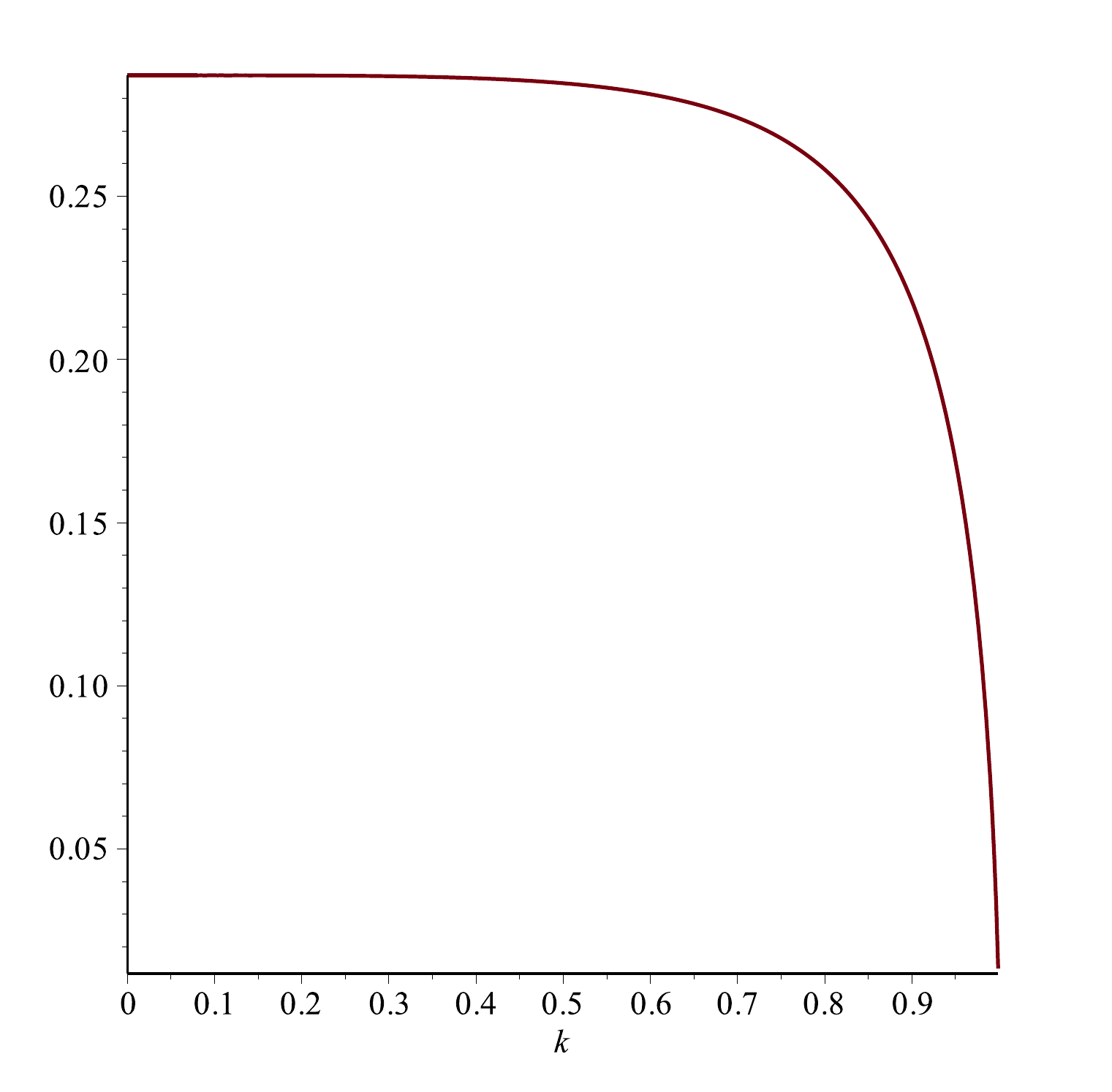}
\caption{  $\mathcal{E}_{\boldsymbol{x}=0}(k)$ }
 \label{ed_origin} 
\end{figure}

\section{The Higgs Field on the coordinate axes}\label{sectionHiggsonaxes}

In this section we shall calculate the Higgs Field on each of the coordinate axes. We have already done this
for the $x_2$ axis using Panagopoulos's formulae for evaluating  the Higgs field and determined $
H(\boldsymbol{x})=\sqrt{-\frac12 \Tr\Phi^2}$. Here we shall employ Theorem \ref{higgsthm}. First though we recall an old approach of 
Brown, Panagopoulos and Prasad \cite{bpp82} sufficient to make comparison with the results here; the full
details are given in Appendix \ref{applame}.

Brown, Panagopoulos and Prasad observed that (a constant) conjugation of the operator $\Delta\sp\dagger$ took the form
$$
\frac{\mathrm{d}}{\mathrm{d} z} 1_4+
\begin {pmatrix} (f_{{3}}+f_{{1}}-f_{{2}})/2&-x
_{{3}}&-x_{{1}}&ix_{{2}}\\ -x_{{3}}&(f_{{3}}-
f_{{1}}+f_{{2}})/2&ix_{{2}}&-x_{{1}}\\ -x_{{1}}&
-ix_{{2}}&(f_{{1}}+f_{{2}}-f_{{3}})/2&x_{{3}}
\\ -ix_{{2}}&-x_{{1}}&x_{{3}}&-(f_{{1}}+f_{
{2}}+f_{{3}})/2\end {pmatrix}
$$
and so on any coordinate axis this reduced to two $2\times2$ matrix equations. They focussed on
the axis joining the two monopoles where they showed that each of the $2\times2$ matrix equations reduced to Lam\'e's equation. For that case they determined the two normalizable solutions to 
$\Delta\sp\dagger v=0$ 
and then showed that two of the three Higgs field components vanished;  denoting the
remaining non-vanishing component by $\phi$
then \cite[6.13] {bpp82} expresses this as
\begin{align}\label{bppx1}
H=-i \phi&=-k' K +\frac{2 k'}{k^2 \sn^2 t -S^2}\left( S-\frac{\sn t}{\cn t \dn t}\frac{dS}{dt}\right),
\intertext{where  \cite[6.11] {bpp82}}
S(t)&=- \frac{\sn t \dn t}{\cn t } \tanh\left( KZ(t)\right)\nonumber
\end{align}
and $t$ is defined through the relation \cite[6.8]{bpp82}
$$4 k'^2 x_{1,BPP}^2 -(1+k^2)= -1-k^2 \cn^2 t
\Longleftrightarrow \sn^2 t= \frac{4 x^2_{1,BE}}{k^2K^2}
.$$

We will relate this to our expressions making use of the results of \S\ref{axesandjacobi} and \S\ref{muandjacobi}. Appendix \ref{applame} performs the analysis for each of the coordinate axes
and relates the conventions and scalings\footnote{
$x_{BPP}= x_{BE}/Kk'$, $\phi_{BE}= \phi_{BPP}/({k' K})$.
} of Brown, Panagopoulos and Prasad to those here; the analysis of the Appendix clarifies some of the
arguments of \cite{bpp82}.

For each coordinate axis we record the calculations of Gram and Higgs$'_i$ ($i=1,2,3$) of Theorem \ref{higgsthm} and then evaluate $H(\boldsymbol{x})$. It is convenient to define
\begin{equation} \label{defDM}
\mathcal{D} (\lambda) = \mathrm{Diag} ( \mathrm{e}^{-2\lambda}, \,
\mathrm{e}^{2 \lambda}  ) ,
\qquad     
\mathcal{M}(\lambda)=
\begin {pmatrix} 
 {{\rm e}^{-2\,\lambda}}\sinh(2\,\lambda )&-1\\
 -1& - \,\ {{\rm e}^{2\,\lambda}}\sinh(2\,\lambda )\\
\end {pmatrix}.
\end{equation}
Then $\Tr\left[\mathcal{D} (\lambda)\mathcal{M}(\lambda)\right]^2/2=\cosh^2(2\,\lambda)$.
We begin first with the simpler case of the $x_2$ axis, recovering our earlier result, and then
turn to the other axes which contain points of bitangency.

\subsection{The \texorpdfstring{$x_2$}{x22} axis} 
With $\mu_1=\lambda_2+i\pi/4$  and
\begin{align*}
\zeta _{{1}}&={\frac {iKk+\sqrt {{K}^{2}{{ k'}}^{2}+4\,{x_{{2}}}^{2}
}}{\sqrt {{K}^{2}+4\,{x_{{2}}}^{2}}}}\\
\intertext{we obtain}
\text{Gram}&= 
 8\, \cosh^3 \left( 2\,\lambda_2 \right){K}^{2}{k}^{2}
 \mathcal{D}(-\lambda_2)
\\
\text{Higgs}'_1&=
8\,i \cosh^2 \left( 2\,\lambda_2 \right) 
\left[K^2k^2 \mathcal{M}(-\lambda_2) +2(K^2+4x_2^2) \sigma_1\right]
\\
\text{Higgs}'_2&= -
{\frac {32\,i  \cosh^3 \left( 2\,\lambda_2 \right)
x_{{2}}{K}^{2}{k}^{2}}{\sqrt {{K}^{2}+4\,{x_{{2}}
}^{2}}\sqrt {{K}^{2}{{k'}}^{2}+4\,{x_{{2}}}^{2}}}}\mathcal{D}(-\lambda_2)
\\
\text{Higgs}'_3&= 
-16\,i \cosh^2(2\lambda_2)(EK+4x_2^2)\,\sigma_1
\end{align*}

Assembling these we again obtain  (\ref{h2}) (for $-\infty<x_2<\infty$)
\begin{equation*}
H^2(0,x_2,0) = \left(\tanh 2\lambda_2+\frac{4 x_2}{W_2}\right)^2 + 
\frac{(K {k'}^2-2E+K)^2}{ K^2 k^4 \cosh^2 2\lambda_2} 
\end{equation*}
with  $ W_2=\sqrt{ (K^2+4x_2^2)(K^2 {k'}^2+4 x_2^2) } $. These are plotted
in Figure \ref{HFIELD2} for different scales using
$\lambda_2=\lambda'_2 -i\pi/4$ where $\lambda'_2$ is given by (\ref{deflambda2p}).

\begin{figure}
\centering
\begin{subfigure}[b]{0.49\linewidth}
\includegraphics[width=0.7\textwidth]{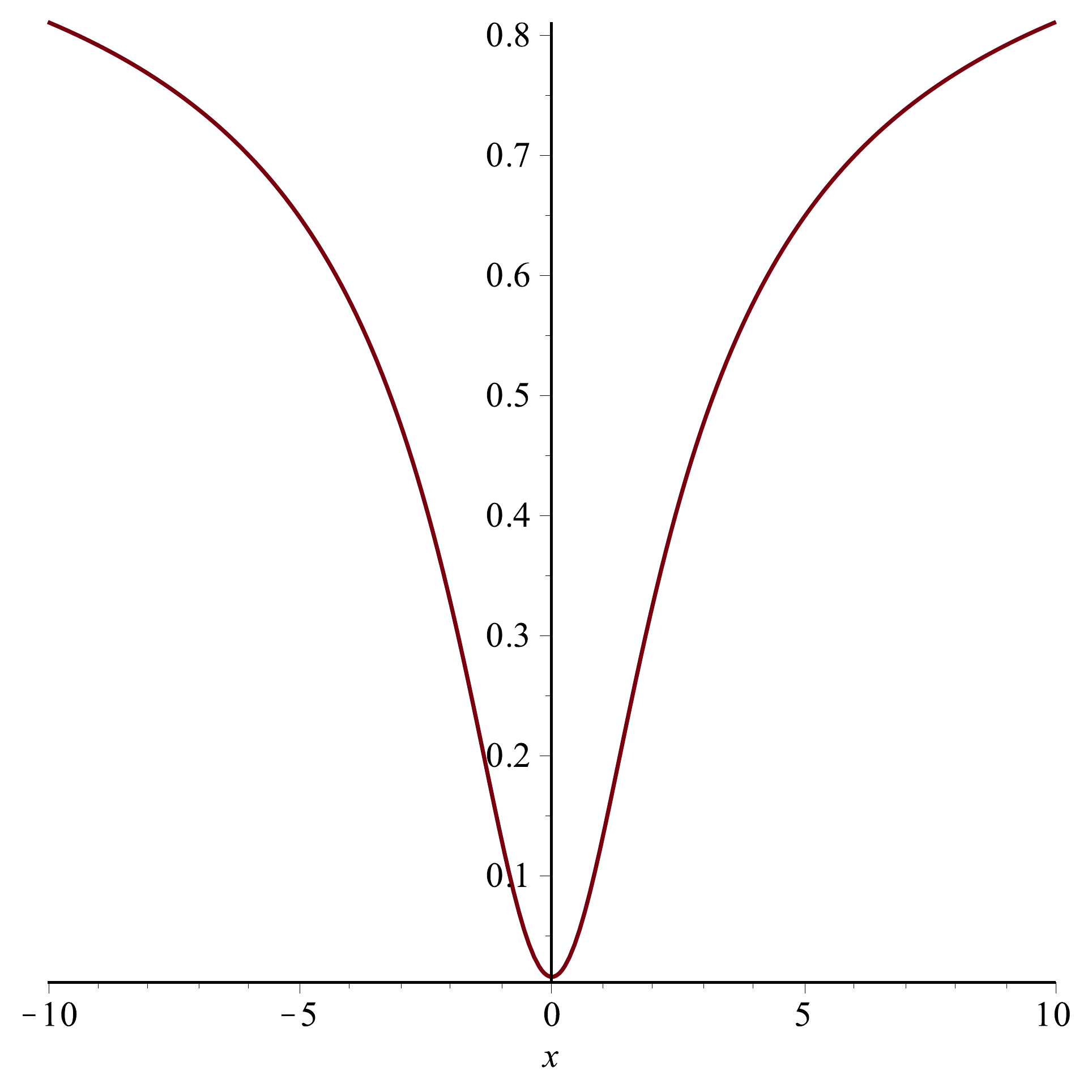}
\end{subfigure}
\begin{subfigure}[b]{0.49\linewidth}
\includegraphics[width=0.7\textwidth]{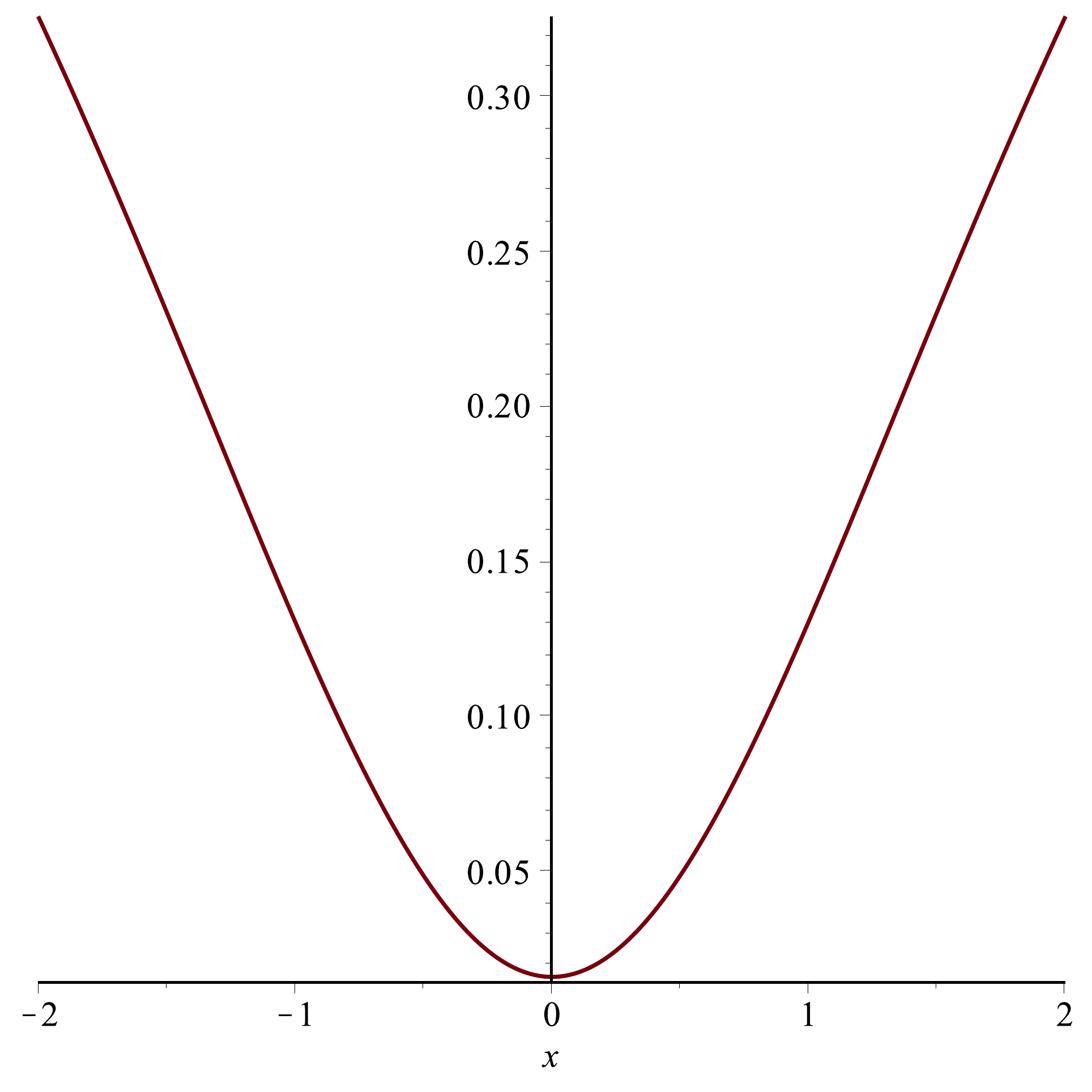}
\end{subfigure}
\caption{$H^2(\boldsymbol{x}):=
 -\frac12 \Tr\, \Phi^2$ restricted to the $x_2$-axis  $k=0.8$ }
 \label{HFIELD2} 
\end{figure}

\subsection{The \texorpdfstring{$x_1$}{x1} axis} We have seen that the points $\pm Kk/2$, $\pm K/2$ of the $x_1$ axis 
correspond to nongeneric points of bitangency to the spectral curve. 
Without including higher order terms
in the expansion of the eigenfunctions we consider separate cases and we show here that these piece
together to one expression for $H^2$. Recall that we identified in (\ref{x1axis})
\[ \zeta_{1}=\zeta_1(x_1)=\frac{Kk'+\imath \sqrt{K^2k^2-4x_1^2}}{\sqrt{K^2-4x_1^2}}\]

\subsubsection{Calculations for \texorpdfstring{$|x_1|<Kk/2$}{x1small}} Here $\mu_1=\lambda_1+ i\pi/4$ with $\lambda_1$ 
given by (\ref{deflambda1}).
\begin{align*}
\text{Gram}&= 
8\,\cosh \left( 2\,\lambda_1 \right)  \left( 
-{K}^{2}{{\ k'}}^{2} \left( \cosh^2 \left( 2\,\lambda_1 \right)  \right) +{K}^{2}-4\,{x_{{1}}}^{2} \right) 
\mathcal{D}(\lambda_1)
\\
\text{Higgs}'_1&= 
8 i\,  \left( 
-{K}^{2}{{\ k'}}^{2} \left( \cosh^2 \left( 2\,\lambda_1 \right)  \right)+{K}^{2}-4\,{x_{{1}}}^{2} \right)\,
\mathcal{ M}(\lambda_1)
\\
\text{Higgs}'_2&= 
- \frac{16\imath K^2 k'^2 \, x_1 \sinh(4\lambda_1)}
{\sqrt{K^2-4x_1^2}\sqrt{K^2k^2-4x_1^2}}\, \mathcal{M}(\lambda_1)
\\
\text{Higgs}'_3&= 
16 i \,
\left( K(E-K)+4
\,{x_{{1}}}^{2} \right)\, \mathcal{M}(\lambda_1)
\end{align*}

\subsubsection{Calculations for $Kk/2<|x_1|<K/2$}

\begin{align*}
\text{Gram}&= 
  -8 \left(K^2 {k'}^2 \sinh^2(2\mu_1)+K^2-4x_1^2  \right) 1_2
\\
\text{Higgs}'_1&= 
 -8\imath \left(K^2 {k'}^2 \sinh^2(2\mu_1)+K^2-4x_1^2  \right) \sigma_1
\\
\text{Higgs}'_2&=  
{\frac {16\imath {K}^{2}{{k'}}^{2}\sinh \left( 4\,\mu_1 \right) x_{{1
}}}{\sqrt {{K}^{2}-4\,{x_{{1}}}^{2}}\sqrt {{K}^{2}{k}^{2}-4\,{x_{{1}}
}^{2}}}}\sigma_1
\\
\text{Higgs}'_3&=- 16 i \,
\left( K(E-K)+4
\,{x_{{1}}}^{2} \right)\, \sigma_1
\end{align*}

\subsubsection{Calculations for $K/2<|x_1|$} With $\mu_1=\lambda''_1$.

\begin{align*}
\text{Gram}&= 8\cosh(2\lambda''_1)( {k'}^2K^2\sinh^2(2\lambda''_1)+K^2-4x_1^2 )\,
\mathrm{Diag} (
{{\rm e}^{-2\,\lambda''_1}}, {{\rm e}^{2\,\lambda''_1}})
\\
\text{Higgs}'_1&=8\imath (K^2{k'}^2\sinh^2(2\lambda''_1)+K^2-4x_1^2)\mathcal{ M}(-\lambda''_1)
 \\
\text{Higgs}'_2&= 
-{\frac {16 \imath {K}^{2}{{ k'}}^{2}\sinh \left( 4\,\lambda''_1 \right) x_{{1
}}}{\sqrt {{K}^{2}-4\,{x_{{1}}}^{2}}\sqrt {{K}^{2}{k}^{2}-4\,{x_{{1}}
}^{2}}}}\,\mathcal{M}(-\lambda''_1)
\\
\text{Higgs}'_3&= 
16\,i \left(K(E-K)+  4\,{x_{{1}}}^{2}
 \right) \mathcal{M}(-\lambda''_1)
\end{align*}

Assembling these pieces yields
\[ H(x_1,0,0)=
\begin{cases}
1+\dfrac{2K(-K {k'}^2 \cosh^2 2 \lambda_1+E) }
{K^2{k'}^2\cosh^2 2\lambda_1-K^2+4x_1^2}-\dfrac{1}{W_1}
\dfrac{2 K^2 {k'}^2 x_1 \sinh 4\lambda_1 }{K^2{k'}^2\cosh^2 2\lambda_1 -K^2+4x_1^2}
&  |x| < \frac{Kk}{2}
\\
1-\dfrac{2K(K {k'}^2 \sinh^2 2\mu_1 +E) }
{K^2{k'}^2\cosh^2 2\mu_1 +K^2k^2-4x_1^2}
+\dfrac{1}{W_1}
\dfrac{2 K^2 {k'}^2 x_1 \sinh 4\mu_1 }{K^2{k'}^2\cosh^2 2\mu_1+K^2k^2-4x_1^2}
& \frac{Kk}{2}<| x| < \frac{K}{2}
\\
1-\dfrac{2K(K {k'}^2 \sinh^2 2\lambda''_1 +E) }
{K^2{k'}^2\cosh^2 2\lambda''_1 +k^2K^2-4x_1^2}
+\dfrac{1}{W_1}
\dfrac{2 K^2 {k'}^2 x_1 \sinh 4\lambda''_1 }{K^2{k'}^2\cosh^2 2\lambda''_1 +K^2k^2-4x_1^2}
& \frac{K}{2}<| x|
\end{cases}
\]
where
\[ W_1= \sqrt{(K^2k^2 - 4 x_1^2)(K^2-4x_1^2)}. \] 
The sign of the square root $W_1$  in regions {\bf II}, {\bf III} requires a little more care and is best done by analytic continuation. We also
note that in all of these formulae there are apparent singularities at $x_1=\pm Kk/2, \pm K/2$ and the zeros
of the denominator 
$$K^2{k'}^2\cosh^2 2\mu_1=4x_1^2-K^2k^2.$$ 
We may in fact solve this transcendental equation: our previous results $\mu_1(\pm Kk/2)=i\pi/4$
and $\mu_1(\pm K/2)=0$ show the roots to again be $x_1=\pm Kk/2, \pm K/2$. 
The signs may be verified as such that give finite values. Analytic continuation away from the axis around these values shows that the signs of region {\bf II}, {\bf III}  coincide with the single
expression for the whole axis
\begin{equation}\label{HBPP}
H(x_1,0,0)=1+\frac{2K(-K {k'}^2 \cosh^2 2\lambda_1 +E) }
{K^2{k'}^2\cosh^2 2\lambda_1 -K^2+4x_1^2}
-\frac{1}{W_1}
\frac{2 K^2 {k'}^2 x_1 \sinh 4\lambda_1 }{K^2{k'}^2\cosh^2 2\lambda_1 -K^2+4x_1^2}
\end{equation}
where $\lambda_1$ is given by (\ref{deflambda1}). Note that $\lambda(x_1)$ is odd and overall $H^2(x_1,0,0,0)$ is even. From the expression when $|x|>K/2$ we see it has the desired asymptotics. We again observe (\ref{higgsorigin}),
a necessary test of consistency. This is plotted in Figure \ref{HFIELD1}.

\begin{figure}
\centering
\begin{subfigure}[b]{0.49\linewidth}
\includegraphics[width=0.7\textwidth]{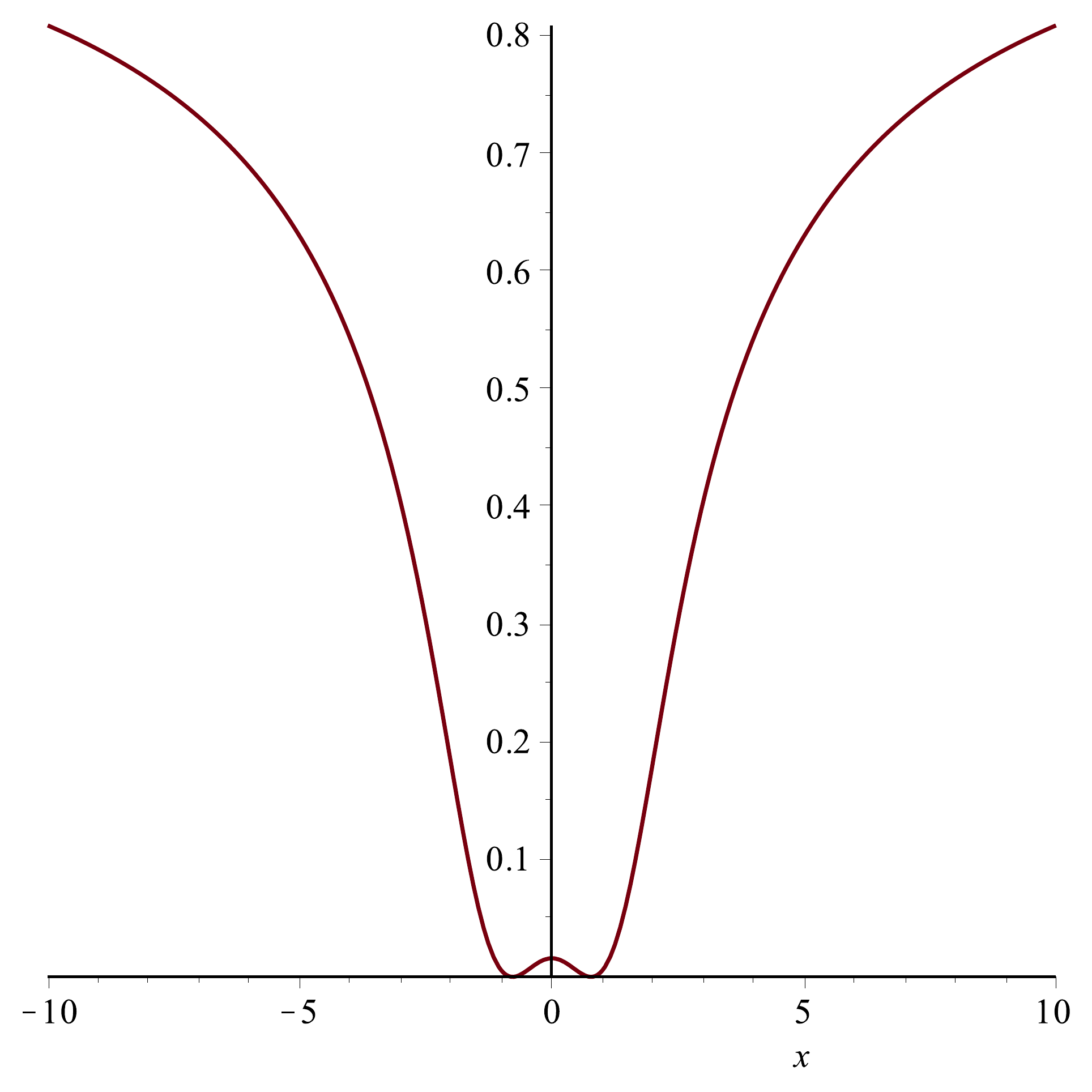}
\end{subfigure}
\begin{subfigure}[b]{0.49\linewidth}
\includegraphics[width=0.7\textwidth]{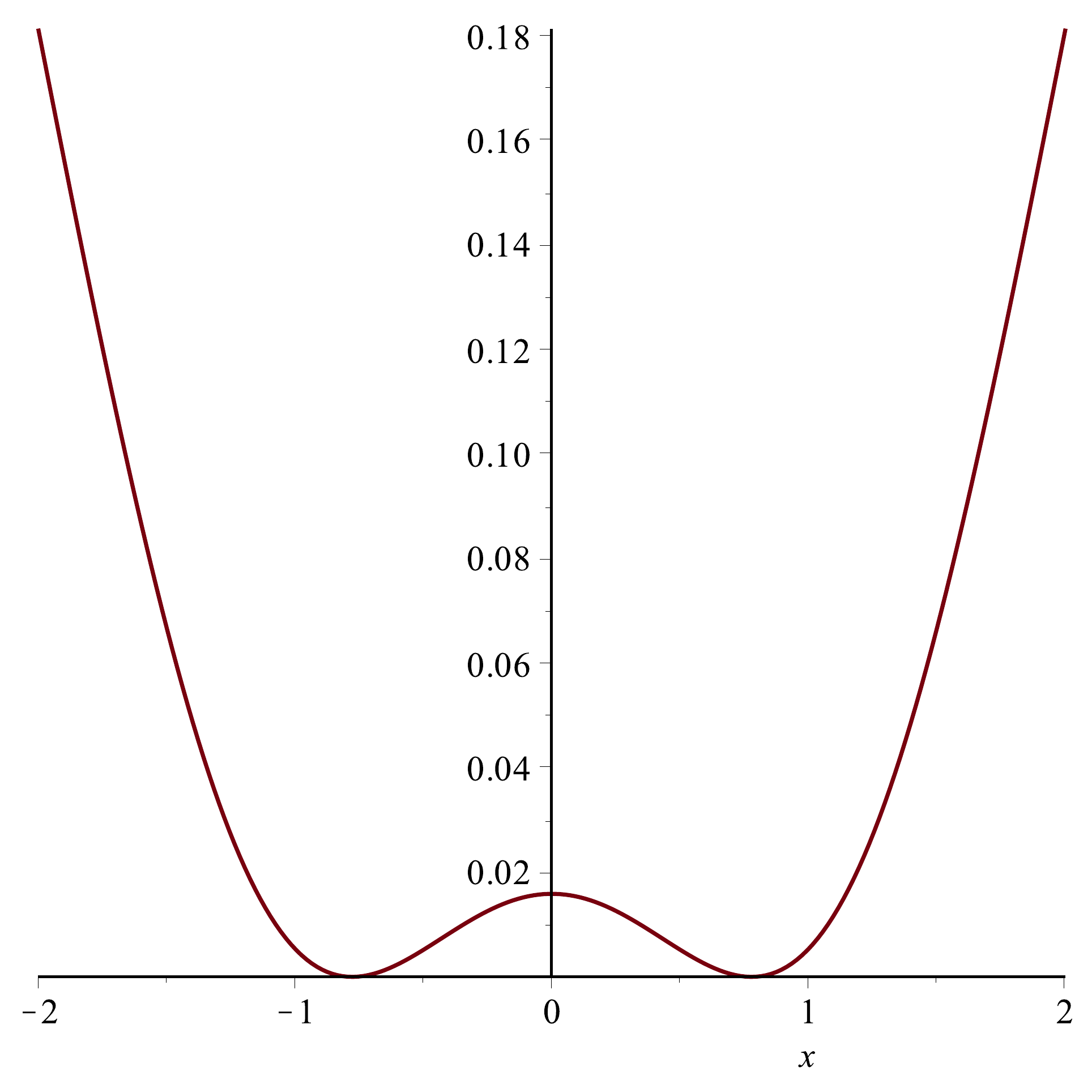}
\end{subfigure}
\caption{$H^2(\boldsymbol{x}):=
 -\frac12 \Tr\, \Phi^2$ restricted to the $x_1$-axis  $k=0.8$ }
 \label{HFIELD1} 
\end{figure}

\subsection{The \texorpdfstring{$x_3$}{x3} axis} 
Again we have the points of bitangency $\pm Kk'/2$ and our analysis proceeds for two intervals which
again may be expressed in terms of a single expression.

\subsubsection{Calculations for $|x_3|<Kk'/2$}  With $\mu_1=i \lambda=\lambda_3+i\pi/4$
and
$$\zeta _{{1}}={\frac {i\sqrt {{K}^{2}{k}^{2}+4\,{x_{{3}}}^{2}}+\sqrt {{
K}^{2}{{ k'}}^{2}-4\,{x_{{3}}}^{2}}}{K}}.
$$

\begin{align*}
\text{Gram}&= 
8 (K^2\sin^2 \left( 2\,\lambda \right) -\,{K}^{2}{{ k'}}^{2}+4\,{x_{{3}}}^{2})1_2
\\
\text{Higgs}'_1&=
8\,i \left( K^2\sin^2 \left( 2\,\lambda \right)-{K}^{2}{{ k'}}^{2}+4\,{x_{{3}}}^{2} \right) 
\sigma_1
\\
\text{Higgs}'_2&= \frac{-16\,\imath\,  \sin(4\lambda) x_3 K^2}{\sqrt{K^2{k'}^2-4x_3^2}\sqrt{K^2k^2+4x_3^2}}\, \sigma_1
\\
\text{Higgs}'_3&= 
-16\,i \left( -\,{K}^{2}{{ k'}}^{2}+KE+4\,{x_{{3}}}^{2} \right) 
\sigma_1
\end{align*}

\subsubsection{Calculations for $Kk'/2<|x_3|$}

\begin{align*}
\text{Gram}&= 
8\,\cosh \left( 2\,\mu_{{1}} \right)  \left(K^2 \cosh^2 \left( 2
\,\mu_{{1}} \right) -\,K^2{k}^{2}-4 x_3^2
 \right) \, \mathcal{D}(-\mu_1)
\\
\text{Higgs}'_1&= 
8\imath  \left(K^2 \cosh^2 \left( 2\,\mu_{{1}} \right) -\,K^2 {k}^{2}
-4\,x_3^2 \right)\,\mathcal{M}(-\mu_1)
\\
\text{Higgs}'_2&=  \frac{16 \sinh(4\mu_1) x_3 K^2}{\sqrt{K^2{k'}^2-4x_3^2}\sqrt{K^2k^2+4x_3^2}}\, \mathcal{M}(-\mu_1)
\\
\text{Higgs}'_3&= 16\imath (-K^2 {k'}^2 +EK + 4 x_3^2)\, \mathcal{M}(-\mu_1)
\end{align*}

Recall that we have shown that for $|x_3|<Kk'/2$ then $\mu_1=\lambda_3+i\pi/4$ with 
$\lambda_3$ given in (\ref{deflambda3}).
These combine to yield the single expression for the $x_3$-axis we have
\begin{equation}
H(0,0,x_3)= 1 - \dfrac{2K(K \cosh^2 2\mu_1 +E-K) } {K^2\cosh^2 2\mu_1 -K^2k^2-4x_3^2}+\dfrac{1}{W_3}\dfrac{2 i \,K^2 x_3 \sinh 4\mu_1 }{K^2\cosh^2 2\mu_1 -K^2k^2-4x_3^2}
\end{equation}
where now
\begin{equation}  
W_3 = \sqrt{ (K^2{k'}^2-4 x_3^2)(K^2k^2+4 x_3^2 ) }.
 \end{equation}
Here the expressions for $\mu_1$ combine with the signs of the square roots $W_3$. The expression
given in fact holds for $x_3>-Kk'/2$ with $\mu_1=\lambda_3+i\pi/4$ and $\lambda_3$ given by (\ref{deflambda3}).
We note that both $W_3$ and the dominators $K^2
\cosh^2 2\mu_1 -K^2k^2-4x_3^2$ vanish at $x_3=Kk'/2$ and again
careful analysis of the function $H^2(0,0,x_3)$ shows these points  it to be regular.
The shape of the field  $H^2(0,0,x_3)$ is shown in Figure \ref{HFIELD3} .
Again we have the consistency check (\ref{higgsorigin})
and the correct asymptotics.

\begin{figure}
\centering
\begin{subfigure}[b]{0.49\linewidth}
\includegraphics[width=0.7\textwidth]{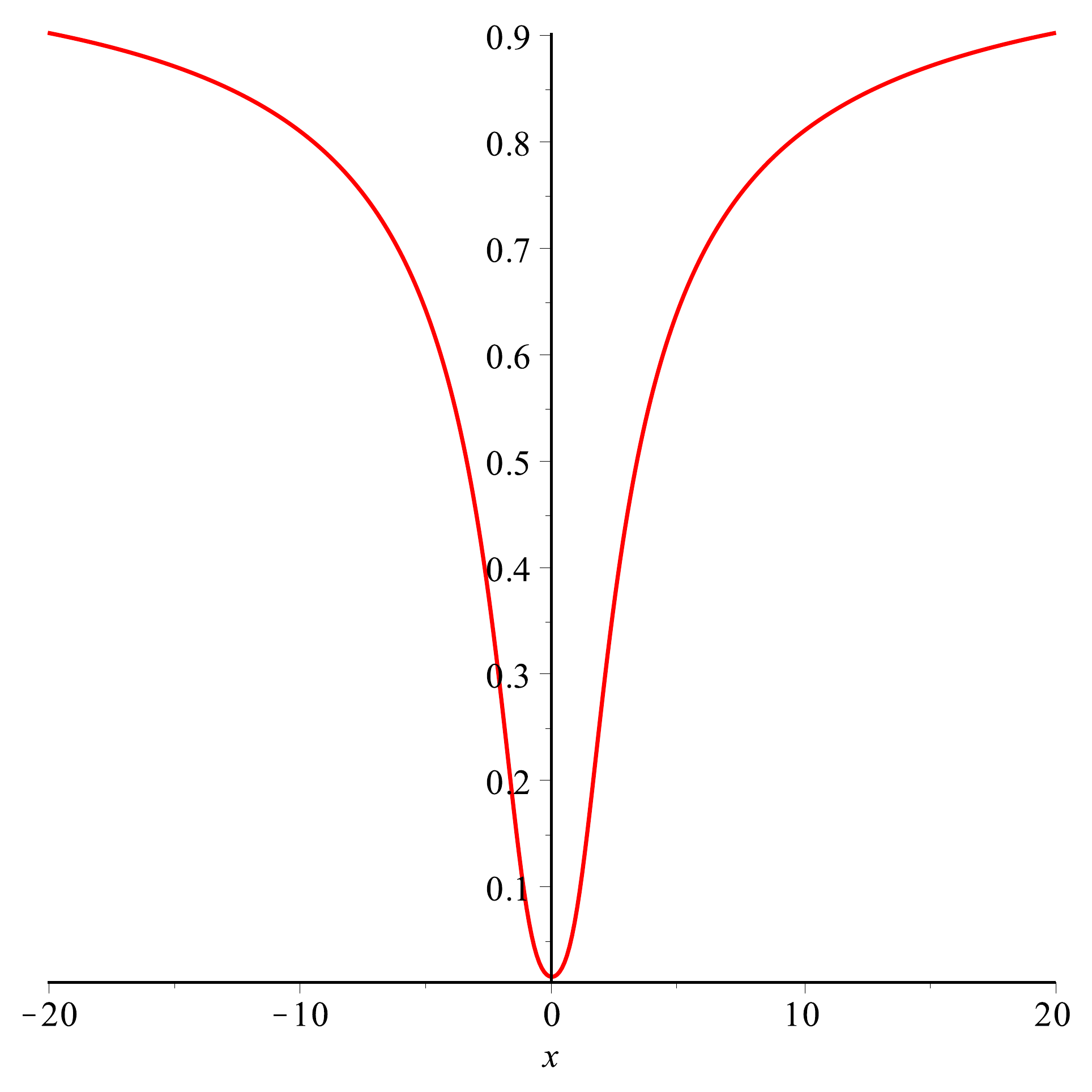}
\end{subfigure}
\begin{subfigure}[b]{0.49\linewidth}
\includegraphics[width=0.7\textwidth]{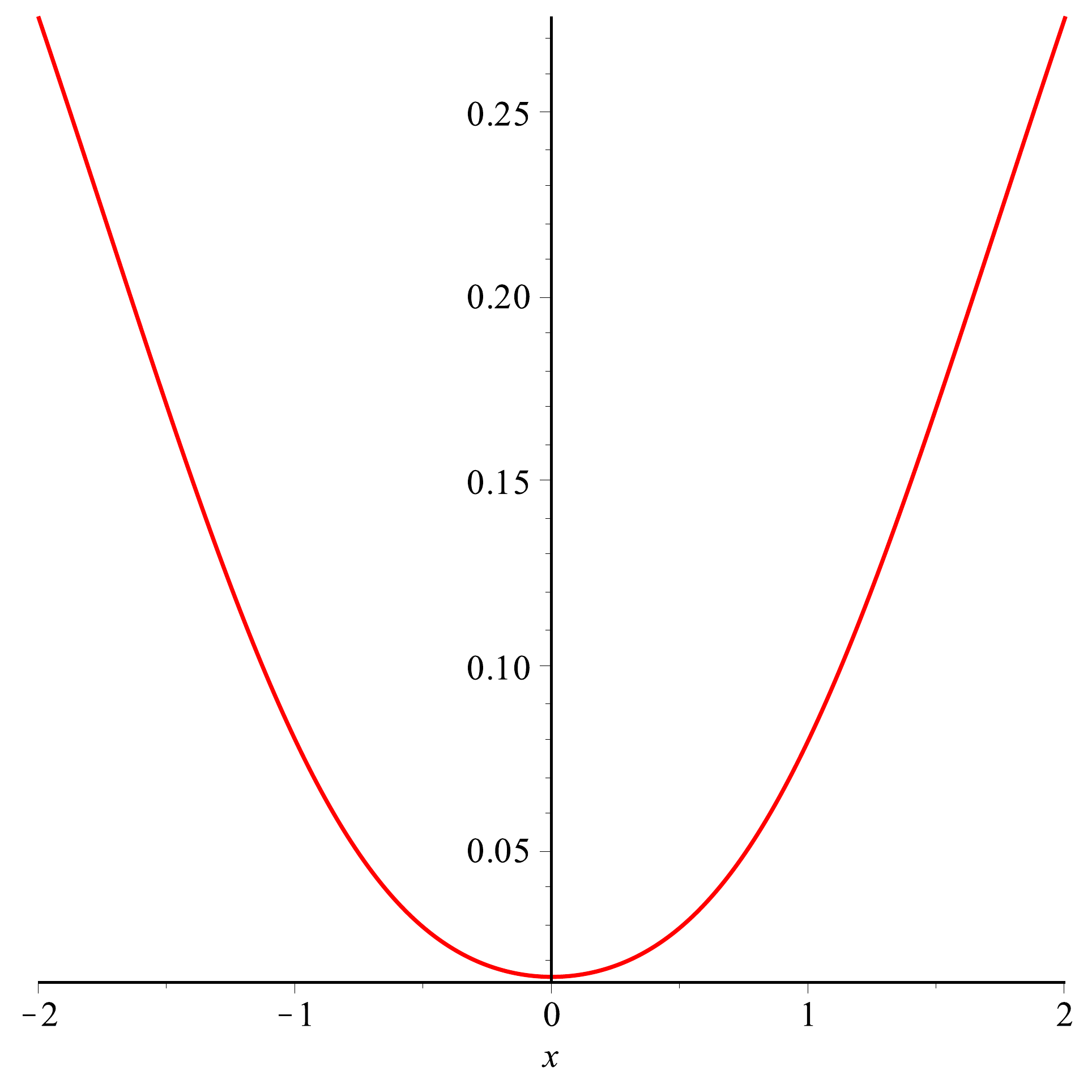}
\end{subfigure}
\caption{$H^2(\boldsymbol{x}):=
 -\frac12 \Tr\, \Phi^2$ restricted to the $x_3$-axis  $k=0.8$ }
 \label{HFIELD3} 
\end{figure}

\subsection{Zeros of the Higgs Field}
The position of the zeros of the Higgs field was an early subject of discussion. Based on the numerical evaluation of their ansatz Forg{\'a}cs,  Horv{\'a}th and Palla \cite[(21)] {forgacs_horvath_palla_82b} gave
these to be (in our units) $\pm kK(k)/2$. Two analytic works then followed. In \cite[\S6, \S9]{ors82}
an analytic expression that needed differentiation was obtained; this \lq very complicated\rq\ expression was evaluated numerically where the zero was found to be \lq very close\rq\ and \lq barely distinguishable\rq\ from $\pm kK(k)/2$. In \cite{bpp82} expansions for the zeros of the Higgs field were given for $k$
near $0$ and $1$ (the latter being situated near $\pm K(k)/2$). In \cite[\S6] {forgacs_horvath_palla_83b}
Forg{\'a}cs,  Horv{\'a}th and Palla expressed that their earlier result was to be viewed as a very good approximation of the zeros.

The position of the zeros of the Higgs field, which lie on the $x_1$ axis, may be found from
(\ref{HBPP}). We have already recorded that at the points of bitangency $\pm k K(k)/2$ and $\pm K(k)/2$
the numerator and denominators of (\ref{HBPP}) vanish, but by l'Hopital's rule for example one sees regular behaviour here. We may express the zero of the Higgs field then as the vanishing of the numerator of
(\ref{HBPP}), but discounting $\pm k K(k)/2$ and $\pm K(k)/2$. With
$$
Y=\exp(4\lambda_1(x)),\qquad W=
\sqrt{(K^2k^2 - 4 x^2)(K^2-4x^2)},
$$
(and $\lambda_1$ again defined in \ref{deflambda1}) we obtain the transcendental equation
\begin{equation}\label{higgszeroposition}
-{Y}^{2} \left( W+4\,x \right) -2\,{\frac {W \left( {K}^{2}{k}^{2}+4\,
EK-3\,{K}^{2}+8\,{x}^{2} \right) Y}{{K}^{2} \left( {k}^{2}-1 \right) }
}-W+4\,x=0.
\end{equation}
The vanishing of $W$ at the points of bitangency makes checking the vanishing of this equation there
straightforward. One finds for $k\in(0,1)$ a further point of vanishing in $(0,k K(k)/2)$. Figure \ref{higgszeroplot} illustrates this zero for which we have found no analytic expression.

\begin{figure}
\centering
\begin{subfigure}[b]{0.49\linewidth}
\includegraphics[width=1.0\textwidth]{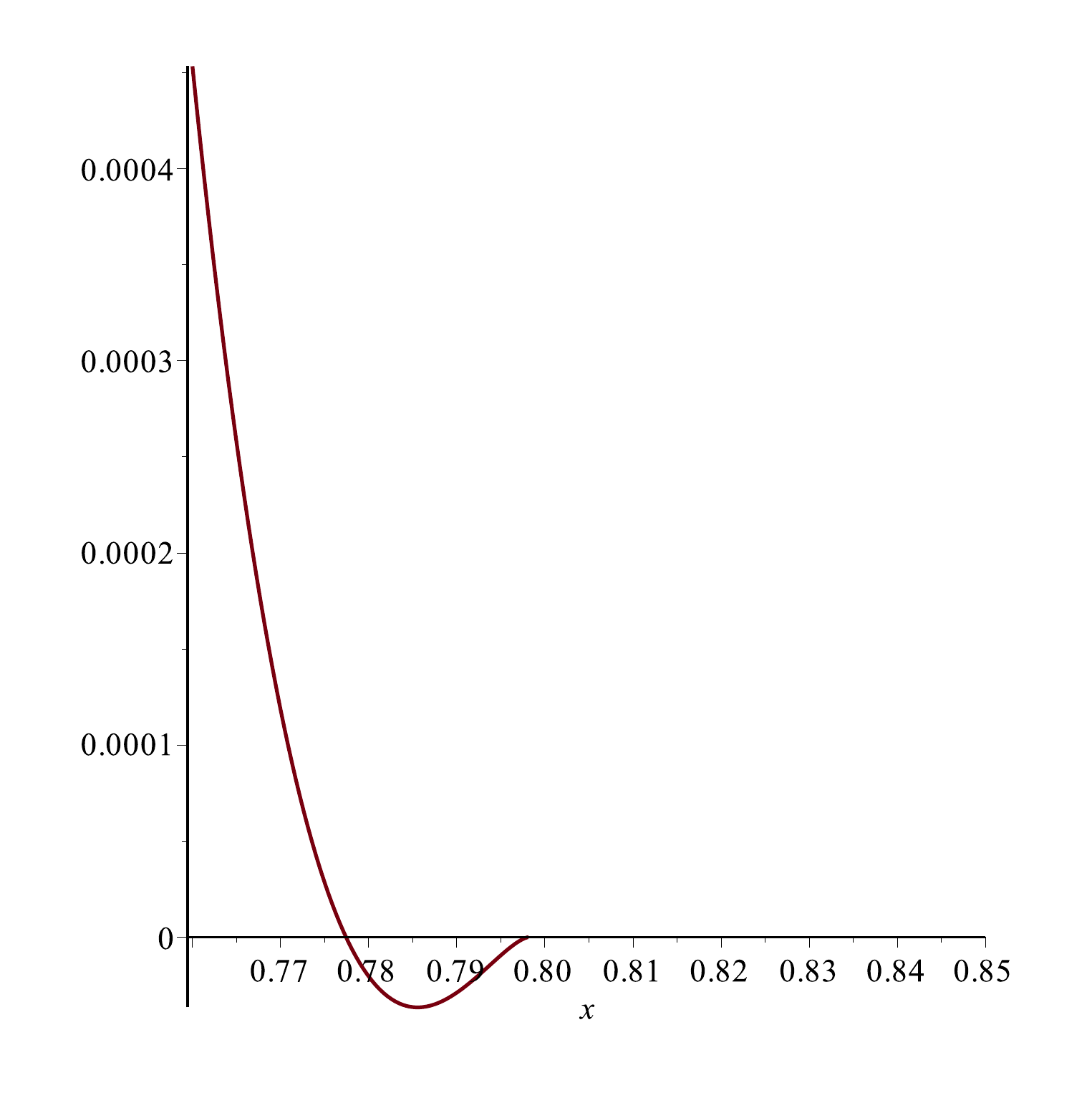}
\end{subfigure}
\begin{subfigure}[b]{0.49\linewidth}
\includegraphics[width=1.0\textwidth]{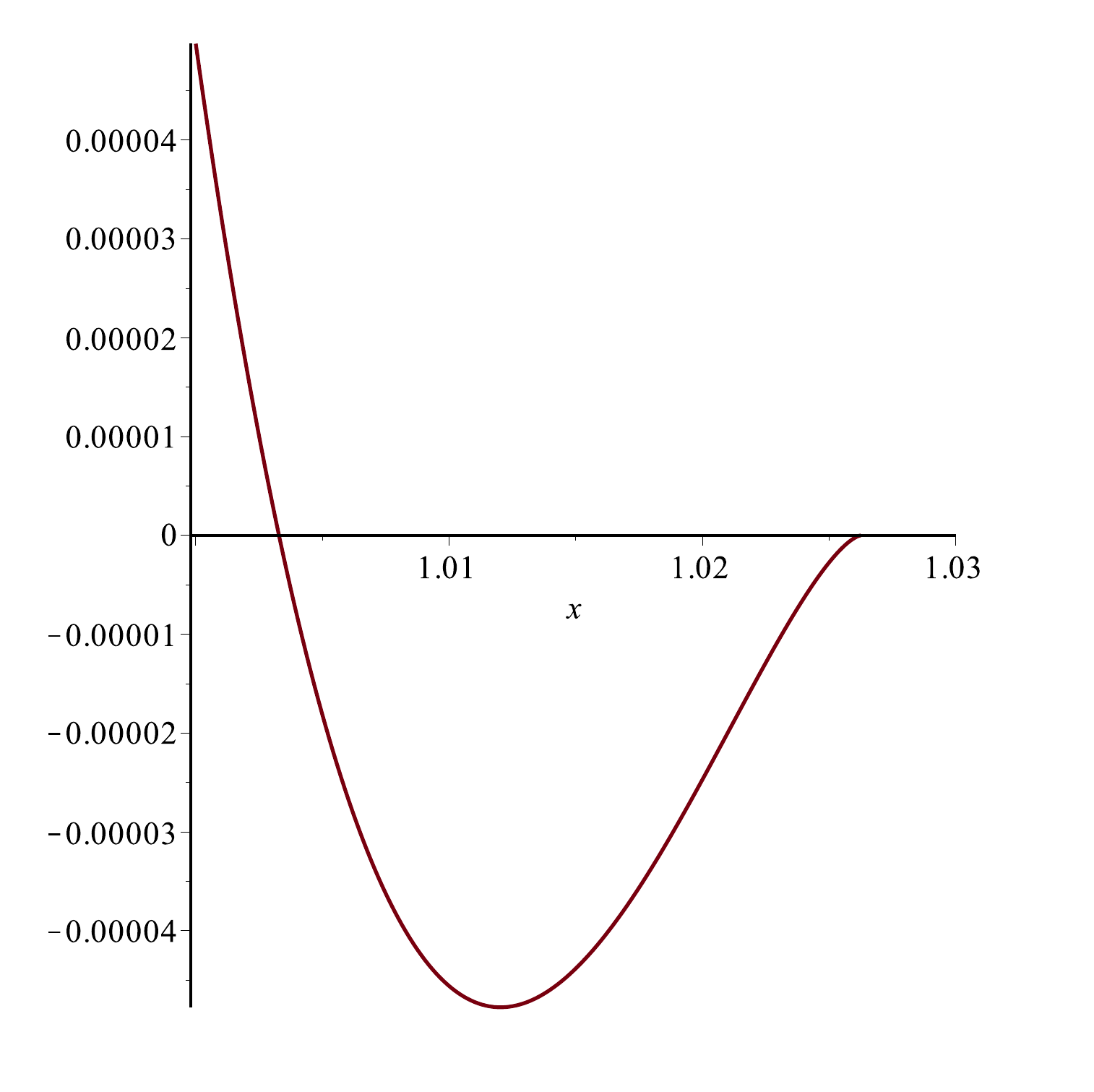}
\end{subfigure}
\caption{Zeros of equation (\ref{higgszeroposition}) for (left) $k=0.8$; $
k K(k)/2=0.798$ is the upper zero and (right) $k=0.9$; $
k K(k)/2=1.026$.  }
 \label{higgszeroplot} 
\end{figure}

\section{ The \texorpdfstring{$k=0$}{k0} limit}\label{sectionk0}
Early analytic studies of monopoles followed work of Manton \cite{manton78} assuming axial symmetry.
One of the surprises discovered was that an axially symmetric monopole corresponded to coincident
charges \cite{hor80}. Ward \cite{Ward1981b} developed the Atiyah-Ward ansatz in the monopole setting and
gave an ansatz that produced a charge $2$ axially symmetric monopole. Our aim in this section is to reproduce Ward's results as the $k\rightarrow 0$ limit of our own. We will first recall Ward's results, then obtain those as a limit and then conclude with a new result (Proposition \ref{higgs14k0x3plane}).

\subsection{Ward's Results}

Ward \cite{Ward1981b}  expresses the Higgs field (with our enumeration of axes) as
\begin{equation}
\Phi= \left( \begin{array}{cc} U&V \mathrm{e}^{-2\imath\psi}\\
W \mathrm{e}^{2\imath\psi}&-U       \end{array}  \right)
\end{equation}
where $x_1+ix_3=|x_1+i x_3|e\sp{i\psi}$. Then on the $x_2$-axis 
\begin{align}\begin{split}
&U=\frac{x_2}{x_2^2+c^2} - \tanh(2z)\\
&V=W=0\end{split} \label{ward1}
\end{align}
while on the  $(x_1,x_3)$-plane
\begin{align}\begin{split}
&U=0\\
&V=W=\frac{c^2 \cosh(2a)[ \sinh(2a)-2a \cosh(2a)  ]}{ a ( a^2 - c^2 \sinh^2(2a))} -1
\end{split}\label{ward2}
\end{align}
where $a=\sqrt{r^2-c^2}$. Ward found that only $c=\pi/4$ gave nonsingular solutions.

\subsection{The \texorpdfstring{$k=0$}{k0} limit of the Atiyah-Ward constraint and relevant quantities}
The $k=0$ limit of the Atiyah-Ward constraint yields the equation
\[\left[ (x_2+\imath x_1)\zeta^2+2x_3\zeta - x_2+\imath x_1\right]^2 + \frac{\pi^2}{16} (\zeta^2-1)^2=0 . \]
With $r^2=x_1^2+x_2^2+x_3^2$ and  $x_{\pm}= x_1\pm \imath x_2 $ 
the solutions are
\begin{align}
\zeta_{1,2}=\frac{ 4\imath x_3 \pm \sqrt{ \pi^2 -16r^2+8\imath\pi x_2  } } { 4 x_--\pi  }, \quad 
\zeta_{3,4}=\frac{ 4\imath x_3 \pm \sqrt{ \pi^2 -16r^2-8\imath\pi x_2  } } { 4 x_-+\pi  },
\label{zeta1}
\end{align}
where we again order the roots according to the conjugation conditions
\begin{equation*}
\zeta_3 = - \frac{1}{\overline{ \zeta}_1}, \quad \zeta_4 = - \frac{1}{\overline{ \zeta}_2}.
\end{equation*}
Noting that $K(0)=E(0)=\pi/2$ one has that the corresponding $\mu_i$ are 
\begin{align}
\begin{split}
\mu_{1,2}&=\left(\frac{\imath \pi}{4}-x_2-\imath x_1\right)\zeta_{1,2}-x_3=\mp \frac{\imath}{4}\sqrt{ \pi^2-16 r^2+8\imath \pi x_2 },\\ 
\mu_{3,4}&=\left(-\frac{\imath \pi}{4}-x_2-\imath x_1\right)\zeta_{3,4}-x_3+\frac{\imath\pi}{2}=
\mp \frac{\imath}{4}\sqrt{ \pi^2-16 r^2-8\imath \pi x_2 }+\frac{\imath\pi}{2}.
\end{split} \label{mu}
\end{align}
One sees that 
\[  \mu_1+\overline{\mu}_3=- \frac{\imath \pi}{2}, \quad   \mu_2+\overline{\mu}_4=- \frac{\imath \pi}{2}. \]
Upon introducing the notation
\begin{align}\label{rprm}
R_+=\sqrt{\pi^2-16 r^2+8\imath\pi x_2 }, \quad R_{-}=\overline{R}_{+}=\sqrt{\pi^2-16 r^2-8\imath\pi x_2 }
\end{align}
 it is convenient to  rewrite the above formulae as
\begin{align}\begin{split}
\zeta_{1,2}&=\frac{4\imath x_3\pm R_+}{4x_--\pi}, \quad
\zeta_{3,4}=\frac{4\imath x_3\pm R_-}{4x_-+\pi},\\
\mu_{1,2}&=\mp \frac{\imath}{4}R_+,\quad \mu_{3,4}=\mp \frac{\imath}{4}R_-+\frac{\imath \pi}{2}.
\end{split}\label{zetamu}
\end{align}

\subsection{The \texorpdfstring{$x_2$}{x2}-axis}
We have previously obtained (\ref{h2}) 
\begin{equation*}
H(0,x_2,0)^2 = \left(\tanh 2\lambda+\frac{4 x_2}{W_2}\right)^2 + 
\frac{(K {k'}^2-2E+K)^2}{ K^2 k^4 \cosh^2 2\lambda} 
\end{equation*}
with 
\[  \lambda = \mu_1-\frac{\imath \pi}{4}, \quad W_2=\sqrt{ (K^2+4x_2^2)(K^2 {k'}^2+4 x_2^2) }. \]
Now on the $x_2$ axis with $k=0$ 
then $
\zeta_{1,2,3,4} = \pm 1$ and $  \mu_{1,2} = \frac{\imath \pi}{4} - x_2 $. In the limit $k\to 0$ the
second term in $H(0,x_2,0)^2$ vanishes and we obtain
\[H(0,x_2,0)_{k=0}^2=\left( -\tanh (2x_2) +  \frac{16x_2}{16 x_2^2+ \pi^2 }\right)^2. \]
This coincides with Ward's result upon his use of $c=\pi/4$.

\subsection{The \texorpdfstring{$(x_1,x_3)$}{x1x3}-plane}
Next we obtain the Higgs field and on the $x_2=0$-plane. We remark that our Higgs field is a gauge transform of Ward's by the gauge transformation $\diag(
e\sp{-i\psi}, e\sp{i\psi})$. In Appendix \ref{higgs13k0app} we establish
\begin{proposition}\label{higgs13k0}
 { The Higgs field at $k=0$  and $x_2=0$ is function of $r=\sqrt{x_1^2+x_3^2}$  
in the whole $(x_1,x_3)$-plane given by}

\begin{align}\begin{split}
&H(x_1,0,x_3)=H(r)\\
&=-1- \frac{2\pi^2 \cos(\frac12\sqrt{\pi^2-16 r^2}  ) 
( 2\sin(\frac12\sqrt{\pi^2-16 r^2}  ) 
-\sqrt{\pi^2-16 r^2} \cos(\frac12\sqrt{\pi^2-16 r^2}  )  )   }
{ \sqrt{\pi^2-16 r^2}(\pi^2\cos^2 
(\frac12\sqrt{\pi^2-16 r^2} )-16r^2)  }
\end{split}\label{Hfield}
\end{align} 

\end{proposition}

One can see that $H(0)=0$ accords with earlier formulae for depth of the well.

\subsection{The \texorpdfstring{$(x_1,x_2)$}{x1x2}-plane}
Next we obtain the  Higgs field and on the $x_3=0$-plane.  In Appendix \ref{higgs13k0app1} we establish
\begin{proposition}\label{higgs14k0x3plane}{ Higgs field on to the 
plane $x_3=0$ is given by  
\begin{align}\begin{split}
H(x_1,x_2,0)^2=  \sum_{j=1}^6 \mathcal{H}_j(x_1,x_2)    
\end{split}\label{HfieldxPlanex3}
\end{align}
with 
\begin{align}\begin{split}
\mathcal{H}_1 &= -\frac{S_+^2S_-^2}{R_+^2R_-^2 G^2} (\pi^2-16 x_2^2)\left( -R_+^4R_-^4+2048\pi^4x_2^2+16 \pi^2(R_+^2+R_-^2)   \right)\\
\mathcal{H}_2 &=\frac{8\pi S_+S_-}{G^2} \left( C_-S_+ (\pi + 4\imath x_2)R_-+C_+S_-(\pi-4\imath x_2)R_+  \right)\\
\mathcal{H}_3 &=-\frac{8\pi S_+}{R_+R_-G^2}(C_+C_-+1)(\pi+4\imath x_2)
(R_-C_-(\pi^2+16r^2)-4(\pi-4\imath x_2)\pi S_- )\\
\mathcal{H}_4 &=-\frac{\pi S_-}{R_-G^2}(\pi-4\imath x_2)\left\{  -16 \left[8(C_+C_-+1)(2C_--3C_+)
             -16S_+S_-(C_--C_+)  \right]r^2\right. \\
              &\hskip0.3cm\left.+8\pi^2C_+(C_-C_++1) +64\imath \pi x_2(C_-+C_+)(C_-C_+-S_-S_++1 )  
\right\}\\
\mathcal{H}_5 &=\frac12 \frac{ S_-S_+ R_-R_+   (C_-C_++1)\pi^2  }{ E_-^2+E_+^2  }
 \left\{ 64 \left[ 2S_-S_+(E_-^2-E_+^2)-(3C_-C_++1)E_-^2\right.\right.\\
&\left.\left.\hskip5.5cm+(2E_-^4-2E_-^2E_+^2-1  )E_+^2 +C_-C_+E_+^2     \right] r^2\right.  \\
&\left.\hskip0.3cm -32\imath(E_-^2E_+^2+1)(C_-+C_+)\pi x_2-4(E_-^2+E_+^2)(C_-C_++1)\pi^2 \right\}\\
\mathcal{H}_6 &=\frac{1}{(E_-^2+E_+^2) G} (C_-C_++1)(\pi^2+16r^2) 
\left\{  - 16\left[ 4C_--4C_++(E_-^2E_+^2-1)E_-^2 \right.\right.\\
&\hskip5cm \left.\left.   + 2(C_--2C_+)E_-^2E_+^2 -C_-C_+(E_-^2+E_+2)\right]r^2\right.\\
&\left.+16\imath (E_-^2E_+^2+1)(C_-+C_+)\pi x_2+(E_-^2+E_+^2)(C_-C_++1)\pi^2\right\}
 \end{split} 
\end{align}\label{Hfieldx3plane}
Here
\begin{align}
S_{\pm}=\sin\left( \frac12 R_{\pm}  \right), 
\quad C_{\pm}=\cos\left( \frac12 R_{\pm}  \right)\quad
E_{\pm}=\mathrm{exp}\left( \frac14 \imath R_{\pm}  \right)\label{notations1}
\end{align}
and $R_{\pm}=\sqrt{\pi^2-16 r^2 \pm \imath \pi x_2}$.
}
\end{proposition}



\section{Acknowledgements}
We wish to thank Paul Sutcliffe for numerous conversations and correspondence over the course of this project. We are particularly grateful to John Merlin for the notes he was able to provide 
on the works
\cite{merlin_ricketts_87,  hey_maerlin_ricketts_vaughn_williams, atiyah_hitchin_merlin_pottinger_ricketts},
to L\'aci Palla and Perter Forg\'acs for comments provided on their computations, and to Ed Corrigan for 
discussion.
We are indebted to Chris Eilbeck and Marko Puzza for computing input, and to Max Ruffert for producing a \lq movie\rq\ of two monopole collisions from our data.

We are particularly grateful for discussions and the funding of a number of research visits when 
certain partd of this investigation were done, including:
Boris Dubrovin, Scuola Internazionale Superiore di Studi Avanzati (SISSA),
Trieste;
Dmitri Korotkin, Department of Mathematics and Statistics of 
Concordia University, Montreal;
John Harnad, Centre de recherches math\'ematiques (CRM), Montreal;
Jutta Kunz of Carl von Ossietzky Universit\"at Oldenburg;
Claus L\"ammerzahl, ZARM, University of Bremen;
Atsushi Nakayashiki, Tsuda University, Tokyo;
 the 
Hanse-Wissenschaftskolleg Institute for Advanced Study in Delmenhorst
for their support (for VZE) over 10 months; and the Simons Center for Geometry and Physics for their
support (for HWB) for 3 months.

\appendix
\section{The Curve}\label{comparisoncurves}

\subsection{Properties of the curve}
We may see that  the monopoles are on the $x_1$ axis (for $k>0$ and at $k=0$ the
monopoles are axially symmetric about the $x_2$ axis in several ways. First
 our spectral curve takes the form
\begin{align*}
0&=\left(\eta+i \frac{K(k)}2\left[\zeta^2+1\right]\right)\left(\eta-i \frac{K(k)}2\left[\zeta^2+1\right]\right)
-K(k)^2k'^2 \zeta^2\\
&=\left(\eta+i \frac{K(k)}2\left[\zeta^2-1\right]\right)\left(\eta-i \frac{K(k)}2\left[\zeta^2-1\right]\right)
+K(k)^2k^2 \zeta^2.
\end{align*}
Now upon noting that $K(k)\sim \ln(4/k')$  as $k\rightarrow1$  then in this limit this behaves as
\begin{align*}
0&\sim \left(\eta+i \frac{K(k)}2\left[\zeta^2+1\right]\right)\left(\eta-i \frac{K(k)}2\left[\zeta^2+1\right]\right) 
\end{align*}
and so upon comparison with the 1-monopole curve we have two widely separated monopoles at $\pm  (\frac{K(k)}2,0,0)$,  on the $x_1$ axis. Alternately, set
$\tilde\eta =\eta-i {K(k)}\left[\zeta^2+1\right]/2$, which corresponds to a shift by 
${K(k)}/2$ along the $x_1$ axis,
 then the curve may be written as
$$0=\frac{\tilde\eta^2}{K(k)}+i(1+\zeta^2)\tilde\eta-K(k)k'^2 \zeta^2.$$
Now again letting $k\rightarrow1$ we find 
$(1+\zeta^2)\tilde\eta=0.$
If $\tilde\eta=0$ we see the second monopole is at the origin, and so both lie on the $x_1$ axis; if
$\zeta=\pm i$ then $\tilde\eta=\eta=2  \mathbf{y}\cdot  \boldsymbol{x}=2(x_2\mp i x_3)$ corresponds to a line
parallel to the  $x_1$ axis through the point $\tilde\eta=\eta$. Finally, we can read off the axis of symmetry from the curve
as follows. If $k=0$ we have
$$0=\left(\eta+i \frac{K(0)}2\left[\zeta^2-1\right]\right)\left(\eta-i \frac{K(0)}2\left[\zeta^2-1\right]\right)
$$
where $K(0)=\pi/2$, and this corresponds to the (complex) points $(0,\pm i K(0)/2, 0)$. A rotation around the $x_2$ axis leaves this invariant.

\subsection{Comparison of Notation}\label{comparingcurves}
 The charge $2$ spectral curve has appeared with many differing conventions. We give here transformations between these to enable comparison with existing results.

\begin{align*}
\mathcal{C}&&0&=
\eta^2+\frac{K^2}{4}(\zeta^4+2(k^2-k'^2)\zeta\sp2+1)
=|\eta -L(\zeta)| &\textbf{Here, Ercolani-Sinha}\\
\mathcal{C}_{FHP}&&0&=y^2+A\left( x^2+\frac{1}{x^2}  \right)+B
&\textbf{Forg\'acs, Horv\'ath \&Palla} \\
\mathcal{C}_{ORS}&&0&=\gamma^2+1-\frac{\varepsilon^2}4\left(\zeta_B-\zeta_B\sp{-1}\right)^2&
\textbf{ O'Raifeartaigh et. al., Brown}\\
\mathcal{C}_{H83}&&w^2 &= r_1z^3 - r_2z^2 - r_1 z, \qquad r_{1,2}\in \mathbb{R},\ r_1\ge0 &\textbf{Hurtubise 83}\\
\mathcal{C}_{H85}&&w^2 &= \kappa (z^2-s^2)(s^2 z^2-1),  \qquad\kappa>0,\ s\in[0,1) &\textbf{Hurtubise 85}\\
\mathcal{C}_{AH}&&\eta_{AH}^2&=K^2\, \zeta_{AH}\left( kk'[\zeta_{AH}^2-1]+(k^2-k'^2) \zeta_{AH}\right)&
\textbf{Atiyah, Hitchin}
\end{align*}

The reality properties of the spectral curve are preserved by the transformations of  $T\PP\sp1$ 
\begin{equation}\label{moebius}
R:=\begin{pmatrix} p&q\\-\bar q&\bar p\end{pmatrix}\in PSU(2),
\qquad
\zeta\rightarrow  \zeta_R:=\dfrac{\bar p\, \zeta-\bar q}{q\, \zeta+p},
\qquad  \eta\rightarrow \eta_R:= \dfrac{\eta}{(q\, \zeta+p)\sp2},
\end{equation}
which correspond to a spatial rotation. In particular with $p=e\sp{-i\theta/2}$, $q=0$, we may rotate
$(\zeta,\eta)\rightarrow e\sp{i\theta}(\zeta,\eta)$, and so the relative signs between $\eta^2$ and the
highest powers ($\zeta^4$ or $\zeta^3$) may be chosen so that the leading coefficients are positive.
This, for example, enabled Hurtubise \cite{hurtubise_83} to choose his coefficient $r_1\ge0$. 

We will describe our procedure for establishing the needed birational correspondence between our curve
and the other curves with the curve of Forg\'acs, Horv\'ath  and Palla as the example.

\subsubsection{Transforming between $\mathcal{C}$  and $\mathcal{C}_{FHP}$} 
First we record that the parameters of the Forg\'acs, Horv\'ath  and Palla curve are related by.
$$\sqrt{B} =\frac{1}{\sqrt{1+\beta}}K\left(\sqrt{ \frac{2\beta}{1+\beta}  }   \right), \quad     A=\frac12\beta B, \quad \beta\in [-1,0].$$

To see that a transformation (\ref{moebius}) between the curves is possible
we compute the Klein absolute invariants 
of both curves, $j_{\mathcal{C}}$,  $j_{\mathcal{C}_{ORS}}$ and find the  allowed relations between the parameters $k$ and $\beta$ given their prescribed domains.
Here, 
\begin{equation}
j_{\mathcal{C}}=\frac{256(k^4-k^2+1)^3}{k^4(k^2-1)^2}, \quad 
j_{\mathcal{C}_{FHP}}=\frac{64(3\beta^2+1)^3}{\beta^2 (\beta^2-1)^2}.
\end{equation}
The equation $j_{\mathcal{C}}=j_{\mathcal{C}_{FHP}}$
admits a number of solutions, and among them exists one suitable, namely, 
\begin{equation}
\beta =- \frac{k^2}{1+{k'}^2}, \qquad   \beta\in [-1,0] \leftrightarrow  [0,1]\ni k^2 \label{beta_k},
\end{equation}    
and equivalently, 
\begin{equation}
k^2=\frac{2\beta}{\beta-1}\label{k_beta}.
\end{equation}
The relation (\ref{beta_k}) enables us to connect the parameters of  the curves,
\begin{equation}
B= \frac{1+{k'}^2}{2} K(k)^2, \quad\quad K(k)=\sqrt{\frac{2B}{1+{k'}^2}}.
\end{equation}

To find the explicit transformation (\ref{moebius}) one can equate the fractional linear transformations
of branch points of $ \mathcal{C}$ with the branch points  of $\mathcal{C}_{ORS}$. 
There are 24 variants of such homogeneous equations with respect to the parameters of the
fractional linear transformations but only four of them admit non-zero solution. An appropriate transformation is given by
\begin{align}\begin{split}
\zeta&=\frac{\imath x -1}{\imath x +1}, \quad 
\eta=  \frac{2\imath xy}{(\imath-x)^2},  \\ 
x&=\imath \frac{\zeta+1}{\zeta-1}, \quad  y =
\frac{2 \eta}{\zeta^2-1 }.
\end{split}\label{Moebius}
\end{align}
Note, that transformation (\ref{Moebius}) maps  the 4 complex branch points 
$\pm k'\pm\imath k$ of the curve $\mathcal{C}$ to four real branch points 
of the curve $\mathcal{C}_{FHP}$ as follows
\begin{align}\begin{split}
&\pm (k'+\imath k)\; 
\longleftrightarrow\;\pm \sqrt{  \frac{-1+\sqrt{1-\beta^2}}{\beta}  }\equiv \pm \frac{1+k'}{k} ,\\
&\pm (-k'+\imath k)\; 
\longleftrightarrow\;\pm \imath \sqrt{  \frac{1+\sqrt{1-\beta^2}}{\beta}  }\equiv \pm \frac{1-k'}{k}.
\end{split}
\end{align}

\subsubsection{Transforming between $\mathcal{C}$  and $\mathcal{C}_{ORS}$} 
We note that Brown's curve \cite{brown83} is the same as O'Raifeartaigh, Rouhani and  Singh's \cite{ors82}
with $d\leftrightarrow\varepsilon$. Now 
\begin{align*}
j_{\mathcal{C}}=\frac{256({k'}^4-{k'}^2+1)^3}{{k'}^4(1-{k'}^2 )^2 } = \frac{256(\varepsilon^4+\varepsilon^2+1)^3}{\varepsilon^4 (\varepsilon^2+1)^2}=j_{\mathcal{C}_{ORS}} 
\end{align*}
has solutions
$\varepsilon^2=  -{k'}^2,\; -k^2,\; - \frac{1}{{k'}^2},\; -\frac{1}{k^2},\;\; \frac{k^2}{{k'}^2},\;\; \frac{{k'}^2}{ k^2}$.
Take
\begin{equation} \varepsilon= - \frac{\imath }{ k} . \label{epsilon_k} \end{equation}
Following the method outlined  above  we find the transformations
\begin{align}\begin{split}
\zeta&=\frac{1-z}{1+z},\quad \eta= \frac{2wz K(k)}{(1+z)^2} ,\\
z&=\frac{1-\zeta}{1+\zeta},\quad w=\frac{2\eta}{K(k)(\zeta^2-1)  },
\end{split}
\end{align}
(where the parameters $k$ and $\epsilon$ are related according (\ref{epsilon_k})) take
 the branch points $ \pm k' \pm \imath k  $  of $\mathcal{C}$ to the branch points of 
$\mathcal{C}_{ORS}$, $( \pm1 \pm \sqrt{\varepsilon^2+1})/ \varepsilon $.

\subsubsection{Transforming between $\mathcal{C}$  and $\mathcal{C}_{H85}$} 

Consider the transformation
$$
\frac1{\sqrt{2}}
\begin{pmatrix} u& \bar u\\ -u & \bar u
\end{pmatrix}, \qquad u=e\sp{i\pi/4};
\qquad
\zeta=-\frac{i z+1}{iz-1}, \quad \eta=\frac{2 i w}{(iz -1)^2}.
$$
This transforms our curve  (\ref{curve})
into
$$
w^2=\frac{K^2}{4}\left( k z^2 +2z +k\right)\left( k z^2 -2z +k\right).
$$
The substitution
$$
k=\frac{2s}{1+s^2},\qquad k'=\frac{1-s^2}{1+s^2},\qquad s=\frac{k}{1+k'},
$$
then yields
$$
w^2 =\frac{K^2}{ (1+s^2)^2} \ (z^2-s^2)(s^2 z^2-1)
$$
which is Hurtubise's curve upon the identification $\kappa={K^2}/{ (1+s^2)^2}$. When we substitute this transformation into our Atiyah-Ward constraint 
$\eta=(x_2-\imath x_1)-2\zeta x_3-(x_2+\imath x_1)\zeta^2$
we obtain
$$w=  -(x_1-i x_3) +2 x_2 z + (x_1+i x_3)z^2
$$
which corresponds to an interchange $x_2\leftrightarrow x_3$ with Hurtubise's conventions. With the
above identifications we find that our curve of bitangency (\ref{bitangency1})
is that of Hurtubise \cite{hurtubise85a} whose  curve of bitangency is
$$\frac{\kappa (s^4-1)^2}{4}= (s^2-1)^2 x_1^2 +(s^2+1)^2 x_2^2, \qquad x_3=0.$$

\subsubsection{Transforming between $\mathcal{C}$  and $\mathcal{C}_{AH}$} 
Then with the rotation $R=\frac1{\sqrt{1+|a|^2}}\begin{pmatrix} -a&1\\-1&-\bar a\end{pmatrix}$ we have
\begin{align*}
\sqrt{\frac{{\bar a}{\bar b}}{ab}}& (\zeta-a)(\zeta+\frac1{\bar a})(\zeta-b)(\zeta+\frac1{\bar b})
\times \frac1{(\zeta-a)^4}
\rightarrow 
 \frac{(b-a) (1+a{\bar b})}{|a||b| (1+|a|^2)^2}\,
\,
\zeta_R\left(\zeta_R+ \frac{1+{\bar a}b}{b-a}\right)\left(\zeta_R+\frac{{\bar a}-\bar{b}}{1+a\bar{b}} \right).
\end{align*}
We have previously ordered the four roots of our curve as
$a=e\sp{i\alpha}=k'+i k$, $-1/{\bar a}=-k'-i k$, $b=e\sp{-i\alpha}=k'-i k$,  $-1/{\bar b}=-k'+i k$,  
where $2\alpha\ge0$ is the angle between the lines. Then $k=\sin\alpha\ge0$ and so
$$|a-b|=2k,\quad |1+{\bar a}b|=2k', \quad
\frac{1+{\bar a}b}{b-a}=i\,\frac{e\sp{-i\alpha}}{\tan\alpha},\quad 
\frac{{\bar a}-\bar{b}}{1+a\bar{b}}=-i\,{e\sp{-i\alpha}}\,{\tan\alpha},
$$
giving
$$-ie\sp{i\alpha}\zeta_R\,kk'\left( \zeta_R^2-ie\sp{-i\alpha}(\tan\alpha-\cot\alpha)\zeta_R+e\sp{-2i\alpha}\right).
$$
Thus the further rotation
$(\zeta_R,\eta_R)\rightarrow(\zeta',\eta'):= i e\sp{i\alpha}(\zeta_R,\eta_R)$
gives
$$
0=\eta^2+\frac{K^2}{4}(\zeta^4+2(k^2-k'^2)\zeta\sp2+1)\rightarrow 0=
\eta'^2-\frac{k k' K^2}{4}\zeta'(\zeta'- k'/k)(\zeta'+k/k'),
$$
where
$ \cot\alpha=k'/k$,  $ \tan\alpha=k/k'$.
Thus we have the Atiyah-Hitchin curve upon the identifications $(\zeta_{AH},\eta_{AH})=
(\zeta', 2\eta')$.

We note that the sign of the term $(\tan\alpha-\cot\alpha)\zeta'=(k^2-k'^2)/kk'\, \zeta'$ depends on whether
$0\le\alpha<\pi/4$ or $\pi/2<\alpha\le\pi/2$. When $\alpha>\pi/4$ the angle between the lines is obtuse.

\subsection{Calculation of Periods and the Abel Map}\label{appendixperiods}
We determine here the periods of the holomorphic differential $\boldsymbol{v}$, the meromorphic
differential $\gamma_\infty$ and express the Abel map in terms of incomplete elliptic integrals.

\subsubsection{The periods of  $\boldsymbol{v}$}
Upon using the substitutions $\zeta=e\sp{i\theta}$ and
$k\sin u=\sin\theta$ on sheet $1$,
$$\frac{d\zeta}{\eta}=i\frac{2}{K}\frac{d\zeta}{\sqrt{(\zeta\sp2-
e\sp{2 i\alpha})(\zeta\sp2- e\sp{-2 i\alpha})}} =
\frac{-1}{k K}\frac{d\theta}{\sqrt{1-\frac{1}{k^2}\sin\sp2\theta}}=
\frac{-1}{K}\frac{du}{\sqrt{1-{k^2}\sin\sp2 u}}.$$ Thus
$$\oint_\mathfrak{a}\frac{d\zeta}{\eta}=
\frac{-2}{K}\int_\alpha\sp{-\alpha}
\frac{d\theta}{\sqrt{{k^2}-\sin\sp2\theta}}=
\frac{4}{K}\int_0\sp{\pi/2}\frac{du}{\sqrt{1-{k^2}\sin\sp2
u}}=4.$$ Similarly (with $\zeta=\exp i(w+\pi/2)$, $\sin w=k'\sin u$)
$$\frac{d\zeta}{\eta}
=i\frac{2}{K}\frac{d\zeta}{\zeta}\frac1{2\sqrt{\sin\sp2 w- k'^2} }
=
\frac{-1}{K}\frac{d\zeta}{\zeta}\frac1{\sqrt{k'^2-\sin\sp2 w} }
=\frac{-i}{K} \frac{dw}
{\sqrt{k'^2-\sin\sp2 w} }=
\frac{-i}{K}\frac{du}{\sqrt{1-{k'^2}\sin\sp2 u}}.$$ 
One determines the sign of the square root in the second equality, for when $\zeta=i$, $w=0$ and $\eta=-iK k'$ on sheet 1 (and crossing no cuts). We then have
\begin{align*}
\oint_\mathfrak{b}\frac{d\zeta}{\eta}&=
-\frac{2i}{K}\int\sp{\alpha-\pi/2}_{\pi/2-\alpha}
\frac{dw}{\sqrt{{k'^2}-\sin\sp2 w}}=
\frac{4i}{K}\int\sp{\pi/2-\alpha}_{0}
\frac{dw}{\sqrt{{k'^2}-\sin\sp2 w}}\\ &=
\frac{4i}{K}\int_0\sp{\pi/2}\frac{du}{\sqrt{1-{k'^2}\sin\sp2
u}}=\frac{4\im\,\mathbf{K}'(k)}{\mathbf{K}(k)}=4\tau,
\end{align*}
where $\tau={\im\,\mathbf{K}'(k)}/{\mathbf{K}(k)}$ is the period matrix.

\subsubsection{The periods of  $\gamma_\infty$} Consider
\begin{equation*}
\gamma_{\infty}(P) = \frac{K^2}{4\eta}\left( \zeta^2 - \frac{2E-K}{K}  \right)\mathrm{d} \,\zeta.
\end{equation*}
Then with the earlier substitutions and again on sheet 1,
\begin{align*}
\frac{K^2}{4} \frac{\zeta^2\mathrm{d} \,\zeta}{\eta}
&=
i\frac{K}{2}\frac{\zeta\sp2 d\zeta}{\sqrt{(\zeta\sp2-
e\sp{2 i\alpha})(\zeta\sp2- e\sp{-2 i\alpha})}} 
=
\frac{-K}{4k }\frac{ e^{2i\theta }\,d\theta}{\sqrt{1-\frac{1}{k^2}\sin\sp2\theta}}\\
&=
\frac{-K}{4}\frac{(1-2k^2\sin^2u +2ik\sin u \sqrt{1-{k^2}\sin\sp2 u}  )du}{\sqrt{1-{k^2}\sin\sp2 u}}.
\intertext{Then}
\frac{K^2}{4} \oint_\mathfrak{a}\frac{\zeta^2\mathrm{d} \,\zeta}{\eta}
&=
K \int_0\sp{\pi/2}
\frac{( 2\left[ 1-{k^2}\sin\sp2 u\right] -1  )du}{\sqrt{1-{k^2}\sin\sp2 u}}=
K\left( 2E-K\right)
\end{align*}
and so $ \oint_\mathfrak{a}\gamma_\infty=0$. Now 

\begin{align*}
\frac{K^2}{4} \oint_\mathfrak{b} &\left(\zeta^2- \frac{2E-K}{K}\right)\frac{d\zeta}{\eta}
=
-i \frac{K}{2} \int\sp{\alpha-\pi/2}_{\pi/2-\alpha}\frac{\left[ - e\sp{2iw} -(2E-K)/K \right] dw }{\sqrt{{k'^2}-\sin\sp2 w}}
\\ &=
-i K  \int_0\sp{\pi/2-\alpha}\frac{\left[  \cos {2w} +(2E-K)/K \right] dw }{\sqrt{{k'^2}-\sin\sp2 w}}
\\ &=
-i K \int_0\sp{\pi/2}
\frac{( 2\left[1-{k'^{2}}\sin\sp2 u\right] -1  +(2E-K)/K)du}{\sqrt{1-{k'^{2}}\sin\sp2 u}}
\\
&=
-i 2\left( K E'+E'K-KK'\right)=-i\pi =2i\pi \boldsymbol{U},
\end{align*}
where we have use Legendre's relation.

\subsubsection{The Abel map}
We may express the Abel maps $\phi$ and $\alpha$ in terms of incomplete elliptic integrals.
Denote
\begin{align*}
a= k'+\imath k, \quad b = k'-\imath k,  \quad c = \frac{2E-K}{K}\in \mathbb{R}.
\end{align*}
Here $a=P_0$ is the base point of the Abel map
$\phi(\zeta)= \frac14\int_{a}^{\zeta}\frac{\mathrm{d}\zeta}{\eta}$.
One can represent $\phi(\zeta)$ in terms of  Jacobian incomplete integrals
\begin{align}
\phi(\zeta) = \frac{\imath}{2 K b}\left( F\left( \frac{\zeta}{a},\frac{a}{b} \right)
 - F\left( 1,\frac{a}{b} \right)   \right),\label{abelmap}
\end{align}
and this representation accords with the relations of \cite{Braden2010d},
\begin{align*}
\phi(\infty_1)= \frac{1+\tau}{4}=-\phi(\infty_2),\quad\phi(0_1)= \frac{1-\tau}{4}=-\phi(0_2).
\end{align*}
We note the relations
\begin{align*}
\frac{a}{b}&=\frac{1+ik/k'}{1-ik/k'}=\frac{1-\dot k'}{1+\dot k'},\quad \dot k'=-\frac{ik}{k'},\ \ \dot k=\frac{1}{k'},
\quad
K\left(\frac{a}{b} \right)=\frac{1+\dot k'}2 K(\dot k)=\frac{b}2 \left( K'(k) +i K(k)\right).
\end{align*}
Using these our normalized Abel map now reads
\begin{align}
\alpha(\zeta(\boldsymbol{x}))= \phi(\zeta(\boldsymbol{x}))- \phi(\infty_1) = \frac{\imath}{2 K} \frac{1}{b}
 F\left(\frac{\zeta(\boldsymbol{x})}{a},\frac{a}{b}\right) - \frac{\tau}{2}  . \label{abelmap1}
\end{align} 

\subsubsection{Numerical Computation}
 We have shown that
\begin{equation*}
\gamma_{\infty}(P) = \frac{K^2}{4\eta}\left( \zeta^2 - \frac{2E-K}{K}  \right)\mathrm{d} \,\zeta.
\end{equation*}
Then
\begin{equation}
\int_{P_0}\sp{P}\gamma_{\infty}(P') =
\frac{\imath K}{2} \int_{a}^{\zeta(\boldsymbol{x})}\frac{(z^2-c) \mathrm{d}z}{\sqrt{(z^2-a^2)(z^2-b^2)}}=
\frac14\left( \frac{\theta_1' \left( \alpha(\zeta(\boldsymbol{x})) \right)  } {\theta_1 \left(  \alpha(\zeta(\boldsymbol{x})) \right)}
+
\frac{\theta_3' \left(  \alpha(\zeta(\boldsymbol{x})) \right)  } {\theta_3 \left(  \alpha(\zeta(\boldsymbol{x})) \right)   } \right)-\frac{\imath \pi}{4}
\end{equation}
with $ \alpha(\zeta(\boldsymbol{x}))$ being the normalized Abel map (\ref{abelmap1}). Now
\begin{align}\begin{split}
& \int_{a}^{\zeta(\boldsymbol{x})}\frac{(z^2-c) \mathrm{d}z}{\sqrt{(z^2-a^2)(z^2-b^2)}}\\
&=b E \left(\frac{a}{b}\right)- \frac{b^2-c}{b}K\left(\frac{a}{b}\right)
-bE\left(\frac{\zeta(\boldsymbol{x})}{a},\frac{a}{b}   \right)+ \frac{b^2-c}{b}F\left(\frac{\zeta(\boldsymbol{x})}{a},\frac{a}{b}   \right) \end{split}\label{mu1}
\end{align}
where $K(\kappa),E(\kappa)$ and  $F(z,\kappa),E(z,\kappa)$ with  $\kappa=a/b$ are standard complete 
and incomplete elliptic integrals of the first and second kind respectively. 
We remark  that some care is needed in keeping track of the sheets when using this representation.

\subsection{Proof of Lemma \ref{parcurve}}\label{appparcuve} To show that  
(\ref{curve}) is parameterized by
\begin{equation*}
\zeta=-i\,\frac{\theta_2[P]\theta_4[P]}{\theta_1[P]\theta_3[P]},\quad
\eta=  
\frac{i\pi\,\theta_3\theta_2^2\theta_4^2}4\,\frac{\theta_3[2 P]}{\theta_1[P]^2\theta_3[P]^2}.
\end{equation*}
we use (with $\theta_i:=\theta_i(0)$)
\begin{align*}
k=\frac{ \theta_2^2}{ \theta_3^2},\quad k'=\frac{ \theta_4^2}{ \theta_3^2},\quad 
 \theta_2^4+\theta_4^4&=\theta_3^4, \quad {K}=\frac{\pi}2\,  \theta_3^2 \\
 \theta_1[P] ^2\theta_3[P]^2+\theta_2[P]^2\theta_4[P]^2 &=
 \theta_2^2\theta_4\,  \theta_4[2 P] ,\\
 2\,\theta_1[P]\theta_2[P]\theta_3[P]\theta_4[P] &= 
 \theta_2 \theta_3\theta_4\, \theta_1[2 P],\\
 \theta_3(x+y)\theta_3(x-y)\,\theta_4^2&=
  \theta_4^2(x) \theta_3^2(y)- \theta_1^2(x) \theta_2^2(y)=
   \theta_3^2(x) \theta_4^2(y)- \theta_2^2(x) \theta_1^2(y);
\end{align*}
the latter with $x=0$ and $y=2\alpha(P)$. Then
\begin{align*}
\zeta^4+2(k^2-k'^2)\zeta^2+1&=
\frac{\theta_3^4\left( \theta_1[P] ^2\theta_3[P]^2+\theta_2[P]^2\theta_4[P]^2  \right)^2
+4(\theta_4^4-\theta_3^4)\,\theta_1[P]^2\theta_2[P]^2\theta_3[P]^2\theta_4[P]^2}
{\theta_3^4\, \theta_1[P]^4\theta_3[P]^4}\\
&=
\frac{\theta_3^4\theta_2^4\theta_4^2\,  \theta_4[2 P] ^2 -\theta_2^6 \theta_3^2\theta_4^2\,
\theta_1[2 P]^2}
{\theta_3^4\, \theta_1[P]^4\theta_3[P]^4}\\
&=
\frac{\theta_2^4\theta_3^2\theta_4^2\left(\theta_3^2\,  \theta_4[2 P] ^2 -\theta_2^2\,
\theta_1[2 P]^2\right)}
{\theta_3^4\, \theta_1[P]^4\theta_3[P]^4}\\
&=
\frac{\theta_2^4\theta_4^4}{\theta_3^2}\,\frac{\theta_3[2 P]^2}{\theta_1[P]^4\theta_3[P]^4}
\end{align*}
so establishing the lemma.

\subsection{Proof of Lemma \ref{intgammainf}}\label{proofintgammainf}
We will show that both sides have the same periodicity, zeros, poles and residues.
Both sides of (\ref{gammatheta}) are constant under shifts of $\mathfrak{a}$-periods. 
A shift in the theta functions under a $\mathfrak{b}$-period is immediate giving  $-\im\pi 
=2\pi i \boldsymbol{U}$ (using $\boldsymbol{U}=-1/2$). That is the right-hand side shifts by ($2\pi \im$ times) the Ercolani-Sinha vector. But 
\[ \oint_{\mathfrak{b}}\gamma_{\infty}=2\pi \im \boldsymbol{U}\]
is fundamental to its definition and follows from a bilinear relation
\cite{Braden2010d}.
Now observe that
\[ d\ln\theta_1\left(\int_{P_\ast}\sp{P}\boldsymbol{v}\right)=\boldsymbol{v}(P)\,
\frac{ \theta_1'\left( \int_{P_\ast}\sp{P}\boldsymbol{v}\right)}{\theta_1\left(\int_{P_\ast}\sp{P}\boldsymbol{v}\right)}
\]
and that for a local parameter $t$ at $P_\ast$,
\[ \boldsymbol{v}=\frac{d\zeta}{4\eta}=[\mu(P_\ast)+O(t)]dt,\quad
 \int_{P_\ast}\sp{P}\boldsymbol{v}=\mu(P_\ast)t+O(t^2).\]
 Thus
 \[
 d\ln\theta_1\left(\int_{P_\ast}\sp{P}\boldsymbol{v}\right)=dt\,[\mu(P_\ast)+O(t)]\left(
 \frac{ \theta_1'\left( 0\right)}{\theta_1\left(\mu(P_\ast)t\right)}+O(t^2)\right)=
 \frac{dt}{t}\, [1+O(t)]\]
 has a simple pole at $P_\ast$. Thus expanding the right-hand side  of (\ref{gammatheta}) at  $\infty_1$,  for example,  gives
\[\frac14\left\{ \frac{1}{-t/(4\rho_1)}+
\frac{\theta_1'[\infty_1-\infty_2]}{\theta_1[\infty_1-\infty_2]}+\ldots\right\}=
\tilde\nu_1-\frac{\rho_1}{t}+\ldots
\]
where
$$\tilde\nu_1=\frac14\frac{\theta_1'[\infty_1-\infty_2]}{\theta_1[\infty_1-\infty_2]}=-\frac{\im\pi}4.
$$
We know  that in the vicinity of $\infty_i$ the left-hand side looks like,
\[\int_{P_0}^P\gamma_{\infty}(P')=\tilde\nu_i-\frac{\eta}{\zeta}.\]
Thus both sides have the same pole and residue and similarly at $\infty_2$. Finally at $P_0$ both sides vanish so establishing the lemma. As remarked after the lemma, this identifying of the vanishing relates the choice of contours on each side of the identity.

\section{Theta function Identitities}\label{thetafunctidentapp}

\subsection{Weierstrass Trisecant \texorpdfstring{$\theta$}{theta}-formulae}
In this appendix we describe the Weierstrass Trisecant $\theta$-formulae that we implemented in the course of calculation. Following \cite{wei885}[p47] we introduce 3 vectors $\boldsymbol{\alpha}=(\alpha_1,\alpha_2,\alpha_3,\alpha_4 )$,  $\boldsymbol{\alpha}'=(\alpha_1',\alpha_2',\alpha_3',\alpha_4' )$,  $\boldsymbol{\alpha}''=(\alpha_1'',\alpha_2'',\alpha_3'',\alpha_4'' )$ that transformed one to another by the rule:
\[  T:  \boldsymbol{\alpha}^T \rightarrow  {\boldsymbol{\alpha}'}^T=\frac12\left( \begin{array}{r} \alpha_1+\alpha_2+\alpha_3+\alpha_4\\ \alpha_1+\alpha_2-\alpha_3-\alpha_4\\\alpha_1-\alpha_2+\alpha_3-\alpha_4\\
-\alpha_1+\alpha_2+\alpha_3-\alpha_4
\end{array}\right)\]
which leads to the relations:  $T(\boldsymbol{\alpha}) =\boldsymbol{\alpha}'$, $T(\boldsymbol{\alpha}') =\boldsymbol{\alpha}''$, $T(\boldsymbol{\alpha}'') =\boldsymbol{\alpha}$.

The following 6 Weierstrass Trisecant $\theta$-relations are valid

\begin{align*}
\begin{split}
\theta_1(\alpha_1)\theta_1(\alpha_2)\theta_1(\alpha_3)\theta_1(\alpha_4)
+\theta_1(\alpha_1')\theta_1(\alpha_2')\theta_1(\alpha_3')\theta_1(\alpha_4')\\
+\theta_1(\alpha_1'')\theta_1(\alpha_2'')\theta_1(\alpha_3'')\theta_1(\alpha_4'')=0,
\end{split} \hskip1cm(W1)
\end{align*}

\begin{align*}
\begin{split}
\theta_i(\alpha_1)\theta_i(\alpha_2)\theta_1(\alpha_3)\theta_1(\alpha_4)
+\theta_i(\alpha_1')\theta_i(\alpha_2')\theta_1(\alpha_3')\theta_1(\alpha_4')\\
+\theta_i(\alpha_1'')\theta_i(\alpha_2'')\theta_1(\alpha_3'')\theta_1(\alpha_4'')=0,
\end{split}\hskip1cm(W2)
\end{align*}

\begin{align*}
\begin{split}
\theta_i(\alpha_1)\theta_j(\alpha_2)\theta_k(\alpha_3)\theta_1(\alpha_4)
+\theta_i(\alpha_1')\theta_j(\alpha_2')\theta_k(\alpha_3')\theta_1(\alpha_4')\\
+\theta_i(\alpha_1'')\theta_j(\alpha_2'')\theta_k(\alpha_3'')\theta_1(\alpha_4'')=0,
\end{split}\hskip1cm(W3)
\end{align*}

\begin{align*}
\begin{split}
\theta_i(\alpha_1)\theta_i(\alpha_2)\theta_j(\alpha_3)\theta_j(\alpha_4)
-\theta_i(\alpha_1')\theta_i(\alpha_2')\theta_j(\alpha_3')\theta_j(\alpha_4')\\
\pm\theta_k(\alpha_1'')\theta_k(\alpha_2'')\theta_1(\alpha_3'')\theta_1(\alpha_4'')=0,
\end{split}\hskip1cm(W4)
\end{align*}

\begin{align*}
\begin{split}
\theta_2(\alpha_1)\theta_2(\alpha_2)\theta_2(\alpha_3)\theta_2(\alpha_4)
-\theta_3(\alpha_1')\theta_3(\alpha_2')\theta_3(\alpha_3')\theta_3(\alpha_4')\\
+\theta_4(\alpha_1'')\theta_4(\alpha_2'')\theta_4(\alpha_3'')\theta_4(\alpha_4'')=0,
\end{split}\hskip1cm(W5)
\end{align*}

\begin{align*}
\begin{split}
\theta_i(\alpha_1)\theta_i(\alpha_2)\theta_i(\alpha_3)\theta_i(\alpha_4)
-\theta_i(\alpha_1')\theta_i(\alpha_2')\theta_i(\alpha_3')\theta_i(\alpha_4')\\
\pm\theta_1(\alpha_1'')\theta_1(\alpha_2'')\theta_1(\alpha_3'')\theta_1(\alpha_4'')=0.
\end{split}\hskip1cm(W6)
\end{align*}

We present here particular cases of these relations that used in our development.
From (W2) it follows that:
\begin{proposition}Let $\alpha_i,\alpha_j,\alpha_k$ , $i,j,k\in \{1,2,3,4\}$   be three arbitrary complex numbers. Then for $n=2,3,4$ and $z\in\mathbb{C}$ the following trisecant addition formula is valid
\begin{align}
\begin{split}
&\theta_1(\alpha_i)\theta_1(\alpha_j-\alpha_k)\theta_n\left(\alpha_i\pm\frac{z}{2}\right)\theta_n\left(\alpha_j+\alpha_k\pm\frac{z}{2}\right)\\
&\quad\quad +\theta_1(\alpha_k)\theta_1(\alpha_i-\alpha_j)\theta_n\left(\alpha_k\pm\frac{z}{2}\right)\theta_n\left(\alpha_i+\alpha_j\pm\frac{z}{2}\right)\\
&\quad\quad\quad\quad +\theta_1(\alpha_j)\theta_1(\alpha_k-\alpha_i)\theta_n\left(\alpha_j\pm\frac{z}{2}\right)\theta_n\left(\alpha_i+\alpha_k\pm\frac{z}{2}\right)=0.
\end{split}
\end{align}
\end{proposition}

From (W3) it follows that:
\begin{proposition}Let $\alpha_1,\alpha_2,\alpha_3,\alpha_4 $ be four arbitrary complex numbers. Then for any $i,j=1,\ldots 4$ and arbitrary $z\in \mathbb{C}$ the following trisecant addition formula is valid
\begin{align}\begin{split}
&\theta_1(\alpha_i)\theta_2\left(\alpha_j\pm\frac{z}{2}\right)
\theta_3(\alpha_j)\theta_4\left(\alpha_i
\pm\frac{z}{2}\right)
-\theta_1(\alpha_j)\theta_2\left(\alpha_i\pm\frac{z}{2}\right)\theta_3(\alpha_i) \theta_4\left(\alpha_j\pm\frac{z}{2}\right)\\
&\qquad\qquad\qquad=\theta_1(\alpha_j-\alpha_i)\theta_2\left(\frac{z}{2}\right)\theta_3(0)
\theta_4\left(\alpha_i+\alpha_j
\pm\frac{z}{2}\right) .\end{split}
\end{align}
\end{proposition}

From (W4) it follows that:
\begin{proposition}
Let $\alpha_i,\alpha_j,\alpha_k$ , $i,j,k\in \{1,2,3,4\}$   be three arbitrary complex numbers.
Then for $p=3,q=4$  or $p=4,q=3$ and $z\in\mathbb{C}$ the following trisecant addition formula is valid
\begin{align}
\begin{split}
&\theta_p(\alpha_i)\theta_p(\alpha_j+\alpha_k)\theta_q\left(\alpha_k\pm\frac{z}{2}\right)\theta_q\left(\alpha_i+\alpha_j-\frac{z}{2}\right)\\  &\qquad\qquad -
\theta_p(\alpha_k)\theta_p(\alpha_i+\alpha_j)\theta_q\left(\alpha_i\pm\frac{z}{2}\right)\theta_q\left(\alpha_j+\alpha_k\pm\frac{z}{2}\right)\\
&\qquad\qquad\qquad\qquad=\theta_2\left(\frac{z}{2}\right)\theta_2\left(\alpha_i+\alpha_j+\alpha_k\pm \frac{z}{2}\right)
\theta_1(\alpha_i-\alpha_k)\theta_1(\alpha_j).
\end{split}
\end{align}
Also for arbitrary four complex numbers $\alpha_1,\ldots,\alpha_4$ , $i,j,k\in \{1,2,3,4\}$ the following trisecant addition formula is valid

\begin{align}\begin{split}
&\theta_p(\alpha_4+\alpha_2)\theta_p(\alpha_1+\alpha_3)\theta_q\left(\alpha_2+\alpha_1\pm\frac{z}{2}\right)
\theta_q\left(\alpha_4+\alpha_3\pm\frac{z}{2}\right)\\
&\qquad -
\theta_p(\alpha_4+\alpha_3)\theta_p(\alpha_2+\alpha_1)\theta_q\left(\alpha_4+\alpha_2\pm\frac{z}{2}\right)
\theta_q\left(\alpha_3+\alpha_1\pm\frac{z}{2}\right)\\
&\qquad\qquad - \theta_2\left(\frac{z}{2}\right)
\theta_2\left(\alpha_1+\alpha_2+\alpha_3+\alpha_4\pm\frac{z}{2}\right)
\theta_1(\alpha_2-\alpha_3)\theta_1(\alpha_1-\alpha_4)=0.
\end{split}\end{align}
\end{proposition}

Suppose we now have that
\begin{equation}   \alpha_1+\alpha_2+\alpha_3+\alpha_4 = N \tau ,\quad   N\in \mathbb{Z}\end{equation}
the last trisecant relation turns to the following:

\begin{align}\begin{split}
&\theta_p(\alpha_4+\alpha_2)\theta_p(\alpha_1+\alpha_3)\theta_q\left(\alpha_2+\alpha_1\pm\frac{z}{2}\right)
\theta_q\left(\alpha_4+\alpha_3\pm\frac{z}{2}\right)\\
&\qquad -
\theta_p(\alpha_4+\alpha_3)\theta_p(\alpha_2+\alpha_1)\theta_q\left(\alpha_4+\alpha_2\pm\frac{z}{2}\right)
\theta_q\left(\alpha_3+\alpha_1\pm\frac{z}{2}\right)\\
&\qquad\qquad =\theta_2\left(\frac{z}{2}\right)^2
\theta_1(\alpha_2-\alpha_3)\theta_1(\alpha_1-\alpha_4)\mathrm{exp}\left\{ -\imath\pi(N^2\tau\pm Nz) \right\}
.\end{split}\end{align}

The combination of relations (W3) written in the form 
 \begin{align}\begin{split}
&\theta_1(\alpha_i) \theta_4\left(\alpha_i-\frac{z}{2}\right)\theta_4\left(\alpha_j\right)\theta_1\left(\alpha_j+\frac{z}{2}\right) \\&\hskip1cm-\theta_3\left(\alpha_i\right)\theta_2\left(\alpha_i-\frac{z}{2}\right)\theta_2\left(\alpha_j\right)\theta_3\left(\alpha_j+\frac{z}{2}\right)       \\
&\hskip2cm+\theta_2\left(\frac{z}{2}\right)\theta_2\left(\alpha_i+\alpha_j\right)\theta_3\left(0\right)
\theta_3\left(\alpha_i-\alpha_j-\frac{z}{2}\right)=0
\end{split}
\end{align}
and
\begin{align}\begin{split}
&\theta_1\left(\alpha_i+\frac{z}{2}\right)\theta_2\left(\alpha_j-\frac{z}{2}\right)  \theta_3\left(\alpha_j\right) \theta_4\left(\alpha_i\right)                     \\
&\hskip1cm-\theta_1\left(\alpha_i\right) \theta_2\left(\alpha_j\right)
\theta_3\left(\alpha_j+\frac{z}{2}\right)\theta_4\left(\alpha_i-\frac{z}{2}\right)       \\
&\hskip2cm+\theta_1\left(\alpha_i-\alpha_j-\frac{z}{2}\right)\theta_2\left(\frac{z}{2}\right)\theta_3\left(0\right)\theta_4\left(\alpha_i+\alpha_j\right)
=0
\end{split}
\end{align}
together with (W4) leads to the addition formula that we implemented to calculate the Gram matrix,
\begin{align}\begin{split}
&\hskip0.45cm\theta_1(\alpha_i)\theta_4\left(\alpha_i-\frac{z}{2}\right)
\displaystyle{\left|\begin{array}{cc} \theta_3\left(\frac{z}{2}\right) \theta_2(0) &\theta_1\left(\frac{z}{2}\right)\theta_4(0)\\\\
\theta_1\left(\alpha_j+\frac{z}{2}\right)
\theta_4(\alpha_i)&\theta_3\left(\alpha_j+\frac{z}{2}\right)\theta_2(\alpha_j) \end{array}\right|}\\\\
&+\theta_3(\alpha_i)\theta_2\left(\alpha_i-\frac{z}{2}\right)
\displaystyle{\left|\begin{array}{cc} \theta_3\left(\frac{z}{2}\right)\theta_2(0)&\theta_1\left(\frac{z}{2}\right)\theta_4(0)\\\\
\theta_3\left(\alpha_j+\frac{z}{2}\right)
\theta_2(\alpha_i)&\theta_1\left(\alpha_j+\frac{z}{2}\right)\theta_4(\alpha_j) \end{array}\right|}\\
&=\theta_2\left(\frac{z}{2}\right)^2\theta_3^2(0)\theta_1(\alpha_i+\alpha_j)\theta_4\left(\alpha_i-\alpha_j-\frac{z}{2}\right) .\end{split}\label{addition}
\end{align}

Finally we note:
\begin{proposition} Let $\alpha_1,\alpha_2,\alpha_3,\alpha_4 $ be four complex numbers satisfying condition
\[   \alpha_1+\alpha_2+\alpha_3+\alpha_4 =    N\tau,\quad N\in\mathbb{Z}   \]
and  $z\in \mathbb{C}$ . Then
\begin{align}\begin{split}
(-1)^N\prod_{m=1}^4 \theta_4(\alpha_m)+(-1)^M\prod_{m=1}^4 \theta_2(\alpha_m)
&=\theta_3(0)\theta_3(\alpha_i+\alpha_j)
\theta_3(\alpha_i+\alpha_k)\theta_3(\alpha_j+\alpha_k)\\
&\times\mathrm{exp}\left\{  -2\imath\pi N \alpha_l +\imath \pi N^2 \tau \right\},
\end{split}
\end{align}
and
\begin{align}\begin{split}
(-1)^{(N+M)}\prod_{m=1}^4 \theta_1(\alpha_m)+\prod_{m=1}^4 \theta_3(\alpha_m)
&=\theta_3(0)\theta_3(\alpha_i+\alpha_j)
\theta_3(\alpha_i+\alpha_k)\theta_3(\alpha_j+\alpha_k)\\
&\times\mathrm{exp}\left\{  -2\imath\pi N \alpha_l +\imath \pi N^2 \tau \right\},
\end{split}
\end{align}
with $i\neq j \neq k \neq l \in \{ 1,2,3,4\} $
\end{proposition}

\begin{proof}

If we use 
$\boldsymbol{\alpha}=(\alpha_3,\alpha_2,\alpha_4,-\alpha_2-\alpha_3-\alpha_4)$,
$\boldsymbol{\alpha'}=(0,\alpha_2+\alpha_3,\alpha_3+\alpha_4,\alpha_2+\alpha_4)$ and
$\boldsymbol{\alpha''}=(\alpha_2+\alpha_3+\alpha_4,-\alpha_4,-\alpha_2,\alpha_3)$ in (W6) for
$i=3$, together with the Abel relation we obtain the second identity while the first similarly follows from (W5).

\end{proof}

\subsection{Periodicities}
\begin{align*}
\theta_1(z+M+N\tau) &= (-1)\sp{N+M+1}\,e\sp{-i\pi[N^2\tau+2Nz]}\,\theta_1(z),\\
\theta_2(z+M+N\tau) &=\quad \quad(-1)\sp{M}\,e\sp{-i\pi[N^2\tau+2Nz]}\,\theta_2(z),\\
\theta_3(z+M+N\tau) &= \qquad\qquad\quad e\sp{-i\pi[N^2\tau+2Nz]}\,\theta_3(z),\\
\theta_4(z+M+N\tau) &= \quad \quad(-1)\sp{N}\,e\sp{-i\pi[N^2\tau+2Nz]}\,\theta_4(z).
\end{align*}

We also note
\begin{align}
\theta_2(z\pm 1/2) &=\mp \theta_1(z), \label{thetashift2}\\
\theta_4(z\pm 1/2)&= \ \ \theta_3(z). \label{thetashift4}
\end{align}

\section{Identities}\label{identitiesapp}

\subsection{Proof of Proposition \ref{abjacprop}}\label{proofabjacprop}
For the curve (\ref{curve}) we may write
\begin{equation}\label{abeldiffinfty}
w(P)=-\eta+(x_2-\imath x_1)-2\zeta x_3-(x_2+\imath x_1)\zeta^2
=c \,\frac{\prod_{i=1}\sp{4}\theta_1[P-P_i]}{\theta_1[P]^2\theta_1[P-\infty_2]^2},
\end{equation}
some constant $c$.
Here we encounter a subtlety referred to earlier when discussing Abel images. In writing the function
in the specified form with the given theta functions, periodicity requires choosing the Abel images so that
$$\alpha_1'+\alpha_2'+\alpha_3'+\alpha_4'=\sum_k \alpha(P_k)= 2\alpha(\infty_1+\infty_2).$$
The first part of the proposition is then proven upon establishing that
$\phi(2[\infty_1-\infty_2])\in\Lambda$, which follows from either
(\ref{abelimages}) or observing that this is the divisor of the function
\[\eta +\im \frac{K}2 \left(\zeta^2+(k^2-k'^2)\right).
\]
Although our numerical calculation of Abel images was such that $\sum_k \alpha_k=N\tau$, the choice of
of sheets in defining (\ref{abeldiffinfty}) has $\sum_k \alpha_k'=-(1+\tau)$; agreement can be 
achieved simply
by shifting the argument of one of the $\theta_1$'s in (\ref{abeldiffinfty}) by the appropriate lattice point,
for example  $$ \alpha_1=\alpha_1',\quad \alpha_2=\alpha_2',\quad \alpha_3=\alpha_3',\quad
\alpha_4=\alpha_4'+(N+1)\tau+1.$$
If we expand $w(P)$ at $\infty_1$ on sheet 1 we see
$$w(P)=\frac{i}2 (K-2x_-)\zeta^2-2x_3\zeta+\ldots
\qquad c_1= \frac{i}2 (K-2x_-) \frac{\theta_1^2[\infty_1-\infty_2] \theta_1^{'\, 2}}{\prod_{i=1}\sp4 \theta_1[\infty_1-P_i]},
$$
while on sheet 2
$$w(P)=-\frac{i}2 (K+2x_-)\zeta^2-2x_3\zeta+\ldots
\qquad c_2= -\frac{i}2 (K+2x_-) \frac{\theta_1^2[\infty_2-\infty_1] \theta_1^{'\, 2}}{\prod_{i=1}\sp4 \theta_1[\infty_2-P_i]}.
$$
Consistency requires that $c_1=c_2$ or that
$$-\frac{K-2x_-}{K+2x_-}=\frac{ \prod_{i=1}\sp4 \theta_1[\infty_1-P_i]}
{\prod_{i=1}\sp4 \theta_1[\infty_2-P_i]}=\exp(i\pi[\sum_k\alpha_k'+\tau])\,
\frac{ \prod_{i=1}\sp4 \theta_1(\alpha_i')}
{\prod_{i=1}\sp4 \theta_3(\alpha_i)}
=-\frac{ \prod_{i=1}\sp4 \theta_1(\alpha_i')}
{\prod_{i=1}\sp4 \theta_3(\alpha_i')}
$$
upon using
$\theta_1[\infty_2-P_k]=-\theta_1(\alpha_k'+(1+\tau)/2)=-\exp(-i\pi[\alpha_k'+\tau/4)\,\theta_3(\alpha_k')$.
Now from  (\ref{sub1a}, \ref{sub1b})
$$
\frac{K-2x_-}{K+2x_-}
=
-\frac{\theta_1(N\tau-\alpha_4)\theta_1(\alpha_1)\theta_1(\alpha_2)\theta_1(\alpha_3)}
{\theta_3(N\tau-\alpha_4)\theta_3(\alpha_1)\theta_3(\alpha_2)\theta_3(\alpha_3)}=
(-1)\sp{N+1}\frac{ \prod_{i=1}\sp4 \theta_1(\alpha_i)}
{\prod_{i=1}\sp4 \theta_3(\alpha_i)}
$$
and the shifts given above establish the needed consistency.

To establish the second identity we use (\ref{residues}). With $\zeta=1/t$ a local parameter we have
at $\infty_1$ on sheet 1 
$$w(P)=\frac{i}{2}(K-2x_-)\frac1{t^2}\left[ 1-\frac{2x_3}{\frac{i}2 (K-2x_-)}t+\ldots\right]
\qquad 
d\ln w(P)=-2\frac{dt}{t}+\frac{4ix_3}{K-2x_-}\,dt+\ldots
$$
while on sheet 2
$$w(P)=-\frac{i}2 (K+2x_-)\frac1{t^2}\left[ 1+\frac{2x_3}{\frac{i}2 (K-2x_-)}t+\ldots\right]
\qquad 
d\ln w(P)=-2\frac{dt}{t}-\frac{4ix_3}{K+2x_-}\,dt+\ldots
$$
while with 
$f(P)=\int_{P_0}\sp{P}\gamma_\infty$,
$$f(P)\sim_{P\sim \infty_j} \tilde\nu_j -\frac{\rho_j}{t}.$$ Thus
$$0=\sum_{\text{Residues}}f(P) d\ln w(P)
=\sum_k f(P_k)-2(\tilde\nu_1+\tilde\nu_2)-\rho_1\frac{4ix_3}{K-2x_-}+\rho_2\frac{4ix_3}{K+2x_-}
$$
giving
$$\sum_k  \int_{P_0}\sp{P_k}\gamma_\infty= 
\sum_k\frac14
\left\{ \frac{\theta_1'(\alpha_k')}{\theta_1(\alpha_k')}+
\frac{\theta_3'[(\alpha_k')}{\theta_3(\alpha_k')} -{\imath\pi} \right\}
=\frac{4K^2 x_3}{K^2-4x_-^2}
$$
establishing (\ref{abel2}).

A consequence of this result is that
$$\sum_k\frac14
\left\{ \frac{\theta_1'(\alpha_k)}{\theta_1(\alpha_k)}+
\frac{\theta_3'[(\alpha_k)}{\theta_3(\alpha_k)}  \right\}+\imath \pi N
=\sum_k \beta_1(P_k)+\imath \pi N
=\frac{4K^2 x_3}{K^2-4x_-^2}
$$
and so
$$
\sum_k \mu(P_k)=3\imath \pi N
-\sum_k(x_3+i x_-\zeta_k)+\frac{4K^2 x_3}{K^2-4x_-^2}
=3\imath \pi N
-4x_3-i x_-\left (\sum_k \zeta_k\right)+\frac{4K^2 x_3}{K^2-4x_-^2}.
$$
Using $\sum_k \zeta_k=-{16\imath x_3 x_-}/{(K^2-4x_-^2)}$
we obtain
\begin{equation}
\sum_k \mu_k= 3\imath \pi N
\end{equation}
establishing the proposition.

\subsection{Proof of Lemma \ref{curveidenta}}
We use
\begin{align*}
\frac{\theta_1'[P]}{\theta_1[P]}-
 \frac{\theta_3'[P]}{\theta_3[P]}&=\frac{\theta_3[P]}{\theta_1[P]} d\left(
 \frac{\theta_1[P]}{\theta_3[P]}\right)=\frac{\theta_3[P]}{\theta_1[P]} 
 \pi \theta_3^2 \frac{\theta_2[P]\theta_4[P]}{\theta_3[P]^2}=2i \boldsymbol{K}\zeta,\\
 \frac{\theta_1''[P]}{\theta_1[P]}-
 \frac{\theta_3''[P]}{\theta_3[P]}&=\frac{d}{d\alpha(P)}\left(\frac{\theta_1'[P]}{\theta_1[P]}-
 \frac{\theta_3'[P]}{\theta_3[P]}\right)+\left(\frac{\theta_1'[P]}{\theta_1[P]}\right)^2-
 \left(\frac{\theta_3'[P]}{\theta_3[P]}\right)^2\\
 &=
 2\boldsymbol{K} 
 \frac{d }{d\alpha(P)}\left(\frac{\theta_2[P]\theta_4[P]}{\theta_1[P]\theta_3[P]}\right)
 +\left(\frac{\theta_1'[P]}{\theta_1[P]}-\frac{\theta_3'[P]}{\theta_3[P]}\right)
 \left(\frac{\theta_1'[P]}{\theta_1[P]}+\frac{\theta_3'[P]}{\theta_3[P]}\right)\\
 &=
  2\boldsymbol{K} \left[
\frac{\theta_4[P]}{\theta_1[P]} \frac{d}{d\alpha(P)}\left(\frac{\theta_2[P]}{\theta_3[P]}\right)  +  
\frac{\theta_2[P]}{\theta_3[P]}\frac{d}{d\alpha(P)} \left(\frac{\theta_4[P]}{\theta_1[P]}\right)
 \right]+2i \boldsymbol{K}\zeta\left(\frac{\theta_1'[P]}{\theta_1[P]}+\frac{\theta_3'[P]}{\theta_3[P]}\right)\\
 &=8i\boldsymbol{K}\eta+
 2i \boldsymbol{K}\zeta\left(\frac{\theta_1'[P]}{\theta_1[P]}+\frac{\theta_3'[P]}{\theta_3[P]}\right)\\
 &=8i\boldsymbol{K}\eta+
 8i \boldsymbol{K}\zeta\,\beta_1(P).
 \end{align*}
In going from the third last to the penultimate line here we are using standard theta function identities
such as $\left({\theta_2[P]}/{\theta_3[P]}\right)'= -\pi \theta_4^2 \theta_1[P] \theta_4[P]/\theta_3[P]^2$.
Examination of the latter proof shows we have in fact also established (\ref{curveident3}).  The final identity follows upon differentiating both sides of
\[\eta^2=-\frac{K^2}4\left( \zeta^4+2(k^2-k'^2)\zeta^2+1\right)\]
and using  (\ref{curveident3}).

\subsection{Proof of Corollary \ref{curveidentb}}
The first set two relations follow upon combining (\ref{curveident1}) with (\ref{betadef2}),
\begin{equation}
\beta_1(P)=\int_{P_0}^P\gamma_{\infty}(P')-\nu_1 =\frac14 \frac{\theta_1'(\alpha)}{\theta_1(\alpha)}+ \frac14 \frac{\theta_3'(\alpha)}{\theta_3(\alpha)}.\label{secondkindint}
\end{equation}
Next, upon differentiating  (\ref{secondkindint}), using (\ref{quot1a}) and upon noting that the normalized second kind differential  $\gamma_{\infty}$ written in the curve coordinates is
\begin{equation}
\gamma_{\infty}(P) = \frac{K^2}{4\eta}\left( \zeta^2 - \frac{2E-K}{K}  \right)\mathrm{d} \,\zeta
\end{equation}
while $\boldsymbol{v}=d\zeta/(4\eta)$ 
we get
\begin{equation}\label{diffbetac}
 \beta_1'(P):=\frac{d \beta_1(P)}{d\alpha(P)}=K^2\left(\zeta^2- \frac{2E-K}{K}  \right)
\end{equation}
and consequently
\begin{equation}
 \frac{\theta_1''(\alpha)}{ \theta_1(\alpha)} + 
 \frac{\theta_3''(\alpha)}{\theta_3(\alpha)} = 
2K^2\zeta^2-4(2E-K)K +8\beta^2(P). \label{quot2}
\end{equation}
Combining this with (\ref{curveident2}) then yields (\ref{quot3a}) and (\ref{quot3b}).  Equally,
upon defining 
\[
 T:=\frac{\theta_1'(\alpha)}{\theta_1(\alpha)}=2\beta_1(P)+\imath K \zeta
 \]
then
\[
\frac{\theta_1''(\alpha)}{\theta_1(\alpha)}=2\beta_1'(P)+\imath K \zeta' +T^2
 =2 K^2\left(\zeta^2- \frac{2E-K}{K}  \right)+4\imath \eta+T^2.
\]
The final results follow upon further differentiation and using the earlier results.

\subsection{Proof of Corollary \ref{curveidentbthreepts}}
Using the constraint $\alpha_1+\alpha_2+\alpha_3+\alpha_4= N\tau$ the periodicity of the $\theta$-functions (Corollary \ref{abelprop}) yields
\begin{align}
\theta_1(\alpha_i+\alpha_j+\alpha_k )&=(-1)^{N+1}\theta_1(\alpha_l)\mathrm{e}^{ -\imath \pi N^2\tau + 2\imath \pi N \alpha_l  } ,\\
\theta_3(\alpha_i+\alpha_j+\alpha_k )&=\theta_3(\alpha_l)\mathrm{e}^{ -\imath \pi N^2\tau + 2\imath \pi N \alpha_l  } ,\\
\theta_1'(\alpha_i+\alpha_j+\alpha_k )&=(-1)^{N}\left[\theta_1'(\alpha_l)
+2\imath\pi N \theta_1( \alpha_l)\right]\mathrm{e}^{ -\imath \pi N^2\tau + 2\imath \pi N\alpha_l  }, \\
\theta_3'(\alpha_i+\alpha_j+\alpha_k )&=-\left[\theta_3'(\alpha_l)
+2\imath\pi N \theta_3( \alpha_l)\right]\mathrm{e}^{ -\imath \pi N^2\tau + 2\imath \pi N \alpha_l  },
\\
\theta_1''(\alpha_i+\alpha_j+\alpha_k )&=(-1)^{N+1}\mathrm{e}^{ -\imath \pi N^2\tau + 2\imath \pi N \alpha_l  }
\times \left( -4\pi^2 N^2 \theta_1(\alpha_l) + 4\imath \pi N\theta_1'(\alpha_l) + \theta_1''(\alpha_l) \right ) ,\\
\theta_3''(\alpha_i+\alpha_j+\alpha_k )&=\mathrm{e}^{ -\imath \pi N^2\tau + 2\imath \pi N \alpha_l  }
\times \left( -4\pi^2 N^2 \theta_3(\alpha_l) + 4\imath \pi N\theta_3'(\alpha_l)+ \theta_3''(\alpha_l) \right).
\end{align}
The Corollary now follows upon employing Corollary \ref{curveidentb}.
We note that 
\begin{align}\begin{split}
&\frac{\theta_1'(\alpha_i+\alpha_j+\alpha_k)}{\theta_1(\alpha_i+\alpha_j+\alpha_k) }+\frac{\theta_3'(\alpha_i+\alpha_j+\alpha_k)}{\theta_3(\alpha_i+\alpha_j+\alpha_k) }= -4 \beta_1(\alpha_l)-4\imath\pi N ,\\\\
&\frac{\theta_1'(\alpha_i+\alpha_j+\alpha_k)}{\theta_1(\alpha_i+\alpha_j+\alpha_k) }-\frac{\theta_3'(\alpha_i+\alpha_j+\alpha_k)}{\theta_3(\alpha_i+\alpha_j+\alpha_k) }=-2\imath K \zeta_l ,
\end{split}
\end{align}
and
\begin{align}\begin{split}
&\frac{\theta_1''(\alpha_i+\alpha_j+\alpha_k)}{\theta_1(\alpha_i+\alpha_j+\alpha_k) }+\frac{\theta_3''(\alpha_i+\alpha_j+\alpha_k)}{\theta_3(\alpha_i+\alpha_j+\alpha_k) } ,\\
&\hskip1cm = 2K^2\zeta_l^2-8\pi^2N^2+16\imath \beta_1(\alpha_l) \pi N -4K(2E-K)+8\beta_1(\alpha_l)^2\\\\
&\frac{\theta_1''(\alpha_i+\alpha_j+\alpha_k)}{\theta_1(\alpha_i+\alpha_j+\alpha_k) }-\frac{\theta_3''(\alpha_i+\alpha_j+\alpha_k)}{\theta_3(\alpha_i+\alpha_j+\alpha_k) }\\
&\hskip1cm =  - 8K(\pi N -\imath \beta_1(\alpha_l))\zeta_l + 8\imath K \eta_l .
\end{split}
\end{align}

\subsection{Proof of Proposition \ref{mux1x2starplane}}\label{proofmux1x2starplane}
To prove the first of these relations we compute 
\begin{align}\begin{split}
&\mu_1(\boldsymbol{x}) + \mu_*(\boldsymbol{x})\\
&=\frac{\imath K}{2}\left(\int_{a}^{\zeta_1(\boldsymbol{x})} + \int_{a}^{\zeta_*(\boldsymbol{x})} \right)\frac{(z^2-c)\mathrm{d}z}{\sqrt{ (z^2-a^2)(z^2-b^2) }}-(x_2+\imath x_1)(\zeta_1(\boldsymbol{x})+\zeta_*(\boldsymbol{x}))+\frac{\imath\pi}{2}\\
&=\frac{\imath K}{2}\left(\int_{a}^{\zeta_1(\boldsymbol{x})} + \int_{a}^{-\zeta_1(\boldsymbol{x})} \right)\frac{(z^2-c)\mathrm{d}z}{\sqrt{ (z^2-a^2)(z^2-b^2) }}+\frac{\imath\pi}{2}\\
&=\frac{\imath K}{2}\left(\int_{a}^{\zeta_1(\boldsymbol{x})} - \int_{-a}^{\zeta_1(\boldsymbol{x})} \right)
\frac{(z^2-c)\mathrm{d}z}{\sqrt{ (z^2-a^2)(z^2-b^2) }}+\frac{\imath\pi}{2}\\
&=\frac{\imath K}{2}\int_{k'+\imath k}^{-k'-\imath k}\frac{(z^2-c)\mathrm{d}z}{\sqrt{ (z^2-a^2)(z^2-b^2) }}+\frac{\imath\pi}{2} .
\end{split}\label{proof}
\end{align}
Taking into account that 
\begin{align*}
0&=\frac{\imath K}{2}\oint_{\mathfrak{a}}\frac{(z^2-c)\mathrm{d}z}{\sqrt{ (z^2-a^2)(z^2-b^2) }}=
\imath K\int_{k'+\imath k}^{k'-\imath k }\frac{(z^2-c)\mathrm{d}z}{\sqrt{ (z^2-a^2)(z^2-b^2) }}\\
&=-\imath K\int_{-k'-\imath k}^{-k'+\imath k }\frac{(z^2-c)\mathrm{d}z}{\sqrt{ (z^2-a^2)(z^2-b^2) }}
\end{align*}
one can transform the last integral in (\ref{proof}) into
\begin{align*}
&\frac{\imath K}{2}\int_{k'+\imath k}^{-k'-\imath k}\frac{(z^2-c)\mathrm{d}z}{\sqrt{ (z^2-a^2)(z^2-b^2) }} +
\frac{\imath K}{2}\int_{-k'-\imath k}^{-k'+\imath k}\frac{(z^2-c)\mathrm{d}z}{\sqrt{ (z^2-a^2)(z^2-b^2) }} \\
&=\frac{\imath K}{4}\oint_{\mathfrak{b}}\frac{(z^2-c)\mathrm{d}z}{\sqrt{ (z^2-a^2)(z^2-b^2) }}=\frac{\imath \pi}{2}
\end{align*}
which completes the proof. The second relation in (\ref{mu13mu24}) can be proved similarly.

\subsection{Proof of Lemma \ref{weierstrassthreept}}
\begin{description}

\item[(\ref{relation1})] 
Let us fix for definiteness $i=1,j=2,k=3$. 
To prove (\ref{relation1}) group and factorize the first term from the left-hand side of (\ref{relation1}) with the first term of the right-hand side and then do the same with next pair to get
\begin{align}
\begin{split}
-\theta_1(\alpha_1)\theta_1(\alpha_2)\theta_1(\alpha_1-\alpha_2)
[&\theta_3(\alpha_1)\theta_3(\alpha_2)                 
 \theta_3(\alpha_1+\alpha_2)\theta_3(2\alpha_3)\\
&-\theta_1(\alpha_1-\alpha_3)\theta_1(\alpha_2-\alpha_3)                       
\theta_1(\alpha_1+\alpha_2+\alpha_3)\theta_1(\alpha_3))],\end{split}
\label{term1}
\end{align}
and
\begin{align}\begin{split}
-\theta_3(\alpha_2)\theta_3(\alpha_3)\theta_1(\alpha_2-\alpha_3)[&\theta_1(\alpha_2)\theta_1(\alpha_3)                  \theta_3(\alpha_2+\alpha_3)\theta_3(2\alpha_1)\\-& \theta_1(\alpha_1-\alpha_3)\theta_1(\alpha_1-\alpha_2)                       \theta_3(\alpha_1+\alpha_2+\alpha_3)\theta_3(\alpha_1))].\end{split}
\label{term2}
\end{align}
Using the Weierstrass trisecants (W6)
\begin{align}\begin{split}
&\theta_1(\alpha_1+\alpha_2+\alpha_3)\theta_1(\alpha_1-\alpha_3)\theta_1(\alpha_2-\alpha_3)\theta_1(\alpha_3)\\
&\qquad\qquad
-\theta_3(\alpha_2)\theta_3(\alpha_1)\theta_3(\alpha_1+\alpha_2)\theta_3(2\alpha_3)\\
&\qquad\qquad\qquad\qquad
+\theta_3(\alpha_2+\alpha_3-\alpha_1)\theta_3(\alpha_1+\alpha_3)\theta_3(\alpha_2+\alpha_3)\theta_3(\alpha_3)=0,
\end{split}\label{WW6}
\end{align}
and (W2),
\begin{align}\begin{split}
& \theta_1(\alpha_1-\alpha_2)\theta_1(\alpha_1-\alpha_3)\theta_3(\alpha_1+\alpha_2+\alpha_3)\theta_3(\alpha_1)\\
&\qquad\qquad -\theta_1(\alpha_2)\theta_1(\alpha_3)\theta_3(\alpha_3+\alpha_2)\theta_3(2\alpha_1)\\
&\qquad\qquad\qquad\qquad +\theta_1(\alpha_2+\alpha_3-\alpha_1)\theta_1(\alpha_1)\theta_3(\alpha_1+\alpha_3)
\theta_3(\alpha_1+\alpha_2)=0.
\end{split}\label{WW2}
\end{align}
Correspondingly we  factorise expressions (\ref{term1}) and (\ref{term2}) to the form 
 \begin{align*}
-\theta_1(\alpha_1)\theta_1(\alpha_2)\theta_1(\alpha_1-\alpha_2)\theta_3(\alpha_3)
\theta_3(\alpha_1+\alpha_2-\alpha_3))\theta_3(\alpha_1+\alpha_3)\theta_3(\alpha_2+\alpha_3),\end{align*}
and
\begin{align*}
-\theta_3(\alpha_3)\theta_3(\alpha_2)\theta_1(\alpha_2-\alpha_3)\theta_3(\alpha_2+\alpha_3-\alpha_1)
\theta_1(\alpha_1)\theta_3(\alpha_1+\alpha_3)\theta_3(\alpha_1+\alpha_2).
\end{align*}
Adding to these two expressions with the remaining term from (\ref{relation1}) we observe the vanishing of the overall sum because of the trisecant (W2), 
\begin{align*}
&\theta_1(\alpha_3)\theta_1(\alpha_2-\alpha_3)\theta_3(2\alpha_1)\theta_3(\alpha_2)\\
&\qquad\qquad + \theta_1(\alpha_1)\theta_1(\alpha_1-\alpha_2)\theta_3(\alpha_1+\alpha_3)\theta_3(\alpha_1+\alpha_2-\alpha_3)\\
&\qquad\qquad\qquad\qquad-\theta_1(\alpha_1-\alpha_3)\theta_1(\alpha_1+\alpha_3-\alpha_2)\theta_3(\alpha_1)\theta_3(\alpha_1+\alpha_2)=0.
\end{align*}

\item[(\ref{relation2a})] This follows from the Weierstrass trisecant (W6).

\end{description}

\subsection{Proof of Lemma \ref{Relations}}
\begin{description}
\item[(\ref{sub1a})] 
Let us fix  $i=1$, $j=2$, $k=3$. Then from (\ref{xm}),
\begin{align*}
x_- = \frac{\imath \eta_1}{ (\zeta_1-\zeta_2)(\zeta_1-\zeta_3) }+
 \frac{\imath \eta_2}{ (\zeta_2-\zeta_1)(\zeta_2-\zeta_3) }+ \frac{\imath \eta_3}{ (\zeta_3-\zeta_1)(\zeta_3-\zeta_2) }.
\end{align*}
Substituting the $\theta$-functional expressions (\ref{eqparam}) into 
this and using (\ref{zetadiff}) we may rewrite $x_-$ in the form 
\begin{align}
x_-=\frac{\pi \theta_3(0)}{4} \frac{\theta_3(2\alpha_1)\theta_1(\alpha_2) \theta_1(\alpha_3)
\theta_3(\alpha_2) \theta_3(\alpha_3) }{\theta_{1}(\alpha_1-\alpha_2)\theta_{1}(\alpha_1-\alpha_3)
\theta_{3}(\alpha_1+\alpha_2)\theta_{3}(\alpha_1+\alpha_3)   }+\text{cyclic}.
\label{xm1a}
\end{align}  
Substituting (\ref{xm1a}) into (\ref{sub1a}) we get 
\begin{align*}
&2 \theta_1(\alpha_1+\alpha_2+\alpha_3) \theta_1(\alpha_1)\theta_1(\alpha_2)\theta_1(\alpha_3)\\
& \hskip2cm = - \theta_3(0)\theta_3(\alpha_1+\alpha_2)\theta_3(\alpha_1+\alpha_3)\theta_3(\alpha_2+\alpha_3)\\
&\hskip 1cm + \frac{ \theta_1(2\alpha_1)\theta_3(\alpha_2+\alpha_3) \theta_1(\alpha_2)
 \theta_1(\alpha_3) \theta_3(\alpha_2)
 \theta_3(\alpha_3)  }{ \theta_1(\alpha_1-\alpha_2)\theta_1(\alpha_1-\alpha_3)  }+ \text{cyclic }.
\end{align*}  
Now using Weierstrass trisecant (W6) written in the form
\begin{align}\begin{split}
& -\theta_3(0)  \theta_1(\alpha_1+\alpha_2)
\theta(\alpha_2+\alpha_3)\theta_1(\alpha_1+\alpha_3)\\
&\hskip1cm =\theta_1(\alpha_1+\alpha_2+\alpha_3)\theta_1(\alpha_1)\theta_1(\alpha_2)\theta_1(\alpha_3)
\\&\hskip2cm -\theta_3(\alpha_1+\alpha_2+\alpha_3) \theta_3(\alpha_1)\theta_3(\alpha_2)\theta_3(\alpha_3)
\end{split}\label{WT6}
\end{align}
 in the first term of the right hand side  we get 
\begin{align*}
& \theta_1(\alpha_1+\alpha_2+\alpha_3) \theta_1(\alpha_1)\theta_1(\alpha_2)\theta_1(\alpha_3)\\
&\hskip2cm +\theta_3(\alpha_1+\alpha_2+\alpha_3) \theta_3(\alpha_1)\theta_3(\alpha_2)\theta_3(\alpha_3)\\
& =\frac{ \theta_1(2\alpha_1)\theta_3(\alpha_2+\alpha_3) \theta_1(\alpha_2)
 \theta_1(\alpha_3) \theta_3(\alpha_2)
 \theta_3(\alpha_3)  }{ \theta_1(\alpha_1-\alpha_2)\theta_1(\alpha_1-\alpha_3)  }+ \text{cyclic}.
\end{align*} 
After multiplication of both sides by $\theta_1(\alpha_1-\alpha_2)\theta_1(\alpha_1-\alpha_3) 
\theta_1(\alpha_2-\alpha_3)$
the last relation becomes the already proven relation (\ref{relation1}).

\item[ (\ref{sub1b})] Is proven in the same way as (\ref{sub1a}).

\item[ (\ref{relation2})] This  follows immediately from (\ref{sub1a}) and  (\ref{sub1b}).

\item[ (\ref{sub2a}, \ref{sub2b}, \ref{relation3a})]  This group of relations represent composition of (\ref{sub1a}) and  (\ref{quot1a}) and (\ref{sub1b}) and  (\ref{quot1b}) with the final identity given by subtracting them and using the definition of $\mu_l$.

\item[ (\ref{sub3a}, \ref{sub3b}, \ref{relation4})]  This group of relations represent composition of (\ref{sub1a}) and  (\ref{quot4a}) and (\ref{sub1b}) and  (\ref{quot4b}) with the final identity given by subtracting them.

\item[ (\ref{relation3}), (\ref{xm1})] These follow directly from those just obtained.
\end{description}

\subsection{Proof of Proposition \ref{twoptrelations}}

Let us prove (\ref{2varlogdiff1}). Fix values $ i=1$,  $j=2$. 
Taking the partial derivatives  with respect to $\alpha_1$ and  $\alpha_2$ of both sides of the equality
\begin{equation}
\theta_3(\alpha_1+\alpha_2)\theta_3(\alpha_1-\alpha_2)=
\theta_4^{-1}(0)(\theta_4^2(\alpha_1)\theta_3^2(\alpha_2)- 
\theta_1^2(\alpha_1)\theta_2^2(\alpha_2)     ) 
\end{equation}
and the adding the results we get 
\begin{align}\begin{split}
&\frac{\theta_3'(\alpha_1+\alpha_2)}{ \theta_3(\alpha_1+\alpha_2) }= \frac{1}{\theta_4^2(\alpha_1)\theta_3^2(\alpha_2)- 
\theta_1^2(\alpha_1)\theta_2^2(\alpha_2) } [ \theta_4(\alpha_1)\theta_3^2(\alpha_2)\theta_4'(\alpha_1)\\
&\qquad \qquad  - \theta_1(\alpha_1)\theta_2^2(\alpha_2)\theta_1'(\alpha_1)  
+\theta_3(\alpha_2)\theta_4^2(\alpha_1)\theta_3'(\alpha_2)  
-\theta_2(\alpha_2)\theta_1^2(\alpha_1)\theta'_2(\alpha_2)] .\end{split}\label{logdiff}
\end{align}
Now we find from taking logarithmic derivatives of (\ref{uniform1}) and using Corollary \ref{curveidentb} that
\begin{equation}
\frac{\theta_2'(\alpha)}{\theta_2(\alpha)}= 2\beta_1(\alpha)+\frac{\imath K+2\eta}{\zeta}, \quad  
\frac{\theta_4'(\alpha)}{\theta_4(\alpha)}= 2\beta_1(\alpha)+\frac{-\imath K+2\eta}{\zeta}.
\end{equation}
Substituting these expressions together with (\ref{quot1a}),  (\ref{quot1b}) and the expressions for $\theta$-squares  (\ref{uniform1}) into (\ref{logdiff}) we obtain after simplification
\begin{align}
\frac{\theta_3'(\alpha_1+\alpha_2)}{ \theta_3(\alpha_1+\alpha_2) } 
=2(\beta_1(\alpha_1)+\beta_1(\alpha_2))-\imath K (\zeta_1+\zeta_2) \frac{1-X}{1+X}\label{dtheta3}
\end{align} 
with 
\begin{equation}  X= \frac{ K(\zeta_1^2-1)-2\imath \eta_1 }{ K(\zeta_1^2+1)+2\imath \eta_1}\,\cdot\, \frac{ K(\zeta_2^2-1)-2\imath \eta_2 }{ K(\zeta_2^2+1)+2\imath \eta_2}\left( \frac{k'}{k}\right)^2 = \frac{\theta_1^2(\alpha_1)\theta_1^2(\alpha_2) }
{ \theta_3^2(\alpha_1)\theta_3^2(\alpha_2) }.  \end{equation}
Taking into account expression for $\mu_k$, $\mu_k=\beta_1(\alpha_k)+\imath\pi N-\imath (x_-\zeta_k-\imath x_3)  $, 
we conclude that the proof of (\ref{2varlogdiff1}) will follow upon establishing that 
\begin{equation}
X\equiv \frac{ (\zeta_1+\zeta_2)(K-2x_-)+4\imath x_3   }{ (\zeta_1+\zeta_2)(K+2x_-)-4\imath x_3  }
\label{equivalence}
\end{equation}
where in the expression for $X$ the variables $\eta_i$ should be expressed in terms of $\zeta_i$ via the mini-twistor 
correspondence. 

To proceed, one can find $k^2$ from the relation 
$$ P(\zeta_1 , x_2-\imath x_1-2\zeta_1 x_3 - (x_2+\imath x_1)\zeta_1^2 )- 
P(\zeta_2, x_2-\imath x_1-2\zeta_2 x_3 - (x_2+\imath x_1)\zeta_2^2  ) =0  $$
giving
\begin{align}
k^2&=-\frac14(\zeta_1^2+\zeta_2^2)+\frac12+\frac{x_-^2}{K^2}(\zeta_1^2+\zeta_2^2)
+\frac{2x_+x_-}{K^2} -\frac{4x_3^2}{K^2} \nonumber
\\
&\qquad
-\frac{4\imath x_3}{ K^2(\zeta_1 + \zeta_2 )}
[ x_-( \zeta_1^2+\zeta_1\zeta_2+\zeta_2^2 ) +x_+  ] .\label{k2}
\end{align}
Expression (\ref{equivalence}) factorises after using (\ref{k2}) with one 
of the  factors vanishing because of relation (\ref{x3eq}).   

The proof of (\ref{2varlogdiff2}) parallels that of   (\ref{2varlogdiff1}).
We find an  expression for 
$\theta_1'(\alpha_1+\alpha_2)/\theta(\alpha_1+\alpha_2)$ similarly to (\ref{logdiff}). Next computing 
\begin{align}
\frac{\theta_1'(\alpha_1+\alpha_2)}{\theta_1(\alpha_1+\alpha_2)}
-\frac{\theta_3'(\alpha_1+\alpha_2)}{\theta_3(\alpha_1+\alpha_2)}
\end{align}
and making all of the above substitutions the result follows.


To prove (\ref{2varlogdiff3}) compute the $\alpha_1$  derivative of both sides of (\ref{logdiff}) and use the 
expressions $\theta_i'(\alpha)/\theta_i(\alpha)$,  $\theta_i''(\alpha)/\theta_i(\alpha)$ from the list 
of formulae (\ref{quot1a})-(\ref{quot3b}) together with the formulae
\begin{align}
\begin{split}
\frac{\theta_2''(\alpha)}{\theta_2(\alpha)} = 4\beta_1^2(\alpha)+\frac{8\eta_1\beta_1(\alpha)}{\zeta(\alpha)}
+4 K^2 {k'}^2-4EK-K^2\zeta^2(\alpha) + \frac{4\imath K \beta_1(\alpha)}{\zeta(\alpha)} ,\\
\frac{\theta_4''(\alpha)}{\theta_4(\alpha)} = 4\beta_1^2(\alpha)+
\frac{8\eta_1\beta_1(\alpha)}{\zeta(\alpha)}
+4 K^2 {k'}^2-4EK-K^2\zeta^2(\alpha) - \frac{4\imath K \beta_1(\alpha)}{\zeta(\alpha)},
\end{split}
\end{align}
and those for $\theta$-squares to get algebraic expression of $\eta_{1,2},\zeta_{1,2}, \beta_1(\alpha_{1,2})$ 
and $(x_{\pm},x_3)$, $K,k$  
For the right hand side of  (\ref{2varlogdiff3}) one can transform from the group of variables labeled by indices 3 and 4 to variables labeled by 1 and 2 using the formulae, 
\begin{align}
\zeta_3+\zeta_4&=-\frac{16\imath x_3 x_-}{K^2-4x_-^2}-\zeta_1-\zeta_2,\\
\mu_3+\mu_4&= -\mu_1-\mu_2 \pmod {i\pi}.
\end{align} 
Now upon subtracting  theses expressions for the left and right hand sides of  (\ref{2varlogdiff3}) and
using the expressions for $\mu_j$ and those for $\eta_j$ following from the mini-twistor correspondence
one obtains a rather cumbersome expression that again factorises as in the proof of (\ref{2varlogdiff1}).
Here we find a vanishing factor 
\begin{align}
(\zeta_1+\zeta_2)[ \zeta_1^2\zeta_2^2(K^2-4x_-^2) -(K^2-4 x_+^2)  ] +16  \imath \zeta_1\zeta_2(x_+\zeta_1\zeta_2-x_+  )x_3
\end{align}
so proving  (\ref{2varlogdiff3}).

The final expression (\ref{2varlogdiff4}) is proved analogously.

\subsection{Proof of Proposition \ref{propmux1axis}}\label{proofmux1axis}
For $S_1^2<0$ and the invariance of the curve under conjugation we have that
\begin{align}
\overline{\alpha(P_1)}&=\overline{\int_{\infty_1}^{P_1} v}
=\int_{\overline{\infty_1}}\sp{(\overline{\zeta},\overline{\eta})}v=
\int_{{\infty_2}}\sp{(\overline{\zeta},\overline{\eta})}v=
-\int_{{\infty_1}}\sp{(\overline{\zeta},-\overline{\eta})}v=
-\alpha(P_2)\label{abelx1axis}.
\end{align}
This together with (\ref{defmu}) and the even/oddness properties of the theta functions shows that on interval {\bf I}
$$\overline{\beta_1(P_1)}= -\beta_1(P_2).$$
Accordingly
$$\overline{\mu}_1(x_1,0)=- \mu_2(x_1,0), 
\qquad \overline{\mu}_2(x_1,0)=- \mu_1(x_1,0).$$
Taken together with Proposition (\ref{mux1x2starplane})
 we  obtain upon noting $\zeta_1+\zeta_3=0$ that
\begin{align*}
 \mu_1=\lambda_1(x_1) +\frac{\imath\pi}4, \
  \mu_2=-\lambda_1(x_1) +\frac{\imath\pi}4, \
 \mu_3=-\lambda_1(x_1) -\frac{\imath\pi}4,\
 \mu_4=\lambda_1(x_1) -\frac{\imath\pi}4.
\end{align*}
where $\lambda_1(x_1) $ is a real function. The initial conditions of $\mu_i$ give $\lambda_1(0)=0$.

Given the reality properties noted earlier, we have that on the remaining intervals
\begin{align*}
{\bf II}& 
 &\overline{\zeta}&=\zeta,&  &\overline{\eta}=-\eta,&
 &\overline{\alpha(P_1)}=-\alpha(P_1),&  &&\overline{\beta(P_1)}&=-\beta(P_1),&
&\overline{\gamma_\infty}= -\gamma_\infty,& \overline{\imath x_1\zeta}= -\imath x_1\zeta,\\
{\bf III}& 
 &\overline{\zeta}&=-\zeta,&  &\overline{\eta}=-\eta,&
  &\overline{\alpha(P_1)}=\alpha(P_1),&  &&\overline{\beta(P_1)}&=\beta(P_1),&
 &\overline{\gamma_\infty}= \gamma_\infty,& \overline{\imath x_1\zeta}= \imath x_1\zeta.
\end{align*}
From these it follows that
\begin{align}
{\bf II}&: \overline{\mu}_1(x_1,0)=- \mu_1(x_1,0), 
&{\bf III}&: \overline{\mu}_1(x_1,0)=\mu_1(x_1,0).
\label{conjugate1a}
\end{align} 
Again taken together with Proposition (\ref{mux1x2starplane})
$$\mu_1=\lambda''+\imath\alpha, \  \mu_2=\lambda''-\imath\alpha-\frac{\imath\pi}2, \
 \mu_3=-\lambda''+\imath\alpha+\frac{\imath\pi}2,\
  \mu_4=-\lambda''-\imath\alpha,
$$
the result follows. The remaining boundary conditions now follow.

\subsection{Proof of Proposition \ref{propmux3axis} }\label{proofpropmux3axis}
For $x_3<{Kk'}/2$ the proof follows that of the the $x_2$-axis. 
For  $x_3>{Kk'}/2$ then $P_2=-P_1$ and
\begin{align*}
\mu(-\zeta,-\eta)&=\frac{ \imath K}{2}\int_a^{-\zeta} \frac{z^2-c}{-\eta}dz-x_3
=\frac{ \imath K}{2}\int_{-a}^{\zeta} \frac{z^2-c}{\eta}dz-x_3.
=\mu(\zeta,\eta)+\int_{-a}^a \gamma_\infty.
\intertext{Then}
-\int_{-a}^a \gamma_\infty&=\int_{k'+\imath k}^{k'-\imath k} \gamma_\infty + 
\int_{k'-\imath k}^{-k'-\imath k} \gamma_\infty 
=\frac12 \oint_{\mathfrak{a}} \gamma_\infty+\frac12 \oint_{\mathfrak{b}} \gamma_\infty
=-\imath\frac{\pi}2
\end{align*}
and the result follows.

\subsection{Proof of Proposition \ref{jacobipar}}\label{proofjacobipar}
We note that in all cases the
possible signs are generated by $z\rightarrow z+1$ and $z\rightarrow z+\tau$ and because of
(\ref{transaxes}) we need only solve this for one axis to determine the answer for each axis.
We begin by focussing on the $x_2$ axis and a choice of signs such that
\begin{equation}\label{trisecantlame2}
\theta_2[P]\theta_4[P] \theta_3 \theta_1(z)=
\theta_1[P]\theta_3[P]  \theta_2 \theta_4(z)       +\theta_1[P]\theta_3[P] \theta_4 \theta_2(z).
\end{equation}
Now we have (for any distinct $i,j,k\in\{2,3,4\}$) the trisecant identity (W3)
$$\theta_i(\alpha_1)\theta_j(\alpha_2)\theta_k(\alpha_3)\theta_1(\alpha_4)+
\theta_i(\alpha_1')\theta_j(\alpha_2')\theta_k(\alpha_3')\theta_1(\alpha_4')+
\theta_i(\alpha_1'')\theta_j(\alpha_2'')\theta_k(\alpha_3'')\theta_1(\alpha_4'')=0
$$
where $\boldsymbol{\alpha}=(\alpha_1,\alpha_2,\alpha_3,\alpha_4)$, $\boldsymbol{\alpha'}$ and 
$\boldsymbol{\alpha''}$ are described in Appendix \ref{thetafunctidentapp}. Then if we take
$\boldsymbol{\alpha}=(0,P,-P,2 P)$ and $(i,j,k)=(3,4,2)$ we recover (\ref{trisecantlame2}) with
$z=-2\,\alpha(P)$. 

For the $x_1$ axis
$$
\zeta= \pm \frac{k' \pm \imath k\,  \mathrm{cn}(t)}{\mathrm{dn}(t)}
=\pm\, 
\frac{ \theta_4 \theta_4(z)\pm \imath\, \theta_2 \theta_2(z)}{\theta_3 \theta_3(z)}
 $$
 we now wish to solve
 \begin{equation}\label{trisecantlame1}
\theta_2[P]\theta_4[P] \theta_3 \theta_3(z)=\pm\left[  
 \theta_2 \theta_2(z)\pm \imath\, \theta_4 \theta_4(z) \right] \theta_1[P]\theta_3[P] 
.
\end{equation}
This may be rewritten as
\begin{align*}
\theta_2[P]\theta_4[P] \theta_3 \theta_2(z+\tau/2)&=\pm\left[  
 \theta_2 \theta_3(z+\tau/2)\pm  \theta_4 \theta_1(z+\tau/2) \right] \theta_1[P]\theta_3[P] 
 \intertext{or}
 \theta_3 \theta_3(\alpha(P)+1/2) \theta_2(z+\tau/2) \theta_2[P]
 &=\pm
  \theta_2 \theta_2(\alpha(P)+1/2) \theta_3(z+\tau/2) \theta_3[P] \\
&\qquad  \pm  \theta_4   \theta_4(\alpha(P)+1/2) \theta_1(z+\tau/2) \theta_1[P].
\end{align*}
This is in the form if the trisecant identity 
$$
\theta_2(\alpha_1)\theta_2(\alpha_2)\theta_3(\alpha_3)\theta_3(\alpha_4)
-\theta_2(\alpha_1')\theta_2(\alpha_2')\theta_3(\alpha_3')\theta_3(\alpha_4')\\
\pm\theta_4(\alpha_1'')\theta_4(\alpha_2'')\theta_1(\alpha_3'')\vartheta_1(\alpha_4'')=0.
$$
We find a solution $\boldsymbol{\alpha}=(-2 P-1/2, P, 0, P+1/2)$, that is
$z=-2\,\alpha(P)-1/2-\tau/2$. 

Finally, for the $x_3$ axis we have
 $$
  \zeta= \pm (  \mathrm{dn}(t) \pm \imath k\, \mathrm{sn}(t) )
=\pm\, 
\frac{ \theta_4 \theta_3(z)\pm \imath\, \theta_2 \theta_1(z)}{\theta_3 \theta_4(z)}
 $$
 and we are led to
 \begin{equation}\label{trisecantlame3}
\theta_2[P]\theta_4[P] \theta_3 \theta_4(z)=\pm\left[  
 \theta_2 \theta_1(z)\pm \imath\, \theta_4 \theta_3(z) \right] \theta_1[P]\theta_3[P] .
\end{equation}
For an appropriate set of signs we may rewrite this as
$$
\theta_1(z+\tau/2)\theta_2[P] \theta_3 \theta_4[P] =
\theta_1[P]  \theta_2 \theta_3[P] \theta_4(z+\tau/2)
+ \theta_1[P]  \theta_2(z+\tau/2)   \theta_3[P]\theta_4$$
which we solved earlier, 
$z=-2\,\alpha(P)-\tau/2$.

\section{The matrices \texorpdfstring{$W$}{W} and \texorpdfstring{$V$}{V}}
\subsection{Proof of Theorem \ref{thmadjpsi}}\label{proofthmadjpsi}
The form of $(\ref{formphi})$ shows that its principal cofactors are either linear or bilinear in  the $\zeta$'s.
Now $\Psi \Adj(\Psi)=|\Psi| 1_4$ and the first two columns are bilinear in the $\zeta$'s while the third and fourth are linear. Let us consider
first the linear case. We find for example that
\begin{align*}\Adj(\Psi)_{13}&=
i\theta_{{1}} \left( \alpha_{{2}} \right) \theta_{{4}} \left( \alpha_{{2}}-
z/2 \right) \theta_{{1}} \left( \alpha_{{3}}-\alpha_{{4}} \right) \theta_{{4}}
 \left( \alpha_{{4}}-z/2+\alpha_{{3}} \right)\theta_3(0)\theta_2(z/2) \zeta _{{2}}\\
 & \qquad  -i\theta_{{1}}
 \left( \alpha_{{3}} \right) \theta_{{4}} \left( \alpha_{{3}}-z/2 \right) 
\theta_{{1}} \left( \alpha_{{2}}-\alpha_{{4}} \right) \theta_{{4}} \left( \alpha_{{4}
}-z/2+\alpha_{{2}} \right) \theta_3(0)\theta_2(z/2) \zeta _{{3}} \\
&\qquad +i\theta_{{1}} \left( \alpha_{{4}}
 \right) \theta_{{4}} \left( \alpha_{{4}}-z/2 \right) \theta_{{1}}
 \left( \alpha_{{2}}-\alpha_{{3}} \right) \theta_{{4}} \left( \alpha_{{3}}-z/2+\alpha_{
{2}} \right)\theta_3(0) \theta_2(z/2) \zeta _{{4}}\\
&:=a \zeta _{2} +b \zeta _{3} +c \zeta _{4}
\intertext{where $a+b+c=0$}
&=a (\zeta _{2} -\zeta _{4} )+b(\zeta _{3} -\zeta _{4} )
\intertext{and upon using (\ref{zetadiff}) we find}
&= -
 \frac{\theta_2(\sum_{k\ne 1}\alpha_k-z/2)}{\prod_{k\ne 1}\theta_3(\alpha_k)}
\theta_2^2(z/2) {\theta_2(0)\theta_3(0)
\theta_4(0)}
 \prod_{ \genfrac{}{}{0pt}{4}{k<l} {k,l\ne 1} }\theta_1(\alpha_k-\alpha_l).
\end{align*}
Here we have used the trisecant identity to simplify each of the coefficients of $\zeta_i$ in obtaining the
first line, and we remark that whichever of $a$, $b$ or $c$ we eliminate in the third step we arrive at the same final
expression. We have then for the third and fourth columns of $ \Adj(\Psi)$ that
\begin{align*}
 \Adj(\Psi)_{i3}&=  \frac{\theta_2(\sum_{k\ne i}\alpha_k-z/2)}{\prod_{k\ne i}\theta_3(\alpha_k)}
 \left[(-1)^{i}
\theta_2^2(z/2) {\theta_2(0)\theta_3(0)
\theta_4(0)}
\prod_{ \genfrac{}{}{0pt}{4}{k<l} {k,l\ne 1} }\theta_1(\alpha_k-\alpha_l)\prod_{ \genfrac{}{}{0pt}{4}{k<l} {k,l\ne 1} }\theta_1(\alpha_k-\alpha_l)
\right],
\\
\Adj(\Psi)_{i4}&=
 \frac{\theta_4(\sum_{k\ne i}\alpha_k-z/2)}{\prod_{k\ne i}\theta_1(\alpha_k)}
 \left[(-1)^{i}
\theta_2^2(z/2) {\theta_2(0)\theta_3(0)
\theta_4(0)}
\prod_{ \genfrac{}{}{0pt}{4}{k<l} {k,l\ne 1} }\theta_1(\alpha_k-\alpha_l)
\right],
\end{align*}
and we note that the terms ${\theta_{2,4}(\sum_{k\ne i}\alpha_k-z/2)}$ may be rewritten using
Proposition \ref{abelprop}.

Let us now consider the quadratic terms. Taking for example $\Adj(\Psi)_{11}$ with the same
$a$, $b$, $c$ appearing as in $\Adj(\Psi)_{13}$ we have
\begin{align*}
 \Adj(\Psi)_{11}&=  i c\,\zeta_2 \zeta_3+ i b\,\zeta_2 \zeta_4+i a\,\zeta_3 \zeta_4\\
 &=ia\, \zeta_3( \zeta_4-\zeta_2)+ib\, \zeta_2( \zeta_4-\zeta_3)\\
 &=-i\left[ a (\zeta _{2} -\zeta _{4} )+b(\zeta _{3} -\zeta _{4} )\right] \zeta_3
 +ib\, (\zeta _{3} -\zeta _{4} )(\zeta _{3} -\zeta _{2} )\\
 &=-i\zeta _{3} \Adj(\Psi)_{13}+(\zeta _{3} -\zeta _{4} )(\zeta _{3} -\zeta _{2} )\times \\
 &\qquad
 \theta_1(\alpha_3) \theta_1(\alpha_2-\alpha_4) \theta_4(\alpha_3-z/2) \theta_4(\alpha_2+\alpha_4-z/2)
  \theta_3(0)\theta_2(z/2) 
  \intertext{or grouping factors differently}
  &=-i\zeta _{2} \Adj(\Psi)_{13}+(\zeta _{2} -\zeta _{3} )(\zeta _{2} -\zeta _{4} )\times \\
 &\qquad
 \theta_1(\alpha_2) \theta_1(\alpha_3-\alpha_4) \theta_4(\alpha_2-z/2) \theta_4(\alpha_3+\alpha_4-z/2)
  \theta_3(0)\theta_2(z/2) \\
&=-i\zeta _{4} \Adj(\Psi)_{13}+(\zeta _{4} -\zeta _{2} )(\zeta _{4} -\zeta _{2} )\times \\
 &\qquad
 \theta_1(\alpha_4) \theta_1(\alpha_2-\alpha_3) \theta_4(\alpha_4-z/2) \theta_4(\alpha_2+\alpha_3-z/2)
  \theta_3(0)\theta_2(z/2).
 \end{align*}
 In general, for $j,k,l\ne i$ and $k<l$ we have
 \begin{align*}
 \Adj(\Psi)_{i1}&= -i\zeta _{j} \Adj(\Psi)_{i3}+(\zeta _{j} -\zeta _{k} )(\zeta _{j} -\zeta _{l} )\times
 {\rm Coeff}( \Adj(\Psi)_{i1}, \zeta _{k} \zeta _{l}),
 \\
  \Adj(\Psi)_{i2}&= -i\zeta _{j} \Adj(\Psi)_{i4}+(\zeta _{j} -\zeta _{k} )(\zeta _{j} -\zeta _{l} )\times 
 {\rm Coeff} (\Adj(\Psi)_{i2}, \zeta _{k} \zeta _{l}).
 \end{align*}
 which may be written as
 \begin{align*}
 \Adj(\Psi)_{i1}&= -i\zeta _{j} \Adj(\Psi)_{i3}-\epsilon_{ijkl}(\zeta _{j} -\zeta _{k} )(\zeta _{j} -\zeta _{l} )\times \\
 &\qquad
 \theta_1(\alpha_j) \theta_1(\alpha_k-\alpha_l) \theta_4(\alpha_j-z/2) \theta_4(\alpha_k+\alpha_l-z/2)
  \theta_3(0)\theta_2(z/2), 
  \\
 \Adj(\Psi)_{i2}&= -i\zeta _{j} \Adj(\Psi)_{i4}-\epsilon_{ijkl}(\zeta _{j} -\zeta _{k} )(\zeta _{j} -\zeta _{l} )\times \\
 &\qquad
 \theta_3(\alpha_j) \theta_1(\alpha_k-\alpha_l) \theta_2(\alpha_j-z/2) \theta_4(\alpha_k+\alpha_l-z/2)
  \theta_3(0)\theta_2(z/2) .
\end{align*}

We find upon using (\ref{zetadiff}) and taking the transpose

\subsection{Proof of Theorem \ref{finiteterm}}\label{appendixfiniteterm}
Towards expanding (\ref{defbarvi})  we first note that
\[
\mathcal{O}\, C^{-1}(z)=\diag(F+G,F-G)\, \mathcal{O}=\diag(p(z),1/p(z))\, \mathcal{O}.
\]
Then, from the integral representation of $p(z)$ and that $f_3(z)$ is even, we have $p(z)p(-z)=1$
and so
\begin{align}\label{defvbarz}
\boldsymbol{\overline v}_i(z)&=
\frac1{\theta_2^2(z/2)}\,\frac1{\sqrt{2}}\,
\left( 1_2\otimes \begin{pmatrix} p(z)&0 \\  0& p(-z) \end{pmatrix} \right)
\left[
\begin{pmatrix} -\im \zeta_j  \\   1\end{pmatrix} \otimes \begin{pmatrix} A -B \\  A+ B \end{pmatrix}
+ 
 \begin{pmatrix} 1 \\   0\end{pmatrix} \otimes \begin{pmatrix} \alpha - \beta \\  \alpha+ \beta \end{pmatrix}
\right].
\end{align}

We shall now define
\begin{align*}
\hat A( \xi ):= A(1- \xi) &=\sum_{s\ge1} A_s\, \xi^{s},\quad
\hat B(\xi) := B(1-\xi) =\sum_{s\ge1} B_s\, \xi^{s},\\
\hat\alpha(\xi)  := \alpha(1-\xi) &=\sum_{s\ge0} \alpha_s\, \xi^{s},\quad
\hat\beta(\xi) := \beta(1-\xi) =\sum_{s\ge0} \beta_s\, \xi^{s},\\
\frac1{\theta^2_2((1-\xi)/2)}&=\frac{c}{\xi^2}\sum_{s\ge0}c_{2s}\,\xi^{2s},\qquad
c=\frac{4}{(\theta_1'(0))^2},\  c_0=1,\ c_2=-\frac1{12}\frac{\theta_1'''(0)}{\theta_1'(0) },\\
p(1-\xi)&=\frac1{\sqrt{\xi}} \frac{\sqrt{2}}{\sqrt{K}}
\sum_{s\ge0}p_{2s}\,\xi^{2s},\qquad p_0=1,\ p_2=\frac1{24}(2k^2-1)K^2, \\
p(-1+\xi)&=
\sqrt{\xi} \frac{\sqrt{2}}{\sqrt{K}} 
\sum_{s\ge0}q_{2s}\,\xi^{2s},\qquad q_0=\frac{K}2,\  q_2=-p_2 q_0, 
\end{align*}
where here the expansion of $A$ and $B$ begin with  $\xi$ because of
\begin{equation*}
\theta_2\left( (1-\xi)/2\right)=\theta_1\left( \xi/2\right)=
\frac{\xi}{2}\,\theta_1'(0)+\frac{\xi^3}{48}\,\theta_1'''(0)+\mathcal{O}(\xi^5)
=\xi\,\frac{\pi \,\theta_{{2}} \theta_{{3}} \theta_{{4}}}{2}\,\left( 1-\frac12 c_2\xi^2+ \mathcal{O}(\xi^4)  \right).
\end{equation*}
We then have  that
\begin{equation}\label{vsexpan}
 \boldsymbol{\overline v}_{i, s}=
 \frac{c}{\sqrt{K}}
 \sum_{2l+2m+n=s} c_{2l}\left[
 \begin{pmatrix} -\im \zeta_j  \\   1\end{pmatrix} \otimes
  \begin{pmatrix} p_{2m}(A _n -B_n)\\  
  q_{2m}(A _{n-1}+B_{n-1})\end{pmatrix}
+ 
 \begin{pmatrix} 1 \\   0\end{pmatrix} \otimes \begin{pmatrix} p_{2m}( \alpha_n - \beta_n )
 \\ q_{2m}( \alpha_{n-1} +\beta_{n-1})\end{pmatrix}
 \right].
\end{equation}

In particular we have
\begin{align}
 \boldsymbol{\overline v}_{i, 0}&= \frac{c}{\sqrt{K}}\, 
  \begin{pmatrix} 1 \\   0\end{pmatrix} \otimes \begin{pmatrix} \alpha_0 - \beta_0 \\0 \end{pmatrix}=0,
  \intertext{where we have used that $\alpha_0=\beta_0$  which follows from  (\ref{thetashift2}, \ref{thetashift2}). Making use of this then yields}
 \boldsymbol{\overline v}_{i, 1}&= \frac{c}{\sqrt{K}}\,
 \begin{pmatrix}  \alpha_1-\beta _1-\im\zeta_2(A_1-B_1)  \\ 
\alpha_0 K \\  
 A_1-B_1  \\
0 \\
\end{pmatrix},\\
 \boldsymbol{\overline v}_{i, 2}&= \frac{c}{\sqrt{K}}\, 
\begin{pmatrix}  \alpha_2-\beta _2-\im\zeta_2(A_2-B_2)  \\ 
\left(\alpha_1+\beta_1-\im \zeta_2(A_1+B_1)\right) K/2 \\  
 A_2-B_2  \\
(A_1+B_1)K/2 \\
\end{pmatrix},\\
\boldsymbol{\overline v}_{i, 3}&= \frac{c}{\sqrt{K}}\, 
\begin{pmatrix}  \alpha_3-\beta _3-\im\zeta_2(A_3-B_3)  \\ 
\left(\alpha_2+\beta_2-\im \zeta_2(A_2+B_2)\right) K/2 \\  
 A_3-B_3  \\
(A_2+B_2)K/2 \\
\end{pmatrix}-
2Kp_2\alpha_0\,\frac{c}{\sqrt{K}}\, \begin{pmatrix} 0\\1\\0\\0 \end{pmatrix}
+
(c_2+p_2)  \boldsymbol{\overline v}_{i, 1}.\label{finitepiece}
\end{align}

\begin{align*}
A_1&=\frac12\,{\frac {\pi \,\theta_{{2}}\theta_{{3}}  \theta_{{4}}\, \theta_{{1}}\left( P_{{4}}+P_{
{3}}+P_{{2}} \right)  }{\theta_{{3}} \left( P_{{3
}} \right) \theta_{{3}} \left( P_{{2}} \right) \theta_{{3}} \left( P_{
{4}} \right) }}\,  {{\rm e}^{-\mu_{{1}}}},&
A_2&=\left( \frac12
 \frac{ \theta_{{1}}' \left( P_{{4}}+P_{{3}}+P_{{2}} \right) }{\theta_{{1}} \left( P_{{4}}+P_{{3}}+P_{{2}} \right) }+\mu_1\right) A_1,\\
 B_1&=\frac12\,{\frac {\pi \,\theta_{{2}}\theta_{{3}}  \theta_{{4}}\, \theta_{{3}}\left( P_{{4}}+P_{
{3}}+P_{{2}} \right)  }{\theta_{{1}} \left( P_{{3
}} \right) \theta_{{1}} \left( P_{{2}} \right) \theta_{{1}} \left( P_{
{4}} \right) }}\,  {{\rm e}^{-\mu_{{1}}}},&
B_2&=\left( \frac12
 \frac{ \theta_{{3}}' \left( P_{{4}}+P_{{3}}+P_{{2}} \right) }{\theta_{{3}} \left( P_{{4}}+P_{{3}}+P_{{2}} \right) }+\mu_1\right) B_1,\\
 \alpha_0&=
 {\frac {\theta_{{2}}  \theta_{{4}}
\theta_{{3}} \left( P_{{4}}+P_{{3}} \right) \theta_{{3}} \left( P_{{3}
}+P_{{2}} \right) \theta_{{3}} \left( P_{{4}}+P_{{2}} \right)  }{\theta_{{1}} \left( P_{{2}} \right) \theta_{{3}}
 \left( P_{{2}} \right) \theta_{{1}} \left( P_{{3}} \right) \theta_{{3
}} \left( P_{{3}} \right) \theta_{{1}} \left( P_{{4}} \right) \theta_{
{3}} \left( P_{{4}} \right) }}\,  {{\rm e}^{-\mu_{{1}}}},&
\alpha_1&=\left( \frac12
 \frac{ \theta_{{3}}' \left( P_{{4}}+P_{{3}} \right) }{\theta_{{3}} \left( P_{{4}}+P_{{3}} \right) }+
 \frac12
 \frac{ \theta_{{3}}' \left( P_{{2}} \right) }{\theta_{{3}} \left(P_{{2}} \right) }+\mu_1\right)  \alpha_0,
\\
\beta_0&=\alpha_0,   &
\beta_1&=\left( \frac12
 \frac{ \theta_{{3}}' \left( P_{{4}}+P_{{3}} \right) }{\theta_{{3}} \left( P_{{4}}+P_{{3}} \right) }+
 \frac12
 \frac{ \theta_{{1}}' \left( P_{{2}} \right) }{\theta_{{1}} \left(P_{{2}} \right) }+\mu_1\right)  \alpha_0,
\end{align*}

 \begin{align*}
 A_3&= \left( \frac18
 \frac{ \theta_{{1}}'' \left( P_{{4}}+P_{{3}}+P_{{2}} \right) }{\theta_{{1}} \left( P_{{4}}+P_{{3}}+P_{{2}} \right) }+\frac{\mu_1 }2 \frac{ \theta_{{1}}'\left( P_{{4}}+P_{{3}}+P_{{2}} \right) }{\theta_{{1}} \left( P_{{4}}+P_{{3}}+P_{{2}} \right) }+ \frac{\mu_1^2 }2 -\frac12 c_2
 \right) A_1,\\
 B_3&= \left( \frac18
 \frac{ \theta_{{3}}'' \left( P_{{4}}+P_{{3}}+P_{{2}} \right) }{\theta_{{3}} \left( P_{{4}}+P_{{3}}+P_{{2}} \right) }+\frac{\mu_1 }2 \frac{ \theta_{{3}}'\left( P_{{4}}+P_{{3}}+P_{{2}} \right) }{\theta_{{3}} \left( P_{{4}}+P_{{3}}+P_{{2}} \right) }+ \frac{\mu_1^2 }2 -\frac12 c_2
 \right) B_1,\\
 \alpha_2&=\left( \frac18
 \frac{ \theta_{{3}}''\left( P_{{4}}+P_{{3}} \right) }{\theta_{{3}} \left( P_{{4}}+P_{{3}} \right) }+
 \frac18
 \frac{ \theta_{{3}}'' \left( P_{{2}} \right) }{\theta_{{3}} \left(P_{{2}} \right) }
 +
 \frac14
 \frac{ \theta_{{3}}'\left( P_{{4}}+P_{{3}} \right) }{\theta_{{3}} \left( P_{{4}}+P_{{3}} \right) }
 \,
 \frac{ \theta_{{3}}' \left( P_{{2}} \right) }{\theta_{{3}} \left(P_{{2}} \right) }
 +
 \frac12\left(
 \frac{ \theta_{{3}}' \left( P_{{4}}+P_{{3}} \right) }{\theta_{{3}} \left( P_{{4}}+P_{{3}} \right) }+
 \frac{ \theta_{{3}}' \left( P_{{2}} \right) }{\theta_{{3}} \left(P_{{2}} \right) }
 \right) \mu_1
 + \frac{\mu_1^2 }2
 \right)  \alpha_0,
\\
\beta_2&=\left( \frac18
 \frac{ \theta_{{3}}''\left( P_{{4}}+P_{{3}} \right) }{\theta_{{3}} \left( P_{{4}}+P_{{3}} \right) }+
 \frac18
 \frac{ \theta_{{1}}'' \left( P_{{2}} \right) }{\theta_{{1}} \left(P_{{2}} \right) }
 +
 \frac14
 \frac{ \theta_{{3}}'\left( P_{{4}}+P_{{3}} \right) }{\theta_{{3}} \left( P_{{4}}+P_{{3}} \right) }
 \,
 \frac{ \theta_{{1}}' \left( P_{{2}} \right) }{\theta_{{1}} \left(P_{{2}} \right) }
 +
 \frac12\left(
 \frac{ \theta_{{3}}' \left( P_{{4}}+P_{{3}} \right) }{\theta_{{3}} \left( P_{{4}}+P_{{3}} \right) }+
 \frac{ \theta_{{1}}' \left( P_{{2}} \right) }{\theta_{{1}} \left(P_{{2}} \right) }
 \right) \mu_1
 + \frac{\mu_1^2 }2
 \right)  \alpha_0,\\
  \alpha_3&=\left[
 \frac18  \left( 
 \frac{ \theta_{{3}}'''\left( P_{{4}}+P_{{3}} \right) }{\theta_{{3}} \left( P_{{4}}+P_{{3}} \right) }+
 3\, \frac{ \theta_{{3}}''\left( P_{{4}}+P_{{3}} \right) }{\theta_{{3}} \left( P_{{4}}+P_{{3}} \right) }
  \frac{ \theta_{{3}}' \left( P_{{2}} \right) }{\theta_{{3}} \left(P_{{2}} \right) }
+ 
3 \, \frac{\theta_{{3}}'\left( P_{{4}}+P_{{3}} \right) }{\theta_{{3}} \left( P_{{4}}+P_{{3}} \right) }
  \frac{ \theta_{{3}}'' \left( P_{{2}} \right) }{\theta_{{3}} \left(P_{{2}} \right) }
+
 \frac{ \theta_{{3}}''' \left( P_{{2}} \right) }{\theta_{{3}} \left(P_{{2}} \right) }
 \right) 
 \right. \\
 &\qquad +
 \frac{3\mu_1}4\left(
  \frac{\theta_{{3}}''\left( P_{{4}}+P_{{3}} \right) }{\theta_{{3}} \left( P_{{4}}+P_{{3}} \right) }+
2\, \frac{ \theta_{{3}}' \left( P_{{4}}+P_{{3}} \right) }{\theta_{{3}} \left( P_{{4}}+P_{{3}} \right) }
 \frac{ \theta_{{3}}' \left( P_{{2}} \right) }{\theta_{{3}} \left(P_{{2}} \right) }+ 
 \frac{ \theta_{{3}}'' \left( P_{{2}} \right) }{\theta_{{3}} \left(P_{{2}} \right) }
  \right) \\
  &\qquad
 +
  \frac{3\mu^2_1}2\left(
   \frac{\theta_{{3}}'\left( P_{{4}}+P_{{3}} \right) }{\theta_{{3}} \left( P_{{4}}+P_{{3}} \right) }+
   \frac{ \theta_{{3}}' \left( P_{{2}} \right) }{\theta_{{3}} \left(P_{{2}} \right) }
   \right) 
 + \mu_1^3 \Bigg]
 \frac{\alpha_0}6, \\
  \beta_3&=\left[
 \frac18  \left( 
 \frac{ \theta_{{3}}'''\left( P_{{4}}+P_{{3}} \right) }{\theta_{{3}} \left( P_{{4}}+P_{{3}} \right) }+
 3\, \frac{ \theta_{{3}}''\left( P_{{4}}+P_{{3}} \right) }{\theta_{{3}} \left( P_{{4}}+P_{{3}} \right) }
  \frac{ \theta_{{1}}' \left( P_{{2}} \right) }{\theta_{{1}} \left(P_{{2}} \right) }
+ 
3 \, \frac{\theta_{{3}}'\left( P_{{4}}+P_{{3}} \right) }{\theta_{{3}} \left( P_{{4}}+P_{{3}} \right) }
  \frac{ \theta_{{1}}'' \left( P_{{2}} \right) }{\theta_{{1}} \left(P_{{2}} \right) }
+
 \frac{ \theta_{{1}}''' \left( P_{{2}} \right) }{\theta_{{1}} \left(P_{{2}} \right) }
 \right) 
 \right. \\
 &\qquad +
 \frac{3\mu_1}4\left(
  \frac{\theta_{{3}}''\left( P_{{4}}+P_{{3}} \right) }{\theta_{{3}} \left( P_{{4}}+P_{{3}} \right) }+
2\, \frac{ \theta_{{3}}' \left( P_{{4}}+P_{{3}} \right) }{\theta_{{3}} \left( P_{{4}}+P_{{3}} \right) }
 \frac{ \theta_{{1}}' \left( P_{{2}} \right) }{\theta_{{1}} \left(P_{{2}} \right) }+ 
 \frac{ \theta_{{1}}'' \left( P_{{2}} \right) }{\theta_{{1}} \left(P_{{2}} \right) }
  \right) \\
  &\qquad
 +
  \frac{3\mu^2_1}2\left(
   \frac{\theta_{{3}}'\left( P_{{4}}+P_{{3}} \right) }{\theta_{{3}} \left( P_{{4}}+P_{{3}} \right) }+
   \frac{ \theta_{{1}}' \left( P_{{2}} \right) }{\theta_{{1}} \left(P_{{2}} \right) }
   \right) 
 + \mu_1^3 \Bigg]
 \frac{\alpha_0}6, 
\end{align*}

Now
\begin{align*}
 A_1-B_1 &=
 \frac12\,{
 \frac {\pi \,\theta_{{2}}\theta_{{3}}  \theta_{{4}} 
 \left(
 \theta_{{1}}\left( P_{{4}}+P_{{3}}+P_{{2}} \right) {\theta_{{1}} \left( P_{{3}} \right) \theta_{{1}} \left( P_{{2}} \right) \theta_{{1}} \left( P_{{4}} \right) }
 -
  \theta_{{3}}\left( P_{{4}}+P_{{3}}+P_{{2}} \right) {\theta_{{3}} \left( P_{{3}} \right) \theta_{{3}} \left( P_{{2}} \right) \theta_{{3}} \left( P_{{4}} \right) }
\right)
         }
        {
\theta_{{1}} \left( P_{{2}} \right) \theta_{{3}} \left( P_{{2}} \right) \theta_{{1}} \left( P_{{3}} \right) \theta_{{3
}} \left( P_{{3}} \right) \theta_{{1}} \left( P_{{4}} \right) \theta_{{3}} \left( P_{{4}} \right)
        }
         }
\,  {{\rm e}^{-\mu_{{1}}}},\\
&=-K\, \alpha_0,
\end{align*}
where we have employed the Weierstrass trisecant identity (W6)  which says for any $\alpha_{1,2,3}$ that
\[
\begin{split}
&\theta_1(\alpha_1+\alpha_2+\alpha_3)\theta_1(\alpha_1)\theta_1(\alpha_2)\theta_1(\alpha_3)-\theta_3(\alpha_1+\alpha_2+\alpha_3) \theta_3(\alpha_1)\theta_3(\alpha_2)\theta_3(\alpha_3) \\
&\qquad\qquad= -\theta_3(0)  \theta_3(\alpha_1+\alpha_2)
\theta_3(\alpha_2+\alpha_3)\theta_3(\alpha_1+\alpha_3)
\end{split}
\]
and $K=\pi \theta_3^2/2$.
Further
\[
 \alpha_1-\beta _1-\im\zeta_2 (A_1-B_1) =
 \left( 
 \frac12 \frac{ \theta_{{3}}' \left( P_{{2}} \right) }{\theta_{{3}} \left(P_{{2}} \right) }-
  \frac12 \frac{ \theta_{{1}}' \left( P_{{2}} \right) }{\theta_{{1}} \left(P_{{2}} \right) }
  +\im\zeta_2 K
\right)  \alpha_0=0
\]
upon making use of (\ref{curveident1}) and we then have the leading order pole term of  $\boldsymbol{\overline v}$ at $z=1-\xi$ behaving as
\[
\frac{1}{\xi^{3/2}}  {c\,\alpha_0 }{\sqrt{K}}\begin{pmatrix} 0\\1\\-1\\ 0\end{pmatrix}.
\]

Now from (\ref{relation2}) we see that
\[
\frac{A_1+B_1}{2\alpha_0}=x_1 -\im x_2.
\]
Writing
\begin{align*}
\frac{\alpha_1+\beta_1-\im \zeta_2(A_1+B_1)}{2\alpha_0}&=
 \frac12 \left(
 \frac{ \theta_{{3}}' \left( P_{{4}}+P_{{3}} \right) }{\theta_{{3}} \left( P_{{4}}+P_{{3}} \right) }+
\frac12  \frac{ \theta_{{3}}' \left( P_{{2}} \right) }{\theta_{{3}} \left(P_{{2}} \right) }
 +
\frac12 \frac{ \theta_{{1}}' \left( P_{{2}} \right) }{\theta_{{1}} \left(P_{{2}} \right) }
 -\frac12
 \frac{ \theta_{{3}}' \left( P_{{2}} \right) }{\theta_{{3}} \left(P_{{2}} \right) }+
 2\mu_1\right)  \\ \\
 &\qquad  -\im \zeta_2[ x_1 -\im x_2]  
 \intertext{and making use of (\ref{quot1a},  \ref{quot1b}, \ref{musum}) yields}
 &= \frac12 \left(
\left[ -\im K +2(x_2+\im x_1)\right]\zeta_2 -2\mu_1+2 x_3 +\im K \zeta_2 +2 \mu_1 \right)
 -\im \zeta_2[ x_1 -\im x_2]  \\
 &=x_3.
 \end{align*}
Similarly making use of (\ref{sub1a}, \ref{sub1b}, \ref{sub2a}, \ref{sub2b}) we find
\[A_2-B_2=\alpha_0 K x_3,
\]

\begin{align*}
\frac{\alpha_2-\beta_2-\im \zeta_2(A_2-B_2)}{\alpha_0}&=\frac{\alpha_2-\beta_2}{\alpha_0}
-\im \zeta_2 K x_3\\
&=
\left( 
 \frac18 \frac{ \theta_{{3}}'' \left( P_{{2}} \right) }{\theta_{{3}} \left(P_{{2}} \right) }
 -
  \frac18 \frac{ \theta_{{1}}'' \left( P_{{2}} \right) }{\theta_{{1}} \left(P_{{2}} \right) }
 +
 \frac14\left[ 
 \frac{ \theta_{{3}}'\left( P_{{4}}+P_{{3}} \right) }{\theta_{{3}} \left( P_{{4}}+P_{{3}} \right) }
 + 2\mu_1 \right] 
 \left[
 \frac{ \theta_{{3}}' \left( P_{{2}} \right) }{\theta_{{3}} \left(P_{{2}} \right) }
 -
 \frac{ \theta_{{1}}' \left( P_{{2}} \right) }{\theta_{{1}} \left(P_{{2}} \right) } 
 \right] 
 \right)\\
 &\qquad\qquad -\im \zeta_2 K x_3
 \intertext{and using ( \ref{quot1a},  \ref{quot1b}, \ref{quot3a}, \ref{quot3b},\ref{musum}) we obtain }
 &=
 -\im K\left( \beta(P_2) \zeta_2 +\eta_2\right)-\frac{\im K}2 \zeta_2\left(
 2(x_2+\im x_1)\zeta_2 +2 x_3 -2\beta(P_2)
 \right)
 -\im \zeta_2 K x_3\\
 &=-\im K \left(  2 x_3\zeta_2 +[ x_2+\im x_1]\zeta_2^2 +\eta_2\right)
 \intertext{and upon making use of the mini-twistor constraint simplifies to}
 &=-\im K (x_2-\im x_1)\\
 &=K (-x_1 -\im x_2).
\end{align*}
Thus we obtain the subleading pole 
\[
\frac{1}{\xi^{1/2}}  {c\,\alpha_0 }{\sqrt{K}}\begin{pmatrix} -x_1 -\im x_2 \\x_3\\ x_3 \\ x_1 -\im x_2\end{pmatrix}
\]
which agrees with (the complex conjugate of) (\ref{asymvp}). At this stage we have shown that the
first column has an expansion
\[
c \sqrt {K}\,{\frac {\theta_{{2}}  \theta_{{4}}  \theta_{{3}} \left( P_{{3}}+P_{{2}} \right) \theta_{{3}}
 \left( P_{{4}}+P_{{2}} \right) \theta_{{3}} \left( P_{{4}}+P_{{3}}
 \right) }{\theta_{{1}} \left( P_{{4}} \right) 
\theta_{{3}} \left( P_{{4}} \right) \theta_{{1}} \left( P_{{2}}
 \right) \theta_{{3}} \left( P_{{2}} \right) \theta_{{1}} \left( P_{{3
}} \right) \theta_{{3}} \left( P_{{3}} \right) }}\,  {{\rm e}^{-\mu_{{1}}}}\,
\left(   
\frac{1}{\xi^{3/2}} \begin{pmatrix} 0\\1\\-1\\ 0\end{pmatrix}+
\frac{1}{\xi^{1/2}} \begin{pmatrix} -x_1 -\im x_2 \\x_3\\ x_3 \\ x_1 -\im x_2\end{pmatrix}+\mathcal{O}
(\xi^{1/2})
\right)
\]
and analogously each column has expansion at $z=1-\xi$
\begin{equation}\label{barviexp1}
\boldsymbol{\overline v}_i =
N_i \left(   
\frac{1}{\xi^{3/2}} \begin{pmatrix} 0\\1\\-1\\ 0\end{pmatrix}+
\frac{1}{\xi^{1/2}} \begin{pmatrix} -x_1 -\im x_2 \\x_3\\ x_3 \\ x_1 -\im x_2\end{pmatrix}+\mathcal{O}
(\xi^{1/2})
\right)
\end{equation}
where
\[
N_i:= c \sqrt {K}\,\theta_{{2}}  \theta_{{4}}\, \frac{\prod_{j<k \  j,k\ne i}\theta_3(P_j+P_k)}
{\prod_{r\ne i}\theta_1(P_r)\theta_3(P_r)}\, {{\rm e}^{-\mu_{{i}}}}.
\]

The remaining terms of the theorem follow from (\ref{finitepiece}) and the relations of Lemma \ref{Relations}
and Lemma \ref{twoptrelations}.
Thus for example the fourth entry of $\boldsymbol{\overline v}_i$ uses (\ref{relation2}) and the second entry
(\ref{relation3}).

\subsection{Proof of Theorem \ref{smonodromy}}\label{proofsmonodromy}
We have seen that
\begin{align*}
\boldsymbol{\overline v}_i(1-\xi)&=
\frac1{\theta_1^2(\xi/2)}\,\frac1{\sqrt{2}}\,
\left( 1_2\otimes \begin{pmatrix} p(1-\xi)&0 \\  0& p(-1+\xi) \end{pmatrix} \right)\,
\lambda_i(\xi)
\intertext{where we now define}
\lambda_i(\xi)&=\begin{pmatrix} -\im \zeta_j  \\   1\end{pmatrix} \otimes 
\begin{pmatrix} \hat A(\xi)-\hat B(\xi)\\  \hat A(\xi)+\hat B(\xi)\ \end{pmatrix}
+ 
 \begin{pmatrix} 1 \\   0\end{pmatrix} \otimes 
 \begin{pmatrix} \hat\alpha(\xi)-\hat\beta(\xi)  \\  \hat\alpha(\xi)+\hat\beta(\xi)   \end{pmatrix}
.
\end{align*}
Next we easily obtain
\begin{lemma}
\begin{align*}
 A(-1+\xi)&=\hat A( -\xi )e\sp{2\mu_i},\quad
B(-1+\xi)=-\hat B(-\xi)e\sp{2\mu_i},\\
\alpha(-1+\xi)&=\hat\alpha(-\xi)e\sp{2\mu_i},\quad
\beta(-1+\xi)=-\hat\beta(-\xi)e\sp{2\mu_i},
\end{align*}
\end{lemma}
Then from (\ref{defvbarz}) and the previous lemma,
\begin{align*}
\boldsymbol{\overline v}_i(-1+\xi)
&=
\frac1{\theta_1^2(\xi/2)}\,\frac1{\sqrt{2}}\,
\left( 1_2\otimes \begin{pmatrix} p(-1+\xi)&0 \\  0& p(1-\xi) \end{pmatrix} \right)
\left(1_2\otimes \begin{pmatrix}0&1\\ 1& 0\end{pmatrix}\right)\, \lambda_i(-\xi) e\sp{2\mu_i}\\
&=
\left( 1_2\otimes \begin{pmatrix}0&1\\ 1& 0\end{pmatrix} e\sp{2\mu_i}\right)
\frac1{\theta_1^2(\xi/2)}\,\frac1{\sqrt{2}}\,
\left( 1_2\otimes \begin{pmatrix} p(1-\xi)&0 \\  0& p(-1+\xi) \end{pmatrix} \right)
\, \lambda_i(-\xi) \\
&=
\left( 1_2\otimes \begin{pmatrix}0&\frac{p(1+\xi)}{ p(1-\xi)}\\ 
\frac{p(1-\xi)}{ p(1+\xi)}& 0\end{pmatrix} e\sp{2\mu_i}\right)
\frac1{\theta_1^2(\xi/2)}\,\frac1{\sqrt{2}}\,
\left( 1_2\otimes \begin{pmatrix} p(1+\xi)&0 \\  0& p(-1-\xi) \end{pmatrix} \right)
\, \lambda_i(-\xi) \\
\intertext{where we have used $p(z)p(-z)=1$. Upon comparing this with}
\boldsymbol{\overline v}_i(1-\xi)&=
\frac1{\theta_1^2(\xi/2)}\,\frac1{\sqrt{2}}\,
\left( 1_2\otimes \begin{pmatrix} p(1-\xi)&0 \\  0& p(-1+\xi) \end{pmatrix} \right)
\lambda_i(\xi)= \sum_{s\ge 0} \boldsymbol{\overline v}_{i, s}\, \xi^{s-5/2}\\
\intertext{we obtain}
\boldsymbol{\overline v}_i(-1+\xi)&=- \im 
\left( 1_2\otimes \begin{pmatrix}0&\frac{p(1+\xi)}{ p(1-\xi)}\\ 
\frac{p(1-\xi)}{ p(1+\xi)}& 0\end{pmatrix} e\sp{2\mu_i}\right)
\sum_{s\ge 0} (-1)^s\,\boldsymbol{\overline v}_{i, s}\, \xi^{s-5/2}.
\end{align*}
Here we use the definition of $p(z)$ and the periodicity of the theta functions to see that
$p^2(1-\xi) =- p^2(1+\xi)$ to give that $p(1+\xi) =\pm \im p(1-\xi)$ yielding
\begin{equation}\label{monodromy}
\boldsymbol{\overline v}_i(-1+\xi)=\pm
\left( 1_2\otimes \begin{pmatrix}0&1\\ -1& 0\end{pmatrix} e\sp{2\mu_i}\right)
\sum_{s\ge 0} (-1)^s\,\boldsymbol{\overline v}_{i, s}\, \xi^{s-5/2}.
\end{equation}
Finally we may use continuity and the explicit formula (\ref{defvbarz}) to determine the overall sign,
which is found to be $1$, so establishing the theorem.

\subsection{Proof of Proposition \ref{Dream}}\label{proofDream}
 Write the (2,1), (3,1) and (4,1)  matrix elements of the matrix equation (\ref{ort1}) which are linear equations with respect to $a_1,b_1,c_1$. Solving these via Kramer's rule we find quantities $a_1,b_1,c_1$ in the form of symmetric functions of $\zeta_{2,3,4}$. Thus, for example,
\begin{equation}
a_1= -\frac{\imath S_+^2}{8} \frac{\zeta_2\zeta_3+\zeta_2\zeta_4+\zeta_3\zeta_4}{\zeta_2\zeta_3\zeta_4} 
+ \frac{\imath}{8}\zeta_2\zeta_3\zeta_4 S_-^2 +x_+x_3 .
\end{equation}
Now from the mini-twistor constraint,  
\begin{align}\begin{split}
&\zeta_1\zeta_2\zeta_3\zeta_4= \frac{K^2-4x_+^2}{K^2-4 x_-^2} ,\\
&\zeta_1+\zeta_2+\zeta_3+\zeta_4=-\frac{16 \imath x_3 x_-}{K^2-4 x_-^2} ,\\
&\zeta_1(\zeta_2+\zeta_3+\zeta_4)+\zeta_2\zeta_3+\zeta_2\zeta_4+\zeta_3\zeta_4=\frac{-8x_+x_-
+16x_3^2+2K^2(1-2{k'}^2)}{K^2-4x_-^2}.\end{split} \label{symmetric}
\end{align}
Using these we may solve to give $a_1$ as given by (\ref{abc}); we similarly obtain $b_1$ and $c_1$.
Now the (1,1) diagonal entry is
\begin{align}
\frac{\imath}{8}S_-^2\zeta_j^3-(c_1-x_-x_3)\zeta_1^2+2\imath b_1 \zeta_j-x_+x_3+a_1+\frac{\imath}{8}\frac{S_+^2}{\zeta_j}=\mathfrak{D}_{1,1}\delta_{1,j} .\label{eqns}
\end{align}
Substituting the expressions for  $a_1$, $b_1$ and $c_1$ leads to (\ref{frakD}).
The same arguments work for the remaining  entries to $\widetilde{\boldsymbol{\bar{\mathfrak{v}}}}_2$. The  mini-twistor constraint is used at all stages of the derivation.

\subsection{Proof of Theorem \ref{constantmatrix}}\label{proofconstantmatrix}
The matrix is skew-hermitian and the block structure is preserved by left and right multiplication by diagonal matrices, so it suffices to show that 
$$ ( {{V}}\sp{\dagger}\,\mathcal{Q}^{-1}\mathcal{H}\,{{V} })\sp{-1}=
{{W}}\sp{\dagger}\,\mathcal{H}^{-1}\mathcal{Q}\,{{W} }= -\frac1{r^2} \,
{{W}}\sp{\dagger}\,\mathcal{Q}\mathcal{H}\,{{W} }
$$
has the desired structure. The constancy of the matrix enables us to choose any convenient $z$ to evaluate this; we will choose $z=0$ where $C(0)=1_2$. Then (\ref{eqparam}, \ref{uniform1}, \ref{formW}) give
\begin{align*}
W_k&= \begin{pmatrix}1\\ i\zeta_k\end{pmatrix}\otimes\mathcal{O}
\begin{pmatrix}-\im \dfrac{K(\zeta_k ^2+1)+2\im \eta_k}{2 K k' \zeta_k} \\  1\end{pmatrix} d_k ,\\
\overline{W_k}&= \begin{pmatrix}\zeta_{\mathcal{J}(k)  }\\ i \end{pmatrix}\otimes\mathcal{O}
\begin{pmatrix}-\im \dfrac{K(\zeta_{\mathcal{J}(k)  } ^2+1)+2\im \eta_{\mathcal{J}(k)  }}{2 K k' \zeta_{\mathcal{J}(k)  }} \\  1\end{pmatrix} d_{\mathcal{J}(k)  }' ,
\end{align*}
for appropriate nonzero $d_k$, $d_k'$. Now 
$$(\mathcal{Q}\mathcal{H})(0)=
\begin{pmatrix}
0&iKx_{{2}}&-K{k'}\, \left( -x_{{1}}+
ix_{{2}} \right) &-x_{{3}}K\\ iKx_{{2}}&0&-x_{{3}}K&
K{k'}\, \left( -x_{{1}}+ix_{{2}} \right) \\ -K{
k'}\, \left( ix_{{2}}+x_{{1}} \right) &x_{{3}}K&0&-iKx_{{2}}
\\ x_{{3}}K&K{k'}\, \left( ix_{{2}}+x_{{1}}
 \right) &-iKx_{{2}}&0\end{pmatrix}.
$$
Substitution of these into ${{W}}\sp{\dagger}\,\mathcal{H}^{-1}\mathcal{Q}\,{{W} }$ and using
(\ref{AtiyahWard}) yields $(i,j)$-matrix entries which, for $j\ne \mathcal{J}(i)$ have the form 
poly$(\zeta_i, \zeta_j)/\zeta_i\zeta_j$ and this polynomial is in the ideal generated by each of the
quartics that  $\zeta_{i,j}$ individually satisfy.\footnote{An elementary way to verify this is as follows.
Let $q_i$ be the quartic that $\zeta_i$ satisfies. Then the resultant of poly$(\zeta_i, \zeta_j)$ and
$q_j$ with respect to $\zeta_j$ is a polynomial in $\zeta_i$ with a factor (amongst others) of $q_i^3$.}

The nonzero elements of the matrix (\ref{constancystructure}) $\mathfrak{f}_j$ $j=1,\ldots,4$ may be represented in the $\theta$ function form (\ref{constancystructuref}) by
using duplication $\theta$-formulae, formulae (7.8) 
representing $\theta$-quotients in term of coordinates of the curve
together with the relations (for all permutations of $\alpha_j$)
\begin{align}
&\theta_1^2(\alpha_1)\theta_3^2(\alpha_1)\frac{ \theta_3(\alpha_2+\alpha_3) \theta_3(\alpha_2+\alpha_4) \theta_3(\alpha_3+\alpha_4) }{ 
 \theta_1(\alpha_2-\alpha_1) \theta_1(\alpha_3-\alpha_1)\theta_1(\alpha_4-\alpha_1) }=\frac{4\zeta_1 (kk')^{3/2} K^4 \mathrm{e}^{-\imath \pi N^2\tau}}{\pi^2 (4x_-x_3\zeta_1^3+\imath
R\zeta_1^2 +12 x_+ x_3 \zeta_1-\imath S_+^2 )} .
\end{align}

\subsection{Proof of Lemma \ref{energydensitylemma}}\label{proofenergydensitylemma}
We want to compute 
\begin{align*}
-4\mathcal{E}(\boldsymbol{x})=
\left\{\frac{\partial^2}{\partial x_1^2}+\frac{\partial^2}{\partial x_2^2}+\frac{\partial^2}{\partial x_3^2} \right\}\Tr
\left(  \mathcal{H}.\mathcal{G}^{-1}. \mathcal{H}.\mathcal{G}^{-1} \right).
\end{align*}
Recall we have defined
\begin{align*}
\mathcal{G}_{1,i} &= \mathcal{G}^{-1}\,. \frac{\partial \mathcal{G}}{\partial x_i} \,. \mathcal{G}^{-1},\quad \mathcal{G}_{2,i} = \mathcal{G}^{-1}\,.
 \frac{\partial^2 \mathcal{G}}{\partial x_i^2} \,. \mathcal{G}^{-1}.
\end{align*}
First observe that 
\begin{align*}
\frac{\partial}{\partial x_i} \left( \mathcal{H}.\mathcal{G}^{-1}\right)&=\frac{\partial \mathcal{H}}{\partial x_i} . \mathcal{G}^{-1}
+\mathcal{H} . \frac{\partial  \mathcal{G}^{-1}}{\partial x_i} 
=\frac{\partial \mathcal{H}}{\partial x_i} . \mathcal{G}^{-1}
-\mathcal{H} . \mathcal{G}^{-1} . \frac{\partial  \mathcal{G}}{\partial x_i}. \mathcal{G}^{-1} 
=\frac{\partial \mathcal{H}}{\partial x_i} . \mathcal{G}^{-1}
-\mathcal{H} . \mathcal{G}_{1,i}   .
\end{align*}
Further, 
\begin{align*}
\frac{\partial^2}{\partial x_i^2}\left(
  \mathcal{H}.\mathcal{G}^{-1}\right)&=\frac{\partial^2 \mathcal{H}}{\partial x_i^2}. \mathcal{G}^{-1}              -2\frac{\partial \mathcal{H}}{\partial x_i}. \mathcal{G}_{1,i}  - \mathcal{H} . \frac{\partial \mathcal{G}_{1,i}}{\partial x_i}  .
\end{align*}
Because 
\begin{align*}
\frac{\partial \mathcal{G}_{1,i}}{\partial x_i} &= \frac{\partial \mathcal{G}^{-1}}{\partial x_i}\,. \frac{\partial \mathcal{G}}{\partial x_i}. \mathcal{G}^{-1}
+ \mathcal{G}^{-1}\,. \frac{\partial^2 \mathcal{G}}{\partial x_i^2} \,. \mathcal{G}^{-1}+ \mathcal{G}^{-1}\,. \frac{\partial \mathcal{G}}{\partial x_i} \,.\frac{\partial \mathcal{G}^{-1}}{\partial x_i}\\
&= -\mathcal{G}_{1,i}. \frac{\partial \mathcal{G}}{\partial x_i}. \mathcal{G}^{-1}
+ \mathcal{G}_{2,i}- \mathcal{G}^{-1}\,. \frac{\partial \mathcal{G}}{\partial x_i} \,.\mathcal{G}_{1,i}
\end{align*}
we get 
\begin{align*}
\frac{\partial^2}{\partial x_i^2} \left( \mathcal{H}.\mathcal{G}^{-1}\right)
&=\frac{\partial^2 \mathcal{H}}{\partial x_i^2}. \mathcal{G}^{-1}-2\frac{\partial \mathcal{H}}{\partial x_i} . \mathcal{G}_{1,i} + \mathcal{H} . \left[   \mathcal{G}_{1,i}. \frac{\partial \mathcal{G}}{\partial x_i}. \mathcal{G}^{-1}
- \mathcal{G}_{2,i}+ \mathcal{G}^{-1}\,. \frac{\partial \mathcal{G}}{\partial x_i} \,.\mathcal{G}_{1,i}   \right] .
\end{align*}
We conclude that
\begin{align*}
\frac{\partial^2}{\partial x_i^2} &\left(  \mathcal{H}.\mathcal{G}^{-1}. \mathcal{H}.\mathcal{G}^{-1} \right)\\
&=\left\{\frac{\partial^2 \mathcal{H}}{\partial x_i^2}. \mathcal{G}^{-1}-2\frac{\partial \mathcal{H}}{\partial x_i} . \mathcal{G}_{1,i} + \mathcal{H} . \left[     \mathcal{G}_{1,i}. \frac{\partial \mathcal{G}}{\partial x_i}. \mathcal{G}^{-1}
- \mathcal{G}_{2,i}+ \mathcal{G}^{-1}\,. \frac{\partial \mathcal{G}}{\partial x_i} \,.\mathcal{G}_{1,i}   \right]\right\}.  \mathcal{H}.\mathcal{G}^{-1}\\
&\quad
+ \mathcal{H}.\mathcal{G}^{-1}.\left\{\frac{\partial^2 \mathcal{H}}{\partial x_i^2}. \mathcal{G}^{-1}
-2\frac{\partial \mathcal{H}}{\partial x_i} . \mathcal{G}_{1,i} + \mathcal{H} . \left[   \mathcal{G}_{1,i}. \frac{\partial \mathcal{G}}{\partial x_i}. \mathcal{G}^{-1}
- \mathcal{G}_{2,i}+ \mathcal{G}^{-1}\,. \frac{\partial \mathcal{G}}{\partial x_i} \,.\mathcal{G}_{1,i}   \right]\right\}\\
&\quad +2 \left[\frac{\partial \mathcal{H}}{\partial x_i} . \mathcal{G}^{-1}
-\mathcal{H} . \mathcal{G}_{1,i}\right]^2 .
\end{align*}
Upon noting
\begin{align*}
\frac{\partial^2}{\partial x_i^2} &\Tr
\left( \mathcal{H}.\mathcal{G}^{-1}. \mathcal{H}.\mathcal{G}^{-1} \right)\\
&=2\,\mathrm{Trace}\left(\left\{\frac{\partial^2 \mathcal{H}}{\partial x_i^2}. \mathcal{G}^{-1}-2\frac{\partial \mathcal{H}}{\partial x_i} . \mathcal{G}_{1,i} + \mathcal{H} . \left[     \mathcal{G}_{1,i}. \frac{\partial \mathcal{G}}{\partial x_i}. \mathcal{G}^{-1}
- \mathcal{G}_{2,i}+ \mathcal{G}^{-1}\,. \frac{\partial \mathcal{G}}{\partial x_i} \,.\mathcal{G}_{1,i}   \right]\right\}.  \mathcal{H}.\mathcal{G}^{-1}\right)\\
&\quad +2\,\mathrm{Trace}\left( \left[\frac{\partial \mathcal{H}}{\partial x_i} . \mathcal{G}^{-1}
-\mathcal{H} . \mathcal{G}_{1,i}\right]^2\right)
\end{align*}
the Lemma follows.

\subsection{Proof of Proposition \ref{energy_density_origin}}\label{appendix_energy_density_origin}

Let matrices $\mathcal{G}$ and $\mathcal{H}$ be given by (\ref{defGH}) and use (\ref{energy}) .
Each of the terms depend on the solutions $\zeta_j$ of the Atiyah-Ward equation and their derivatives. 
We note that at $\boldsymbol{x}=0$  we have $\zeta_j =\pm k' \pm \imath k  $. Then,

\begin{align*}
\mathcal{G}(\boldsymbol{0})&=8k^2K^2\left( \begin{array}{cc} 1&0\\0&1 \end{array} \right),\\
\left.\frac{\partial}{\partial x_1} \mathcal{G}(\boldsymbol{x})\right|_{\boldsymbol{x}=0}&=16Kk(E-K)
\left( \begin{array}{rr} -1&0\\0&1 \end{array} \right) ,\\
\left.\frac{\partial}{\partial x_2} \mathcal{G}(\boldsymbol{x})\right|_{\boldsymbol{x}=0}&=16EK\frac{k^2}{{k'}^2}
\left( \begin{array}{rr} -1&0\\0&1 \end{array} \right),\\
   \left.\frac{\partial}{\partial x_3} \mathcal{G}(\boldsymbol{x})\right|_{\boldsymbol{x}=0}&=0,\\
\left.\frac{\partial^2}{\partial x_1^2} \mathcal{G}(\boldsymbol{x})\right|_{\boldsymbol{x}=0}&=16EK\frac{64}{k^2}
\left((E-K)^2(2k^2-1)-k^2\right)
\left( \begin{array}{rr} 1&0\\0&1 \end{array} \right),\\
\left.\frac{\partial^2}{\partial x_2^2} \mathcal{G}(\boldsymbol{x})\right|_{\boldsymbol{x}=0}&=\frac{128E^2k^2}{{k'}^2}
\left( \begin{array}{rr} 1&0\\0&1 \end{array} \right) ,\\
\left.\frac{\partial^2}{\partial x_3^2} \mathcal{G}(\boldsymbol{x})\right|_{\boldsymbol{x}=0}&=-\frac{64}{k^2{k'}^2}
\left[ ( E - {k'}^2K )^2-k^2{k'}^2  \right]
\left( \begin{array}{rr} 1&0\\0&1 \end{array} \right).
\end{align*}

Also

\begin{align*}
\mathcal{H}(\boldsymbol{0})&=8\imath K [ K(1+{k'}^2)-2E ]\left( \begin{array}{cc} 0&1\\1&0 \end{array} \right),\\
\left.\frac{\partial}{\partial x_1} \mathcal{H}(\boldsymbol{x})\right|_{\boldsymbol{x}=0}&=16\imath (E-K)[ 2E-K(1+{k'}^2) ]
\left( \begin{array}{rr} -1&0\\0&1 \end{array} \right),\\
\left.\frac{\partial}{\partial x_2} \mathcal{H}(\boldsymbol{x})\right|_{\boldsymbol{x}=0}&=\frac{16\imath k^2}{k'}(EK-2)
\left( \begin{array}{rr} 1&0\\0&-1 \end{array} \right),\\
   \left.\frac{\partial}{\partial x_3} \mathcal{H}(\boldsymbol{x})\right|_{\boldsymbol{x}=0}&=0,\\
\left.\frac{\partial^2}{\partial x_1^2} \mathcal{H}(\boldsymbol{x})\right|_{\boldsymbol{x}=0}&=\frac{64\imath}{Kk^2}[ K(1+{k'}^2)-2E ]
\left( \begin{array}{rr} 1&0\\0&1 \end{array} \right)
\\
&+ \frac{64\imath}{K k^2}\left[ (K(E-K)-2)(E-K){k'}^2-Kk^2  \right]\left( \begin{array}{rr} 0&1\\1&0 \end{array} \right) ,\\
\left.\frac{\partial^2}{\partial x_2^2} \mathcal{H}(\boldsymbol{x})\right|_{\boldsymbol{x}=0}&=\frac{64\imath k^2E}{K{k'}^2}(2-EK)
\left( \begin{array}{rr} 1&0\\0&1 \end{array} \right)\\
&+\frac{64\imath E^2}{K {k'}^2} [K(1+{k'}^2)-2E] \left( \begin{array}{rr} 0&1\\1&0 \end{array} \right), \\
\left.\frac{\partial^2}{\partial x_3^2} \mathcal{H}(\boldsymbol{x})\right|_{\boldsymbol{x}=0}&=
-\frac{64\imath}{K{k'}^2k^2} [ K(E-K{k'}^2)^2+Kk^2{k'}^2+2(K{k'}^2-E)  ]
\left( \begin{array}{rr} 0&1\\1&0 \end{array} \right).
\end{align*}
Bringing all these together in (\ref{energy}) yields (\ref{energy_density0}).
To find (\ref{energy_density_limits}) we use for $k=0,1$ the corresponding expansions of $K(k)$ and $E(k)$,

\begin{align}\begin{split}
\mathcal{E}_{\boldsymbol{x}=0,k\sim 0} &=\frac{8}{\pi^4}(\pi^2-8)^2+\frac{(\pi^2-8)(\pi^2-16)}{4\pi^4} k^4 +O(k^6),\\
\mathcal{E}_{\boldsymbol{x}=0,k'\sim 0} &=32 {k'}^2+O({k'}^4).
\end{split}
\end{align}

\subsection{Proof of Proposition \ref{higgs13k0}}\label{higgs13k0app}
In the case $x_2=0$ we have from (\ref{rprm}) that $R_+=R_-=R=\sqrt{\pi^2-16r^2}$, $x_{+}=x_{-}=x_1$. We order the solutions
of Atiyah-Ward equation in such the way that in the limit $x_2=x_3=0$ they coincide with our case {\bf II}, 
$|x_1|\leq \frac{\pi}{4}$ of the $x_1$-axis, namely, 
\begin{align}
\zeta_1=\frac{4\imath x_3-R}{4 x_1-\pi},\; \zeta_2=\frac{4\imath x_3+R}{4 x_1+\pi},\; 
\zeta_3=\frac{4\imath x_3-R}{4 x_1+\pi},\; \zeta_4=\frac{4\imath x_3+R}{4 x_1-\pi}.
\end{align}
The corresponding exponentials $\mu_i$ are
\begin{align}
\mu_1=\lambda,\;\; \mu_2=-\lambda+\frac12\imath\pi,\;\; \mu_3=\lambda+\frac12\imath\pi,\;\; \mu_4=-\lambda,
\end{align} 
with
\begin{equation}
\lambda= \frac{\imath}{4}R=\frac{\imath}{4}\sqrt{\pi^2 -16 r^2 }, \quad   r=\sqrt{x_1^2+x_3^3}.
\end{equation}
We then obtain
\begin{align} \begin{split}
\text{Gram}&=32r^2-2\pi^2\cosh^2(2\lambda)1_2\\
\text{Higgs}'_1&=\imath(32r^2-2\pi^2\cosh^2(2\lambda)) \sigma_1\\
\text{Higgs}'_2&=\frac{4\pi^2}{R}\sinh(4\lambda) \sigma_1\\
\text{Higgs}'_3&=-64\imath r^2\sigma_1
\end{split}
\end{align}
and the corresponding normalized Higgs field is
\begin{equation}
\Phi_{\mathrm{norm}}=\imath \sigma_1\left[
\frac{ \cosh(2\lambda)^2\pi^2R+2 \imath\pi^2\sinh(4\lambda)+16r^2R } 
{R(\pi^2\cosh^2(2\lambda)-16 r^2)}\right]
\end{equation}
and formula (\ref{Hfield}) follows.

Although this formula was derived under the assumption $16r^2 < \pi^2$ the discussion of \S11 shows that
 (\ref{Hfield}) admits  a continuation  to the area   $16r^2 > \pi^2$. Indeed the points $r= \pm \pi/4$ are ordinary points of the function
$H(r)$, namely $$H(\pm \pi/4)=-\frac{12-\pi^2}{3(\pi^2-4)}\sim 0.121.$$ 


\subsection{Proof of Proposition \ref{higgs14k0x3plane}} \label{higgs13k0app1}
In the case $x_3=0$ we have from (\ref{rprm})
\begin{equation}
R_+=\sqrt{\pi^2-16 r^2 +8\imath \pi x_2}, 
\quad R_-=\sqrt{\pi^2-16 r^2 -8\imath \pi x_2}, \quad r^2=x_1^2+x_2^2.
\end{equation}
Solutions of Atiyah-Ward equation are ordered as follows
\begin{align}
\zeta_1=\frac{R_+}{\pi-4 x_-}, \;\zeta_2=\frac{R_-}{\pi+4 x_-},\;
\zeta_3=-\frac{R_-}{\pi+4 x_-}, \;\zeta_4=-\frac{R_+}{\pi-4 x_-}.
\label{zetaPlaneX3}
\end{align}
The associated $\mu_i$ are then
\begin{align}
\mu_1=\frac{1}{4}\imath R_+,\; \mu_2=-\frac{\imath}{4}R_-+\frac12\imath \pi,\; 
\mu_3=\frac{\imath}{4}R_++\frac12\imath \pi, \; \mu_4=\frac{1}{4}\imath R_-. \label{muPlaneX3}
\end{align}

The usual calculation leads to the following expression for Gram matrix, 
\begin{align}
\begin{split}
\text{Gram}=\frac12 G(x_1,x_2)\mathrm{Diag}
\left(\mathrm{exp}\left( \frac{\imath}{2} ( R_+-R_- ) \right),\mathrm{exp}\left( -\frac{\imath}{2} ( R_+-R_- ) \right)   \right)
\end{split}\label{GramX3}
\end{align}
where
\begin{equation}
G(x_1,x_2)= (\pi^2+16 r^2)\sin\frac{R_+}{2} \sin\frac{R_-}{2} 
-R_+R_-\left(\cos\frac{R_+}{2} \cos\frac{R_-}{2}+1 \right) .\label{GramX3A}
\end{equation} 

It is easy to see that $  G(0,x_2) \equiv 0$ and we have the
series expansion  in $\xi\sim 0$,
\begin{equation}  
G(\xi,x_2)=-\frac{32\mathrm{cosh}^2(2x_2)\pi^2}
{\pi^2+16 x_2^2}\xi^2+ O(\xi^4) .\label{GramExpan}      \end{equation}
Introducing the shorthand, 
\begin{align}
S_{\pm}=\sin\left( \frac12 R_{\pm}  \right), 
\quad C_{\pm}=\cos\left( \frac12 R_{\pm}  \right)\quad
E_{\pm}=\mathrm{exp}\left( \frac14 \imath R_{\pm}  \right) ,\label{notations2}
\end{align}
we arrive at the formula (\ref{higgsthm}) upon
substitution  of (\ref{zetaPlaneX3}), (\ref{muPlaneX3}) 
into the Higgs field $\Phi$.

To check the result obtained, consider the reduction of   (\ref{HfieldxPlanex3})   
to the $x_1$ and $x_2$ axes.
To reduce expression (\ref{HfieldxPlanex3})  $x_1$-axis set $x_2=0$ and 
\begin{equation}  R_+=R_-=\sqrt{\pi^2-16x_1^2}=R .   \end{equation}
Substituting these into the general formula together with $k=0$, $K=E=\pi/2$ one obtains
Ward's expression (\ref{ward2}) with $r=x_1$,
\begin{equation}
H(x_1,0,0)=-1 - \frac{2\pi^2 C(2S-RC)}{R(\pi^2C^2-16 x_1^2)},
\end{equation}
where  $S=\sin(R/2), C=\cos(R/2)$. 

To reduce expression (\ref{HfieldxPlanex3}) to the  $x_2$-axis set $x_1=0$ and 
\begin{equation}  R_+=\pi+4\imath x_2, \quad  R_-=\pi-4\imath x_2.  \end{equation}
Now in this case $G(0,x_2)=0$ for all $x_2 \in \mathbb{R}$ while also the numerator of the 
expression vanishes and therefore the value of $H(0,x_2,0)$ results from the limit as $x_1\to 0$. 
To do that expand the quantities $R_{\pm}$ with $x_1=\xi$ up to order 4,  
\begin{equation}
R_{\pm}=\pi \pm 4\imath x_2 - \frac{8\xi^2}{\pi \pm 4\imath x_2}- \frac{32\xi^4}{(\pi \pm 4\imath x_2)^3}+O(\xi^6).
\end{equation}
Substituting these into the numerator and denominator one finds that both vanish to order $\xi^4$ and the quotient reads
\begin{equation}
H(0,x_2,0)^2=\left(-\tanh(2x_2)+\frac{16 x_2}{16 x_2^2 + \pi^2}\right)^2
\end{equation}  
which again coincides with Ward's answer in this case.

\section{Lam\'e's Equation}\label{applame}
Here we adopt the approach of Brown, Panagopoulos and Prasad \cite{bpp82}  to conjugate the equation $\Delta\sp\dagger v=0$ into a convenient form. The aim of this appendix is to show that
for each coordinate axis we may reduce the matrix differential equation $\Delta\sp\dagger v=0$ to solving
\[
\frac {d^{2}u}{d{z}^{2}
} \left( z \right) 
+\mathcal{U} (z)\,u(z)=\lambda_1\,u(z), 
\]
or the same equation for a shifted argument giving $w(z)$. Here
$$ \lambda_1=x_1^2-\frac14 {(1+ {k}^{2}){K}^{2} }, \qquad 
\lambda_2=x_2^2-\frac14 {{k}^{2}{K}^{2} }, \qquad 
\lambda_3=x_3^2-\frac14 {{K}^{2} }.
$$
With $Kz = 3K+2iK'+2 s$ we may put $\mathcal{U} (z)$ into the standard Lam\'e form
$$\left(\frac{d^2}{ds^2}-2 k^2  \mathrm{sn}^2(s)\right)F(s)=-\lambda\, F(s)$$
with $\lambda=1+k^2  \mathrm{cn}^2(t)=-4\lambda_j/K^2$ and the parameterization of
\S \ref{axesandjacobi}:
\begin{align*}
x_1:&\quad    \mathrm{sn}^2(t)=\frac{4x_1^2}{k^2 K^2},\qquad
x_2:\quad   \mathrm{dn}^2(t)=-\frac{4x_2^2}{K^2},\qquad
x_3:\quad    \mathrm{cn}^2(t)=-\frac{4x_3^2}{k^2 K^2}.
\end{align*}

We caution at the outset that  both the order of our tensor products and our Nahm data differ from those of \cite{bpp82}  and we shall relate our conventions shortly.
Set
$$\mathcal{P}:=\begin {pmatrix} 1&0&0&0\\ 0&0&1&0
\\0&1&0&0\\ 0&0&0&1\end {pmatrix},
\quad
\mathcal{Q}:=\frac1{\sqrt{2}}\begin {pmatrix} 1&1&0&0\\ 0&0&1&1
\\0&0&1&-1\\ 1&-1&0&0\end {pmatrix},
\quad
\mathcal{R}:=\frac1{\sqrt{2}}\begin {pmatrix} -i &0&0&i\\ 0&-1&-1&0
\\-i&0&0&-i\\ 0&-1&1&0\end {pmatrix}.
$$
Then we have $\mathcal{P}\left(\sum_{j=1}^3 \sigma_j \otimes x_j 1_2 \right) \mathcal{P}=
\sum_{j=1}^3 x_j 1_2  \otimes \sigma_j $. Then with
\begin{align*}\Delta\sp\dagger v&=
\left[\frac{\mathrm{d}}{\mathrm{d} z}+\frac12 \left(\sum_{j=1}^3 \sigma_j\otimes \sigma_j f_j(z) \right)
-\left(\sum_{j=1}^3 \sigma_j \otimes x_j 1_2 \right) \right]
\boldsymbol{v}(z,\boldsymbol{x})\\
&=
\left[\frac{\mathrm{d}}{\mathrm{d} z} 1_4+
\begin {pmatrix} f_{{3}}/2-x_{{3}}&0&-x_{{1}}+ix_{{2}}&
f_{{1}}/2-f_{{2}}/2 \\ 0&-f_{{3}}/2-x_{{3}}&
f_{{1}}/2+f_{{2}}/2&-x_{{1}}+ix_{{2}}\\ -x_{{1
}}-ix_{{2}}&f_{{1}}/2+f_{{2}}/2&-f_{{3}}/2+x_{{3}}&0
\\ f_{{1}}/2-f_{{2}}/2&-x_{{1}}-ix_{{2}}&0&
f_{{3}}/2+x_{{3}}\end {pmatrix} 
\right]\boldsymbol{v}(z,\boldsymbol{x})
\end{align*}
the conjugation
\begin{align}
{\tilde\Delta}\sp\dagger& :=
\mathcal{Q}\sp{-1}\mathcal{P}\,\Delta\sp\dagger\,
\mathcal{P}\mathcal{Q}
\nonumber
\\&=
\frac{\mathrm{d}}{\mathrm{d} z} 1_4+
\begin {pmatrix} (f_{{3}}+f_{{1}}-f_{{2}})/2&-x
_{{3}}&-x_{{1}}&ix_{{2}}\\ -x_{{3}}&(f_{{3}}-
f_{{1}}+f_{{2}})/2&ix_{{2}}&-x_{{1}}\\ -x_{{1}}&
-ix_{{2}}&(f_{{1}}+f_{{2}}-f_{{3}})/2&x_{{3}}
\\ -ix_{{2}}&-x_{{1}}&x_{{3}}&-(f_{{1}}+f_{
{2}}+f_{{3}})/2\end {pmatrix}
\label{conjdeltadaggerV}
\end{align}
brings this to the \emph{form} (6.2) of \cite{bpp82}.  This is the form of the equation to be studied.
The key observation of \cite{bpp82} is that for a coordinate axis (there the $x_2$-axis) the $4\times4$
problem (\ref{conjdeltadaggerV}) reduces to two $2\times2$ uncoupled equations. These in turn may be reduced to an  $n=1$ Lam\'e equation. Before turning to each of these reductions we first comment on the relation of the solutions to our earlier general solution and relate the conventions of \cite{bpp82} with those here.

In general the solutions $\mathcal{V}_{BPP}$ to ${\tilde\Delta}\sp\dagger\mathcal{V}_{BPP}=0$ are related to our earlier solutions by
\begin{align*}
\mathcal{V_{BPP}}&= \mathcal{Q}\sp{-1}\mathcal{P}\, V = \mathcal{Q}\sp{-1}\mathcal{P}\,
\left(1_2\otimes \mathcal{O}\,C\sp{-1}(z)\right)\,\frac1{\theta_2^2(z/2)}\,\Lambda_i 
\end{align*}
where (the conjugate of) $\Lambda_i$ was given in (\ref{LAMBDABAR}). Thus
\begin{align*}
{\overline {\mathcal{V}}}_{BPP}&= 
\mathcal{Q}\sp{-1}\mathcal{P}\,
\left(1_2\otimes \diag(p(z),p^{-1}(z))\,\mathcal{O}\right)\,\frac1{\theta_2^2(z/2)}\,{\overline{ \Lambda_i}}
\\
&=\frac1{\theta_2^2(z/2)}\,\frac12 \,
\begin {pmatrix} 
p \left( z \right) &-p \left( z \right) &
p \left( -z \right) &p \left( -z \right) \\ \noalign{\medskip}p
 \left( z \right) &-p \left( z \right) &-p \left( -z \right) &-p
 \left( -z \right) \\ \noalign{\medskip}p \left( -z \right) &p \left( 
-z \right) &p \left( z \right) &-p \left( z \right) 
\\ \noalign{\medskip}-p \left( -z \right) &-p \left( -z \right) &p
 \left( z \right) &-p \left( z \right)
\end{pmatrix} 
\,{\overline {\Lambda_i}}
\\
&=
\frac1{\theta_2^2(z/2)}\,
\begin {pmatrix} 
 \left( 1/2\,iB-1/2\,iA \right) p \left( z
 \right) \zeta_i+ \left( 1/2\,\alpha-1/2\,\beta \right) p \left(z
 \right) + \left( 1/2\,B+1/2\,A \right) p \left( -z \right) 
\\ \noalign{\medskip} \left( 1/2\,iB-1/2\,iA \right) p \left( z
 \right) \zeta_i+ \left( 1/2\,\alpha-1/2\,\beta \right) p \left( z
 \right) + \left( -1/2\,B-1/2\,A \right) p \left( -z \right) 
\\ \noalign{\medskip} \left( -1/2\,iB-1/2\,iA \right) p \left( -z
 \right) \zeta_i+ \left( -1/2\,B+1/2\,A \right) p \left( z \right) +
 \left( 1/2\,\alpha+1/2\,\beta \right) p \left( -z \right) 
\\ \noalign{\medskip} \left( 1/2\,iB+1/2\,iA \right) p \left( -z
 \right) \zeta_i+ \left( -1/2\,B+1/2\,A \right) p \left( z \right) +
 \left( -1/2\,\alpha-1/2\,\beta \right) p \left( -z \right) 
\end{pmatrix}
\end{align*}
With our expansion
\begin{equation}\label{asymptildev1}
V(1-\xi)\sim
\begin{pmatrix}
{\dfrac {-x_{{1}}+ix_{{2}}}{\sqrt {\xi}}}+
\sqrt {\xi}\,a\\ \noalign{\medskip}{\xi}^{-3/2}+{\dfrac {x_{{3}}}{\sqrt {
\xi}}}+\sqrt {\xi} \left( b-1/2\,{r}^{2} \right) \\ \noalign{\medskip}
-{\xi}^{-3/2}+{\dfrac {x_{{3}}}{\sqrt {\xi}}}+\sqrt {\xi} \left( b+1/2
\,{r}^{2} \right) \\ \noalign{\medskip}{\dfrac {x_{{1}}+ix_{{2}}}{
\sqrt {\xi}}}+\sqrt {\xi}\,c
\end{pmatrix}
\quad\textrm{then}\quad
{{\mathcal{V}}}_{BPP}(1-\xi)\sim
\begin{pmatrix}
\dfrac{\left( a+c \right)}{\sqrt {2} }\sqrt {\xi
}+{\dfrac {i\sqrt {2}x_{{2}}}{\sqrt {\xi}}}\\
\dfrac{ \left( a-c \right)}{\sqrt {2} } \sqrt {\xi}-{\dfrac {\sqrt {2}x_{{1}}}{
\sqrt {\xi}}}
\\ 
\sqrt {2}\,b\,\sqrt {\xi}+{\dfrac {\sqrt 
{2}x_{{3}}}{\sqrt {\xi}}}
\\ \dfrac{ {r}^{2}}{\sqrt {2} }
\sqrt {\xi}-{\dfrac {\sqrt {2}}{{\xi}^{3/2}}}\end{pmatrix}
\end{equation}
and we find ${\tilde\Delta}\sp\dagger{{\mathcal{V}_{BPP}}}=O(\xi^{1/2})$.
Similarly
\begin{equation}\label{asymptildev2}
V(\xi-1)\sim
\begin{pmatrix}
{\xi}^{-3/2}-{\dfrac {x_{{3}}}{\sqrt {\xi}}}+
\sqrt {\xi} \left( b-1/2\,{r}^{2} \right) \\ \noalign{\medskip}{\dfrac 
{-x_{{1}}+ix_{{2}}}{\sqrt {\xi}}}-\sqrt {\xi}a\\ \noalign{\medskip}-{
\dfrac {x_{{1}}+ix_{{2}}}{\sqrt {\xi}}}+\sqrt {\xi}c
\\ \noalign{\medskip}{\xi}^{-3/2}+{\dfrac {x_{{3}}}{\sqrt {\xi}}}-
\sqrt {\xi} \left( b+1/2\,{r}^{2} \right)
\end{pmatrix}
\quad\textrm{then}\quad
{{\mathcal{V}}}_{BPP}(\xi-1)\sim
\begin{pmatrix}
-\dfrac{{r}^{2}}{\sqrt {2}}{\sqrt {\xi}}+{\dfrac {\sqrt {2}}{{\xi}^{3/2}}}
\\  \sqrt {2}\,b\sqrt {\xi}-{
\dfrac {\sqrt {2}x_{{3}}}{\sqrt {\xi}}}
\\ -\dfrac{\left( a-c \right) }{\sqrt {2}} \sqrt {\xi}-{\dfrac {\sqrt {2}x_{{1}}}{
\sqrt {\xi}}}\\ 
\dfrac{ \left( a+c \right) }{\sqrt{2}}
\sqrt {\xi}-{\dfrac {i\sqrt {2}x_{{2}}}{\sqrt {\xi}}}
\end{pmatrix}.
\end{equation}

\subsection{Comparison of Notation}
To compare with \cite{bpp82}  we note their
choice of the functions $f_j(z)$ are cyclically shifted from ours and their  spatial coordinate are the opposite  
of ours. Denoting the \cite{bpp82} choices by $\tilde f_j$ , $\tilde x_j$ and $\tilde z$ then 
$$\tilde f_1=f_3/Kk',\quad \tilde f_2=f_1/Kk',\quad \tilde f_3=f_2/Kk'.$$
By comparing 
\begin{align*}
\mathcal{R}&\,\Delta\sp\dagger\,
\mathcal{R}\sp{-1}
=
\frac{\mathrm{d}}{\mathrm{d} z} 1_4
\\&\quad
 +
 \begin {pmatrix}\dfrac{Kk'}2 (\tilde f_{{3}}+\tilde f_{{1}}-\tilde f_{{2}})&-x
_{{2}}&-x_{{3}}&ix_{{1}}\\ -x_{{2}}&\dfrac{Kk'}2(\tilde f_{{3}}-
\tilde f_{{1}}+\tilde f_{{2}})&ix_{{1}}&-x_{{2}}\\ -x_{{3}}&
-ix_{{1}}&\dfrac{Kk'}2(\tilde f_{{1}}+\tilde f_{{2}}-\tilde f_{{3}})&x_{{2}}
\\ -ix_{{1}}&-x_{{3}}&x_{{2}}&-\dfrac{Kk'}2(\tilde f_{{1}}+\tilde f_{
{2}}+\tilde f_{{3}})\end {pmatrix}
\end{align*}
with \cite{bpp82} we see that
$$\tilde x_1=-x_3/Kk',\quad \tilde x_2=-x_1/Kk',\quad \tilde x_3=-x_2/Kk', \qquad \tilde z =Kk'z.$$

\subsection{The \texorpdfstring{$x_2$}{x2} axis}
We shall now reduce  (\ref{conjdeltadaggerV}) for each of the coordinate axes in turn to a
$n=1$ Lam\'e's equation; having done that we will solve for this, and so solve
 (\ref{conjdeltadaggerV}) for the three axes. First the $x_2$ axis.
With $x_1=0=x_3$ the equations ${\tilde\Delta}\sp\dagger\mathcal{V}=0$ decouple into
\[
\begin{pmatrix}
{\dfrac {d}{dz}}v_{{1}} \left( z \right) +
 \left( 1/2\,f_{{3}} \left( z \right) +1/2\,f_{{1}} \left( z \right) -
1/2\,f_{{2}} \left( z \right)  \right) v_{{1}} \left( z \right) +ix_{{
2}}v_{{2}} \left( z \right) \\ \noalign{\medskip}{\dfrac {d}{dz}}v_{{2}
} \left( z \right) -ix_{{2}}v_{{1}} \left( z \right) + \left( -1/2\,f_
{{1}} \left( z \right) -1/2\,f_{{2}} \left( z \right) -1/2\,f_{{3}}
 \left( z \right)  \right) v_{{2}} \left( z \right) \end {pmatrix}=0
\]
and
\[
\begin{pmatrix} 
{\dfrac {d}{dz}}w_{{1}} \left( z \right) +
 \left( 1/2\,f_{{3}} \left( z \right) -1/2\,f_{{1}} \left( z \right) +
1/2\,f_{{2}} \left( z \right)  \right) w_{{1}} \left( z \right) +ix_{{
2}}w_{{2}} \left( z \right) \\ \noalign{\medskip}{\dfrac {d}{dz}}w_{{2}
} \left( z \right) -ix_{{2}}w_{{1}} \left( z \right) + \left( 1/2\,f_{
{1}} \left( z \right) +1/2\,f_{{2}} \left( z \right) -1/2\,f_{{3}}
 \left( z \right)  \right) w_{{2}} \left( z \right)\end {pmatrix}=0.
\]
Consider the first of these. Solving the first entry for $v_2$ yields
\[
v_2(z)=
\frac {i}{2x_{{2}}}\left( v_{{1}} \left( z \right) f_{{3}} \left( z
 \right) +v_{{1}} \left( z \right) f_{{1}} \left( z \right) -v_{{1}}
 \left( z \right) f_{{2}} \left( z \right) +2\,{\frac {d}{dz}}v_{{1}}
 \left( z \right)  \right) 
\]
which after substituting in the second component gives
a second order equation of the form 
$$v_1''(z)-f_2(z)\, v_1'(z)+\tilde  A(z)v_1(z)=0.$$
Now
\begin{equation}
f_1 = \frac{d}{dz}\ln(f_2+f_3),\quad 
f_2 = \frac{d}{dz}\ln(f_3+f_1),\quad 
f_3 = \frac{d}{dz}\ln(f_1+f_2),
\end{equation}
and so upon introducing the integrating factor
$$v_1(z)=u(z)\,\exp(\frac12 \int\sp{z}f_2(s)ds )=u(z)\,\sqrt{f_1(z)+f_3(z)}$$
we obtain the equation
\[
\frac {d^{2}u}{d{z}^{2}
} \left( z \right) 
+A (z)\,u(z)=0
\]
where
\[ A(z)=
 -1/4\, \left( f_{{3}} \left( z \right)  \right) ^{2}-1/2\,f_{{
1}} \left( z \right) f_{{3}} \left( z \right) -1/4\, \left( f_{{1}}
 \left( z \right)  \right) ^{2}+1/2\,f_{{1}} \left( z \right) f_{{2}}
 \left( z \right) -{x_{{2}}}^{2}+1/2\,f_{{3}} \left( z \right) f_{{2}}
 \left( z \right) .
\]
We shall show that this is a Lam\'e equation. In terms of $u(z)$ we have
\begin{equation}
v_2(z)=\frac{i}{2 x_2}\left[ 2 u'(z)+(f_3(z)+f_1(z))\, u(z)\right]\,\sqrt{f_1(z)+f_3(z)}.
\end{equation}

When performing the same eliminations for the second set of equations we find that the integrating factor is the inverse of that found. We have
\[ w_1(z)= \frac{w(z)}{ \sqrt{f_1(z)+f_3(z)} },\quad
w_2(z)=\frac{i}{2 x_2}\frac{ 2w'(z)+   (f_3(z)-f_1(z))\, w(z)   }{ \sqrt{f_1(z)+f_3(z)} },
\]
where now
\begin{align*}
\frac {d^{2}w}{d{z}^{2}
} \left( z \right) 
+B(z)\,w(z)=0,\quad
B(z)=A(z)+f_3(z)\left[ f_1(z)-f_2(z)\right].
\end{align*}

At this stage we have a solution of the form
\[\mathcal{V}=
c_1\mathcal{V}_1+
c_2\mathcal{V}_2
\]
where
\begin{align}\label{v1x2}
\mathcal{V}_1&=
\begin{pmatrix}u(z)\,\sqrt{f_1(z)+f_3(z)} \\0\\0\\
\dfrac{i}{2 x_2}\left[ 2 u'(z)+(f_3(z)+f_1(z))\, u(z)\right]\,\sqrt{f_1(z)+f_3(z)} \end{pmatrix},
\\ \label{v2x2}
\mathcal{V}_2&=
\begin{pmatrix}0\\ \dfrac{w(z)}{ \sqrt{f_1(z)+f_3(z)} }\\ \dfrac{i}{2 x_2}\dfrac{ 2w'(z)+   (f_3(z)-f_1(z))\, w(z)   }{ \sqrt{f_1(z)+f_3(z)} }\\ 0 \end{pmatrix}
\end{align}
and this has precisely the asymptotics of (\ref{asymptildev1}, \ref{asymptildev2}) for regular solutions $u$ and $w$.

To proceed we now construct the relevant Lam\'e equations. Although our choices for the $f_j$'s differ from
those of \cite{bpp82} we will obtain the same equation though for shifted arguments.
First, noting that $f_{{3}} ^{\,2} \left( z \right) =-{k}^{2}{K}^{2}+ f_{{1}}^{\,2} \left( z \right) 
$
we obtain the equation
\begin{equation}\label{lame1}
\frac {d^{2}u}{d{z}^{2}
} \left( z \right) 
+\mathcal{U} (z)\,u(z)=\lambda_2\,u(z), \qquad  \lambda_2=x_2^2-\frac14 { {k}^{2}{K}^{2} },
\end{equation}
with
\begin{align*}
\mathcal{U} (z)&=
\frac12\left[ f_{{2}} \left( z \right) f_{{3}} \left( z \right) -f_{{1}} \left( z \right) f_
{{3}} \left( z \right) +f_{{1}} \left( z \right) f_{{2}} \left( z
 \right)- f
_{{1}}  ^{\,2}  \left( z \right)\right],
\end{align*}
and also 
\[
\frac {d^{2}w}{d{z}^{2}
} \left( z \right) 
+\mathcal{W} (z)\,w(z)=\lambda_2\,w(z), \qquad  \lambda_2=x_2^2-\frac14 { {k}^{2}{K}^{2} },
\]
with
\begin{align*}
\mathcal{W} (z)&=
\frac12\left[ -f_{{2}} \left( z \right) f_{{3}} \left( z \right) +f_{{1}} \left( z \right) f_
{{3}} \left( z \right) +f_{{1}} \left( z \right) f_{{2}} \left( z
 \right)- f
_{{1}}  ^{\,2}  \left( z \right)\right].
\end{align*}

Let us record the translation properties of our functions $f_j(z)$ (\ref{nahmsolution}):
\begin{align*}
f_1(\pm z)&=f_1(z),&& f_1(2\pm z)=-f_1(z), && f_1(2\tau\pm z)=f_1(z),\\
f_2(\pm z)&=\pm f_2(z),&& f_2(2\pm z)=\pm f_2(z), && f_2(2\tau\pm z)=\mp f_2(z),\\
f_3(\pm z)&=f_3(z),&& f_3(2\pm z)=-f_3(z), && f_3(2\tau\pm z)=-f_3(z).
\end{align*}
Thus
\[ \mathcal{W} (z)=\mathcal{U} (2\tau-z),\qquad w(z)=u (2\tau-z),
\]
and our analysis reduces to the study of $\mathcal{U}(z)$.

Observe that from our asymptotics (\ref{expandfs}, \ref{expandfsm}) of $f_j(z)$ that
$$\mathcal{U}(1-\xi)=O(\xi), \qquad \mathcal{U}(\xi-1)\sim -\frac{2}{\xi^2}-\frac{1+k^2}6 K^2,$$
which means that from (\ref{lame1})
\begin{align}\label{asux2p}
u(1-\xi)&= \textrm{constant} + O(\xi),\\ \label{asux2m}
u(\xi-1)&=\frac1\xi-\left[ \frac1{12}(1+k^2)K^2+\frac12\lambda_2\right]\xi+O(\xi^2)\nonumber \\ &=
\frac1\xi-\left[\frac{x_2^2}2+ \frac1{24}(2-k^2)K^2\right]\xi+O(\xi^2).
\end{align}
When substituted into (\ref{v1x2}) this yields (\ref{asymptildev1}). We also see that both  $w(1-\xi)$ and
$w(\xi-1)$ are regular.

Finally, let
$$Kz = 3K+2iK'+2 s.$$
Then
$$\mathcal{U}(z)=-\frac12 k^2 K^2 \mathrm{sn}^2(s)$$
and we arrive at the first Lam\'e equation
\[
\frac {d^{2}F}{d{s}^{2}
} (s)
- 2 k^2  \mathrm{sn}^2(s)\,F(s)
=\left(\frac{4 x_2^2}{K^2}-k^2\right)\,F(s):=-\lambda\,F(s), \qquad F(s)=u(z).
\]
If we set (this is the parameterization of \S \ref{axesandjacobi})
$$\lambda=1+k^2  \mathrm{cn}^2(t)=k^2-\frac{4 x_2^2}{K^2}=-\frac{4\lambda_2}{K^2},\quad\textrm{or}\quad 
 \mathrm{dn}^2(t)=-\frac{4 x_2^2}{K^2},
$$
and noting that if $F(s)$ is a solution so too is  $F(-s)$ (because $\mathrm{sn}^2(-s)=\mathrm{sn}^2(s)$),
then in section \S\ref{lamesect} we show that Hermite's eigenfunctions in this case are
\[ \frac{H(s+t)}{\Theta(s)}\exp\left\{-s Z(t)\right\}, \qquad \frac{H(-s+t)}{\Theta(s)}\exp\left\{s Z(t)\right\}.\]
We note that
$$x_2=0 \Longleftrightarrow  t= K+iK'.$$

\subsection{The \texorpdfstring{$x_1$}{x1} axis}
With $x_2=0=x_3$ the equations ${\tilde\Delta}\sp\dagger\mathcal{V}=0$ decouple into
\[
\begin{pmatrix}
{\dfrac {d}{dz}}v_{{1}} \left( z \right) +
 \left( 1/2\,f_{{3}} \left( z \right) +1/2\,f_{{1}} \left( z \right) -
1/2\,f_{{2}} \left( z \right)  \right) v_{{1}} \left( z \right) -x_{{
1}}v_{{2}} \left( z \right) \\ \noalign{\medskip}{\dfrac {d}{dz}}v_{{2}
} \left( z \right) -x_{{1}}v_{{1}} \left( z \right) + \left( 1/2\,f_
{{1}} \left( z \right) +1/2\,f_{{2}} \left( z \right) -1/2\,f_{{3}}
 \left( z \right)  \right) v_{{2}} \left( z \right) \end {pmatrix}=0
\]
and
\[
\begin{pmatrix} 
{\dfrac {d}{dz}}w_{{1}} \left( z \right) +
 \left( 1/2\,f_{{3}} \left( z \right) -1/2\,f_{{1}} \left( z \right) +
1/2\,f_{{2}} \left( z \right)  \right) w_{{1}} \left( z \right) -x_{{
1}}w_{{2}} \left( z \right) \\ \noalign{\medskip}{\dfrac {d}{dz}}w_{{2}
} \left( z \right) -x_{{1}}w_{{1}} \left( z \right) + \left( -1/2\,f_{
{1}} \left( z \right) -1/2\,f_{{2}} \left( z \right) -1/2\,f_{{3}}
 \left( z \right)  \right) w_{{2}} \left( z \right)\end {pmatrix}=0.
\]

Following the same steps as before we now have a solution of the form
\[\mathcal{V}=
c_1\mathcal{V}_1+
c_2\mathcal{V}_2
\]
where
\begin{align*}
\mathcal{V}_1&=
\begin{pmatrix}
 \dfrac{u(z)}{\sqrt{f_2(z)+f_3(z)}  }\\
0\\
\dfrac{1}{2 x_1}\dfrac{2 u'(z)+(f_3(z)-f_2(z))\, u(z)}{\sqrt{f_2(z)+f_3(z)}}\\
0\\
  \end{pmatrix},
\\
\mathcal{V}_2&=
\begin{pmatrix}  
 0\\
 {w(z)}\,{ \sqrt{f_2(z)+f_3(z)} }\\ 
0\\ 
\dfrac{1}{2 x_1}\left[ 2w'(z)+   (f_3(z)+f_2(z))\, w(z)   \right]{ \sqrt{f_2(z)+f_3(z)} }
 \end{pmatrix}.
\end{align*}
Here
\[
\frac {d^{2}u}{d{z}^{2}
} \left( z \right) 
+\mathcal{U} (z)\,u(z)=\lambda_1\,u(z), \qquad  \lambda_1=x_1^2-\frac14 {(1+ {k}^{2}){K}^{2} },
\]
where again
\begin{align*}
\mathcal{U} (z)&=
\frac12\left[ f_{{2}} \left( z \right) f_{{3}} \left( z \right) -f_{{1}} \left( z \right) f_
{{3}} \left( z \right) +f_{{1}} \left( z \right) f_{{2}} \left( z
 \right)- f
_{{1}}^{\,2} \left( z \right)   \right],
\end{align*}
and also
\[
\frac {d^{2}w}{d{z}^{2}
} \left( z \right) 
+\mathcal{W} (z)\,w(z)=\lambda_1\,w(z), \qquad  \lambda_1=x_1^2-\frac14 {(1+ {k}^{2}){K}^{2} },
\]
with
\begin{align*}
\mathcal{W} (z)&=
\frac12\left[ -f_{{2}} \left( z \right) f_{{3}} \left( z \right) +f_{{1}} \left( z \right) f_
{{3}} \left( z \right) +f_{{1}} \left( z \right) f_{{2}} \left( z
 \right)- f
_{{1}} ^{\,2} \left( z \right)  \right].
\end{align*}

Now we have the parameterisation (this is the parameterization of \S \ref{axesandjacobi})
$$\lambda=1+k^2  \mathrm{cn}^2(t)=\frac4{K^2}\left( \frac14(1+k^2)K^2-x_1^2 \right)
\quad\textrm{or}\quad 
 \mathrm{sn}^2(t)=\frac{4 x_1^2}{k^2 K^2}.
$$
We note that depending on whether $x_1<Kk/2$ or not, the value of $t$ changes from real to 
a general complex number. This effects the nature of the functions $u(z)$ and the argument of Brown
\emph{et al.} for a real solution breaks down.

\subsection{The \texorpdfstring{$x_3$}{x3} axis}
There are a few differences in this case. With $x_1=0=x_2$ we have the decoupled equations
\[
\begin{pmatrix}
{\dfrac {d}{dz}}v_{{1}} \left( z \right) +
 \left( 1/2\,f_{{3}} \left( z \right) +1/2\,f_{{1}} \left( z \right) -
1/2\,f_{{2}} \left( z \right)  \right) v_{{1}} \left( z \right) -x_{{
3}}v_{{2}} \left( z \right) \\ \noalign{\medskip}{\dfrac {d}{dz}}v_{{2}
} \left( z \right) -x_{{3}}v_{{1}} \left( z \right) + \left( 1/2\,f_
{{3}} \left( z \right) -1/2\,f_{{1}} \left( z \right) +1/2\,f_{{2}}
 \left( z \right)  \right) v_{{2}} \left( z \right) \end {pmatrix}=0
\]
and
\[
\begin{pmatrix} 
{\dfrac {d}{dz}}w_{{1}} \left( z \right) +
 \left( 1/2\,f_{{1}} \left( z \right) +1/2\,f_{{2}} \left( z \right) -
1/2\,f_{{3}} \left( z \right)  \right) w_{{1}} \left( z \right) +x_{{
3}}w_{{2}} \left( z \right) \\ \noalign{\medskip}{\dfrac {d}{dz}}w_{{2}
} \left( z \right) +x_{{3}}w_{{1}} \left( z \right) + \left( -1/2\,f_{
{1}} \left( z \right) -1/2\,f_{{2}} \left( z \right) -1/2\,f_{{3}}
 \left( z \right)  \right) w_{{2}} \left( z \right)\end {pmatrix}=0.
\]

Again  solutions have  the form
\[\mathcal{V}=
c_1\mathcal{V}_1+
c_2\mathcal{V}_2
\]
where now
\begin{align*}
\mathcal{V}_1&=
\begin{pmatrix} \dfrac{u(z)}{\sqrt{f_1(z)+f_2(z)}  }\\
\dfrac{1}{2 x_3}\dfrac{2 u'(z)+(f_1(z)-f_2(z))\, u(z)}{\sqrt{f_1(z)+f_2(z)}}  \\0\\0 \end{pmatrix},
\\
\mathcal{V}_2&=
\begin{pmatrix} 
\\0\\0\\
 {w(z)}\,{ \sqrt{f_1(z)+f_2(z)} }\\ 
 - \dfrac{1}{2 x_3}\left[ 2w'(z)+   (f_1(z)+f_2(z))\, w(z)   \right]{ \sqrt{f_1(z)+f_2(z)} } \end{pmatrix}.
\end{align*}
Here
\[
\frac {d^{2}u}{d{z}^{2}
} \left( z \right) 
+\mathcal{U} (z)\,u(z)=\lambda_3\,u(z), \qquad  \lambda_3=x_3^2-\frac14 {{K}^{2} },
\]
where again
\begin{align*}
\mathcal{U} (z)&=
\frac12\left[ f_{{2}} \left( z \right) f_{{3}} \left( z \right) -f_{{1}} \left( z \right) f_
{{3}} \left( z \right) +f_{{1}} \left( z \right) f_{{2}} \left( z
 \right)- f
_{{1}} ^{\,2} \left( z \right)  \right],
\end{align*}
but now
\[
\frac {d^{2}w}{d{z}^{2}
} \left( z \right) 
+{\tilde{\mathcal{W} }}(z)\,w(z)=\lambda_3\,w(z), \qquad  \lambda_3=x_3^2-\frac14 {{K}^{2} },
\]
with
\begin{align*}
{\tilde{\mathcal{W} }}(z)&=\mathcal{W} (2K+z)=
\frac12\left[ f_{{2}} \left( z \right) f_{{3}} \left( z \right) +f_{{1}} \left( z \right) f_
{{3}} \left( z \right) -f_{{1}} \left( z \right) f_{{2}} \left( z
 \right)- f
_{{1}} ^{\,2}  \left( z \right) \right].
\end{align*}
Thus
\[ {\tilde{\mathcal{W} }}(z)=\mathcal{U} (2K+2iK'-z),\qquad w(z)=u (2K+2iK'-z),
\]
and our analysis again reduces to the study of $\mathcal{U}(z)$.

\subsection{\texorpdfstring{$n=1$}{n1} Lam\'e Equation}\label{lamesect}
We have shown that our matrix differential equation may be  reduced to the same Lam\'e equation for each of the coordinate axes and here we recall the solutions to this. 
We have that
\[\Phi(u,a)=\frac{\sigma(u+a)}{\sigma(u)\sigma(a)}\exp\{-u\,\zeta(a)\} \]
solves
\[ \left[\dfrac{d^2}{du^2}-2\wp(u)\right] \Phi(u,a)= \wp(a)\,\Phi(u,a). \]

The Jacobi functions are
\[ H(u)=\theta_1(u/\theta_3^2(0)), \quad \Theta(u)=\theta_4(u/\theta_3^2(0)),\quad Z(u)=\frac{ \Theta'(u)}{ \Theta(u)}.
\]
Set
$\omega_1=\frac12\, \theta_3^2(0) =K$. Then
\begin{align*}
 \sigma(u)&= \frac{2\omega_1}{ \theta_1'(0)}\, \theta_1( u/(2\omega_1) )\,\exp(\eta_1 u^2/(2\omega_1) ),\\
 \zeta(u)&=\frac{\eta_1 u}{\omega_1}+\frac1{2\omega_1}\frac{ \theta_1'( u/(2\omega_1) )}{\theta_1( u/(2\omega_1) )}
\end{align*}
and
\[\Phi(u,a) = c\, \frac{\theta_1( [u+a]/(2\omega_1) )}{\theta_1( u/(2\omega_1) )}\,
\exp\left\{-\frac{u}{2\omega_1}\, \frac{ \theta_1'( a/(2\omega_1) }{\theta_1( a/(2\omega_1) } \right\}.
\]
Thus
\begin{align*}
\Phi(u-\omega_1 \tau, a+\omega_1 \tau)&=
c\, \frac{H(u+a)}{\theta_1( u/(2\omega_1) -\tau/2)}\,
\exp\left\{-\left(\frac{u}{2\omega_1}-\frac\tau2\right)\, \frac{ \theta_1'( a/(2\omega_1)+\tau/2 )}{\theta_1( a/(2\omega_1+\tau/2) } \right\}
\\
&=c'\, \frac{H(u+a)}{\theta_4( u/(2\omega_1) )\exp(i\pi  u/(2\omega_1) )}\,
\exp\left\{-\frac{u}{2\omega_1} 
\left[\frac{ \theta_4'( a/(2\omega_1) )}{\theta_4( a/(2\omega_1) )} -i \pi \right]\right\}
\\
&=c'\, \frac{H(u+a)}{\Theta(u)}\, \exp\left\{-u Z(a)\right\}
\end{align*}
and using $\tau=\omega_3/\omega_1$,
\[\wp(u+\omega_3)=-\frac13(1+k^2)+k^2\mathrm{sn}^2(u)\]
we obtain Hermite's solution.

We may write our solutions to (\ref{lame1}) as
\begin{equation}\label{lame1sol}
u(z)=
\dfrac{\theta_4\left( \frac{z+1}4 +\frac{t}{2K} \right)}{\theta_1\left( \frac{z+1}4  \right)}\,
\exp\left\{- \frac{z+1}4 \dfrac{ \theta_4'\left( \frac{t}{2K} \right) }   { \theta_4\left( \frac{t}{2K} \right) }              \right\}
\end{equation}
Now
\begin{align*}
\frac{\theta_4\left(v+\alpha\right)  \theta_1'(0) }{\theta_1\left(v\right)\theta_4\left(\alpha\right)}
\exp\left\{  -v \frac{ \theta_4'\left(\alpha\right)}{ \theta_4\left(\alpha\right)}
\right\}&=
\frac1{v}+\frac12\left[  \frac{ \theta_4''\left(\alpha\right)}{ \theta_4\left(\alpha\right)}  
-\left( \frac{ \theta_4'\left(\alpha\right)}{ \theta_4\left(\alpha\right)} \right)^2
-\frac13 \frac{ \theta_1'''(0)}{ \theta_1'\left(0\right)}
\right]v+O(v^2)\\
&=\frac1{v}-{2K^2}\,\wp(2K\alpha+\omega_3)\,v +O(v^2)
\\
&=\frac1{v}-16\left[ \frac1{12}(1+k^2)K^2+\frac12\lambda_2\right]\,v +O(v^2)
\end{align*}
which gives (\ref{asux2m}).

Using our observation that if $F(s)$ is a solution of Lam\'e's equation then so is $F(-s)$ we may construct a solution vanishing at $z=-1$ by taking
$$u(z)=
\dfrac{\theta_4\left( \frac{z+1}4 +\frac{t}{2K} \right)}{\theta_1\left( \frac{z+1}4  \right)}\,
\exp\left\{- \frac{z+1}4 \dfrac{ \theta_4'\left( \frac{t}{2K} \right) }   { \theta_4\left( \frac{t}{2K} \right) }              \right\}
-
\dfrac{\theta_4\left( - \frac{z+1}4 +\frac{t}{2K} \right)}{\theta_1\left( \frac{z+1}4  \right)}\,
\exp\left\{ \frac{z+1}4 \dfrac{ \theta_4'\left( \frac{t}{2K} \right) }   { \theta_4\left( \frac{t}{2K} \right) }              \right\}
.
$$

\section{Monopole Numerics and Visualisation\\ (by
David E. Braden, Peter Braden and H.W. Braden)}\label{appBBB}

\begin{center}{\small{
The numerical evaluation and visualisation of the charge $2$ monopole is described.}}

\end{center}

We describe here the numerical evaluation and visualisation of the charge $2$ monopole. The code and
numerical evaluation of the energy density are available via
\href{https://github.com/harrybraden/monopole}{GitHub}.
The numerical evaluation of the Higgs field and energy density implements the functions of the main text in python: the main procedures are described below. The key procedures are those that calculate $-\frac12 \Tr \Phi^2$ and the energy density $\mathcal{E}$ 
for a given $k$ and point in space $(x_1, x_2, x_3)$. These are then utilised to calculate the same quantities
on planes or axes as desired. We have given the energy density output in the directory \emph{python\_smoothed} for $k$ from $k=0.01$ to $k=0.99$ in steps of $0.01$.\footnote{Each of the subdirectories,
for example \emph{python\_smoothed/$k=0.01$}, contains $60$ files, each with the 
results of an $xy$-plane with a specified $z$-value from $z=0.025$ to $z=2.975$ in steps of $0.05$. 
The  $xy$-plane themselves are $60\times60$ arrays of double precision output for $x,y$ from
$0.025$ to $2.975$ in steps of $0.05$. The procedures allow arbitrary grids to be specified.
}

We have provided three tools to visualize the five dimensional datasets ($\mathcal{E},k,x_1,x_2,x_3$): two are interactive, and the third graphical. These may be also used for more complicated monopole configurations. We consider the data as three dimensional volumetric data.
Each of the interactive viewers allow data to be dragged around and resized. The
first \href{https://www.maths.ed.ac.uk/~hwb/browse.html}{https://www.maths.ed.ac.uk/~hwb/browse.html}
(see the first of Figure  \ref{bothvisualisers})
uses the energy density as opacity, and hides all volumes below a specified value in order to look inside the
volume. One may vary $k$ and the threshold energy density value. The opacity here is on a $0-255$
scale and the double precision energy density is converted to byte format. (The code also includes
options for \emph{ab initio} creating  energy density in byte form.)

The second visualizer defines a threshold above
which to consider as solid, and uses the Marching Cubes algorithm \cite{lh87}
to construct
a mesh of that threshold's contour.  These meshes can be visualised with many 
mesh viewers, or even 3D printed (see  the second of Figure  \ref{bothvisualisers}). The procedure \emph{
generatemesh.py} will take the value $k=0.6$ and threshold $0.55$ to produce a standard .obj file via \lq
{\tt python generatemesh.py 0.6 0.55 > test.obj}\rq. The resolution of the cubes is that coming from the numerical evaluation (in the Figure these are cubes of size $0.05^3$).

The third method of visualizing the data is a \lq Tomogram\rq\ that takes slices through the volume. We can plot the contours on
these images, or use colour to represent the density at that slice; Figure \ref{bothtoms} shows Tomograms with uniform and nonuniform colourings.
The second last column of these figures correspond to the $k$ value of Figure  \ref{bothvisualisers}.

\begin{center}
{\bf{Numerical and Visualisation Scripts}}
\end{center}
The key scripts will now be described, breaking these into the numerical determination of the relevant
quantities and then their visualisation.

\noindent{\bf{Numerical Scripts.}} The requirements for running these are
python 2.7
and the following pip packages: numpy, scipy, mpmath.

\begin{description}

\item[higgs\_squared.py] For input $(k, x_1, x_2, x_3)$
calculates $-\frac12 \Tr \Phi^2$ at the point $(x_1, x_2, x_3)$ of space and a parameter $k$ (between $0$  and $1$). This file also calculates $-\frac12 \Tr \Phi^2$ for various planes. It may be run by

{{ \tt python -c "from higgs\_squared import higgs\_squared; print(higgs\_squared(0.8, 0.1, 0.2, 0.3))"}}

\item[energy\_density.py]	 For input $(k, x_1, x_2, x_3)$
calculates the trace of the Higgs field squared for a point $(x_1, x_2, x_3)$ of space and a parameter $k$ (between $0$  and $1$). This file also calculates $-\frac12 \Tr \Phi^2$ for various planes. The file tests if
the spatial point corresponds to a multiple root or branch point.  It may be run by

{
{ \tt python -c "from energy\_density import energy\_density; print(energy\_density(0.8, 0.1, 0.2, 0.3))"}}

\item[Scripts for basic Functions] These functions are in the previous scripts. Given $k$, determining the curve, and a point in space $(k, x_1, x_2, x_3)$ the elementary operations are

\begin{itemize}
\item  \emph{quartic\_roots}$(k, x_1, x_2, x_3)$ gives (unordered) solutions $\zeta_i$  to the Atiyah-Ward constraint.

\item  \emph{order\_roots(roots)} orders the roots using the real structure (\ref{ordering}).

\item  \emph{calc\_zeta}$(k, x_1, x_2, x_3)=$ \emph{order\_roots(quartic\_roots}$(k, x_1, x_2, x_3))$

\item  \emph{calc\_eta}$(k, x_1, x_2, x_3)$ gives the corresponding $\mu_i$'s.

\item \emph{calc\_abel(k, zeta, eta)}  calculates the Abel image of a point ${P}=(\zeta, \eta)$. To get the correct choice of contour we compare the $\eta$-value given by the theta function (\ref{eqparam}) given by \emph{calc\_eta\_by\_theta}$(k, z)$ and use \emph{abel\_select}: if they agrees the Abel image is accepted and if not it is shifted by the half period to the correct sheet.

\item  \emph{calc\_mu}$(k, x1, x2, x3, \zeta, abel)$ calculates the one transcendental function $\mu$.

\item \emph{is\_awc\_multiple\_root}$(k, x_1, x_2, x_3)$  tests if there are multiple roots.

\item \emph{is\_awc\_branch\_point}$(k, x_1, x_2, x_3)$   tests if we get a branch point as a roots; these are numerically unstable.

\end{itemize}

\item[python\_expressions, modified\_expressions] These directories contains the code for such  expressions as the Gram matrix,
the Higgs, and the various first and second derivatives of $\zeta_i$'s, $\mu_i$'s. When the the expressions are very long, the appropriate matrix element of the matrices is given.

\item[python\_smoothed] As described above, this directory gives the double precision output for 
the energy density  for $k$ from $k=0.01$ to $k=0.99$ in steps of $0.01$.

\end{description}

\noindent{\bf{Visualisation Scripts.}}
The dependencies for the visualisation tools may be installed via the Makefile. The README.md contains instructions for using the tools.%
We have
\begin{description}

\item[contours-image.py] Generates the Tomogram image.

\item[generatemesh.py] Generates a 3d {\tt .obj} file from the data.
\item[generate-image-data.py] Generates the image data for the interactive web visualiser.

\item[visualise] This directory contains the interactive visualiser.

\end{description}

\noindent{\bf{File Handling and Smoothing.}} A number of files deal with file handling.

\begin{description}

\item[array\_tools.py, array\_tools\_float.py] General file handling and reflection of first quadrant data.

\item[data.py] Load and manipulate the data.

\item[files.py] read and write floating point files.

\item[file\_converter.py] converts floating point to bytes.

\item[file\_smoother.py] smooths the data on the exceptional loci arising from bitangency.

\item[simplify\_script.py] modifies a number of python expressions in order to evaluate them faster; the resulting files are in the modified\_expressions directory.

\item[smoothing\_tools.py] for smoothing arrays.

\end{description}


\bibliographystyle{plain}

\def\cprime{$'$} \def\cdprime{$''$} \def\cprime{$'$}

\end{document}